\pdfoutput=1
\documentclass{article}

\usepackage{arxiv}
\usepackage[utf8]{inputenc} 
\usepackage[T1]{fontenc}    
\usepackage{url}            
\usepackage{booktabs}       
\usepackage{amsfonts}       
\usepackage{nicefrac}       
\usepackage{microtype}      
\usepackage{lipsum}
\usepackage{natbib}

\usepackage[english]{babel}
\usepackage{amsmath}
\usepackage{graphicx}
\usepackage{multirow}
\usepackage{amsthm}
\usepackage{amssymb}
\usepackage{bbm}
\usepackage[colorinlistoftodos]{todonotes}
\usepackage{hyperref}
\hypersetup{colorlinks,allcolors=black}
\usepackage{float}
\usepackage{subcaption}
\usepackage{multirow}
\usepackage{makecell}
\usepackage{mathtools}
\usepackage{algorithmic}
\usepackage{bm}
\usepackage{physics}
\usepackage{appendix}
\usepackage{tikz}
\usepackage{diagbox}
\usetikzlibrary{bayesnet}
\allowdisplaybreaks

\newtheorem{lemma}{Lemma}

\numberwithin{defn}{section}
\numberwithin{property}{section}
\numberwithin{lemma}{section}
\numberwithin{cor}{section}
\numberwithin{equation}{section}
\newcommand{\zEps}{\ensuremath{z_\epsilon}}
\newcommand{\indFct}[1]{\ensuremath{\mathbbm{I}\left(#1\right)}}
\newcommand{\E}[1]{\ensuremath{\mathbbm{E}\left(#1\right)}}

\newcommand{\Var}[1]{\ensuremath{\textrm{Var}\left(#1\right)}}
\newcommand{\invlogit}[1]{\ensuremath{\psi^{-1}\left(#1\right)}}

\newcommand{\vM}[3]{\ensuremath{f_{\mathcal{V}}\left(#1 \mid #2, #3\right)}}
\newcommand{\PN}[3]{\ensuremath{f_{N}\left(#1 \mid #2, #3\right)}}
\renewcommand{\v}[1]{\ensuremath{\bm{#1}}}

\long\def\comment#1{}

\begin{document}

\title{Modeling Random Directions in 2D Simplex Data}

\author{
  Rayleigh Lei\\
  Department of Statistics\\
  University of Michigan\\
  Ann Arbor, MI 48109 \\
  \texttt{rayleigh@umich.edu} \\
  \And
 XuanLong Nguyen \\
  Department of Statistics\\
  University of Michigan\\
  Ann Arbor, MI 48109 \\
  \texttt{xuanlong@umich.edu} \\
}

\maketitle

\begin{abstract}
We propose models and algorithms for learning about random directions in two-dimensional simplex data, and apply our methods to the study of income level proportions and their changes over time in a geostatistical area. There are several notable challenges in the analysis of simplex-valued data: the measurements must respect the simplex constraint and the changes exhibit spatiotemporal smoothness while allowing for possible heterogeneous behaviors. To that end, we propose Bayesian models that rely on and expand upon building blocks in circular and spatial statistics by exploiting a suitable transformation based on the polar coordinates for circular data. Our models also account for spatial correlation across locations in the simplex and the heterogeneous patterns via mixture modeling. We describe some properties of the models and model fitting via MCMC techniques. Our models and methods are illustrated via a thorough simulation study, and applied to an analysis of movements and trends of income categories using the Home Mortgage Disclosure Act data. 
\end{abstract}

\section{Introduction}
In many real-world problems, some or all the primary quantities of interest are non-negative proportions that sum up to one and thus lie in a probability simplex. For instance, to study a collection of rocks and sediments, geologists might analyze the chemical, mineral, and/or fossil percentages of each sample to determine the effects of various natural processes or which samples are related \citep{AitchisonStatisticalAnalysisCompositional1982}. As another example, microbiolologists might examine microbiome data from a set of volunteers because of the microbiome's effect on human health \citep{AllabandEtAlMicrobiome101Studying2019}. Here, the microbiome refers to the proportions of various microbes in the human gut \citep{AllabandEtAlMicrobiome101Studying2019}. 

Modeling the variation and change of measurements that lie on the simplex can be challenging. To be concrete, we let $\Delta^D$ denote the $D$-dimensional probability simplex, i.e., the subset of elements in $\mathbbm{R}^{D + 1}$ whose components are non-negative and sum to one. Given a data set represented by a collection of random samples $x_1,\ldots,x_N \in \Delta^D$, each data point $x_{i}$ is composed of the $D$ components $x_{i,j}$ for $j=1,\ldots, D$. 
The most immediate difficulty is the simplex constraint imposed by the proportions because they must sum up to one \citep{vpg2007}. For this constraint to be maintained, a change in one proportion necessarily involves a change in the other proportions. If this is not handled correctly, the model might detect what Pearson called "spurious correlation" \citep{PearsonMathematicalContributionsTheory1896}. Another difficulty may arise if any of the proportions are zero. Customary composition data approaches rely on a standard transformation technique, e.g., taking the ratio of $x_{i, j}$ and the geometric mean of $x_i$ or the ratio of $x_{i, j}$ and $x_{i, J}$ for $j \in 1, 2, \ldots D$ and for some reference $J \in 1, 2, \ldots D$ and then applying the log transform to these ratios \citep{vpg2007}. As a result, these "log-ratio" transformations convert the data from $\Delta^D$ to $\mathbbm{R}^{D - 1}$. Note that if any of the proportions are zero, then some or all of these ratios are either $0$ or $\infty$ and the transformation becomes undefined. 

%

To deal with this difficulty, we might model how the data shift as random movements of points within the simplex. Despite the potential challenge in modeling movements that respect the simplicial constraint, such an approach has several fundamental advantages. It would avoid picking up "spurious correlation" because all components of an observation would be dealt with simultaneously. Moreover, while the range of movements is restricted, boundary points could be treated in a similar framework as interior points.

To help develop such a model for the random movement, this paper will examine an important aspect: modeling the random directions of such movements. This is a key step for two main reasons. First, at a high level, any movement in the simplex, including points on the boundary, can be decomposed into the direction and the magnitude. For a movement to respect the simplex constraint, the magnitude of a point's movement is bounded by the maximum distance between the point and the boundary point in a given direction. Thus, not only is there greater flexibility in modeling the two aspects separately, but also doing so allows us to model the magnitude with greater care. The second reason is that even this modeling task presents potential non-trivial difficulties. An illustrative example is a data set of income proportions for all Census tracts or "neighborhoods" in Los Angeles County from 1990 to 2010. The proportions have been binned into three categories: \$0 to \$100,000, \$100,000 to \$200,000, and greater than \$200,000. Using the Census tract IDs, we can track how the income proportions for a neighborhood changes from one year to the next. As we will demonstrate later, this shift appears dependent on the current income proportions and not on the neighborhood's physical location. Then, Figure \ref{real_data_example} shows that we cannot assume a uniform random direction in this change because there may be clearly preferred directions. We need a way to assign probability to these random directions. In order to start thinking about how to accomplish this, we first examine the 2D simplex case, which shall be the main focus of our paper. There, the simplex is essentially a two dimensional surface for our points to move in. It suffices to assume that these random directions are random angles and lie in the interval $[0, 2\pi)$. One naive approach to assign probability to such an interval is by assigning probability to some random variable $z \in \mathbbm{R}$ and using the inverse logit function to transform $z$ to that interval. Such an approach would be problematic because the endpoints of the interval and values near the end points are mapped near their respective $\pm\infty$ and are far apart. Meanwhile, the end points for the random angle's interval, $0$ and $2\pi$, denote the same direction and should not be so far apart. The issue here, mathematically, is that the inverse logit function that maps $\mathbb{R}$ to the angles $[0,2\pi)$, while continuous and one-to-one, does not possess a continuous inverse.

\noindent\textbf{Existing work}
Several useful techniques for modeling (random) angles and directions come from directional statistics \citep{MardiaJuppDirectionalStatistics2010}. In particular, we can think of the random angle as circular data because these angles can be mapped to a point on the unit circle and circular data are observations that lie on a unit circle. There are several basic distributions on circular data that we can utilize for our purposes: \textit{the von Mises distribution} and \textit{the projected normal distribution} in two dimensions. The former can be loosely viewed as the circular version of the normal distribution \citep{LeeCircularData2010}. Meanwhile, the latter is a bivariate Gaussian distribution that is transformed to a distribution on the circle using the usual polar coordinate transform and integrating out the radius \citep{MardiaJuppDirectionalStatistics2010}. There is another distribution that we might use. As seen in Figure \ref{real_data_example}, it is possible that nearby locations move in mostly similar directions. In other words, the random directions, depending on where the start points are, appear to be spatially correlated. Because the change in the pattern of random direction is also smooth as we pass across the simplex in this figure, an appropriate model might be the Gaussian process \citep{RasmussenWilliamsGaussianProcessesMachine2006}. Even though we cannot use the inverse logit function to transform a Gaussian process because of what we discussed earlier about angles, we can use the circular version of the Gaussian process, i.e. the projected Gaussian process \citep{WangGelfandDirectionalDataAnalysis2013}. 


\begin{figure}[!tbp]
\centering
\begin{subfigure}{.24\textwidth}
\centering
\includegraphics[width = 0.9\textwidth]{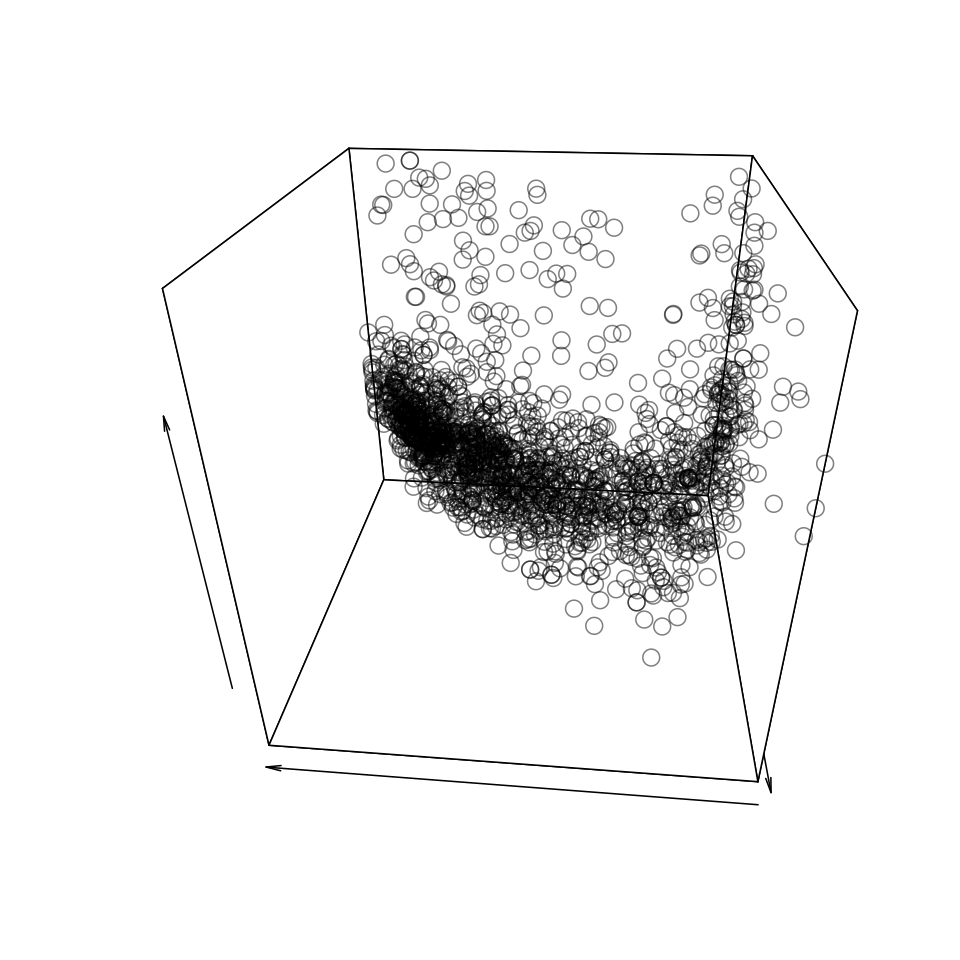}
\caption{2003-2004}
\label{fig:real_data_example_one_comp}
\end{subfigure}
\begin{subfigure}{.24\textwidth}
\centering
\includegraphics[width = 0.9\textwidth]{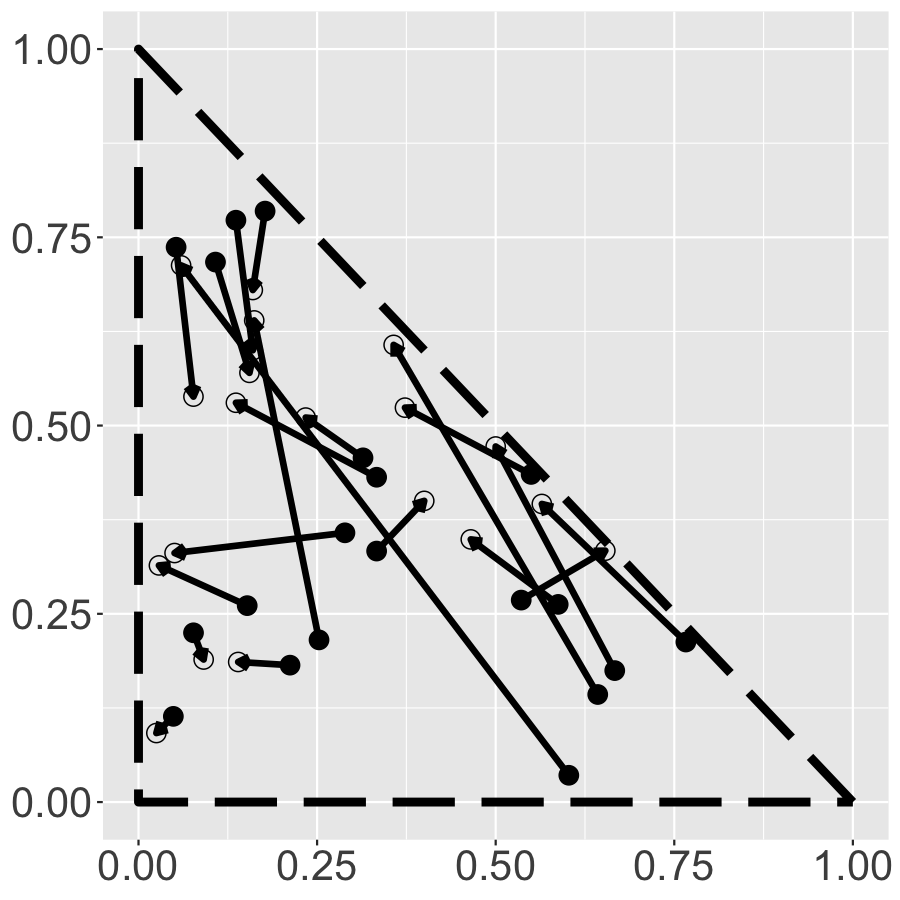}
\caption{2003-2004}
\label{fig:real_data_example_one_comp_start_end}
\end{subfigure}
\begin{subfigure}{.24\textwidth}
\centering
\includegraphics[width = 0.9\textwidth]{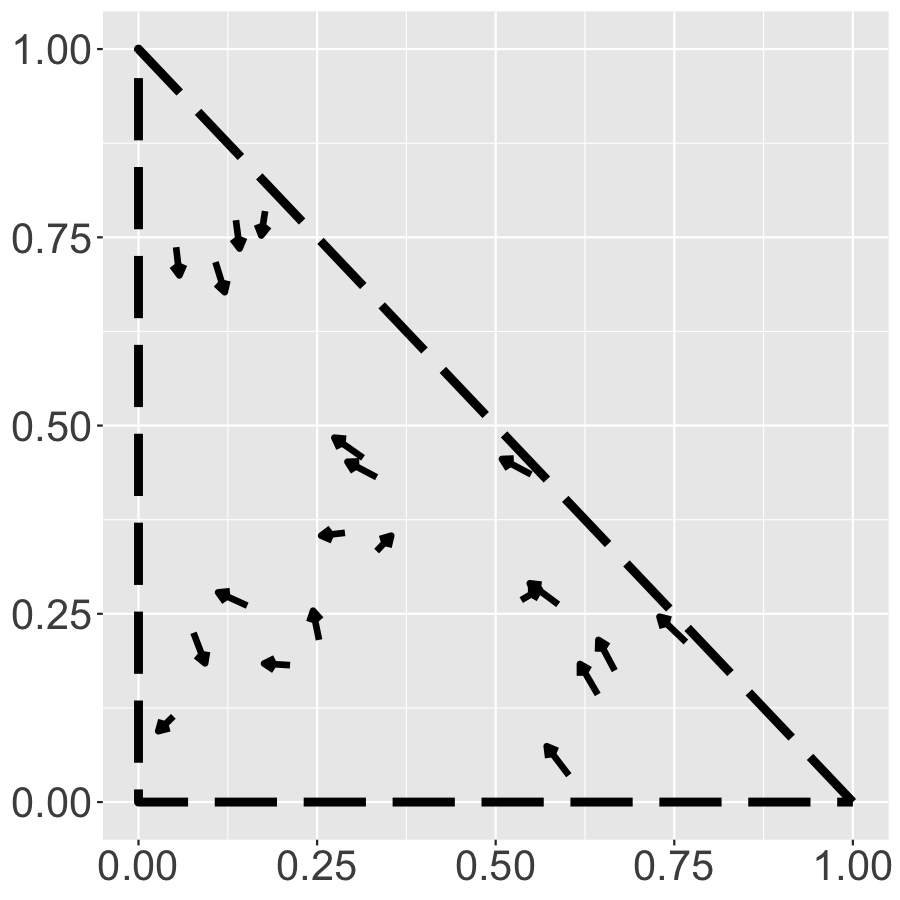}
\caption{2003-2004}
\label{fig:real_data_example_one_comp_subset}
\end{subfigure}
\begin{subfigure}{.24\textwidth}
\centering
\includegraphics[width = 0.9\textwidth]{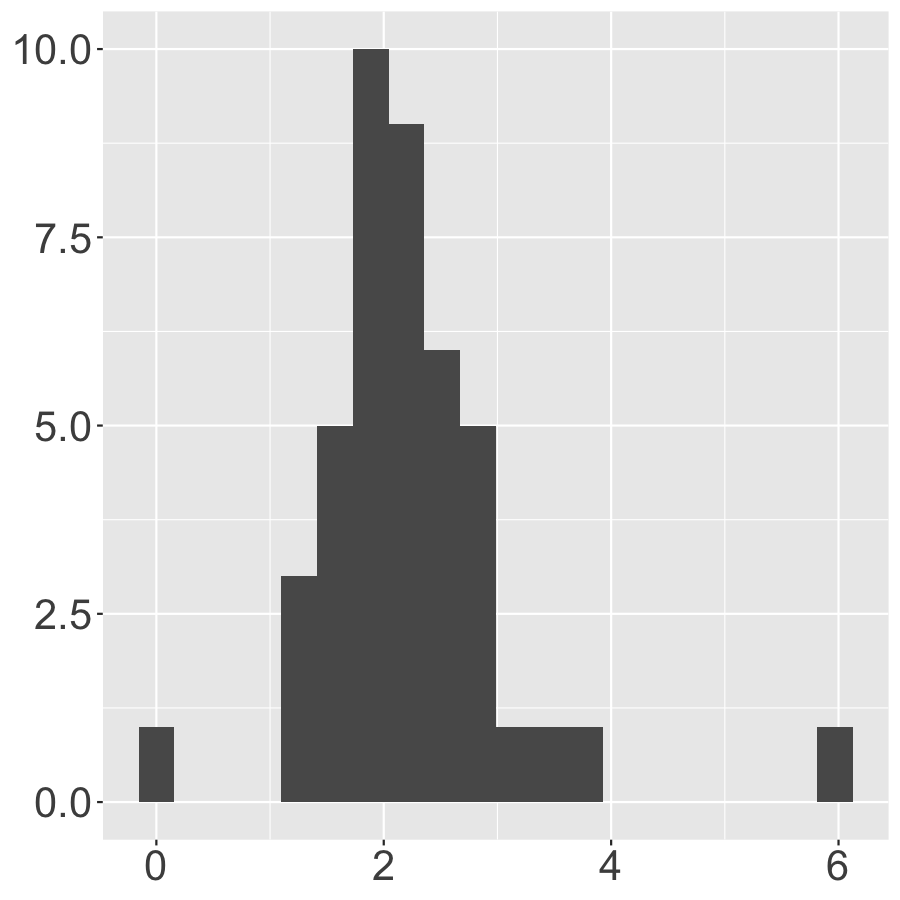}
\caption{2003-2004}
\label{fig:real_data_example_one_comp_hist}
\end{subfigure}\\
\begin{subfigure}{.24\textwidth}
\centering
\includegraphics[width = 0.9\textwidth]{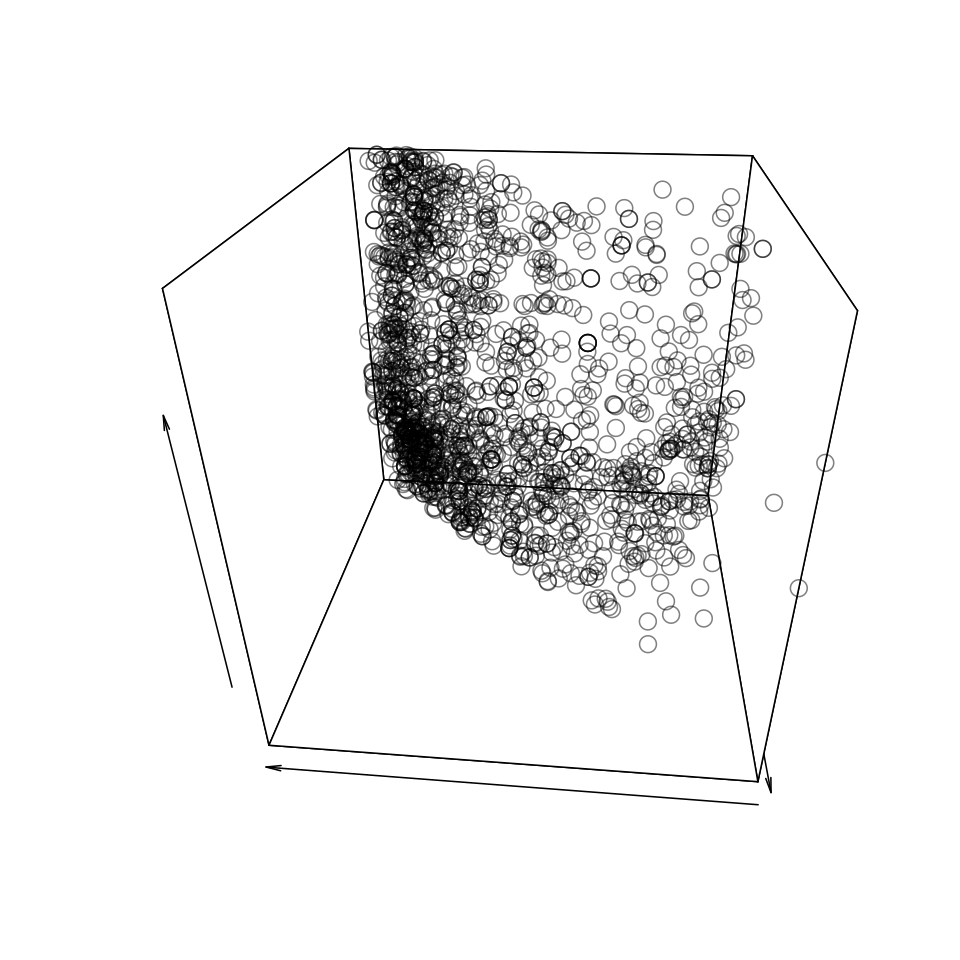}
\caption{1998-1999}
\label{fig:real_data_example_mixture}
\end{subfigure}
\begin{subfigure}{.24\textwidth}
\centering
\includegraphics[width = 0.9\textwidth]{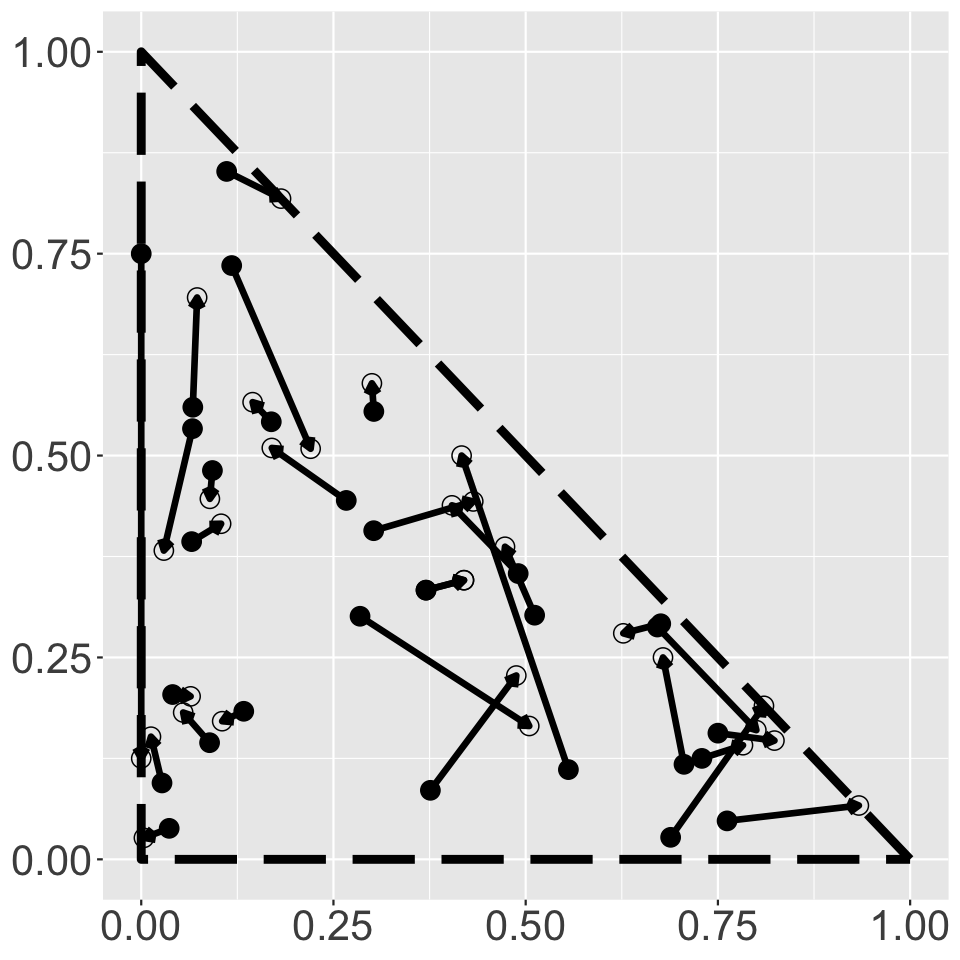}
\caption{1998-1999}
\label{fig:real_data_example_mixture_start_end}
\end{subfigure}
\begin{subfigure}{.24\textwidth}
\centering
\includegraphics[width = 0.9\textwidth]{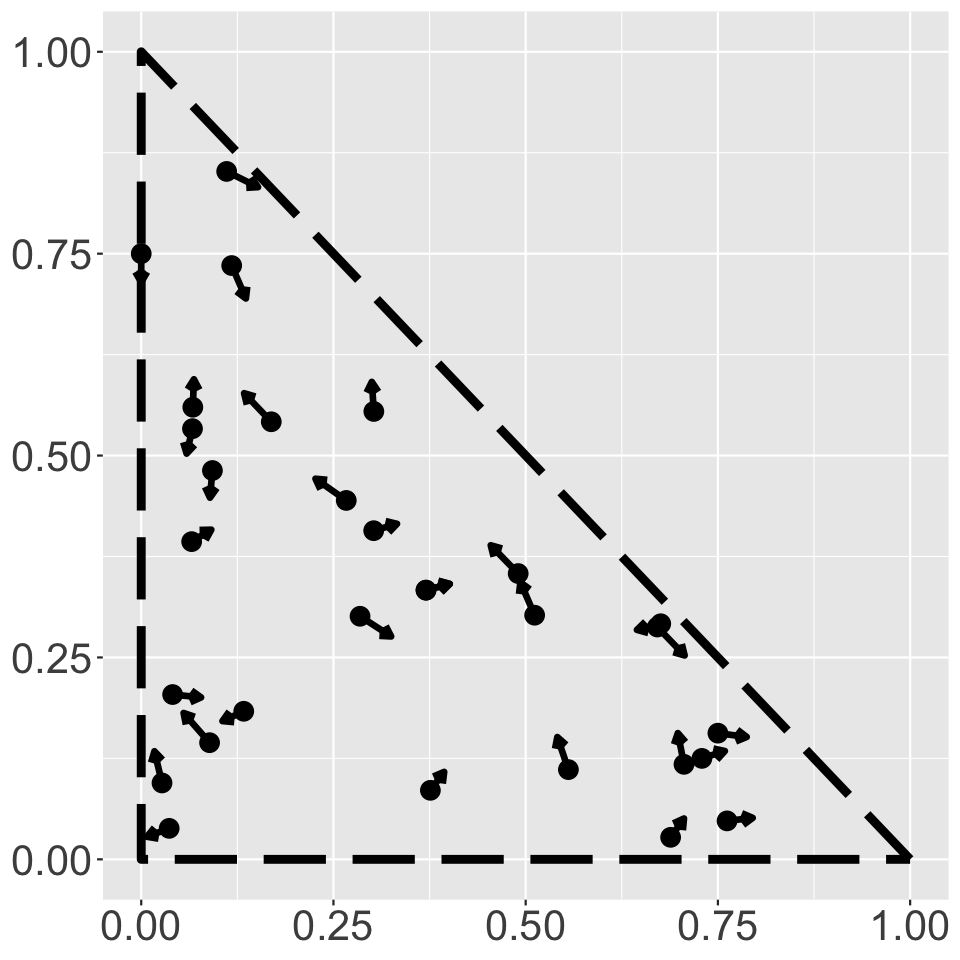}
\caption{1998-1999}
\label{fig:real_data_example_mixture_subset}
\end{subfigure}
\begin{subfigure}{.24\textwidth}
\centering
\includegraphics[width = 0.9\textwidth]{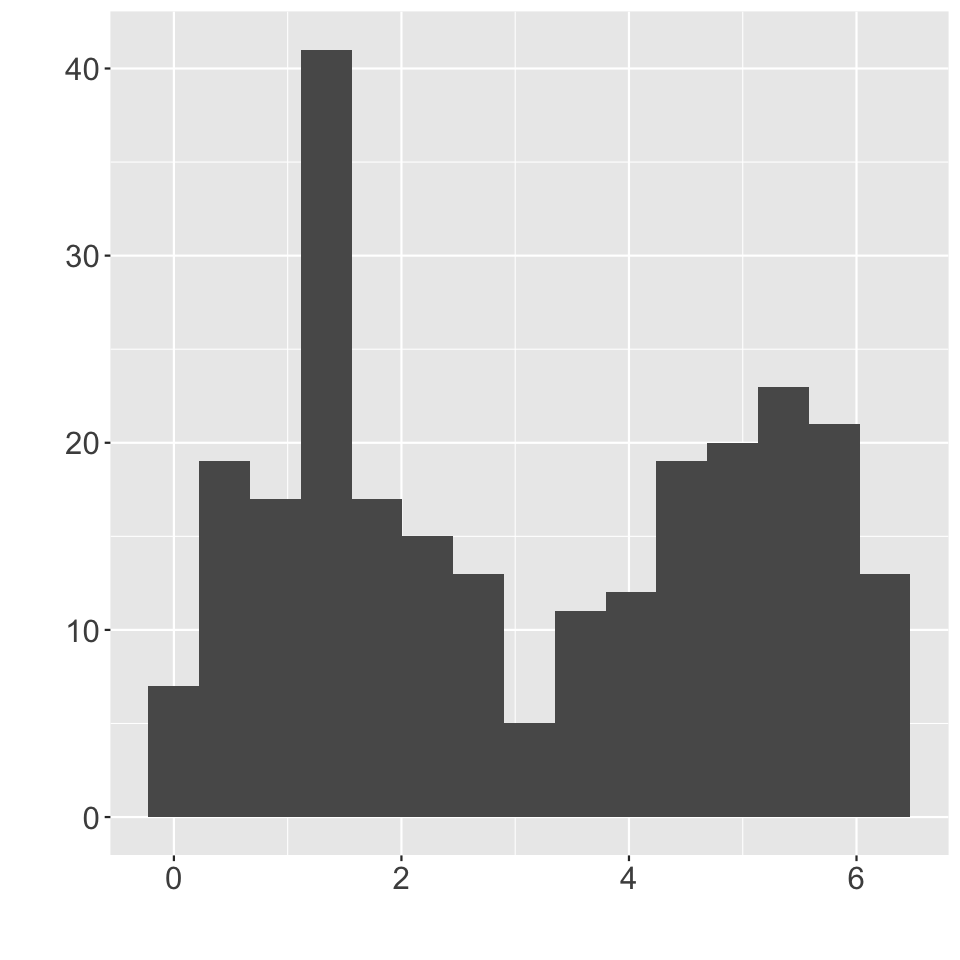}
\caption{1998-1999}
\label{fig:real_data_example_mixture_hist}
\end{subfigure}
\caption{Random direction plots for different years. The $x-y$ coordinates represent the income proportion in the first two categories. If the plot is three dimensional, the z coordinate is the random direction divided by $2\pi$ to lie in $[0, 1)$. The two leftmost plots show all random directions for the years listed below. The middle plots show the movement for a subset of 20-30 "locations" and the corresponding extracted directions. The two rightmost plots display a histogram of the random directions associated with locations that are within $0.05$ $L_2$-distance of $(0.39, 0.56, 0.05)$ for 2003-2004 and of $(0.85, 0.10, 0.05)$ for 2008-2009.}
\label{real_data_example}
\end{figure}

\noindent \textbf{Our approach} To model the random direction of movements for simplex-valued data, we will leverage and expand upon the building blocks advocated by~\cite{MardiaJuppDirectionalStatistics2010,RasmussenWilliamsGaussianProcessesMachine2006} and the techniques of~\cite{WangGelfandDirectionalDataAnalysis2013,WangGelfandModelingSpaceSpaceTime2014}. As the first step of our approach, we assume that the observed random directions are distributed according to a von Mises distribution. We then correlate the von Mises distributions' means with a projected Gaussian process. In other words, this is a circular version of the Gaussian process with Gaussian white noise. This modeling choice seems appropriate because there appears to be an unimodal empirical distribution for the random directions associated with the nearby locations for any location in the simplex. For example, Figure \ref{fig:real_data_example_one_comp_hist} shows such a distribution for the location in the simplex, $(0.39, 0.56, 0.05)$. Such a choice also makes sense because the mean of the von Mises distribution can be thought of as a vector on the unit circle, which is what a projected Gaussian process outputs. In addition, it allows both the prior and likelihood to recognize the geometry of angles. Not only can this basic model harness the power of Gaussian processes to spatially correlate random directions of nearby locations and handle noisy directions, but also it does so in an interpretable way. After fitting our model, we can make a posterior prediction of the mean preferred direction for any location. 

Extensions to the basic model described above are required to account for "heterogenous" patterns for random directions, which are evidently illustrated in Figure \ref{fig:real_data_example_mixture}. Because the basic approach proves to be inadequate, the modeling extensions form a substantial portion of this paper. To capture such heterogeneity, we shall extend the basic model via several mixture modeling techniques. In particular, the observation is distributed according to one of $K\geq 1$ von Mises distributions with some probability. Each of the $K$ distributions has a projected Gaussian process to correlate its mean. This simple extension still retains the advantages mentioned previously because each component now represents a mean preferred direction. We also discuss another modeling alternative. While there still are $K$ von Mises distributions, we extend the inverse logit to transform $K - 1$ Gaussian processes to the mixing probability for each component. 


We have applied and evaluated these models to a simulated data set and the aforementioned data set of income proportions. For model selection, we compare the fitted models' posterior predictive probabilities calculated on data withheld from each year. The chosen models enable us to discover several patterns of interest on the year-to-year income proportion movements. There appear to be four phases based on the patterns we discover, and interestingly enough, these phases correspond to the economic cycles observed during that time period. The first two phases charts the end of one cycle, the third phase follows the dot com bubble and the recession after the bubble, and the last phase represents the subsequent housing bubble and the housing market crash. Because we can also assign meaning to the random directions with respect to the income categories, interpreting the results from the models gives us further information about the year to year change during these phases. 

These modeling techniques can also be extended to random directions in higher dimensions. In addition, it might be applied to data that lie on manifolds other than the simplex because any change in the data can be represented as a random movement, which can be decomposed accordingly.  
However, in this paper we shall restrict our attention to two dimensional simplex data in order to focus on the heterogeneous modeling and inference in the circular representation, which is already quite intricate. The rest of the paper is organized as follows. In Section \ref{background}, we describe how the polar coordinate transformation can be used to transform distributions in $\mathbbm{R}^2$ to distributions on the circle, such as the von Mises distribution or the projected normal distribution in two dimensions. We also discuss some properties of these transformed distributions and extend the projected bivariate normal distribution to the projected Gaussian process. While it is not a focus of this paper, we then introduce the higher dimension version of these distributions. Next, we more formally introduce our models in the Section \ref{Model} and describe a few properties of these models using ideas from directional statistics. We then briefly examine how to fit these models in Section \ref{sec:model_fitting}. We fit them to simulated and real data in Section \ref{results} and discuss these results. Section~\ref{sec:conclusion} discusses possible extensions.

\section{Polar coordinate transformations for circular data}
\label{background}
In this section, we examine how the polar coordinate transformation can be used to transform distributions in Euclidean domains, such as $\mathbbm{R}^2$, to distributions on non-Euclidean domains, such as the unit circle. 
Because our paper is focused on random directions in $\Delta^2$, we begin by discussing tools from circular statistics and then introduce the higher dimension analogues. In circular statistics, the data are observations that lie on an unit circle \citep{MardiaJuppDirectionalStatistics2010}, so they can be represented by their angles, which will be denoted by random variable $Y\in [0,2\pi)$. 


Given any $(z_1,z_2) \in \mathbbm{R}^2 \setminus \{0\}$, there is a unique polar representation in terms of the corresponding radius and angles, namely $(r,y) \in \mathbbm{R}^+ \times [0,2\pi)$, where $z_1 = r\cos y$ and $z_2 = r\sin y$. In fact, this representation establishes a bijection between $(r, y) \in \mathbbm{R}^+ \times [0,2\pi)$ and $(z_1, z_2) \in \mathbbm{R}^2\setminus \{0\}$. This suggests a natural recipe for assigning probability for (random) angles $y \in [0,2\pi)$: this can be done by endowing a distribution for $(z_1,z_2)$, which induces a distribution for the pair $(r,y)$. We obtain a valid distribution on angles $[0,2\pi)$ by either \emph{marginalizing} out $r$ or \emph{conditioning} on a particular value of $r$.

The aforementioned bijection between $(z_1,z_2)$ and $(r,y)$ induces a function that maps $(z_1,z_2)\in \mathbbm{R}^2\setminus \{0\}$ to $y\in [0,2\pi)$. This function is denoted by $y = \arctan^*(z_1, z_2)$, where the notation $\arctan^*$ is used to indicate a modified version of the standard arctan function. Recall that the standard $\textrm{arctan}$ function maps the real line $\mathbb{R}$ to the open interval $(-\pi/2,\pi/2)$. We may extend this continuous function to the closed domain $\overline{\mathbb{R}} := [-\infty,+\infty]$ and the closed range of angles $[0,2\pi)$. The function $\arctan^*$ satisfies the following (which is commonly used as a definition): for $(z_1,z_2) \in \mathbb{R}^2\setminus \{0\}$,
\begin{align}
    \textrm{arctan}^*(z_1, z_2) &=
    \begin{cases}
        \textrm{arctan}(\frac{z_2}{z_1}) & z_1 \geq 0, z_2 \geq 0\\
        \textrm{arctan}(\frac{z_2}{z_1}) + 2\pi & z_1 \geq 0, z_2 < 0 \\
        \textrm{arctan}(\frac{z_2}{z_1}) + \pi & z_1 < 0.
    \end{cases}
\label{fct:arctan_star}
\end{align}


Using these representations, we can derive the von Mises and projected normal distributions, which will be useful building blocks in our subsequent modeling.

\subsection{von Mises Distributions}
\label{ssection:vM}
A natural choice for constructing a probability distribution of $(z_1,z_2) \in \mathbbm{R}^{2}$ is the bivariate Gaussian distribution with mean, $(\cos(\alpha), \sin(\alpha))^T$, and variance, $\rho^{-1}\mathbbm{I}_{2 \cross 2}$, for some $\alpha \in [0, 2\pi)$ and $\rho \in \mathbbm{R}^+$. 
The precise meaning of parameters $\alpha$ and $\rho$ will be clear shortly. Applying change of variables, we obtain that for $r>0$,
\begin{align*}
    p(r, y \mid (\cos(\alpha), \sin(\alpha))^T, \rho^{-1}\mathbbm{I}_{2 \times 2}) &= r\frac{\rho}{2\pi}\exp\left(-\frac{\rho}{2}\big((r\cos(y) - \cos(\alpha))^2 + \right.\\
    & \qquad\qquad\qquad\qquad (r\sin(y) - \sin(\alpha))^2\big)\bigg)\\ 
    &= \frac{r\rho}{2\pi}\exp\left(-\frac{\rho}{2}(r^2 - 2r\cos(y - \alpha) + 1)\right)
\end{align*}
\[
p(r, y \mid (\cos(\alpha), \sin(\alpha))^T, \rho^{-1}\mathbbm{I}_{2 \cross 2}) = \frac{r\rho}{2\pi}\exp\left(-\frac{\rho}{2}(r^2 - 2r\cos(y - \alpha) + 1)\right).
\]
By conditioning on $r=1$, we arrive at the von Mises distribution on the interval $[0,2\pi)$ \citep{MardiaJuppDirectionalStatistics2010}:
\[p(y|r=1, \alpha, \rho) \propto 
\exp(\rho\cos(y-\alpha)).
\]


In its original form, the von Mises distribution is defined by two parameters: a mean angle, $\alpha \in [0, 2\pi)$, and a concentration parameter, $\rho \in \mathbbm{R}^+$. Given these parameters, the von Mises distribution has a density function of the following form:
\begin{equation}
    \vM{y}{\alpha}{\rho} := \frac{\textrm{e}^{\rho \cos(y - \alpha)}}{2\pi I_0(\rho)}.
\end{equation}
Appearing in the normalizing constant, $I_0(\rho)$ is the modified Bessel function of the first kind and of order 0. In general, the modified Bessel function of the first kind and of order $n$ on the interval of $[0, 2\pi)$ is defined to be
\begin{equation}
    I_n(\rho) := \frac{1}{2\pi} \int_0^{2\pi} \cos(n y) \textrm{e}^{\rho \cos(y)} d y.
    \label{eq:bessel}    
\end{equation}

Since the domain of the angles, the unit circle, is non-Euclidean, one has to be careful when speaking of notions such as mean and variance. Here are several properties we wish to highlight. First, the distribution is unimodal and symmetric around its mean, $\alpha$, if $\rho > 0$. This makes sense because as discussed earlier, the density function is proportional and restricted to a bivariate Gaussian centered at a point on the unit circle. If $\rho = 0$, then the von Mises distribution becomes the uniform distribution on an interval of length $2\pi$. Next, the mean angle, $\alpha$, is different from the typical definition because for the von Mises distribution with $\rho > 0$, it is the $\alpha \in [0,2\pi)$ such that the following quantity
\begin{equation}
\E{e^{iY}} := \int e^{i y} \vM{y}{\alpha}{\rho} dy, 
\end{equation}
where $i$ denotes the complex number, satisfies the following identity
\begin{equation}
    \E{e^{iY}} = \frac{I_1(\rho)}{I_0(\rho)}e^{i\alpha}.
    \label{eq:vM_mean}
\end{equation}
Intuitively, the circular mean is a weighted average of points on the circle and the angles that correspond to them instead of the values of the angles. Indeed, if $f_{\mathcal{C}}(y)$ is some distribution on the circle, the circular mean in general is the $\alpha \in [0, 2\pi)$ such that
\begin{equation}
    \E{e^{iY}} := \int e^{i y} f_{\mathcal{C}}(y) dy
\end{equation}
satisfies the following equation for some $r \in \mathbbm{R}^+$,
\begin{equation}
    \E{e^{iY}} = r e^{i\alpha}.
    \label{eq:circ_mean}
\end{equation}
For convenience, we may refer to the circular mean as the mean or average later on. It should be obvious from context if we refer to the typical or circular mean.

Finally, despite the circular variance still using the mean direction $\alpha$, the circular variance is different as well. It is defined on the interval $[0, 2\pi)$ to be 
\begin{equation}
\Var{Y} := 1 - \int \cos(y - \alpha) f_{\mathcal{C}}(y) dy.
\label{eq:circ_var}
\end{equation}
We may also write $\Var{Y} = 1 - \E{\cos(Y - \alpha)}$ \citep{MardiaJuppDirectionalStatistics2010}, where the left hand side refers to the circular variance, and the right hand side makes use of the usual expectation (of a scalar random variable). It should be obvious from context which type of variance we are discussing.
Due to how circular variance is defined, it takes values between 0 and 1. At a high level, if $Y$ is tightly concentrated around the mean direction $\alpha$, the variance will be close to zero. Conversely, if $Y$ is dispersed around the circle, $\E{\cos(Y - \alpha)}$ will be 0 so the variance will be 1. Because the variance indicates the spread of the data, $\E{\cos(Y - \alpha)}$ can be seen as a measure of concentration for the distribution $f_{\mathcal{C}}$. For a von Mises distribution, the circular variance is
\begin{equation}
\Var{Y} = 1 - \frac{I_1(\rho)}{I_0(\rho)}.
\label{eq:vM_var}
\end{equation}

\subsection{Projected Normal Distribution}
\label{ssection:pn2}
The von Mises distribution is not the only distribution that we can derive from transforming a bivariate Gaussian distribution to a distribution on random angles. Let us consider a bivariate Gaussian with any mean, $\v{\mu} \in \mathbbm{R}^2$, covariance matrix, $\Sigma$.
Then, if we apply the polar coordinate transformation to this distribution, we have that 
\begin{align*}
    p(r, y \mid \v{\mu}, \Sigma) &= \frac{r}{2\pi\abs{\Sigma}^{-\frac{1}{2}}} \exp\left(-\frac{1}{2}\big((r\cos(y), r\sin(y))^T - \v{\mu}\big)^T\Sigma^{-1} \big(r\cos(y), r\sin(y))^T - \v{\mu}\big)\right).
\end{align*}
The projected normal distribution in two dimensions immediately follows from this because it is the marginal distribution of $Y$, i.e.
\begin{equation}
\PN{y}{\v{\mu}}{\Sigma} = \int_{r > 0} p(r, y \mid \v{\mu}, \Sigma) dr.
\label{eq:pn_2}
\end{equation}

Even though this integral is tractable \citep{MardiaJuppDirectionalStatistics2010}, we focus on the case $\Sigma = \mathbbm{I}_{2\times 2}$ because its properties are well-understood. Given $\v{\mu}$, let $\mu_0 \in \mathbbm{R}^+$ and $\alpha \in [0, 2\pi)$ be the polar coordinate transform of $\v{\mu}$. As shown by Wang and Gelfand \citep{WangGelfandDirectionalDataAnalysis2013}, the distribution in this scenario has the following closed-form expression:
\begin{align*}
    \PN{y}{\v{\mu}}{\mathbbm{I}_{2\times 2}} &= \int_{r > 0} \frac{r}{2\pi} \exp\left(-\frac{1}{2}\big((r\cos(y) - \mu_0\cos(\alpha))^2 + (r\sin(y) - \mu_0\sin(\alpha))^2\big)\right) dr\\
    &= \frac{1}{\sqrt{2\pi}}\exp\left(-\frac{1}{2}\mu_0^2\sin^2(y - \alpha)\right)\frac{1}{\sqrt{2\pi}}\int_{r > 0} r\exp\left(-\frac{1}{2}\big(r - \mu_0\cos(y - \alpha)\big)^2\right) dr\\
    &= \phi(\mu_0\sin(y - \alpha))\big(\phi(\mu_0\cos(y - \alpha)) + \mu_0\cos(y - \alpha)\Phi(\mu_0\cos(y - \alpha))\big)  \stepcounter{equation}\tag{\theequation}\label{eq:pn_2_id}.
\end{align*}
Here, $\phi(x)$ and $\Phi(x)$ denote the standard normal PDF and CDF. If we also then set $\v{\mu} = 0$, the projected normal distribution becomes uniform. For $\v{\mu} \neq 0$, the projected normal distribution with the identity matrix is symmetric and unimodal around $\alpha$ like the von Mises distribution. It is not surprising that based on this fact, the circular mean is $\alpha$ \citep{WangGelfandDirectionalDataAnalysis2013}. If $\alpha$ is undefined due to $\v{\mu} = 0$, then the mean angle is $\pi$ because the projected normal distribution is a uniform distribution. Meanwhile, we can use Kendall's technique to compute the circular variance \citep{KendallPoleSeekingBrownianMotion1974}. Despite Kendall only considering the case in which $\mu_1 \in \mathbbm{R}^{+}$ and $\mu_2 = 0$, it is straightforward to extend his result and show that the circular variance is $1 - \frac{1}{2}\sqrt{2\pi\beta}\exp(-\beta)(I_0(\beta) + I_1(\beta))$, where $\beta = \frac{\mu_0^2}{4}$ and $I_0(\beta)$ and $I_1(\beta)$ are the modified Bessel functions defined in \eqref{eq:bessel}. Again, like the von Mises distribution, the circular variance is independent of the mean angle. Then, we can use the asymptotics of the Bessel function to show that the circular variance is 1 when $\beta \rightarrow 0$ and 0 when $\beta \rightarrow \infty$ \citep{WatsonTreatiseTheoryBessel1995}. As discussed before, the former occurs when $\mu_0 \rightarrow 0$ and the distribution becomes uniform. The latter happens when $\mu_0 \rightarrow \infty$ and the probability (density) of other angles tend to 0.

We cannot use these results to derive closed-form expressions for the circular mean and variance in the general case because the projected normal distribution in that case has a more complicated closed-form expression. Fortunately, as we will demonstrate later in this paper, these results do allow us to compute the circular mean and variance of the models that we will introduce.

\subsection{Projected Gaussian process}
\label{ssection:pgp}
It is of interest to define a stochastic process of random directions indexed in a general domain
$\Omega$. This is possible by exploiting the polar coordinate transformation as above, whereas the bivariate normal distribution for $(z_1,z_2)$ can be extended into a stochastic process defined on $\Omega$; a common choice is to use the Gaussian process. This idea was studied by \citep{WangGelfandModelingSpaceSpaceTime2014}. We will present a simpler version of this idea.

Gaussian process is a powerful modeling tool for spatio-temporal data \citep{CressieWikleStatisticsSpatiotemporalData2011,BanerjeeEtAlHierarchicalModelingAnalysis2015,RasmussenWilliamsGaussianProcessesMachine2006}.
Given observations $z_1, z_2, \dots z_N$ indexed by the corresponding locations $x_1, x_2, \dots, x_N \in \Omega$, we assume that these observations are realizations of a stochastic process $\{Z(x) |x \in \Omega\}$, i.e., $z_\ell = Z(x_{\ell})$ for $\ell=1, 2, \ldots, N$. 
To account for the spatial dependence of these observations, one may assume that $Z$ is a Gaussian process, which is parameterized by a mean function $m$ and a covariance function $K$ on $\Omega$. Abusing notation, we let $\v{m}$ denote the mean function applied to every location such that $m_\ell = m(x_\ell)$. If $\Sigma$ is a $\mathbbm{R}^{N \cross N}$ matrix such that $\Sigma_{\ell, \ell'} = K(x_{\ell}, x_{\ell'})$ for $\ell, \ell' \in 1, 2, \ldots, N$, we will denote this as $\v{Z} \sim \textrm{GP}(\v{m}, \Sigma)$. We write it this way because $p(\v{Z} \mid \v{m}, \Sigma) = \mathcal{N}(\v{Z} \mid \v{m}, \Sigma)$.

We describe how to convert the Gaussian process to the projected Gaussian process. Suppose that $\v{Z_1} \sim \textrm{GP}(\v{m_1}, \Sigma_1)$ and $\v{Z_2} \sim \textrm{GP}(\v{m_2}, \Sigma_2)$ with $\v{Z_1}, \v{Z_2} \in \mathbbm{R}^N$. For identifiability of the angle-valued stochastic process, we might assume that $\begin{pmatrix}
\Sigma_1 & 0 \\
0 & \Sigma_2
\end{pmatrix}$ is equal to $\begin{pmatrix}
1 & 0 \\
0 & \sigma^2
\end{pmatrix} \otimes \Sigma$ for 
some fixed $\sigma^2 > 0$ and covariance matrix, $\Sigma$. While this may appear restrictive, it is similar in spirit to what Wang and Gelfand proposed to identify the projected normal distribution in two dimensions \citep{WangGelfandDirectionalDataAnalysis2013}. 
Then, let $\v{\mathcal{A}} \in [0, 2\pi)^N$ denote the angle-valued stochastic process observed at the $N$ index points so that $\mathcal{A}_{\ell} = \mathcal{A}(x_{\ell})$ for $\ell = 1, 2, \ldots, N$ and $\v{r} \in (\mathbbm{R}^{+})^{N}$ be a latent variable. We apply the polar coordinate transformation element-wise such that $z_{1, \ell} = r_\ell\cos(\mathcal{A}_{\ell})$ and $z_{2, \ell} = r_\ell\sin(\mathcal{A}_{\ell})$ for $\ell = 1, 2, \ldots, N$. As a shorthand, define $\cos(\mathcal{A})$ and $\sin(\mathcal{A})$ to be the vectors such that $\cos(\mathcal{A})_\ell = \cos(\mathcal{A}_\ell)$ and $\sin(\mathcal{A})_\ell = \sin(\mathcal{A}_\ell)$ for $\ell = 1, 2, \ldots, N$. Since the Jacobian of this transformation is $\prod_\ell r_\ell$, we have that
\begin{align}
    p(\v{\mathcal{A}}&, \v{r} \mid \v{m_1}, \v{m_2}, \Sigma_1, \Sigma_2) = \frac{1}{\sqrt{(2\pi)^N\abs{\Sigma_1}}}\exp\left(-\frac{1}{2}(\v{r}\cos(\v{\mathcal{A}}) - \v{m_1})^T\Sigma_1^{-1}(\v{r}\cos(\v{\mathcal{A}}) - \v{m_1})\right)\nonumber\\
    &\quad\frac{1}{\sqrt{(2\pi)^N\abs{\Sigma_2}}}\exp\left(-\frac{1}{2}(\v{r}\sin(\v{\mathcal{A}}) - \v{m_2})^T\Sigma_2^{-1}(\v{r}\sin(\v{\mathcal{A}}) - \v{m_2})\right) \left(\prod_{\ell} \v{r}\right)
\end{align}
The projected Gaussian process is the marginal distribution of $\v{\mathcal{A}}$:
\begin{align}
    p(\v{\mathcal{A}} \mid m_1, m_2, \Sigma_1, \Sigma_2) &= \int_{\v{r} > 0} p(\v{\mathcal{A}}, \v{r} \mid \v{m_1}, \v{m_2}, \Sigma_1, \Sigma_2) d\v{r}.
\end{align}


While this integral is intractable, we can still calculate its properties marginally for $\mathcal{A}_{\ell}$ to understand this process. Because for $\ell = 1, 2, \dots, N$, $\mathcal{A}_\ell$ is marginally distributed according to $f_N\left((m_{1, \ell}; m_{2, \ell})^T, \left(\begin{smallmatrix} (\Sigma_1)_{\ell, \ell} & 0 \\ 0 & (\Sigma_2)_{\ell, \ell}\end{smallmatrix}\right)\right)$. If $(\Sigma_1)_{\ell, \ell}, (\Sigma_2)_{\ell, \ell} = 1$, we can use the previous section's discussion to do so. Otherwise, we will explain the computation later.

\subsection{Higher Dimension Analogues}
\label{ssection:higher_dim}
An alternative way to describe these distributions is to view them as distributions on vectors "projected" onto the unit sphere. In other words, given $\v{z} \in \mathbbm{R}^2$, the von Mises distribution and the projected normal distribution in two dimensions are distributions on $\frac{\v{z}}{\norm{\v{z}}_2}$. From this perspective, because $r = \norm{\v{z}}_2$ in the polar coordinate transform, it further explains why the von Mises distribution is proportional to a particular bivariate Gaussian conditioned on $r = 1$. Meanwhile, the projected normal distribution arises by computing the probability of $\frac{\v{z}}{\norm{\v{z}}_2}$ if $\v{z}$ is distributed according to a bivariate Gaussian with mean $\v{\mu} \in \mathbbm{R}^2$ and $\Sigma \in \mathbbm{R}^2 \times \mathbbm{R}^2$. In other words, we have that
\begin{align}
    \PN{\frac{\v{z}}{\norm{\v{z}}_2}}{\v{\mu}}{\Sigma} &= \int_{\frac{\v{z'}}{\norm{\v{z'}}_2} = \frac{\v{z}}{\norm{\v{z}}_2}} p(\v{z'} \mid \v{\mu}, \Sigma) d\v{z'}.
    \label{eq:pn_2_alt}
\end{align}
This calculation is easier to do with polar coordinates.

This perspective also suggests how to generalize the von Mises distribution and the projected normal distribution for higher dimensions. First, for the von Mises distribution, suppose that $\v{z} \in \mathbbm{R}^D$, $D > 2$. Let $\v{z}$ be distributed according to a Gaussian distribution of dimension $d$ with mean $\v{\mu}$ and covariance matrix, $\rho^{-1} \mathbbm{I}_D$. Here, $\v{\mu} \in \mathbbm{R}^D$ such that $\norm{\v{\mu}}_2 = 1$ and $\rho \in \mathbbm{R}^+$. Then, 
\begin{align*}
p(\v{z} \mid \v{\mu}, \rho^{-1} \mathbbm{I}_D) &= \left(\frac{\rho}{2\pi}\right)^{\frac{d}{2}}\exp(-\frac{\rho}{2}(\v{z} - \v{\mu})^2)\\
&= \left(\frac{\rho}{2\pi}\right)^{\frac{d}{2}}\exp(-\frac{\rho}{2}(\v{z}^2 + 1 - 2 \v{\mu}^T \v{z})).
\end{align*}
If we condition on $\norm{\v{z}}_2 = 1$, we get the von Mises-Fisher distribution:
\begin{align*}
    p(\v{z} \mid \norm{\v{z}}_2 = 1, \v{\mu}, \rho) &\propto \exp(\rho \v{\mu}^T \v{z}).
\end{align*}
The normalization constant for this distribution is $(\frac{\rho}{2\pi})^{\frac{D}{2}}\frac{1}{\rho I_{\frac{D}{2} - 1}(\rho)}$ where $I_{\frac{D}{2} - 1}(\rho)$ is the modified Bessel function defined in \eqref{eq:bessel} \citep{MardiaJuppDirectionalStatistics2010}. Note that this is identical to the von Mises distribution in two dimensions because if $\v{\mu} = (\cos(\alpha), \sin(\alpha))$ and $\v{z} = (\cos(y), \sin(y))$ for $y, \alpha \in [0, 2\pi)$,  
\begin{align*}
    \v{\mu}^T \v{z} &= (\cos(\alpha), \sin(\alpha))^T (\cos(y), \sin(y))\\
    &= \cos(\alpha)\cos(y) + \sin(\alpha)\sin(y)\\
    &= \cos(y - \alpha).
\end{align*}

The projected normal distribution in higher dimension can be derived in a similar manner from its lower dimensional analogue. The projected normal distribution is given by computing the probability of $\frac{\v{z}}{\norm{\v{z}}_2}$ if $\v{z}$ is now distributed according to a Gaussian of dimension $D$ with mean, $\v{\mu} \in \mathbbm{R}^D$, and covariance matrix, $\Sigma \in \mathbbm{R}^D \times \mathbbm{R}^D$. The calculation is the same as \eqref{eq:pn_2_alt} except with $\v{\mu}$ and $\Sigma$ of the appropriate dimension. While this distribution might not be identifiable \citep{WangGelfandDirectionalDataAnalysis2013}, $\Sigma$ can be set in certain ways to simplify fitting the distribution to data \citep{Hernandez-StumpfhauserEtAlGeneralProjectedNormal2017}.

This then provides a template for defining the projected Gaussian process in higher dimensions. Suppose that for all $d = 1, 2, \ldots, D$, $\v{Z}_d \sim GP(\mu_d, \Sigma_d)$. Again, $\v{\mu}_d \in \mathbbm{R}^N$ and $\Sigma_d \in \mathbbm{R}^N \times \mathbbm{R}^N$ for any $d$. Set $\mathcal{C} \in \mathbbm{R}^N \times \mathbbm{R}^D$, $\norm{\v{\mathcal{C}_{n, \cdot}}} = 1$ to be a stochastic process on the unit circle. We then element-wise project each to the unit circle such that $\mathcal{C}_{n, d} = \frac{Z_{d, n}}{\sqrt{\sum_{d'} Z_{d', n}^2}}$. If we abuse notation and denote this element-wise projection as $\mathcal{C}(\v{Z_1}, \v{Z_2}, \ldots, \v{Z_D})$, then the probability is the following:
\begin{align*}
    p(\mathcal{C} \mid \v{\mu_1}, \v{\mu_2}, \ldots, \v{\mu_D}, \Sigma_1, \Sigma_2, \ldots, \Sigma_D) &= \int_{\mathcal{C}(Z_1, Z_2, \ldots, Z_D) = \mathcal{C}} \prod_d p(\v{Z_d} \mid \v{\mu_d}, \Sigma_d) d \v{Z_d}.
\end{align*}

\section{Modeling Random Directions} \label{Model}
\subsection{Model Description}
\label{ssection:model_description}
The purpose of our models is to elucidate the patterns observed in random directions. 
We first introduce the notation that we intend to use.

 We are given a location $x_\ell \in \Delta^2$ and an observation $y_\ell \in [0, 2\pi)$ for that location for $\ell = 1, 2, \dots, N$. We use $\ell$ as our index because of the spatial aspects of the data and model. To compare two observations, we will use $\ell'$ to indicate a location and observation different from the location and observation corresponding to $\ell$. If there are $K$ von Mises distributions indexed by $k$, $k \in 1, 2, \dots, K$, each observation, $y_\ell$, is assumed to be distributed according to one of the von Mises distributions with mean parameter, $m_{k, \ell} \in [0, 2\pi)$, and a concentration parameter, $\rho_{k, \ell} \in \mathbbm{R}^+$. We will use $\zeta_\ell$ to indicate which von Mises observation is associated with. If there is one von Mises distribution, we will drop $k$ from the parameters' index and will not use $\zeta_\ell$. We will remove $\ell$ from the index if the parameters are the same regardless of the location or if we wish to refer to the entire vector. While this overloading is unfortunate, any vector will be displayed in bold. Even though context should make it clear which definition we are referring to, we will specify the space the variable lies in.

We also need a set of notation for Gaussian processes because all the models that we will introduce use it. We let $\v{z_k} \in \mathbbm{R}^{N}$ be a draw from a Gaussian process with mean $\v{\mu_k} \in \mathbbm{R}^{N}$ and covariance matrix $\Sigma_k$. For models that need two Gaussian processes for each component, we will further index these vectors as $\v{z_{k, 1}}$ and $\v{z_{k, 2}}$ and their parameters as $\v{\mu_{k, 1}}$, $\v{\mu_{k, 2}}$, $\Sigma_{k, 1}$, and $\Sigma_{k, 2}$. The specific value corresponding to observation $y_\ell$ will be denoted by $z_{k, \ell}$ or $z_{k, 1, \ell}$ and $z_{k, 2, \ell}$. Again, we may drop $k$ if there is only one Gaussian process. In both cases, our choice of indices is to make the link between the specific Gaussian process and the particular von Mises distribution and observation more explicit.

\begin{table}[!tbp]
    \centering
    \begin{tabular}{c c | c | c}
         & & \multicolumn{2}{c}{\textbf{Observed Pattern}}\\
         & &  Homogeneous & Heterogeneous \\
         \hline
         \multirow{9}{*}{\rotatebox{90}{\textbf{Spatial}}} & & &  \\
         & \multirow{2}{*}{Ind.} & \multirow{2}{*}{\textit{iV} \eqref{model:iV}} & \multirow{2}{*}{\textit{iVM} \eqref{model:iVM}}\\
         & & &  \\ 
         & & &  \\ \cline{2-4}
         & & &  \\
         & \multirow{2}{*}{Var.} & \multirow{2}{*}{\textbf{\textit{SvM}} \eqref{model:SvM}} & \multirow{2}{*}{\shortstack[c]{\textbf{\textit{SvM-c}} \eqref{model:SvM-c}\\
         \textbf{\textit{SvM-p}} \eqref{model:SvM-p}, \eqref{model:SvM-p-2}}} \\
         & & &  \\ 
         & & & \\
    \end{tabular}
    \caption{A table showing the relationship between the models discussed in Section \ref{ssection:model_description} and the supplementary material. The models introduced in this paper are in bold.}
    \label{table:all_models_relation}
\end{table}

\noindent\textbf{Heterogenous Spatial Random Direction Model} There are two ways to incorporate spatial information for heterogeneous spatial random direction patterns. The first way is to integrate spatial information in the means of the von Mises distributions. We call this model the \textit{Spatially varying von Mises component mixture} model or \textit{SvM-c}. More specifically,
\begin{align*}
    \v{z_{k, 1}} &\sim \textrm{GP}(\cdot \mid \mu_{k, 1}, \Sigma_k), & k = 1, 2, \ldots, K\\
    \v{z_{k, 2}} &\sim \textrm{GP}(\cdot \mid \mu_{k, 2}, \Sigma_k), & k = 1, 2, \ldots, K\\
    \v{m_k} &= \textrm{arctan}^*(\v{z_{k, 1}}, \v{z_{k, 2}}), & k = 1, 2, \ldots, K \stepcounter{equation}\tag{\theequation}\label{model:SvM-c}\\
    \v{\varphi_k} &\stackrel{iid}{\sim} \textrm{N}(\cdot \mid \nu_k, \varsigma^2), & k = 1, 2, \ldots, K\\
    \v{\rho_k} &= \exp{\v{\varphi_k}}, & k = 1, 2, \ldots, K\\
    \zeta_\ell \mid \lambda_1, \lambda_2, \ldots, \lambda_K &\stackrel{iid}{\sim} \textrm{Cat}(\cdot \mid \lambda_1, \lambda_2, \ldots, \lambda_K)\\
    y_\ell \mid \zeta_\ell = k, m_{k, \ell}, \rho_{k, \ell} &\sim \vM{\cdot}{m_{k, \ell}}{\rho_{k, \ell}}, & \ell = 1, 2, \ldots, N.
\end{align*}

According to this model specification, each observation may be distributed by one of $K$ von Mises distributions with probability $\lambda_k$ regardless of its location. Each distribution's mean parameters, $\v{m_k} \in [0, 2\pi)^N$, are transformed from two draws, $\v{z_{k, 1}}, \v{z_{k, 2}} \in \mathbbm{R}^N$, from one of $K$ Gaussian process with potentially its own mean $\v{\mu}_k \in \mathbbm{R}^N$ and covariance matrix $\Sigma_k$. This transformation is accomplished using the $\textrm{arctan}^*$ function element-wise. A distribution's concentration parameters, $\v{\rho_k}$, are again random variables, $\v{\varphi_k}$, that have been exponentiated. These random variables are distributed according to a hierarchical normal distribution. At a lower level, they are distributed according to a normal distribution with the same standard deviation, $\varsigma$, but with different hierarchical means, $\nu_k$. These hierarchical means, $\nu_k$, are given the same hyperprior, $\textrm{N}(0, \tau)$.

In other words, for the concentration parameter in \textit{SvM-c}, we use the following hierarchical prior:
\begin{align*}
    \nu_k &\sim \textrm{N}(\cdot \mid 0, \tau^2),\\
    \v{\varphi_k} &\stackrel{iid}{\sim} \textrm{N}(\cdot \mid \nu_k, \varsigma^2).
\end{align*}
We use the hierarchical prior for the concentration parameter because it is a compromise between assigning an individual and a global concentration parameter. Using a global parameter $\rho$ might affect the estimates of the mean if the variances differ significantly because the model cannot adjust the concentration parameter. Conversely, assigning an individual parameter makes the model too flexible. This might negatively affect the model's ability to spatially correlate the observations. This concern also leads us to set the standard deviation for the lower term, $\varsigma$, to a small value instead of sampling for it. Even with a tight prior on $\varsigma$, the variance of the lower terms will be greater if we sample for the standard deviation. In addition, we do not use another Gaussian process to model the variance parameter $\varphi_\ell$ for computation reasons and for fears of making the model too rich. Still, a normal distribution is useful because it will allow us to separately sample the hierarchical mean, $\nu_k$, from the lower term, $\varphi_{k, \ell}$.

\textit{SvM-c} included spatial information through the components. An alternative approach to meld spatial information into our models is through the mixing probability. In this approach, we assume that observations nearby are likely to belong to the same von Mises distribution instead of being distributed according to a von Mises distribution with similar means. This alternative model is still interpretable because the results describe the preferred von Mises distributions of certain regions.

More explicitly, we call this model the \textit{Spatially varying mixing probability for von Mises distributions} model. Keeping with our previous convention, we will shorten it to \textit{SvM-p}. For this model, $y_\ell$ is assumed to be distributed to according to one of $K$ von Mises distribution with mean parameter, $m_k \in [0, 2\pi)$, and concentration parameter, $\rho_k \in \mathbbm{R}^+$. Both parameters are assumed to be the same for a particular von Mises distribution regardless of the observations' location. The mean parameters across all von Mises distributions are consequently given a prior of $\textrm{Unif}(0, 2\pi)$ whereas the concentration parameters are given a prior of $\Gamma(1, 1)$. If there is more prior knowledge about either parameter or if we want to force identifiability among the parameters, a von Mises distribution and a stronger Gamma distribution could instead be used for the mean parameter and concentration parameter respectively. Meanwhile, because the mixing probabilities sum up to 1, we only need to use the generalized inverse logit function, $\Psi^{-1}()$, to element-wise transform draws from $K - 1$ Gaussian processes with mean $\v{\mu_k}$ and covariance matrix $\Sigma_k$ to the mixing probabilities. We define the generalized inverse logit function, $\Psi^{-1}(z_{1, \ell}, z_{2, \ell}, \ldots, z_{k - 1, \ell})$, as a mapping from $\mathbbm{R}^{K - 1}$ to $\Delta^{K}$ in the following manner:  
\begin{equation*}
    \begin{aligned}
        \Psi^{-1}(z_{1, \ell}, z_{2, \ell}, \ldots, z_{k - 1, \ell})_k &= \frac{\exp(z_{k, \ell})}{1 + \sum_{k' = 1}^{K - 1} \exp(z_{k', \ell})}, & k = 1, 2, \ldots, K - 1 \\ 
        \\Psi^{-1}(z_{1, \ell}, z_{2, \ell}, \ldots, z_{k - 1, \ell})_K &= \frac{1}{1 + \sum_{k' = 1}^{K - 1} \exp(z_{k, \ell})}.
    \end{aligned}
    \label{fct:gen_inv_logit}
\end{equation*}
We will still use $k$ as our index for the Gaussian process draws to make clear which draw links to which von Mises distribution. In short, \textit{SvM-p} can be described as following.
\begin{align*}
    m_k &\stackrel{iid}{\sim} \textrm{Unif}(\cdot \mid 0, 2\pi), & k = 1, 2, \ldots, K\\
    \rho_k &\stackrel{iid}{\sim} \Gamma(\cdot \mid 1, 1), & k = 1, 2, \ldots, K\\
    \v{z}_k &\sim \textrm{GP}(\cdot \mid \v{\mu_k}, \Sigma_k), & k = 1, 2, \ldots, K - 1 \stepcounter{equation}\tag{\theequation}\label{model:SvM-p}\\
    \v{\lambda_1}, \v{\lambda_2}, \dots, \v{\lambda_K} &= \Psi^{-1}(\v{z_1}, \v{z_2}, \ldots, z_{k - 1}) \\ 
    \zeta_\ell \mid \lambda_{1, \ell}, \lambda_{2, \ell}, \dots, \lambda_{K, \ell},  &\sim \textrm{Cat}(\cdot \mid \lambda_{1, \ell}, \lambda_{2, \ell}, \dots, \lambda_{K, \ell})\\
    y_\ell \mid \zeta_\ell = k, m_k, \rho_k &\sim \vM{\cdot}{m_k}{\rho_k}, & \ell = 1, 2, \ldots, N.
\end{align*}

If there are only two von Mises distributions, the general inverse logit function simplifies to the inverse logit and the model only needs one Gaussian process. The inverse logit function, $\psi^{-1}(z)$, is defined to be $\psi^{-1}(z) = \frac{\exp(z)}{1 + \exp(z)}$. The function is a map from $\mathbbm{R}$ to $(0, 1)$. \textit{SvM-p} reduces down to the following, which we will denote as \textit{SvM-p-2}.
\begin{align*}
    m_k &\stackrel{iid}{\sim} \textrm{Unif}(\cdot \mid 0, 2\pi), & k = 1, 2\\
    \rho_k &\stackrel{iid}{\sim} \Gamma(\cdot \mid 1, 1), & k = 1, 2\\
    \v{z} &\sim \textrm{GP}(\cdot \mid 0, \Sigma), \stepcounter{equation}\tag{\theequation}\label{model:SvM-p-2}\\
    \lambda_{1, \ell} &= \invlogit{z_\ell}, & \ell = 1, 2, \ldots, N\\
    P(\zeta_\ell = 1) &= \lambda_{1, \ell}, & \ell = 1, 2, \ldots, N\\
    y_\ell \mid \zeta_\ell = k, m_k, \rho_k &\sim \vM{\cdot}{m_k}{\rho_k}, & \ell = 1, 2, \ldots, N.
\end{align*}

We make a few more comments before we discuss the properties of the models. \textit{SvM-c} and \textit{SvM-p} can also be thought of in the framework of the Dependent Dirichlet Process or DDP \citep{MacEachernStevenN.DependentDirichletProcess2000}. To understand the DDP, we need to briefly introduce the Dirichlet Process. Given a base distribution, $F_0$, and for $i \in 1, 2, \dots$, a Dirichlet process assigns weights, $w_i$, to atoms or draws, $\xi_i$, from $F_0$. If $v_i \sim \textrm{Beta}(1, m)$, then the weight for $\xi_i$ is given by $v_i \prod_{j = 1}^{i - 1} (1 - v_j)$ in the stick breaking representation \citep{SethuramanJayaramConstructiveDefinitionDirichlet1994}. A Dependent Dirichlet Process extends the Dirichlet Process by attaching separate stochastic processes to $v_i$ and $\xi_i$. If a stochastic process is based on location, nearby distributions that use DDP atoms can resemble each other \citep{MacEachernStevenN.DependentDirichletProcess2000}. While we assume a finite number of components, we view $\xi_i$ as the means of a von Mises distributions for our models. Then, in this framework, \textit{SvM-c} only attaches a stochastic process to $\xi_i$ whereas \textit{SvM-p} only places a stochastic process on $v_i$. Note that we do not endow stochastic processes to both $\xi_i$ and $v_i$ in this work because such a model might be too rich to depend only on the spatial information. 

We also wish to discuss some notation. Because of our convention with \textit{SvM-p-2}, we will denote the number of von Mises distributions after the model if we need to specify $K$. For instance, \textit{SvM-c-3} indicates the model \textit{SvM-c} with $K = 3$. The one exception to this guideline is \textit{SvM}, which is \textit{SvM-c} with $K = 1$ and introduced in the supplementary material 
(cf, \eqref{model:SvM}).

Finally, despite the focus on the circular case, the framework we introduce can easily be extended into directions of higher dimensions. Based on the discussion in Section \ref{ssection:higher_dim}, we can plug in the von Mises-Fisher in place of the von Mises distribution. We can also use the higher dimension version of the projected Gaussian process as a prior for the von Mises-Fisher for \textit{SvM-c}. The output from that process is marginally a unit vector, which can serve as the mean parameter for a von Mises-Fisher distribution.

\subsection{Model Properties for Spatially Dependent von Mises Model}
\label{ssection:model_prop}

To better understand our models, we shall derive the prior circular mean, circular variance, and circular correlation for \textit{SvM} and \textit{SvM-p-2}. Due to space constraints we do not display the results for \textit{SvM-c} because one can easily combine the iterated expectation formula with the results for the special case \textit{SvM} to obtain such results. Due to lack of conjugacy, deriving properties for posterior distributions is difficult. Nonetheless, the following results are useful because practitioners can use them when setting informative priors. While the definition for the circular mean and variance are given in equations \eqref{eq:circ_mean} and \eqref{eq:circ_var} respectively, we will calculate a notion of correlation proposed by \citep{JammalamadakaSarmaCorrelationCoefficientAngular1988}. If $y_\ell$ and $y_{\ell'}$ have means $\alpha_\ell$ and $\alpha_{\ell'}$ respectively, then the correlation is defined to be
\begin{equation}
\textrm{Corr}(y_\ell, y_\ell') = \frac{\E{\sin(y_\ell - \alpha_\ell)\sin(y_{\ell'} - \alpha_{\ell'})}}{\sqrt{\E{\sin^2(y_\ell - \alpha_\ell)}\E{\sin^2(y_{\ell'} - \alpha_{\ell'})}}}.
\end{equation}
These expectations are taken over the interval $[0, 2\pi)$ with respect to some distribution on the circle, $f_{\mathcal{C}}(y_\ell, y_{\ell'})$, or its marginal distributions. As noted by Jammalamadaka and Sarma, 
this is equivalent to
\begin{equation}
    \textrm{Corr}(y_\ell, y_\ell') = \frac{\E{\cos(y_\ell - \alpha_\ell - (y_{\ell'} - \alpha_{\ell'})) - \cos(y_\ell - \alpha_\ell + y_{\ell'} - \alpha_{\ell'})}}{2\sqrt{\E{\sin^2(y_\ell - \alpha_\ell)}\E{\sin^2(y_{\ell'} - \alpha_{\ell'})}}}.
\end{equation}
At a high level, the circular correlation is comparing how tightly $y_\ell - \alpha_\ell$ concentrates around $y_{\ell'} - \alpha_{\ell'}$ versus how tightly $y_\ell - \alpha_\ell$ concentrates around $-(y_{\ell'} - \alpha_{\ell'})$. Then, the correlation defined in this manner has the usual desired properties \citep{JammalamadakaSarmaCorrelationCoefficientAngular1988}. For instance, if $y_\ell$ and $y_{\ell'}$ are independent, the correlation is 0. 

\noindent\textbf{\textit{SvM}}
\label{ssection:SvM_circular_properties}
For \textit{SvM} specified in \eqref{model:SvM}, we will assume that there is a global $\rho$. 
This is not unreasonable because when we use a hierarchical prior for $\rho_\ell$, we set $\varsigma^2$ to a small value. Then, the circular mean, variance, and correlation are the following.
\begin{lemma}
If $Y_{\ell}$ and $Y_{\ell'}$ are generated according to \textit{SvM} outlined in \eqref{model:SvM} with the random variables associated with them labeled accordingly and $(Z_1, Z_2) \sim \mathcal{N}((\mu_1, \mu_2), \sigma^2\mathbbm{I}_2)$ with $\mu_1 = \mu_0\cos(\alpha_\mu)$ and $\mu_2 = \mu_0\cos(\alpha_\mu)$, $\mu_0 \in \mathbbm{R}^+$, $\alpha_\mu \in [0, 2\pi)$, then
\begin{align}
    \E{Y_{\ell}} &= \alpha_\mu,
    \label{lemma:SvM_e}\\
    \Var{Y_{\ell}} &= 1 - \frac{I_1(\rho)}{I_0(\rho)}\left(\frac{\pi\beta}{2}\right)^{\frac{1}{2}}\exp(-\beta)(I_0(\beta) + I_1(\beta)) \qquad \beta = \frac{\mu_0^2}{4\sigma^2},
    \label{lemma:SvM_var}\\
    \textrm{Corr}(Y_\ell, Y_{\ell'}) &= \frac{\left(\frac{I_1(\rho)}{I_0(\rho)}\right)^2(\E{\cos(m_\ell - m_{\ell'})} - \E{\cos(m_\ell + m_{\ell'} - 2\alpha_\mu)})}{\sqrt{\left(1 - \frac{I_2(\rho)}{I_0(\rho)}\E{\cos(2(m_\ell - \alpha_\mu))}\right)\left(1 - \frac{I_2(\rho)}{I_0(\rho)}\E{\cos(2(m_{\ell'} - \alpha_\mu))}\right)}}.
    \label{lemma:SvM_corr}
\end{align}
\label{lemma:SvM_model_prop}
\end{lemma}

The above lemma reveals the following characteristics of \textit{SvM}. First, the prior circular mean for $m_\ell$ is $\alpha_\mu$. It is nice that this mean can be set so easily. Based on this, we say that if $\mu_1 = c_1\mathbbm{1}$ and $\mu_2 = c_2\mathbbm{1}$, then $\textit{SvM}$ is \textit{centered at} or has mean $\arctan^*(\frac{c_2}{c_1})$. Because this holds for any $\mu_{k, 1} = c_{k, 1}\mathbbm{1}$ and $\mu_{k, 2} = c_{k, 2}\mathbbm{1}$, we might also say that $m_k$ of $\textit{SvM-c}$ is centered at $\arctan^*(\frac{c_{k, 2}}{c_{k, 1}})$. Second, the prior circular variance involves multiplying the concentrations of the von Mises' distribution and mean parameter according to the projected normal distribution. This demonstrates a few ways the Gaussian process can change the prior circular variance without changing the prior circular mean. Recall that $\beta = \frac{\mu_0}{4\sigma^2}$ in \eqref{lemma:SvM_var}. As $\beta \rightarrow 0$, $\Var{Y_\ell} \rightarrow 1$ whereas the asymptotics of Bessel functions can be used to show that $\Var{Y_\ell} \rightarrow 0$ when $\beta \rightarrow \infty$ and $\rho \rightarrow \infty$ \citep{WatsonTreatiseTheoryBessel1995}. The former occurs if $(\mu_1, \mu_2) = (0, 0)$ or $\sigma^2 \rightarrow \infty$. As a result, the respective distribution for $m_\ell$ or $Z_1$ and $Z_2$ becomes uniform. On the other hand, $\beta \rightarrow \infty$ when $\sigma^2 \rightarrow 0$ or $\mu_0 \rightarrow \infty$ while $\alpha_\mu$ remains unchanged. In other words, $(Z_1, Z_2)$ becomes more concentrated at $(\mu_1, \mu_2)$ or $\textrm{P}(m_\ell \neq \alpha_\mu) \rightarrow 0$. Note however that even if $m_\ell = \alpha_\mu$ with probability 1, but $\rho$ is finite, the circular variance of $Y_\ell$ is that of a von Mises distribution with concentration parameter, $\rho$. If $\rho = 0$, the circular variance is 1 and $Y_\ell$ will be uniformly distributed. This further supports our earlier observation that when \textit{SvM} is considered in a generative sense, the von Mises distribution adds noise. Finally, it is difficult to compute the exact form of the circular correlation due to the interactions between the terms in the expectation and the projected normal distribution. Fortunately, it is straightforward to compute these values by simulation. Wang and Gelfand briefly explored this in their 2014 paper on projected Gaussian processes. 

\noindent\textbf{\textit{SvM-p}}
\label{ssection:SvM-p_circular_properties}
We next compute the circular properties of \textit{SvM-p}. Unlike before, there is no explicit distribution for the inverse logit or general inverse logit transformation of a normally distributed variable. Consequently, we will first derive a formula for any choice of hyperparameters to allow practitioners to understand how their choice will affect these quantities. We do so under the assumption that the von Mises parameters, $m_1, m_2, \ldots, m_K$, and $\rho_1, \rho_2, \ldots, \rho_K$, are fixed. This assumption is a simplification because we can set $m_1$, $m_2$, $\rho_1$, and $\rho_2$ to their prior expected values. We then will calculate them for the choices we specified in this paper for the \textit{SvM-p-2} model. Under these assumptions, we have the following lemma.

\begin{lemma}
If $Y_{\ell}$ and $Y_{\ell'}$ are generated according to \textit{SvM-p} specified in \eqref{model:SvM-p} with the random variables associated with them labeled accordingly, $\arctan^*$ is defined in \eqref{fct:arctan_star}, $\alpha = \E{Y_{\ell}}$, which is given in \eqref{lemma:SvM-p_e}, and $s(\zeta_\ell) = 1 - \sum\limits_k P(\zeta_\ell = k) \frac{I_2(\rho_k)}{I_0(\rho_k)}\cos(2(m_k - \alpha))$, 
\begin{align}
    \E{Y_{\ell}} &= \arctan^*\left(\frac{\sum\limits_{k} P(\zeta = k)\frac{I_1(\rho_k)}{I_0(\rho_k)}\sin(m_k)}{\sum\limits_{k} P(\zeta = k)\frac{I_1(\rho_k)}{I_0(\rho_k)}\cos(m_k)}\right),
    \label{lemma:SvM-p_e}\\
    \Var{Y_\ell} &= 1 - \sum\limits_k P(\zeta = k)\frac{I_1(\rho_k)}{I_0(\rho_k)}\cos(m_k - \alpha),
    \label{lemma:SvM-p_var}\\
    \textrm{Corr}(Y_\ell, Y_{\ell'}) &= \frac{1}{\sqrt{s(\zeta_\ell) s(\zeta_{\ell'})}} \sum\limits_{k'}\sum\limits_{k} P(\zeta_\ell = k, \zeta_{\ell'} = k')\frac{I_1(\rho_k)}{I_0(\rho_k)}\frac{I_1(\rho_{k'})}{I_0(\rho_{k'})}\nonumber\\
    & \qquad\qquad\qquad\qquad\qquad\quad(\cos(m_k - m_{k'}) - \cos(m_k + m_{k'} - 2 \alpha)).
    \label{lemma:SvM-p_corr}
\end{align}
\label{lemma:SvM-p_model_prop}
\end{lemma}
We now discuss the properties in Lemma \ref{lemma:SvM-p_model_prop}. In $\eqref{lemma:SvM-p_e}$, \textit{SvM-p}'s mean direction is an average of the components' mean direction weighted by how concentrated each component is around its mean direction and the probability that an observation belongs to a component. If there are two components that are equally concentrated around their mean direction and equally likely of being selected, we can use trigonometric sum to product formulas to show that $\alpha$ will be the average of the components' mean angles. Otherwise, simulation must be used to understand how changing the hyperparameters of the Gaussian process affects $P(\zeta = k)$ and the mean angle $\alpha$. The circular variance in \eqref{lemma:SvM-p_var} is similar to the circular variance of \textit{SvM} described in \eqref{lemma:SvM_var}. From each component, we get the product of how tightly the observation concentrates around that component's mean direction and how close that component's mean direction is to the model's mean direction. The circular variance in \eqref{lemma:SvM-p_var} is also like the circular mean of \textit{SvM-p} in \eqref{lemma:SvM-p_e} because the probability parameters again serves as the weights for a convex combination of each component's concentrations. 

The numerator of the circular correlation in \eqref{lemma:SvM-p_corr} is of interest. Each term in the numerator's sum is comprised of three things. First, it has a comparison between the means concentrating around each other and the difference of one mean from the model mean $\alpha_2$ concentrating around its negative counterpart. Next, it holds the product of the observations' concentration around the mean directions with respect to the observations' cluster. Finally, it contains the joint probability of the observations' cluster memberships. It is this probability that will be affected by the spatial information and Gaussian process. At a high level, if observations are "nearby", they are likely to belong to the same cluster. More weight will be placed on the terms coming from the observations belonging to the same cluster. What observations are considered closed is determined by the hyperparameters of the Gaussian process.

\noindent\textbf{\textit{SvM-p-2}}
We now calculate these properties for the \textit{SvM-p-2} model. Using the notation of the model \eqref{model:SvM-p-2}, we will also assume that $Z, Z_{\ell}, Z_{\ell'} \sim \textrm{N}(0, 1)$ and $s = \textrm{Cov}(Z_\ell, Z_{\ell'})$. One technical issues arising in computing these quantities is that $P(\zeta = k) = \E{\invlogit{Z_k}}$, which is intractable. Instead, we will bound any expectation involving the logistic function with the functions, $f(z_\ell)$ and $g(z_\ell)$. We define $f(z_\ell) = g(z_\ell) + \frac{1}{2}$, where $g(z_\ell)$ is defined as following:
\begin{equation}
g(z_\ell) = 
\begin{cases}
-\frac{1}{2} \quad z_\ell < -\zEps\\
\frac{1}{2\zEps} z_\ell \quad z_\ell \in [-\zEps, \zEps]\\
\frac{1}{2} \quad z_\ell > \zEps,
\end{cases}.
\label{eq:logit_approx}
\end{equation}
We will additionally assume that $\zEps \in (0, 2]$ to ensure that $f(z_\ell) < \invlogit{z_\ell}$ for $z < 0$ and $f(z_\ell) > \invlogit{z_\ell}$ for $z > 0$. Under such assumptions, $\E{f(Z)} = \E{\invlogit{Z}}$. As a result, these functions, $f(z_\ell)$ and $g(z_\ell)$, are not a bad choice.

We can also use these functions to prove that $\E{\invlogit{Z}} = \frac{1}{2}$ under our assumptions. This fact and Lemma \ref{lemma:SvM-p_model_prop} then gives us the following lemma.
\begin{lemma}
If $Y_{\ell}$, and $Y_{\ell'}$ are generated according to \textit{SvM} outlined in \eqref{model:SvM-p-2} with the random variables associated with them labeled accordingly, $Z \sim N(0, 1)$, $(Z_\ell, Z_{\ell'}) \sim N(0, \Sigma)$, $\zEps \in (0, 2]$, and $\alpha = \E{Y_\ell}$, which is given in \eqref{lemma:SvM-p_version_e}, then
\begin{align}
    \E{Y_\ell} &= \arctan^*\left(\frac{\frac{I_1(\rho_1)}{I_0(\rho_1)}\sin(m_1) + \frac{I_1(\rho_2)}{I_0(\rho_2)}\sin(m_2)}{\frac{I_1(\rho_1)}{I_0(\rho_1)}\cos(m_1) + \frac{I_1(\rho_2)}{I_0(\rho_2)}\cos(m_2)}\right)
    \label{lemma:SvM-p_version_e}\\
    \Var{Y_\ell} &= 1 - \frac{1}{2}\left(\frac{I_1(\rho_1)}{I_0(\rho_1)}\cos(m_1 - \alpha) + \frac{I_1(\rho_2)}{I_0(\rho_2)}\cos(m_2 - \alpha)\right)
    \label{lemma:SvM-p_version_var}\\
    \textrm{Corr}(Y_{\ell}, Y_{\ell'}) &= \frac{2\E{\sin(y_\ell - \alpha)\sin(y_{\ell'} - \alpha)}}{1 - \frac{1}{2}\left(\frac{I_2(\rho_1)}{I_0(\rho_1)}\cos(2(m_1 - \alpha)) + \frac{I_2(\rho_2)}{I_0(\rho_2)}\cos(2(m_2 - \alpha))\right)}.
    \label{lemma:SvM-p_version_corr}
\end{align}
\label{lemma:SvM-p_version_prop}
\end{lemma}
We make the following comments based on the lemma above. First, because both component's are equally favored, $m_1$, $m_2$, $\rho_1$, and $\rho_2$ for the prior mean and variance become more important. The \textit{SvM-p-2} model's prior mean direction is completely determined by these parameters. For the model's prior circular variance, the concentration term is the average of the two components' concentration terms, which underscores the importance of $\rho_1$ and $\rho_2$. If we want the Gaussian process to have more of an effect on the prior mean or variance, it is necessary to change the mean for the Gaussian process. It can be shown that $\E{f(Z)} = \E{\invlogit{Z}}$ for $Z$ distributed according to any normal distribution with mean 0. Next, $\E{\sin(y_\ell - \alpha)\sin(y_{\ell'} - \alpha)}$ in \eqref{lemma:SvM-p_version_corr} can be expanded and written as a function of $\E{\invlogit{z_\ell}\invlogit{z_\ell'}}$. Doing so shows that the numerator again represents the difference between the concentration of two observations belonging to the same component and the concentration of two observations belonging to different components. It is then possible to bound $\E{\invlogit{z_\ell}\invlogit{z_\ell'}}$ with $\E{f(z_\ell)f(z_\ell')}$. Due to the length and technical aspects, we leave the details of this derivation and the expansion to the supplementary material.

\section{Posterior Inference}
\label{sec:model_fitting}
We now describe how to fit our models. Fitting \textit{SvM-c} can be difficult because of the underlying non-Euclidean geometry of its parameter space. We will first describe the basic sampling schemes based on the Hamiltonian MCMC approach, and then discuss techniques for exploiting the geometric structures via a suitable application of  elliptical slice sampling.

\noindent\textbf{Sampling algorithms} We use the following approaches to sample from our models. For \textit{SvM-p}, we use Hamiltonian Monte Carlo (HMC) to sample from the following posterior:
\[
\prod_{\ell = 1}^N \left(\sum_{k = 1}^K \lambda_k \vM{y_\ell}{m_k}{\rho_k}\right) \prod_{k = 1}^K \Gamma(\rho_k \mid 1, 1) \prod_{k = 1}^K \textrm{N}\left(\widetilde{z}_k \mid 0, \mathbbm{I}_{N \times N}\right).
\]
Here, $\v{\widetilde{z}}_k$ is the vector such that if $\Sigma = LL^T$, then $\v{z}_k = L\v{\widetilde{z}}_k$. Further, note that the probability for a single observation $y_\ell$ is the following:
\[
p(y_\ell \mid \v{\widetilde{z}_{., \ell}}, \v{m}, \v{\rho}) = \sum\limits_k \lambda_{k, \ell}\vM{y_\ell}{m_k}{\rho_k}.
\]
In other words, we marginalize out the label. As a result, we are using HMC to sample for $m_1, m_2, \ldots, m_K$; $\rho_1, \rho_2, \ldots, \rho_K$; and $\v{\widetilde{z}}_1, \v{\widetilde{z}}_2, \ldots, \v{\widetilde{z}}_{K - 1}$ according to the non-centered parametrization.

Then, the gradients for the log likelihood and thus the update for the momentum vector are given below.
\begin{itemize}
    \item For $k = 1, 2, \ldots, K - 1$ and $\ell = 1, 2, \ldots, N$, the gradient for $\widetilde{z}_{k, \ell}$ is the following:
    \begin{align}
    -\frac{\partial U}{\partial q_\ell}(q(t)) &= \sum_{\ell'} L_{\ell', \ell} \lambda_{k, \ell'}\frac{\vM{y_{\ell'}}{m_k}{\rho_k} - \sum_{k' = 1}^K\lambda_{k', \ell'}\vM{y_{\ell'}}{m_{k'}}{\rho_{k'}}}{p(y_\ell' \mid \v{\widetilde{z}_{., \ell'}}, \v{m}, \v{\rho})}  -  \widetilde{z}_\ell..
    \label{lemma:mixture_prob_noncentered_m_update}
    \end{align}
    In the two parameter case, this reduces to the following:
    \begin{align}
         -\frac{\partial U}{\partial q_\ell}(q(t)) &= \sum_{\ell'} \lambda_{1, \ell'} (1 -  \lambda_{1, \ell'}) \frac{\vM{y_{\ell'}}{m_1}{\rho_1} - \vM{y_{\ell'}}{m_2}{\rho_2}}{p(y_{\ell'} \mid z_{\ell'}, m_1, m_2, \rho_1, \rho_2)} L_{\ell', \ell}  -  \widetilde{z}_\ell.
        \label{lemma:mixture_prob_noncentered_m_update-2}
    \end{align}
    \item For $k = 1, 2 \ldots, K$, the gradient for $m_k$ is the following:
    \begin{align}
    -\frac{\partial U}{\partial q_k}(q(t)) &= \sum_{\ell'} \frac{\vM{y_{\ell'}}{m_k}{\rho_k}\sin(y_{\ell'} - m_k)}{p(y_\ell' \mid \v{\widetilde{z}_{., \ell'}}, \v{m}, \v{\rho})}.
    \label{lemma:mixture_prob_angle_m_update}
    \end{align}
    \item For $k = 1, 2 \ldots, K$, the gradient for $\rho_k$ is the following:
    \begin{align}
    -\frac{\partial U}{\partial q_k}(q(t)) &= \sum_{\ell'} \frac{\vM{y_{\ell'}}{m_k}{\rho_k}}{p(y_\ell' \mid \v{\widetilde{z}_{., \ell'}}, \v{m}, \v{\rho})}\left(\cos(y_{\ell'}- m_k) - \frac{\textrm{I}_{-1}(\rho_k)}{\textrm{I}_{0}(\rho_k)}\right) - 1.
    \label{lemma:mixture_prob_rho_update}
    \end{align}
\end{itemize}

On the other hand, for \textit{SvM-c}, we use the following blocked Gibbs approach: 


\begin{itemize}
    \item For $\ell = 1, 2, \ldots, N$, sample $\zeta_\ell$ given $y_\ell$; $z_{1, 1, \ell}, z_{1, 2, \ell}, z_{2, 1, \ell}, z_{2, 2, \ell} \ldots, z_{K, 1, \ell}, z_{K, 2, \ell}$; $\v{\varphi}_{1, \ell}, \v{\varphi}_{2, \ell}, \ldots, \v{\varphi}_{K, \ell}$; and $\nu_1, \nu_2, \ldots, \nu_K$ according to the following probability:
    \[
    p(\zeta_\ell = k \mid y_\ell; z_{k, 1, \ell}, z_{k, 2, \ell}; \v{\varphi}_\ell; \nu_\ell) \propto \lambda_k\vM{y_n}{m_{k, \ell}}{\rho_{k, \ell}}.
    \]
    \item For $k = 1, 2, \ldots K$, sample $(\v{z_{k, 1}} - \v{\mu_{k, 1}}, \v{z_{k, 2}} - \v{\mu_{k, 2}})$ given $\v{y}$, $\zeta_1, \zeta_2, \ldots, \zeta_N$; $\v{\varphi}_1, \v{\varphi}_2,$ $\ldots \v{\varphi}_K$; and $\nu_1, \nu_2, \ldots, \nu_K$ using the elliptical slice sampling outlined in Algorithm 2 \citep{MurrayEtAlEllipticalSliceSampling2010}. If we use the terminology of the paper, the normal prior and the likelihood that we are sampling are $\mathcal{N}\left(0, \begin{pmatrix} \Sigma_{k, 1} & 0\\
    0 & \Sigma_{k, 2}
    \end{pmatrix}\right)$ and $\prod_\ell (\vM{y_n}{m_{k, \ell}}{\rho_{k, \ell}})^{\indFct{\zeta_\ell = k}}$ respectively.
    \item Sample $\v{\varphi}_1, \v{\varphi}_2, \ldots, \v{\varphi}_K$ and $\nu_1, \nu_2, \ldots, \nu_K$ given $\v{y}$; $\zeta_1, \zeta_2, \ldots, \zeta_N$; and $\v{z_{1, 1}}, \v{z_{1, 2}}, \ldots,$ $\v{z_{K, 1}}, \v{z_{K, 2}}$ using HMC and the following probability:
    \begin{align*}
        p(\v{\varphi}_1, & \v{\varphi}_2, \ldots, \v{\varphi}_K, \nu_1, \nu_2, \ldots, \nu_K \mid \v{y}; \zeta_1, \zeta_2, \ldots, \zeta_N; \v{z_{1, 1}}, \v{z_{1, 2}}, \ldots, \v{z_{K, 1}}, \v{z_{K, 2}}) \propto\\
        & \prod_{\ell} (\vM{y_\ell}{m_{k, \ell}}{\rho_{k, \ell}})^{\indFct{\zeta_\ell = k}} \prod_k \left(\prod_\ell \textrm{N}(\varphi_{k, \ell} \mid \nu_k, \varsigma^2)\right) \textrm{N}(\nu_k \mid 0, \tau^2).
    \end{align*}
\end{itemize}
Unlike $\textit{SvM-p}$, we sample for the labels, $\zeta_1, \zeta_2, \ldots, \zeta_N$, because the mixing appeared to be better according to the trace plots.

To help the samplers, we used initial values obtained via a regularized version of Expectation Maximization algorithm derived from \textit{SvM-c} and \textit{SvM-p}. 
We leave the details for these algorithms to the supplementary material.

\noindent\textbf{Exploiting geometric structure} To sample \textit{SvM-c} in an efficient manner, we need to sample for the mean angles while respecting the spatial information, which is embedded in the Gaussian process' covariance matrix. HMC has trouble doing both even if we change the model parametrization. In one parametrization, the mean angles can be sampled using HMC, but at the cost of having to invert the covariance matrix. An alternative parametrization allows HMC to propose moves according to the covariance matrix without inverting it. However, these moves marginally lie in $\mathbbm{R}^2$ and may not result in new angles being sampled even if the proposal is accepted. For example, a marginal move along a ray from the origin might be accepted even though $m_{\ell}$ remains the same. Note that this is not the case for \textit{SvM-p} because the Gaussian processes are mapped to probabilites via the general inverse logit function. Hence, each marginal move represents a different mixing probability for each observation. More discussion of why HMC struggles for \textit{SvM-c} and not \textit{SvM-p} and the exact parametrization discussed are given in the supplementary material. 

\begin{figure}[!tp]
\captionsetup[subfigure]{justification=centering}
\centering
\begin{subfigure}{.4\textwidth}
\centering
\includegraphics[width = 1\textwidth]{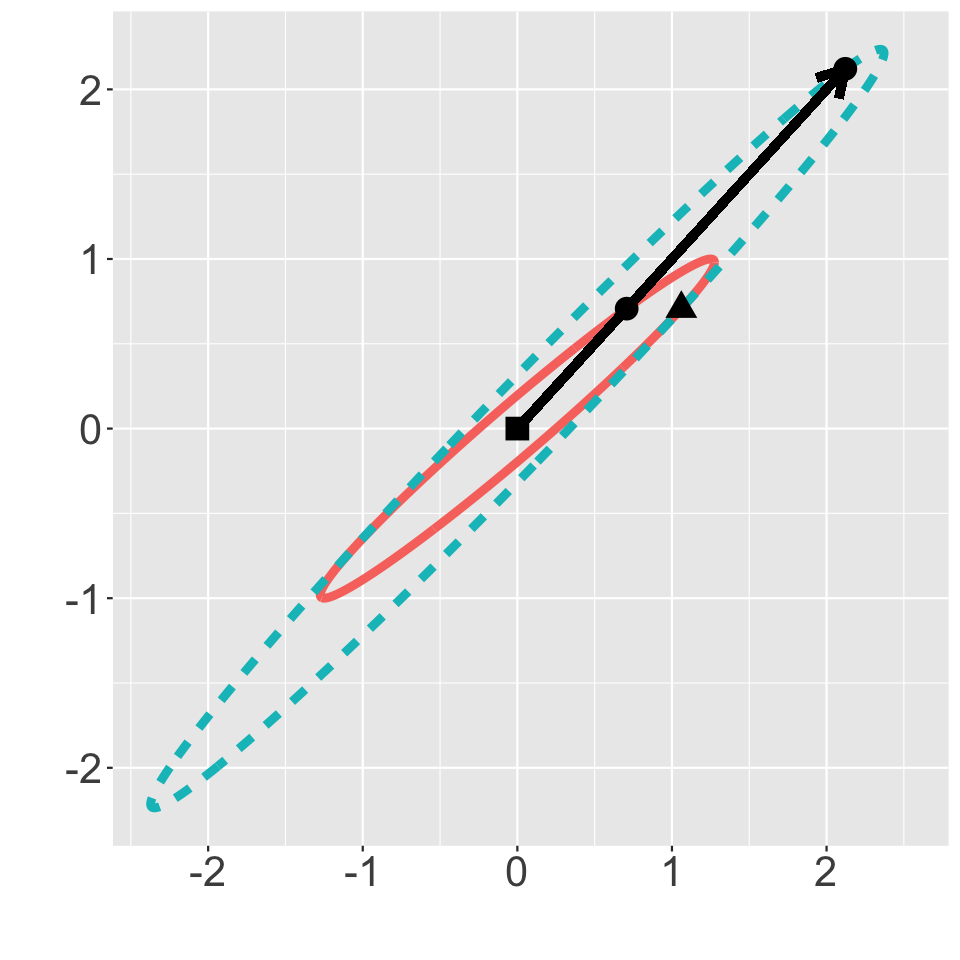}
\end{subfigure}
\begin{subfigure}{.4\textwidth}
\centering
\includegraphics[width = 1\textwidth]{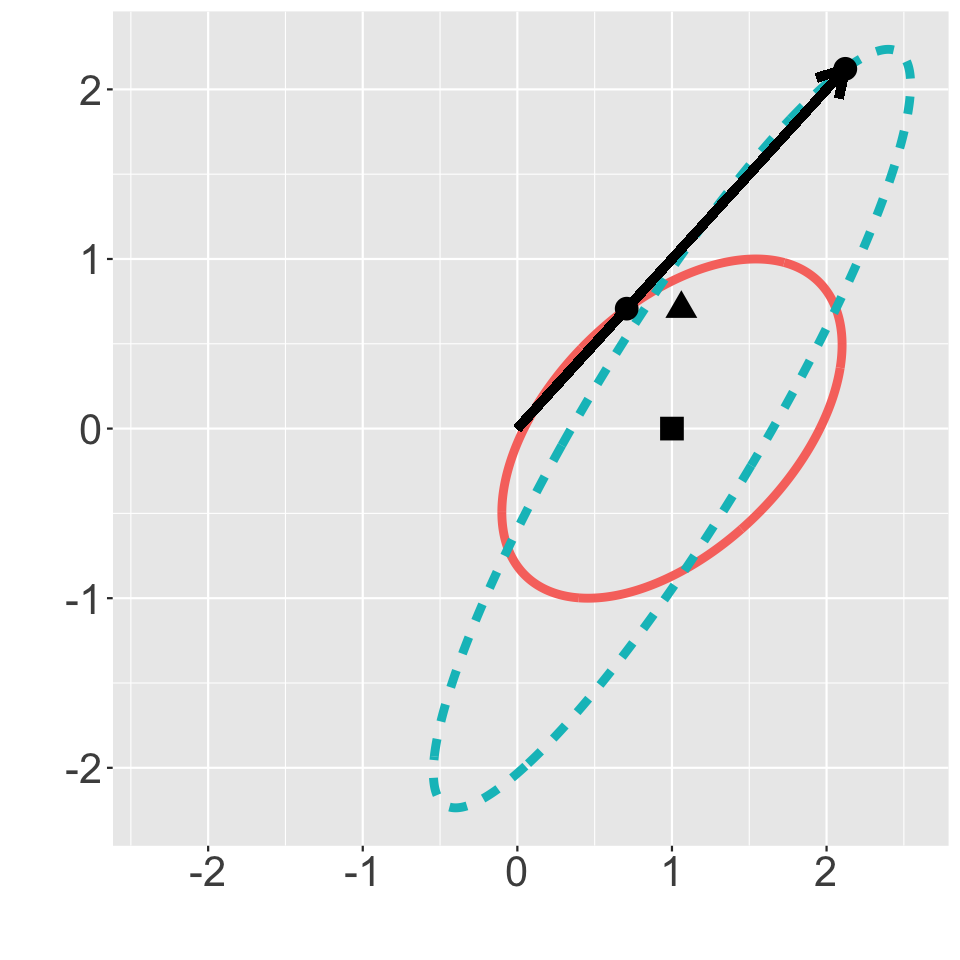}
\end{subfigure}\\
\caption{Plots showing the elliptical slice sampler's proposal ellipse for different locations (shown as circles) and means for the Gaussian prior (shown as the squares), but the same elliptical proposal point (shown as the triangle). These ellipses represent the set of angles that can be proposed by the elliptical slice sampler. The arrow demonstrates that the locations have the same angle in polar coordinates, but different radii.}
\label{fig:ess_proposal_ellipse}
\end{figure}

We turn to elliptical slice sampling instead \citep{MurrayEtAlEllipticalSliceSampling2010}. This technique allows us to draw samples for a random variable distributed according to a Gaussian prior with zero mean and any likelihood. 
For \textit{SvM-c}, we use the elliptical slice sampler to sample for $(Z_{k, 1} - \mu_{k, 1}, Z_{k, 2} - \mu_{k, 2}) \mid \v{\varphi}_k, \nu_k, \v{\zeta}$ for $k \in 1, 2, \ldots K$. In this case, the Gaussian prior has covariance $\mathbbm{I}_2 \otimes \Sigma$. Following the notation of the paper, we let $\v{f}$ represent the current location of $(Z_{k, 1} - \mu_{k, 1}, Z_{k, 2} - \mu_{k, 2})$. If we use the transformation specified in \eqref{fct:arctan_star} on $\v{f} + (\mu_1, \mu_2)$ to get $\v{m}$, then the likelihood is defined as follows: The likelihood becomes the following:
\[
L(\v{y} \mid \v{m}_k, \v{\varphi}_k, \nu_k, \v{\zeta}) = \prod_{\ell} (\vM{y_\ell}{m_\ell}{\rho_\ell})^{\indFct{\zeta_\ell = k}}.
\]
Note that even if $y_\ell$ is not included in the likelihood because $\zeta_\ell \neq k$, we still propose $Z_{k, 1, \ell}$ and $Z_{k, 2, \ell}$. This will prove useful for predictive purposes.

As seen in Figure \ref{fig:ess_proposal_ellipse}, the elliptical slice sampler is marginally a sensible way to explore the space if $\mu_1 = 0$ and $\mu_2 = 0$. We can then explore the entire space of $m_\ell$ with each run of the slice sampler. Unfortunately, if we set $\mu_1$ and $\mu_2$ to be $0$, $m_\ell$ will be marginally distributed according to a uniform distribution. By not setting $\mu_1$ and $\mu_2$ to be $0$, Figure \ref{fig:ess_proposal_ellipse} shows us that not only can the sampler potentially explore a subset of random angles, but also the subset might depend on its current location in $\mathbbm{R}^2$ and the point used to determine the ellipse. This might affect the next draw because for the same value of $m_\ell$, the range of the next $m_\ell$ may be different. Still, even with this limitation, the elliptical slice sampler allows us to propose a valid move in the space of $\v{m}$ that respects the covariance information from the Gaussian process. Due to it being hard to characterize the dependence of the next draw on the location, we thin our MCMC chain to mitigate this issue. Luckily, the elliptical slice sampler runs reasonably quickly so this does not impose too much of an additional computational burden.

\section{Simulations and Results}
\label{results}
In this section we shall present a simulation study for the introduced models and then an application to the analysis of the income proportion data set. 
\subsection{Simulation study}
\label{ssection:sim_results}
\textbf{Simulation overview} The simulation study is conducted with following goals in mind. First, we wanted evidence that if we used our models for inference, our models could correctly recover the model parameters for data generated according to their respective models. Next, with this menagerie of simulated data, we could examine how these models behave in mis-specified settings and how to address model selection. Then, while fitting these models, we might notice inefficiencies in our sampling scheme and develop better approaches. Finally, inspecting the results might yield additional insights about the behaviors of these models.

To that end, we used a uniform distribution to generate 500 "locations" in $\Delta^2$. We created observations in six different ways. First, we generated observations using Von Mises($\pi$, 5). Next, we created them using a mixture of Von Mises($\frac{\pi}{2}$, 5) and Von Mises($\frac{3\pi}{2}$, 10) with mixing probability of 0.3 and 0.7 respectively. As discussed in the supplementary material, these models are the \textit{iV} and \textit{iVM} models respectively. While these two methods did not use the spatial information, the remaining methods do. Third, we simulated observations according to \textit{SvM} specified in Model \eqref{model:SvM} with a prior mean of $\pi$ for all locations. Fourth, we created them according to \textit{SvM-c} specified in Model \eqref{model:SvM-c}. One component's mean was centered at $\frac{\pi}{2}$ and the other component's mean at $\frac{3\pi}{2}$. There was equal probability of using either component. Next, we used the \textit{SvM-p} model, \eqref{model:SvM-p}, to simulate observations. We again used means of $\frac{\pi}{2}$ and $\frac{3\pi}{2}$. Finally, we simulated according to \textit{SvM} specified in Model \eqref{model:SvM} with a prior mean of $0$  for all locations. This will test whether our model can handle data that appears to be separated due to where we set zero, but is actually connected.

\noindent\textbf{Simulation model fitting} We fit the \textit{SvM}, \textit{SvM-c}, and \textit{SvM-p} models to these simulated observations using the approaches outlined in the previous section. Because there was at most two components, we only fit the two component versions of \textit{SvM-c} and \textit{SvM-p} to our simulated observations. Then, for the Gaussian processes, we used the squared exponential kernel with $\omega = 0.1$. For \textit{SvM-p}, we set $\sigma = 1$ and $\mu$ = 0. Meanwhile, for \textit{SvM} and \textit{SvM-c} model, we set $\sigma = 0.5$. We set $\mu = (-1, 0)$ for \textit{SvM} and $\mu_1 = (0, 1)$ and $\mu_2 = (0, -1)$ for \textit{SvM-c}. These values correspond to $\pi$, $\frac{\pi}{2}$, and $\frac{3\pi}{2}$ respectively. We set $\varsigma = 0.05$ and $\tau = 5$ for \textit{SvM} and \textit{SvM-c} concentration parameter's hierarchical prior. Our goal for these choices of parameters was to capture a larger range of probabilities for \textit{SvM-p} whereas we wanted components to be more distinct for \textit{SvM} and \textit{SvM-c} models. 


\begin{figure}[!tb]
\captionsetup[subfigure]{justification=centering}
\centering
\begin{subfigure}{.23\textwidth}
\centering
\includegraphics[width = 1\textwidth]{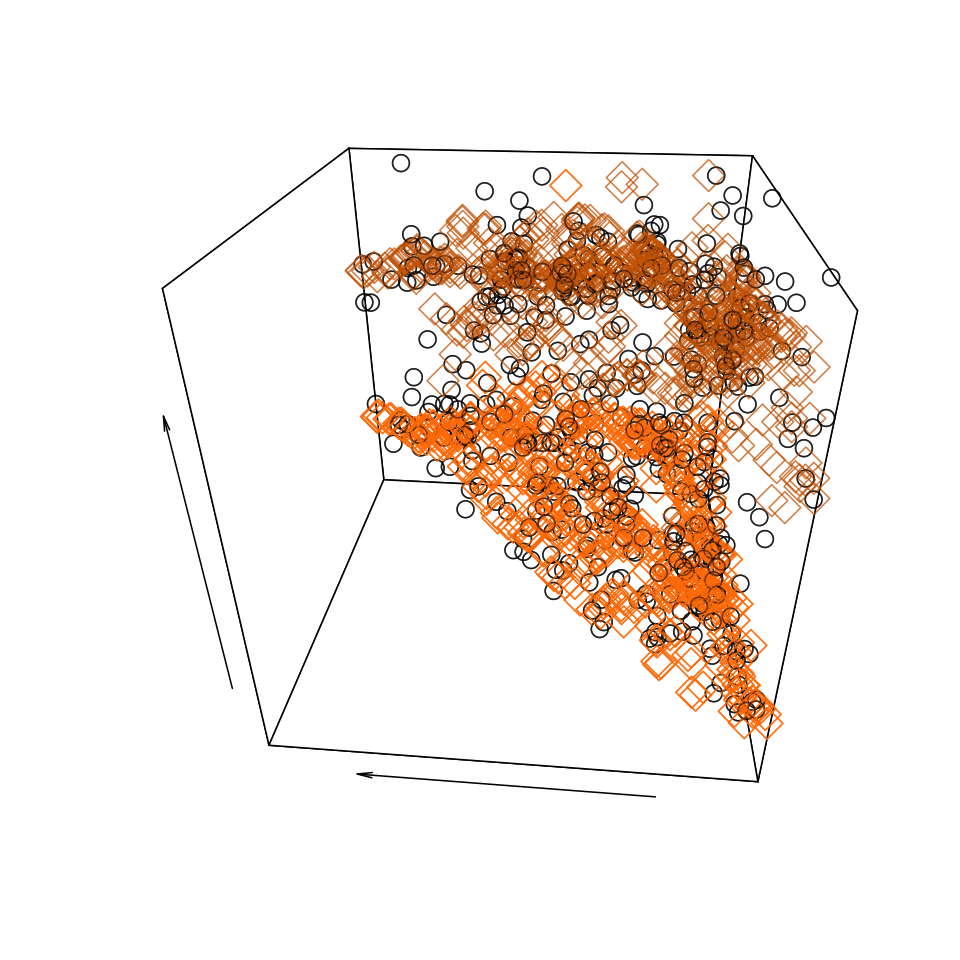}
\caption{Simulated data with\\means in orange}
\end{subfigure}
\begin{subfigure}{.23\textwidth}
\centering
\includegraphics[width = 1\textwidth]{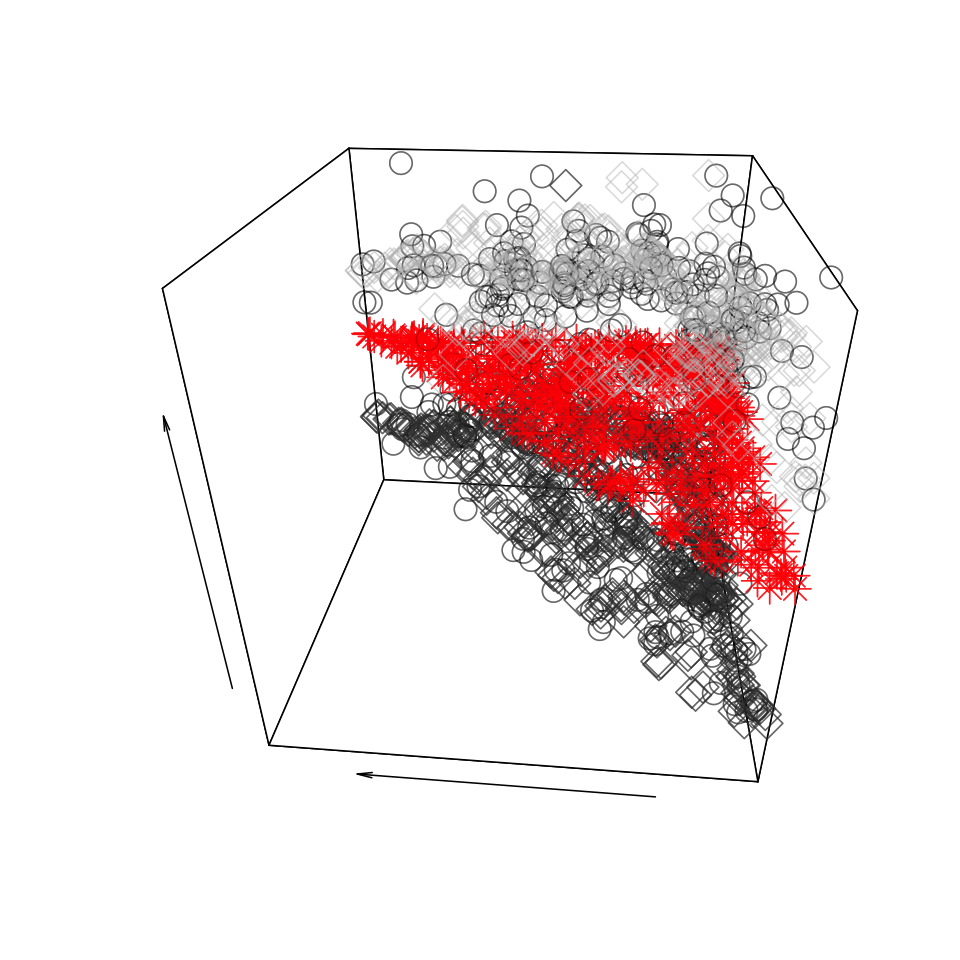}
\caption{SvM model (cf.~\eqref{model:SvM})\\fitted means in red}
\end{subfigure}
\begin{subfigure}{.23\textwidth}
\centering
\includegraphics[width = 1\textwidth]{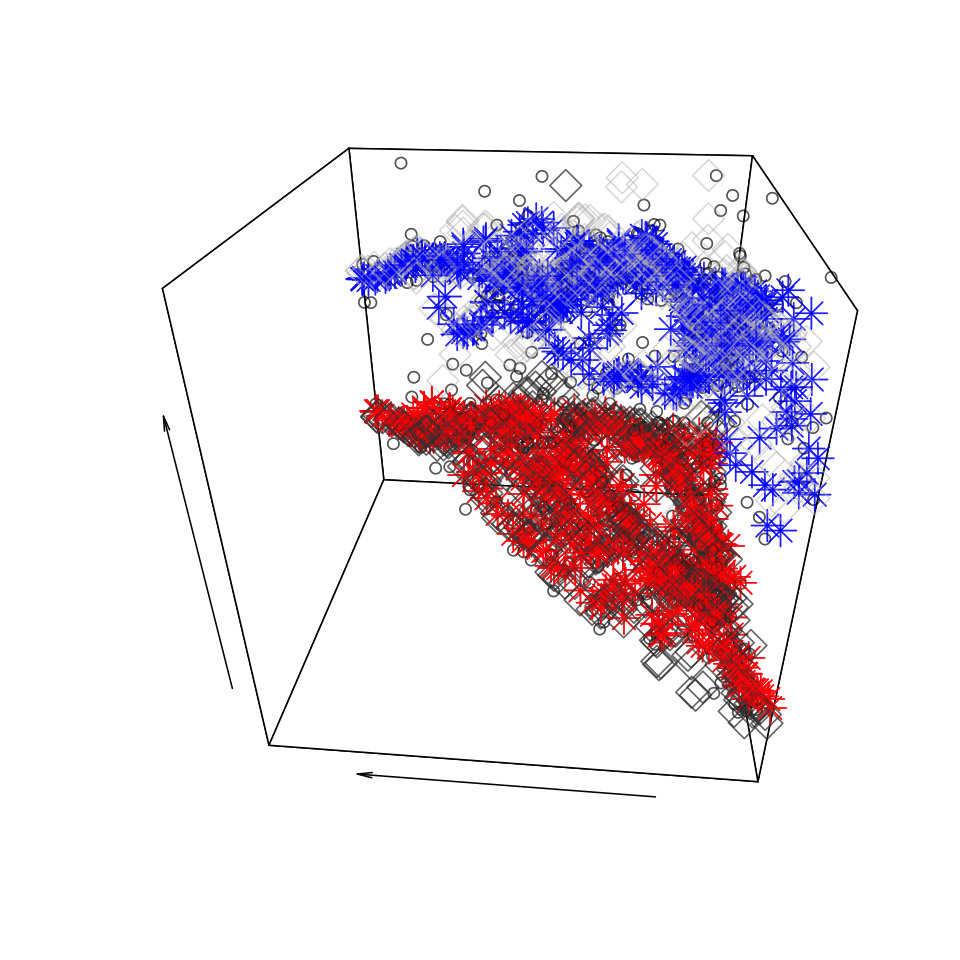}
\caption{SvM-c model (cf.~\eqref{model:SvM-c})\\fitted means in red and blue}
\end{subfigure}
\begin{subfigure}{.23\textwidth}
\centering
\includegraphics[width = 1\textwidth]{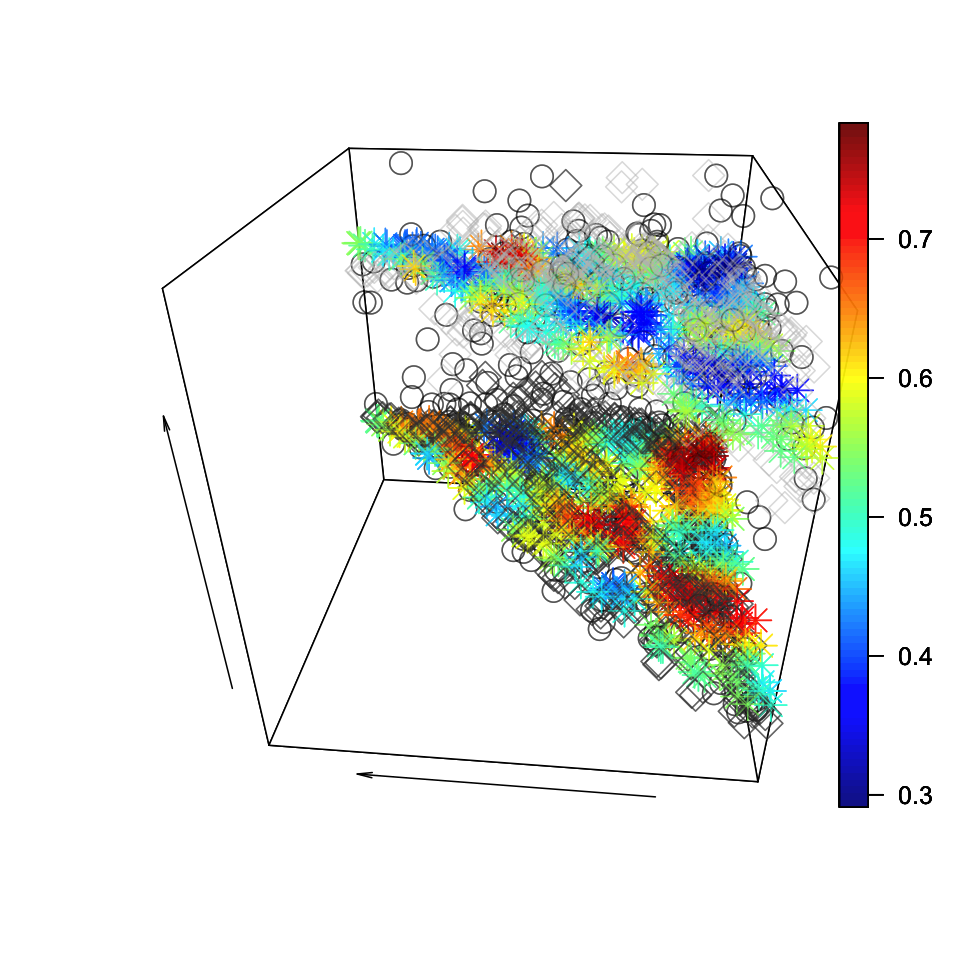}
\caption{SvM-p model (cf.~\eqref{model:SvM-p})\\fitted means colored by $\lambda_1$}
\end{subfigure}\\
\caption{The right three plots show the mean component fitted by different models for random directions simulated according to \textit{SvM-c} in \eqref{model:SvM-c}. The fitted mean components for the different models are shown in colored asterisks with the simulated random directions as circles and the simulated means as grey rhombuses. The leftmost plot is just the simulated data and the mean surface used to generate that data. The x-y axis show the locations in the first two dimensions of a three dimensional simplex; the z-axis represent the direction.}
\label{fig:model_fitted_SvM-c_plots}

\resizebox{0.9\textwidth}{!}{
\centering
\begin{tabular}{@{\extracolsep{5pt}} cccc} 
\\[-1.8ex]\hline 
\hline \\[-1.8ex] 
 & $\bm{m}$ & $\bm{\rho}$ & $\bm{\lambda}$ \\ 
\hline \\[-1.8ex] 
\multirow{2}{*}{\textbf{Simulation}} & $\overline{m}_1$ = 1.36 (0.41, 2.16) & $\overline{\rho}_1 = 3.01 (2.56, 3.50)$ & $\lambda_1 = 0.50$\\
& $\overline{m}_2$ = 4.85 (3.53, 5.64) & $\overline{\rho}_2 = 8.00 (6.98, 9.00)$ & $\lambda_2 = 0.50$\\
& & &\\
\textit{SvM} cf.~\eqref{model:SvM} & $\overline{m}$ = 3.15 (1.93, 4.35) & $\overline{\rho}$ = 0.01 (0.00, 0.04) & ---\\
& & &\\
\multirow{2}{*}{\shortstack[c]{\textit{SvM-c} cf.~\eqref{model:SvM-c}}} & $\overline{m}_1$ = 1.42 (0.90, 1.97) & $\overline{\rho}_1$ = 2.79 (2.17, 3.52) & $\lambda_1$ =  0.54 (0.49, 0.58)\\
& $\overline{m}_2$ = 4.79 (4.38, 5.20) & $\overline{\rho}_2$ = 8.01 (6.11, 10.52) & $\lambda_2$ = 0.46 (0.42, 0.51)\\
& & &\\
\multirow{2}{*}{\shortstack[c]{\textit{SvM-p} cf.~\eqref{model:SvM-p}}} & $m_1$ = 1.38 (1.24, 1.51) & $\rho_1$ = 1.44 (1.07, 1.91) & $\overline{\lambda_1}$ = 0.58 (0.32, 0.82)\\
& $m_2$ = 4.77 (4.68, 4.86) & $\rho_2$ = 4.08 (2.75, 5.55) & $\overline{\lambda_2}$ = 0.42 (0.18, 0.68)\\
& & &\\
\hline \\[-1.8ex] 
\end{tabular} 
}
\captionof{table}{Circular mean posterior values and 95\% credible intervals for the parameters in von Mises distribution and mixing probability are shown for the simulation data in Figure \ref{fig:model_fitted_SvM-c_plots}. Parameters with bars over them were averaged across all locations.}
\end{figure}

For posterior inference, any step involving HMC was implemented in Stan \citep{StanDevelopmentTeam.StanCoreLibrary2018}. Otherwise, we implemented our approach in R. Using this code\footnote{Our code can be found at \url{https://github.com/rayleigh/housing\_price}.}, we ran four chains and 2000 iterations of the HMC sampler for \textit{SvM-p} with 1000 iterations as burn-in. Meanwhile, for \textit{SvM} and \textit{SvM-c} sampler, we ran four chains and 10000 iterations, using every fifth iteration after the first 5000 iterations. On a cluster using 2x 3.0 GHz Intel Xeon Gold 6154 as its processor and with 192 GB as its RAM, \textit{SvM-c} took between four and a half to eight and a half hours whereas \textit{SvM-p} took between 8 and 20 minutes. We checked for convergence by primarily looking at trace plots of parameters and examining Rhat values for non-circular values as computed by Stan \citep{GelmanRubinInferenceIterativeSimulation1992}.

\noindent\textbf{Simulation model findings} These are a number of interesting results to report from these experiments. First, as exemplified by Figure \ref{fig:model_fitted_SvM-c_plots} and other such figures in the appendix, the models appeared to correctly capture the mean surfaces for data simulated from their respective models. For example, the fitted mean surface for \textit{SvM} matches the simulated mean surface, even when the simulated surface is centered at $0$. While the numerical summaries of the random variables were slightly off, this could be due to the posterior credible intervals being averaged across all locations. Next, homogeneous models may have trouble capturing heterogeneous patterns. For instance, the fitted mean surface from \textit{SvM} tries to match the observations as much as possible in the case of \textit{SvM-p} or essentially gives up and becomes a uniform distribution in the case of \textit{SvM-c}. It does not recognize the simulated mean pattern of either model. Finally, the results demonstrate that \textit{SvM} and \textit{SvM-c} are in a different class of models compared to the class of models \textit{SvM-p} belongs to. As seen in Figure \ref{fig:model_fitted_SvM-c_plots}, \textit{SvM-p} tries to put the von Mises distributions at the two components' overall mean with small concentration parameters in the case of \textit{SvM-c}. The fitted model is trying to make up for the variation in the observations by returning two diffuse von Mises distribution. A similar result happens in the \textit{SvM} case, except that the two von Mises are placed at an "lower" and "upper" mean. In contrast, \textit{SvM-c} correctly recovers the simulated mean surfaces in the \textit{SvM-p} case, but it can only roughly capture the average probability of each surface. Interestingly enough, in the \textit{SvM} case, \textit{SvM-c} matches one fitted mean surface to the simulated mean surface and essentially "zeroes" out the other surface.

\noindent\textbf{Simulation model selection} We decided to compute the posterior predictive probability for another 50 random locations and their observations simulated according to the different scenarios in order to model select. Let $x^*$ represent the withheld locations, $y^*$ the withheld data, $\theta$ a posterior draw for the parameters based on $x$ and $y$, and $\theta^*$ a draw for the parameters for $x^*$ and $y^*$. The posterior predictive probability is given as follows.
\begin{align}
p(\v{y^*} \mid x, x^*, \v{y}) := \int \int \textrm{p}(y^* \mid \theta^*) \textrm{p}(\theta^* \mid \theta, x, x^*)  \textrm{p}(\theta \mid x, y) d\theta^* d\theta.
\label{eq:post_pred_prob}    
\end{align}
We discuss how to compute this probability in the supplementary material because it is straightforward to compute $\textrm{p}(y^* \mid \theta^*)$ according to our models defined in $\eqref{model:SvM-c}$ and $\eqref{model:SvM-p}$. It is also simple to sample from $\textrm{p}(\theta^* \mid \theta, x, x^*)$ due to the conditional formula for the Multivariate normal distribution to generate these draws from the Gaussian processes. We take advantage of this fact by sampling 100 times from $\textrm{p}(\theta^* \mid \theta, x, x^*)$ to create our posterior predictive draws. To ensure that both the posterior draws and the posterior predictive draws speak equally, we first averaged across posterior draws. We then averaged these means across the posterior predictive draws. This gives us a sense how likely our posterior is on withheld data.

\begin{table}[!tp]
    \centering
    \begin{tabular}{c c c c c c c}
    \hline \\[-1.8ex]
    & iV & iVM & \textit{SvM} Mean $\pi$ & \textit{SvM-c} & \textit{SvM-p} & \textit{SvM} Mean $0$\\
    \hline \\[-1.8ex] 
    \textit{iV} & \textbf{-39.85} & -78.20 & -68.08 & -91.99 & -92.05 & -$89.20^*$\\
    & \multicolumn{5}{c}{}\\
    \textit{iVM} & \textbf{-39.92} & \textbf{-47.60} & -68.24 & -86.47 & \textbf{-61.80} & -85.69\\
    & \multicolumn{5}{c}{}\\
    \textit{SvM} & -42.42 & -84.28 & \textbf{-59.34} & -92.43 & -97.86 & -63.49 \\
    & \multicolumn{5}{c}{}\\
    \textit{SvM-c} & -45.78 & -56.79 & -62.19 & \textbf{-73.40} & -70.90 & \textbf{-60.67}\\
    & \multicolumn{5}{c}{}\\
    \textit{SvM-p} & \textbf{-39.98} & -49.60 & -66.61 & -90.03 & -62.78 & -73.54\\
    & \multicolumn{5}{c}{}\\
\end{tabular}
\caption{The log of the posterior predictive probability given in Equation \eqref{eq:post_pred_prob} for 50 draws averaged across posterior draws and then posterior predictive draws.}
\label{table:sim_post_pred_results}
\end{table}

Table \ref{table:sim_post_pred_results} shows the results from doing so. There are a few interesting things to note. First, the posterior predictive probability confirms that the homogeneous models performs poorly in any heterogeneous scenarios. Next, the posterior predictive probability calculated after fitting \textit{SvM-p} is higher when the means of the von Mises used to generate the data are constant. On the other hand, the posterior predictive probability calculated after fitting the \textit{SvM} and \textit{SvM-c} models is higher when the von Mises components' means are spatially correlated. We might expect the former because spatially correlated observations might be overfitted to the data. For the latter, \textit{SvM-p} might have difficulty in capturing the variability and thus cannot accurately predict the next value. Finally, while we do not report these cases, there are scenarios in which \textit{SvM-c} can perform poorly according to the predictive posterior probability. When we plot them, we see that the worst performing observations lie between the fitted mean surface. Even though the fitted mean surface and concentration parameter may be mostly correct for the observed data, the model is being penalized for being "over-certain" about the withheld data. This means that if we desire well-separated components, the observations must also be well-separated. Even with this caveat, we use the posterior predictive probability to model select because it allows us to compare models that are different in their approaches and generally picks the model used to generate the data in simulation.

\subsection{Data analysis}
\begin{figure}[!tbp]
\centering
\begin{subfigure}{.4\textwidth}
\centering
\includegraphics[width = 0.8\textwidth]{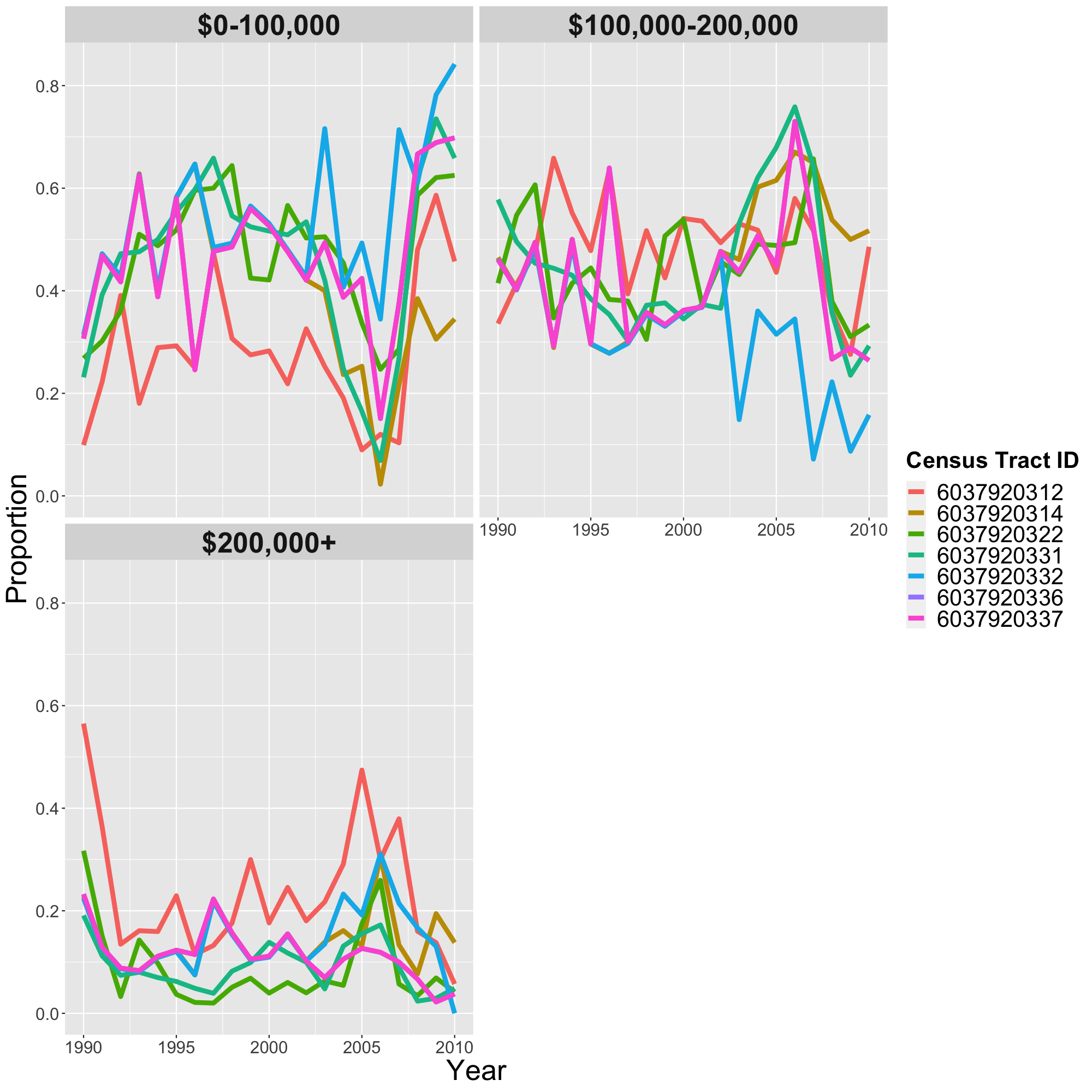}
\label{fig:real_data_neighbor_evo}
\end{subfigure}
\begin{subfigure}{.4\textwidth}
\centering
\includegraphics[width = 0.9\textwidth]{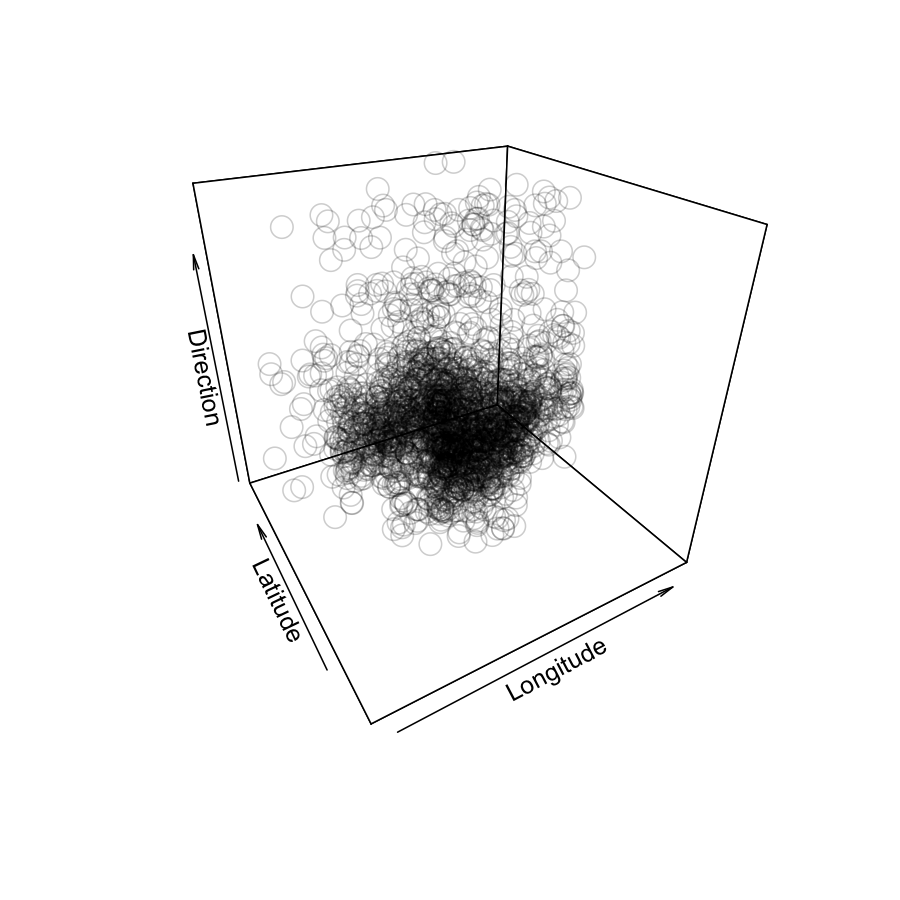}
\label{fig:real_data_Census_tract_dir}
\end{subfigure}
\caption{Plots showing the change in income proportion for Census tract 6037920336 and its neighboring tract on the left and the random direction of the movements for 2003-2004 by Census tract on the right. The center of the Census tract in longitude and latitude is used to represent the Census tract.}
\label{fig:real_data_Census_tract_example}
\end{figure}



\noindent\textbf{Data overview} We now fit the proposed models to the random directions observed in the year to year income proportion changes in the Home Mortgage Disclosure Act (HMDA) data for each Census tract of Los Angeles County. While HMDA data is publicly available, the dataset we worked with is not because it is fused with data purchased from a private company. For this paper, there are three income categories: \$0 to \$100,000, \$100,000 to \$200,000, and greater than \$200,000. We choose to examine these proportions because the number of mortgages recorded in a year differ per Census tracts. Analyzing proportions allows us to potentially ignore the biases that might arise from these differences. We also assume that these proportions observed are the true income proportions for a Census tract. This assumption may not be too unreasonable because people are likely to move into neighborhoods with demographic characteristics similar to their own.


To model how these income proportions change, we also decided to use the proportions as locations because as seen in Figure \ref{fig:real_data_Census_tract_example}, it can be challenging to do so otherwise. If we examine how the income proportions evolve for Census tract 6037920336 and its adjacent tracts, we see that the income proportion for the top category change in a similar manner. However, the proportions for the lower two income categories only mostly shift in a similar manner after 2005. Further, there is one adjacent Census tract that diverges from the other tracts for the second income category during this time period. As a result, it is hard to imagine a model that only uses Census tract information because it might need to borrow strength for one income category, but potentially ignore it for other categories. If we then extract the random direction using the procedure discussed in the next paragraph, Figure \ref{fig:real_data_Census_tract_example} also shows that it might not be helpful to model these directions using Census tracts. Even though the random direction is not uniformly distributed, the pattern appears to be the same across the different Census tracts.


\noindent\textbf{Random direction} Because we have the income proportion for a Census tract for a given year and the next year, we can extract the random direction using the following procedure. Let $x_1$ be the income proportion for one year and $x_2$ be the proportion for the next. Set $\theta_1$ and $\phi_1$ to be the spherical coordinates for $\sqrt{x}$. Then, define $O_p$ to be the following matrix:
\[
\begin{pmatrix}
\cos\theta_1\cos\phi_1 & -\sin\phi_1 & \sin\theta_1\cos\phi_1\\
\cos\theta_1\sin\phi_1 & \cos\phi_1 & \sin\theta_1\sin\phi_1\\
-\sin\theta_1 & 0 & \cos\theta_1\\
\end{pmatrix}
\]
This matrix transforms $(0, 0, 1)$ to $x_1$. It is orthogonal because the first column is the spherical coordinates for $\theta_1 + \frac{\pi}{2}$ and the second column is the spherical coordinates for $\theta_1 = \frac{\pi}{2}$ and $\phi_1 + \frac{\pi}{2}$. Its inverse allows us to define a spherical coordinate system based on $x_1$. If $O_p x'_2 = x_2$ and $\theta'_2$ and $\phi'_2$ represent the spherical coordinates of $x'_2$, then $\phi'_2$ is the random direction and $\theta'_2$ represents how "far" $x_1$ goes in that random direction. 

This random direction has some nice properties.  Even if $x_1$ is zero for any proportions, the random direction can still be extracted and examined. It is also interpretable. In spherical coordinates, changing $\phi'_2$ for a fixed $\theta'_2$ only affects the first two coordinates. Due to the transformation defined in the previous paragraph, a change in the first two coordinates changes how the first two columns of the matrix that transforms $(0, 0, 1)$ to $x_1$ are weighted. To understand this random direction, we need to comprehend what the first two columns represent and how the random direction interacts with these columns. Because the first column represents the spherical coordinates for ($\theta_1 + \frac{\pi}{2}$, $\phi$), the simplest way to understand this column is to compare $\theta_1 = 0$ and $\theta_1 = \frac{\pi}{2}$. As $\theta_1 = 0$ corresponds to $(0, 0, 1)$ and $\theta_1 = (\cos(\phi), \sin(\phi), 0)$ for some $\phi \in [0, 2\pi)$, this first column represents a push away from the third income category. On the other hand, if we use a similar logic to compare $\theta = \frac{\pi}{2}$ and $\phi_1 = 0$ against $\theta = \frac{\pi}{2}$ and $\phi_1 = \frac{\pi}{2}$, then we see that the second column represents a pull toward the second income category. Because the first and second coordinate include $\cos(\phi'_2)$ and $\sin(\phi'_2)$ respectively, we will examine $0$, $\frac{\pi}{2}$, $\pi$, and $\frac{3\pi}{2}$. Note that at $0$ and $\pi$, the first coordinate will be $1$ and $-1$ and the second will be zero by definition. A random direction of $0$ represent a push away from the third income category whereas $\pi$ represents a pull toward. The reverse is true for $\frac{\pi}{2}$ and $\frac{3\pi}{2}$. We can say that a random direction of $\frac{\pi}{2}$ and $\frac{3\pi}{2}$ represent a pull toward and push away from the second income category. It is the opposite because the first column is a push away from a category whereas the second column is a pull toward a category.


We fitted our models to these random directions with one additional pre-processing step. We removed duplicated directions so that each location has at most one observation. This was done because duplicated directions at the same location should happen with probability zero according to our model and they can be easily identified. Indeed, their properties are further discussed in the supplementary material. As a result of this step, the number of observations are reduced from between 2302 to 2356 per year to between 1704 to 2164 per year. However, while we used the same priors that we utilized for simulation, it was still unclear what to set the hyperparameters for the squared exponential kernels. To pick, we used the same hyperparameters as simulation as a baseline. We then compared it against $\sigma = 0.25, 1$, and $2$ and $\omega = 0.05, 0.15, 0.2$, and $0.25$ for \textit{SvM-c} and $\sigma = 0.25, 0.5$, and $2$ and $\omega = 0.05, 0.15, 0.2$, and $0.25$ for \textit{SvM-p}.

\begin{figure}[!tbp]
\centering
\begin{subfigure}{.2\textwidth}
\centering
\includegraphics[width = 1\textwidth]{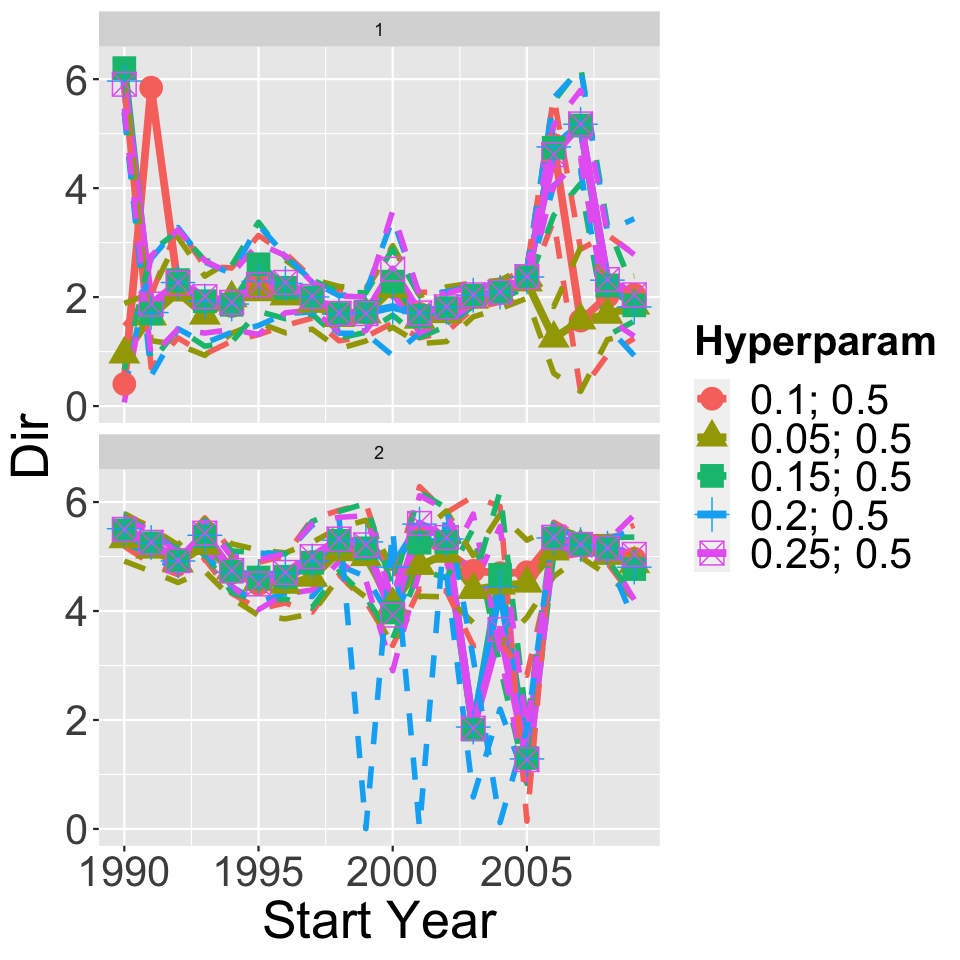}
\end{subfigure}
\begin{subfigure}{.2\textwidth}
\centering
\includegraphics[width = 1\textwidth]{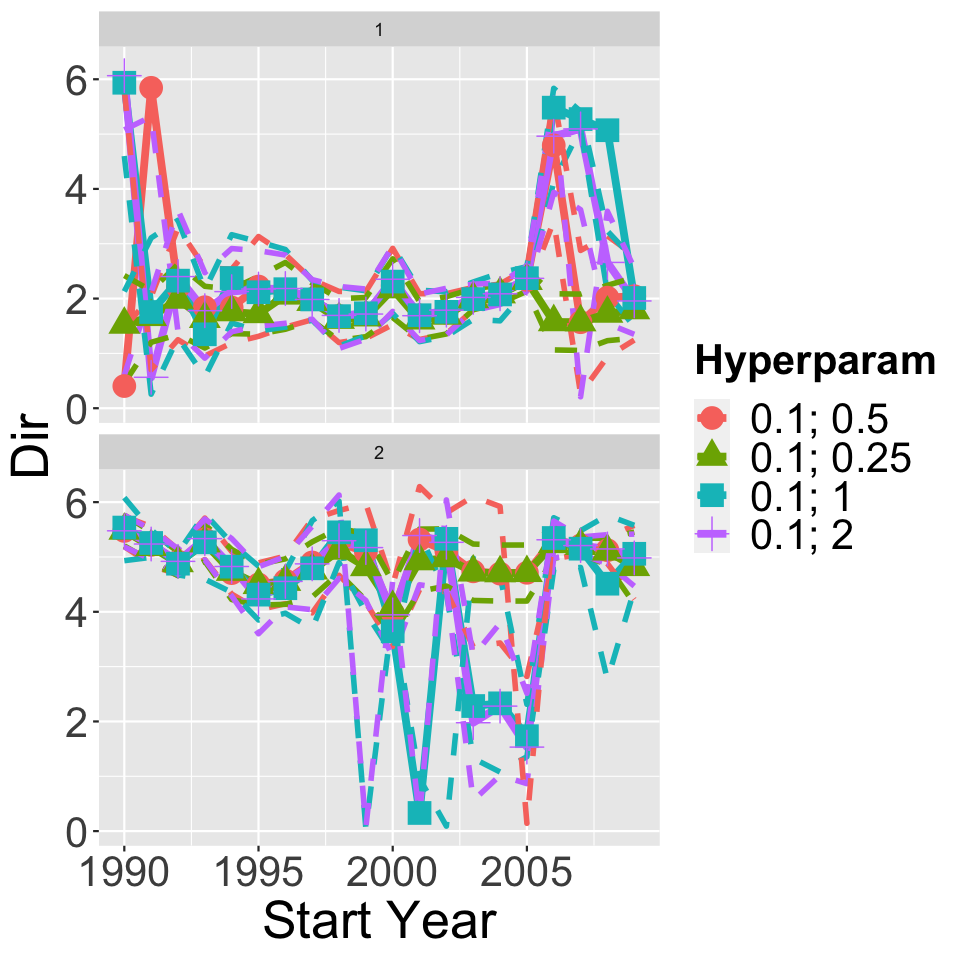}
\end{subfigure}
\begin{subfigure}{.2\textwidth}
\centering
\includegraphics[width = 1\textwidth]{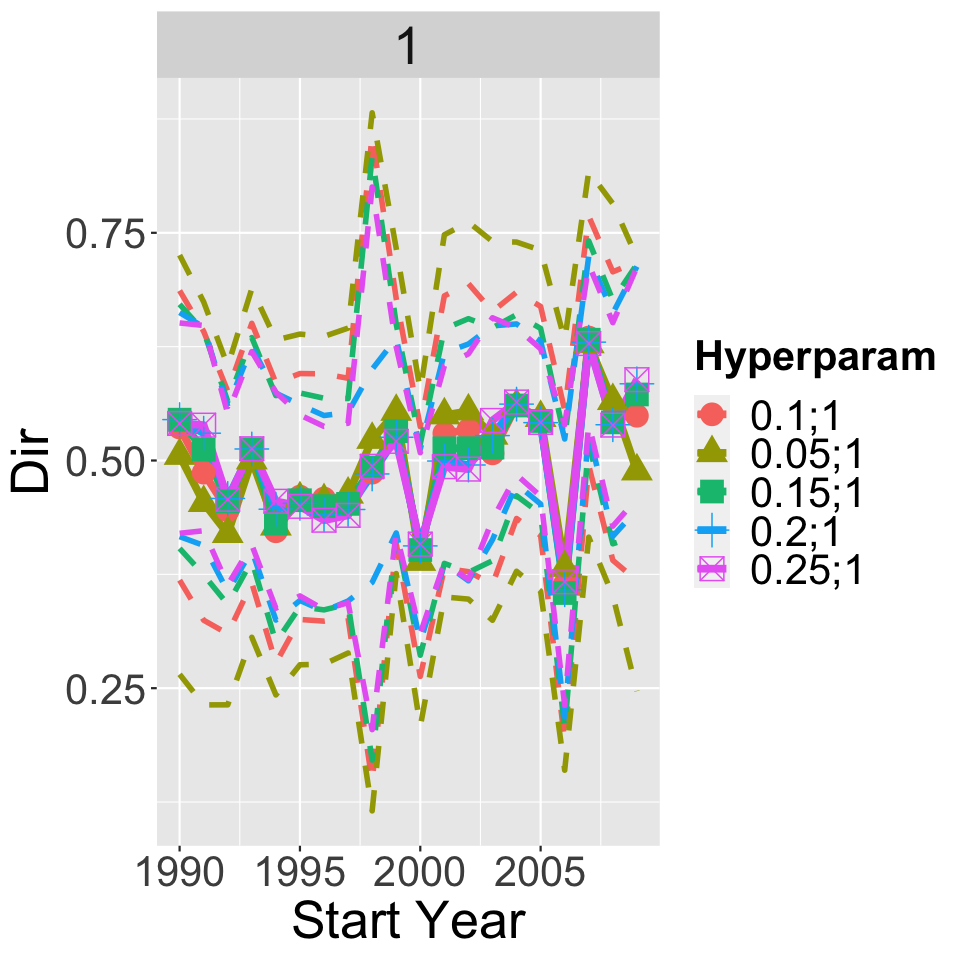}
\end{subfigure}
\begin{subfigure}{.2\textwidth}
\centering
\includegraphics[width = 1\textwidth]{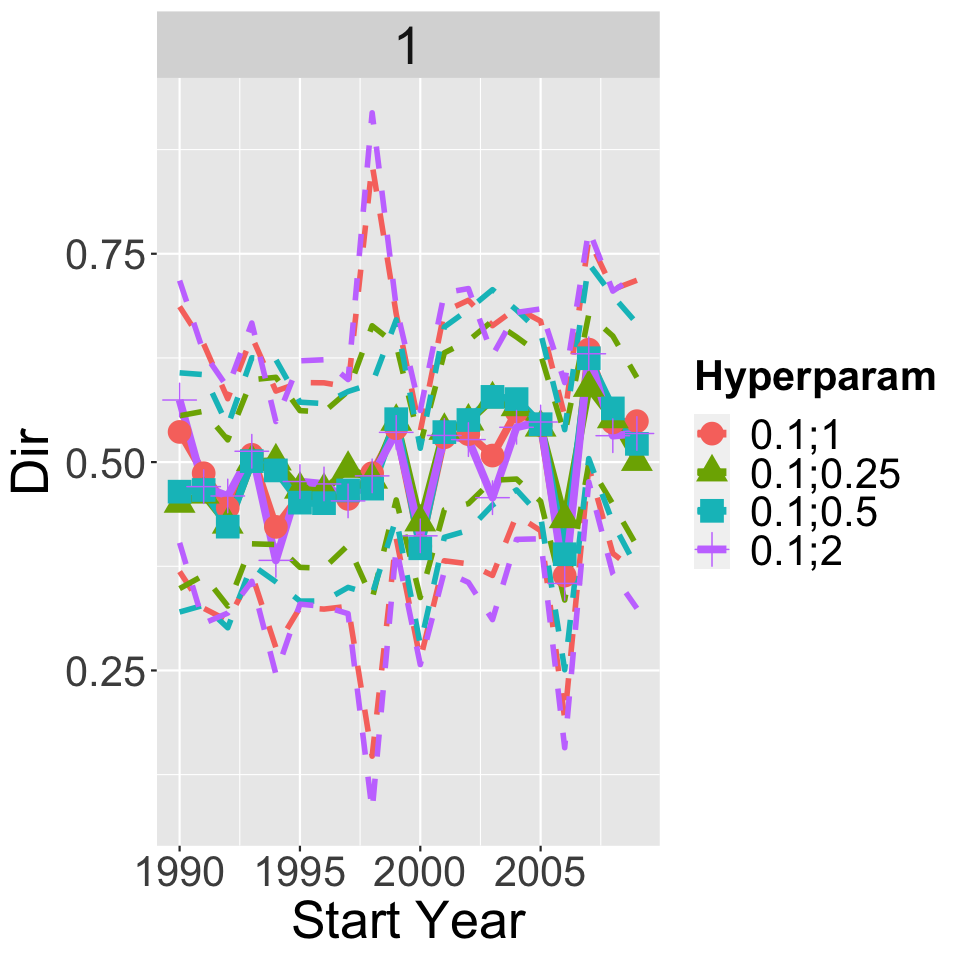}
\end{subfigure}
\caption{Plots showing the mean direction and credible intervals averaged over all locations for various hyperparameters for \textit{SvM-c} compared to $\omega = 0.1$ and $\sigma = 0.5$ on the left and showing the probability and credible intervals averaged over all locations for various hyperparameters for \textit{SvM-p} compared to $\omega = 0.1$ and $\sigma = 1$ on the right. The value for $\omega$ is given first.}
\label{fig:real_data_exp_results}
\end{figure}

\noindent\textbf{Hyperparameter and model selection} Some of our findings from these experiments can be seen in Figure \ref{fig:real_data_exp_results}. If we examine \textit{SvM-p} first, we see that while the overall mixing probability averaged across all locations remain mostly the same, the credible intervals averaged across all locations expands as $\sigma$ increases or $\omega$ decreases. This is to be expected because increasing $\sigma$ allows for a wider range of values for the Gaussian process whereas decreasing $\omega$ means the mixing probability is more sensitive to the observed direction at each location. Next, if we examine \textit{SvM-c}, changing the hyperparameters does not appear to affect the first mean averaged by location between 1995 and 2005. However, there is greater stability in the second mean averaged by location when $\omega$ is changed, but not when $\sigma$ is altered. Indeed, the latter only appears to have stability before 1998. Still, there are differences in the mean surfaces hidden by this statistic. If we focus on the results from 1996 to 1997, we see that while the second mean surface is mostly similar when we change $\sigma$, the first mean surface varies more when $\sigma$ increases. For instance, the surface is mostly flat when $\sigma = 0.25$, but the surface appears to be a series of thinly connected pillars that span zero to $2\pi$ when $\sigma = 2$. This makes sense that neighboring values on the mean surface can vary drastically when $\sigma$ grows. On the other hand, the overall shape of the mean surface in 1996 mostly remains the same as we increase $\omega$ from 0.05 to 0.25. This is again particularly true for the second mean surface. However, the first mean surface becomes smoother as $\omega$ rises. This suggests that for this year, neighboring points are similar in value and the mean surface must reflect this fact.

After running these experiments, we selected the hyperparameters that had the average highest log posterior predictive probability across all years. Here, we take the average with respect to the number of observations. For \textit{SvM-c}, this was $\sigma = 0.5$ and $\omega = 0.2$. The hyperparameters selected for $\textit{SvM-p-2}$ was $\sigma = 2$ and $\omega = 0.1$. We then used these results to explore the choice of kernel. The other kernels were the Matern with three halves and five halves degrees of freedom. The Matern with one half degree of freedom was not compared because previous experiments demonstrated that kernels leading to results that are less smooth did not perform as well and the Matern kernel with one half is such a kernel \citep{RasmussenWilliamsGaussianProcessesMachine2006}. For \textit{SvM-p}, the results remained similar compared to the results from the optimal squared exponential kernel model. Indeed, the posterior predictive probability is similar between these models. Meanwhile, choosing the Matern kernels or the optimal squared exponential kernel for \textit{SvM-c} result in similar mean directions averaged by location. However, if we again look at the mean surface from 1996 to 1997, we notice differences. While the lower mean surface is similar from the Matern kernel with five halves degrees of freedom and squared exponential kernel, the surface is more dispersed for the Matern kernel with three half degrees of freedom. However, for non-lower income neighborhoods. the top mean surface is more scattered for the Matern kernel than for the squared exponential kernel. Interestingly enough, the surface from the Matern kernel with five halves degrees of freedom is more diffuse even though it is smoother than the surface from the Matern kernel with three halves degrees of freedom. One other key difference is that the \textit{SvM-c} models with the Matern kernels performed worse than the same model with the squared exponential kernel. This suggests that while the mean surface for the random direction might be sensitive to the choice of kernel, regions that are likely to move in the same direction are not sensitive.

Finally, we also wanted to compare the results from the square exponential hyperparameter experiments against $\textit{SvM}$ and \textit{SvM-c-3} using the hyperparameters selected for \textit{SvM-c} and \textit{SvM-p-3} with the hyperparameters selected according to $\textit{SvM-p-2}$. Note that for years 1991-1993, 1996-1997, 1998-2000, 2005-2010, we used a strong von Mises prior on the mean component to force identifiability for \textit{SvM-p-3}. Otherwise, we used an uniform prior on the mean component. Further, because of extreme estimates in the concentration parameters, we also ran a version of \textit{SvM-c-3} with the concentration parameter constrained to be between 1e-4 and 20. According to the log posterior predictive probability given in $\eqref{eq:post_pred_prob}$, \textit{SvM-c-2} perform the best for all years except 1991-1994, 2003-2005, 2006-2007, 2008-2010, and potentially 2008-2009. For those years, \textit{SvM-c-3} or the constrained version performs the best. However, the posterior predictive probability for \textit{SvM-c-3} is close to that of \textit{SvM-p-3} for 2006-2007 and of \textit{SvM-c-2} for 2008-2009. This suggests that the average observed change in income proportion is better explained by one of two random direction surfaces that varies depending on the income proportion proportions. Further, it tells us that having two components is sufficient in most cases.
\begin{figure}[!tb]
\centering
\begin{subfigure}{.24\textwidth}
\centering
\includegraphics[width = 1\textwidth]{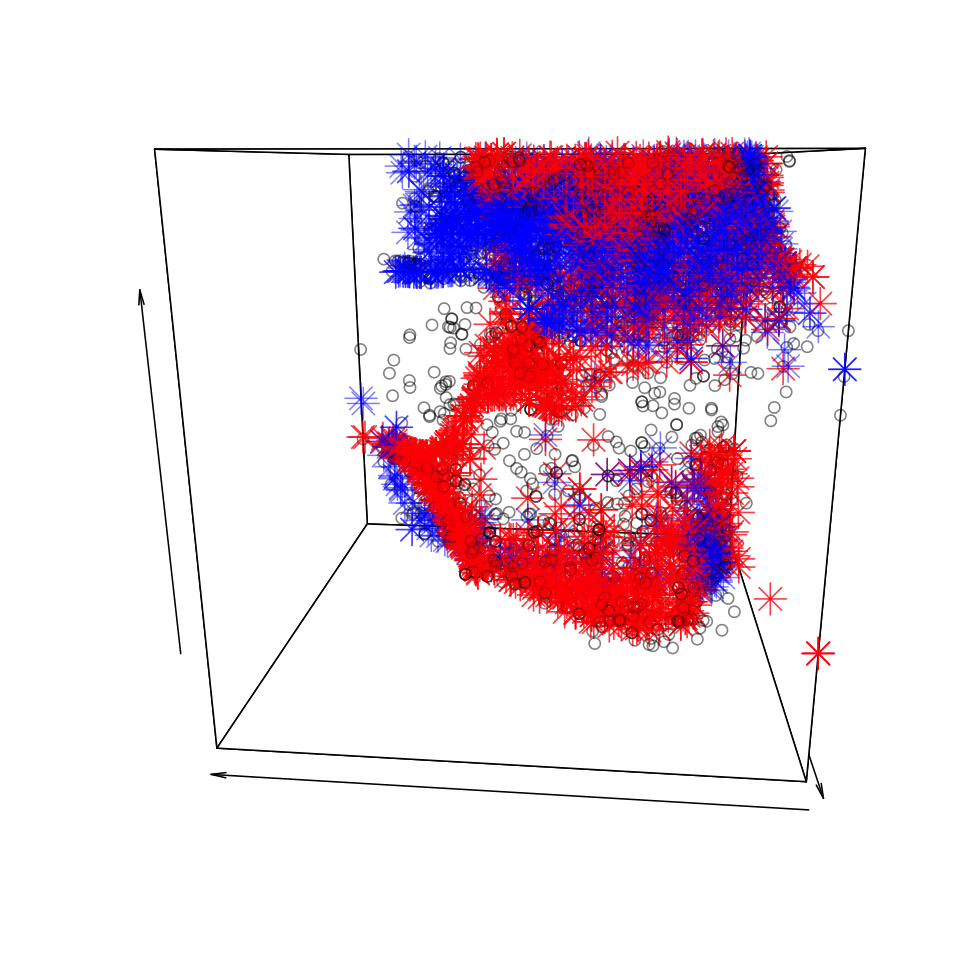}
\caption{1990-1991\\ (\textit{SvM-c-2})}
\label{fig:real_data_fitted_results_1}
\end{subfigure}
\centering
\begin{subfigure}{.24\textwidth}
\centering
\includegraphics[width = 1\textwidth]{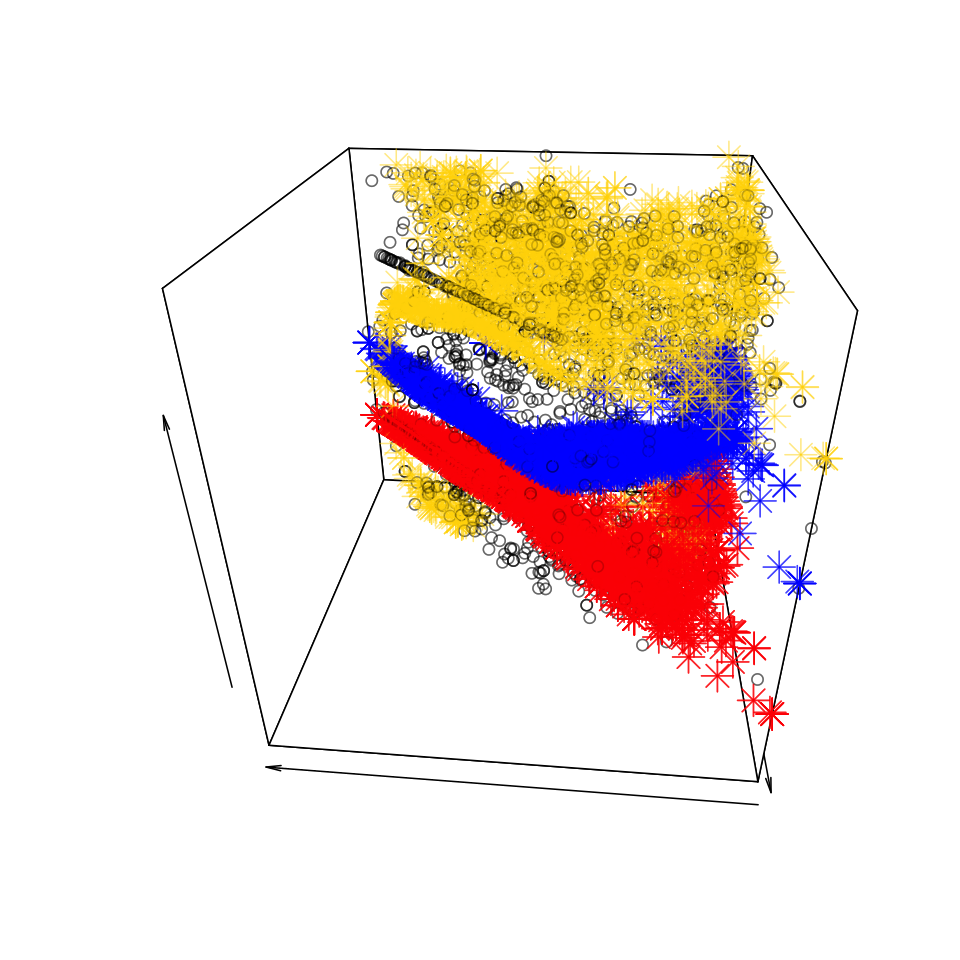}
\caption{1991-1992\\(\textit{SvM-c-3})}
\label{fig:real_data_fitted_results_3}
\end{subfigure}
\centering
\begin{subfigure}{.24\textwidth}
\centering
\includegraphics[width = 1\textwidth]{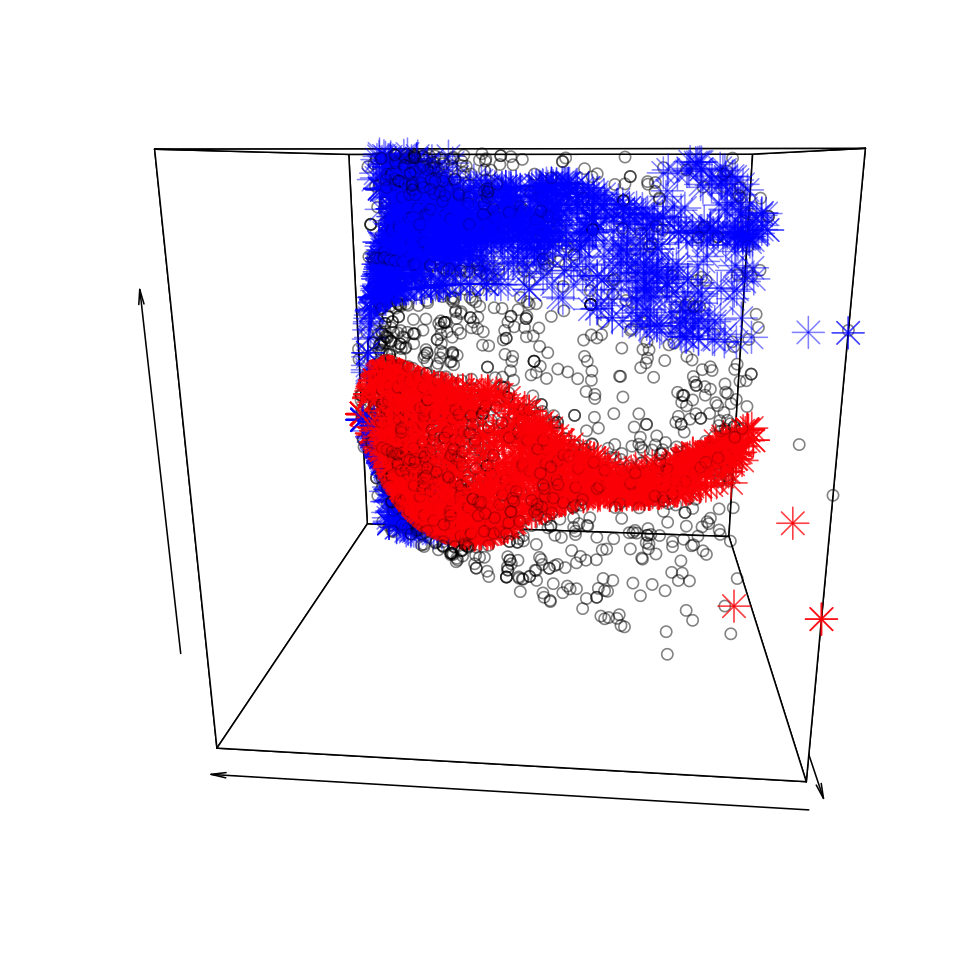}
\caption{1998-1999\\(\textit{SvM-c-2})}
\label{fig:real_data_fitted_results_9}
\end{subfigure}\\
\centering
\begin{subfigure}{.24\textwidth}
\centering
\includegraphics[width = 1\textwidth]{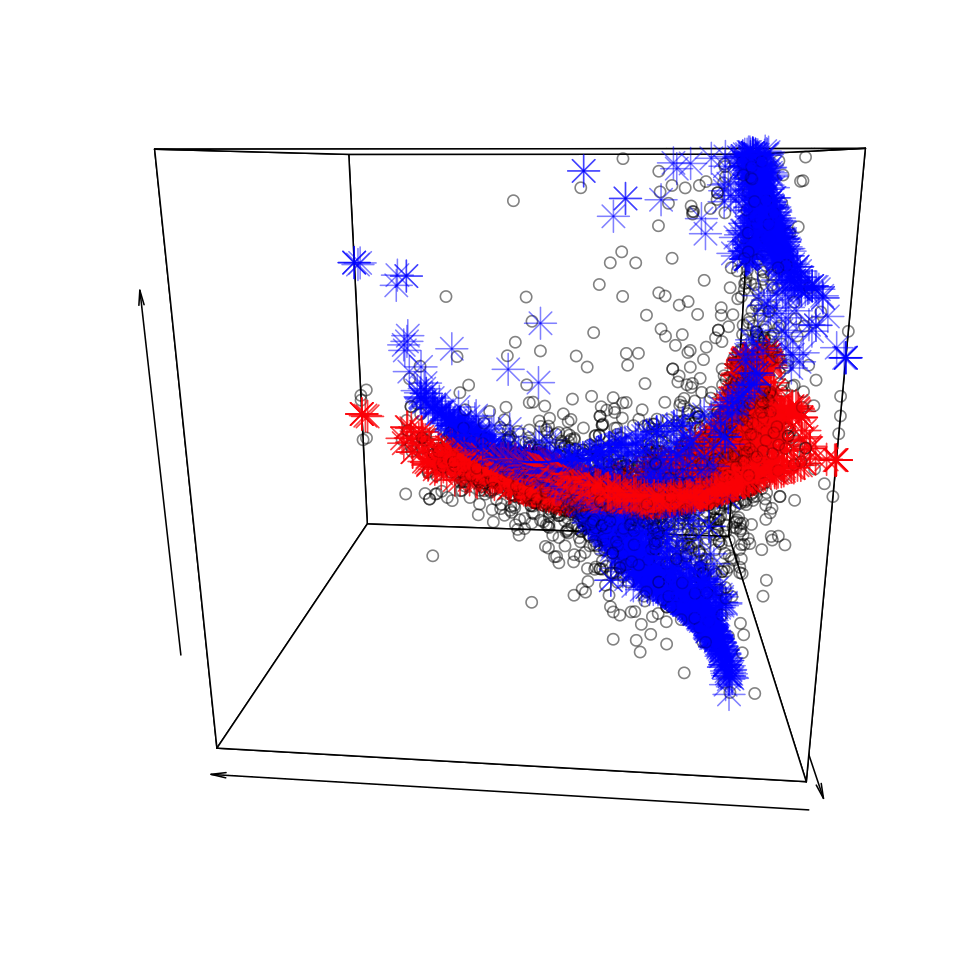}
\caption{2005-2006\\(\textit{SvM-c-2})}
\label{fig:real_data_fitted_results_16}
\end{subfigure}
\centering
\begin{subfigure}{.24\textwidth}
\centering
\includegraphics[width = 1\textwidth]{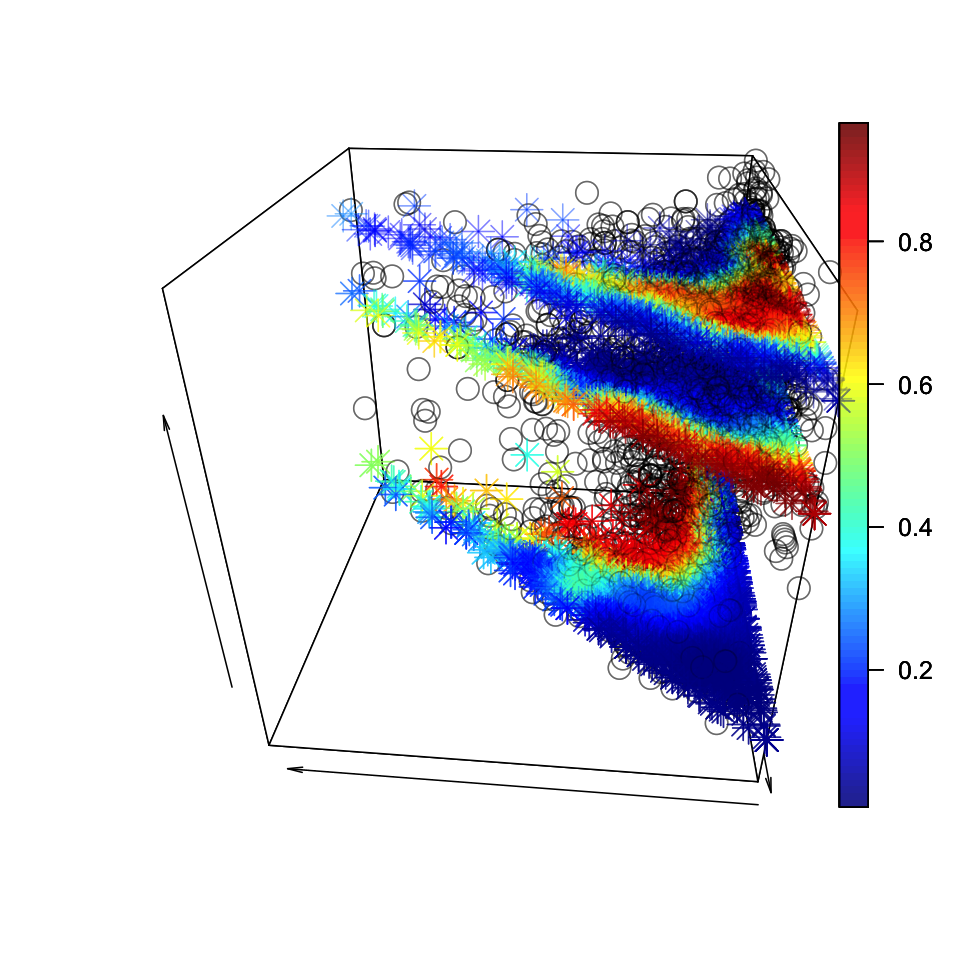}
\caption{2006-2007\\(\textit{SvM-p-3})}
\label{fig:real_data_fitted_results_17_svm_p}
\end{subfigure}
\centering
\begin{subfigure}{.24\textwidth}
\centering
\includegraphics[width = 1\textwidth]{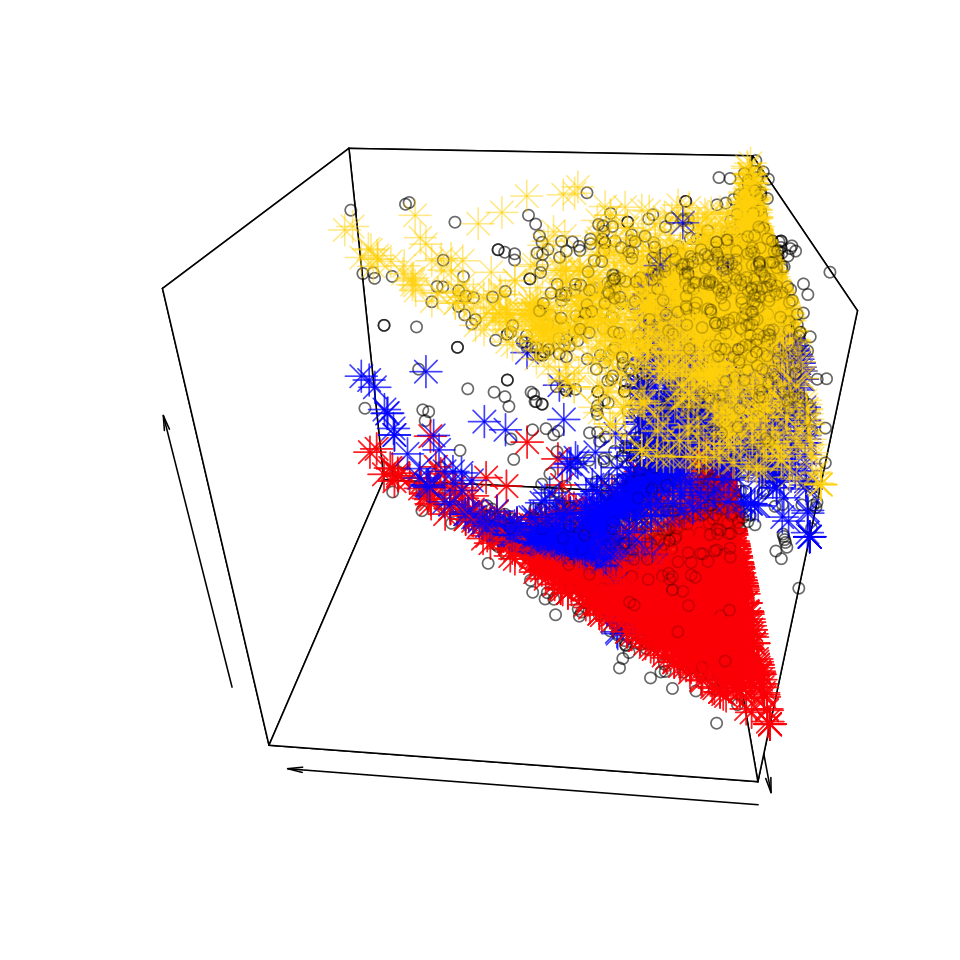}
\caption{2006-2007\\(\textit{SvM-c-3})}
\label{fig:real_data_fitted_results_17_svm_c}
\end{subfigure}
\caption{Plots showing the observed random direction not withheld and the fitted mean surface of the model selected by the posterior predictive log probability in \eqref{eq:post_pred_prob}. The front axis represents the proportion in the first income category and the side axis represents the proportion in the second income category. This is reversed in the alternative view. Regardless, the up-down axis represents the direction. The start of the arrow indicates a value of zero whereas the end indicates a value of $2\pi$ for angles and 1 for proportions.}
\label{fig:real_data_fitted_results_example}
\end{figure}

\noindent\textbf{Findings and interpretation} Based on the Table \ref{table:real_data_results_param_summary} and Figure \ref{fig:real_data_fitted_results_all} in the supplementary material, the fitted mean surface can be grouped into four phases: 1990-1991, 1991-1993, 1993-2003, 2003-2010. These phases correspond to the early 1990 recession, the transition from President George HW Bush to President Bill Clinton, the economic boom in the 1990s, and the recovery from the dot com bubble with the subsequent housing market crash respectively. Representative examples of each phase can be found in Figure \ref{fig:real_data_fitted_results_example}. As an example of the stories these surfaces can tell us, consider 2005 to 2006 and 2006 to 2007. Like the two years before 2005, there is one surface that is an curved, half spiral increasing from $\frac{\pi}{2}$ to around $\pi$ if we follow it from neighborhoods with income proportions largely below \$100,000 to neighborhoods with income proportions largely between \$100,000 and \$200,000 and then to neighborhoods with income proportions largely greater than \$200,000. This suggests that for these years, the income distributions for all neighborhoods are being pulled up a category. However, there is another surface that while similar for neighborhoods of lower income, is centered around zero for neighborhoods with income proportions in the second and third categories. This suggests that there already is a push away from the third income category for these neighborhoods, two years before the housing market crash. Because the posterior predictive probability is similar for \textit{SvM-c} with three components and \textit{SvM-p} with three components, it is interesting to see these two aforementioned mean surfaces re-partition into three. As hinted by the lower half of the non-spiral surface, there is one component at either $\frac{\pi}{3}$ or $\frac{\pi}{4}$. This suggests a medium to semi-strong pull toward the second income category. Then, for the upper half of the non-spiral surface, \textit{SvM-p} takes two components. Interestingly enough, as seen in Figure \ref{fig:real_data_fitted_results_17_svm_p} and \ref{fig:real_data_fitted_results_17_svm_c}, the areas of high probability for \textit{SvM-p}'s two components match the third component of \textit{SvM-c} with two components. These two mean components are centered at roughly $\frac{4\pi}{3}$ and $\frac{5\pi}{3}$, which represent semi-strong pushes away from the second and third income categories. In other words, the data already suggests the potential recession that might happen in the next year. Finally, the spiral mean surface changes as well and adopts characteristics of the other mean surface for the third mean surface for \textit{SvM-c-3} in 2007. As seen in \ref{fig:real_data_fitted_results_17_svm_c}, it matches the first surface for neighborhoods with the lowest income proportions before rising rapidly to meet the other mean surface for neighborhoods with the second highest income proportions. It then comes back down to $\pi$, i.e. a pull toward the second income category, for neighborhoods with the highest income proportion.

\section{Conclusion and Future Directions}
\label{sec:conclusion}
In this paper, we have introduced two types of hierarchical models to draw inference about the random directions for simplex-valued measurements and discuss how these models might be utilized. These approaches creatively use the data's location within the simplex to do so. Because the average random direction across all location or the "mean surface" and the mixing probability are important, the "spatial" information is naturally fed into them. Indeed, \textit{SvM-c} is the model when this information is fed into the former whereas \textit{SvM-p} is the model when this information is fed into the latter. Then, to better understand how to set the hyperparameters of these models, we derived the models' prior circular means, variances, and correlations. We discussed how to fit these models using sampling schemes that respect the geometry of the parameters of interest. For instance, we used elliptical slice sampling to sample for the mean surface in \textit{SvM-c}. Using these findings, we then applied the models to simulated data and then analyzed a data set of income proportions and random directions observed for a set of Census tracts in LA County from 1990 to 2010. There is evidence to suggest that the random direction is associated with the Census tract's income proportions observed in a given year and not the tract's physical location. This means that the data set is a highly relevant real world example to apply our models to. Consequently, it is noteworthy that the patterns our models discovered matches and potentially clarifies real world economic trends during the same time period.

There are several directions worth exploring moving forward. We might extend these models to random directions in dimensions greater than two. This would allow us to analyze the random directions for data that lie in $\Delta^D$ for $D > 2$. The LA County data set has sixteen income categories before we reduced them down to three. Fortunately, as discussed in the paper, both the von Mises distribution and projected normal distribution can be extended to angles in higher dimensions. These higher dimensional distributions can be plugged into our models in place of their lower dimensional analogues. Next, to better sample \textit{SvM-c} or the higher dimensional equivalents, we might combine Multiple Try MCMC with elliptical slice sampling. This could provide a principled approach to consider the entire range of angles during each proposal step for the next mean angles and thus result in better mixing. Finally, we might merge our random direction models with a model for the magnitude in order to model the random movement. As discussed earlier, the two components of the movement were separated to achieve greater modeling flexibility and to better deal with the challenges posed by them. For instance, modeling the magnitudes and directions on the simplex's boundary require special care because the range of valid movements and directions are limited. In addition, the magnitudes in the interior have to be treated carefully. In order for the simplicial constraint to be respected, certain magnitudes in certain directions are not possible.



\bibliographystyle{unsrt} 
\bibliography{my_refs}  

\section*{Acknowledgements}
This research is supported by the NSF Graduate Research Fellowship Program (Grant No. DGE 1256260). Any opinions, findings, and conclusions or recommendations expressed in this material are those of the author and do not necessarily reflect the views of the National Science Foundation. We also want to thank Professor Elizabeth Bruch for introducing us to the data set and for her discussions and Lydia Wileden for preparing the data. XN gratefully acknowledges support from NSF grants DMS-1351362, CNS-1409303 and DMS-2015361.


\newpage
\section{Additional Models}
\subsection{Independent Random Direction Models} 
The most elementary model based on the von Mises distribution is the independent von Mises or \textit{iV} model. For this model, the observations are assumed to be identically and independently distributed according to one von Mises distribution with mean $m$ and concentration parameter $\rho$. A von Mises prior on the interval $[0, 2\pi)$ with mean $u \in [0, 2\pi)$ and concentration parameter $c \in \mathbbm{R}^+$ is placed on $m$. As discussed earlier in Section \ref{ssection:vM}, setting $c = 0$ is equivalent to putting a uniform prior on $m$. Meanwhile, the concentration parameter $\rho$ has a Gamma prior with parameters $a \in \mathbbm{R}^+$ and $b \in \mathbbm{R}^+$. In summary, the \textit{iV} model is the following.
\begin{equation}
    \begin{aligned}
    m &\sim \vM{\cdot}{u}{c}, &\\
    \rho &\sim \Gamma(\cdot \mid a, b), &\\
    y_\ell \mid m, \rho &\stackrel{iid}{\sim} \vM{\cdot}{m}{\rho}, & \ell = 1, 2, \ldots, N.\\
    \end{aligned}
\label{model:iV}
\end{equation}

If there is heterogeneity in the observations, an extension of the model is the \textit{Independent von Mises Mixture} or \textit{iVM} model.
\begin{equation}
    \begin{aligned}
    m_k &\sim \vM{\cdot}{u_k}{c_k}, & k = 1, 2, \ldots, K\\
    \rho_k &\sim \Gamma(\cdot \mid a_k, b_k), & k = 1, 2, \ldots, K\\
    \zeta_\ell | \lambda_1, \lambda_2, \ldots, \lambda_K & \stackrel{iid}{\sim} \textrm{Cat}(\cdot|\lambda_1, \lambda_2, \dots, \lambda_K), & \\
    y_\ell \mid \zeta_\ell = k, m_k, \rho_k & \sim \vM{\cdot}{m_k}{\rho_k}, & \ell = 1, 2, \ldots, N.\\
    \end{aligned}
\label{model:iVM}
\end{equation}
An observation can now be distributed according to one of $K$ von Mises distribution in the interval of $[0, 2\pi)$. Each distribution has its own mean $m_k$ and concentration parameter $\rho_k$. Moreover, each distributions' parameters are given different priors.

\subsection{Homogeneous Spatial Random Direction Model} 
We proceed to modeling spatial patterns of random directions. For homogeneous spatial random direction patterns, we incorporate spatial information through the components, i.e., the parameters of the von Mises distribution. In particular, we link a Gaussian process to the mean parameters because it is a straightforward way to collate nearby observations and borrow strength. This also leads to fairly interpretable draws from the Gaussian process. Due to the von Mises distribution being concentrated at its mean, results from these models can describe the preferred random direction for a given location.

We introduce a model that applies this approach. This is the \textit{Spatially varying von Mises distribution} model or \textit{SvM}. Each observation $y_\ell$ has a mean parameter, $m_\ell \in [0, 2\pi)$, and a concentration parameter, $\rho_\ell \in \mathbbm{R}^+$. We use the modified arctan function defined in \eqref{fct:arctan_star} to element-wise transform two draws from the Gaussian process, $\v{Z_1}, \v{Z_2} \in \mathbbm{R}^N$, into the mean parameters, $\v{m} \in [0, 2\pi)^N$. The positive-valued concentration parameter, $\rho_\ell$, is transformed via exponentiation from a real-valued $\varphi_\ell$, which is then endowed with a normal prior.  
In summary, \textit{SvM} is the following.
\begin{align*}
    \v{z_1} &\sim \textrm{GP}(\cdot \mid \mu_1, \Sigma),\\
    \v{z_2} &\sim \textrm{GP}(\cdot \mid \mu_2, \Sigma),\\
    \v{m} &= \textrm{arctan}^*(\v{z_1}, \v{z_2}), \stepcounter{equation}\tag{\theequation}\label{model:SvM}\\
    \v{\varphi} &\stackrel{iid}{\sim} \textrm{N}(\cdot \mid \nu, \varsigma^2),\\
    \v{\rho} &= \exp{\v{\varphi}},\\
    y_\ell \mid m_\ell, \rho_\ell &\sim \vM{\cdot}{m_\ell}{\rho_\ell}, & \ell = 1, 2, \ldots, N.
\end{align*}

\section{Model Properties}
\subsection{SvM model}
We begin by discussing how we derive Lemma \ref{lemma:SvM_model_prop_a}. After using the law of total expectation and trigonometric identities, we see that calculating the properties of $Y_{\ell}$'s are reduced to that of the corresponding $m_{\ell}$. Because $(z_{1. \ell}, z_{2, \ell}) \sim \mathcal{N}((\mu_1, \mu_2), \sigma^2\mathbbm{I}_2)$ and $m_{\ell} = \arctan^{*}(z_{1. \ell}, z_{2, \ell})$, $m_{\ell}$ is distributed according to $f_N((\mu_1, \mu_2), \sigma^2 \mathbbm{I}_2)$ \citep{MardiaJuppDirectionalStatistics2010}. We can extend the techniques in \cite{WangGelfandDirectionalDataAnalysis2013} to demonstrate that $m_{\ell}$ is equivalently distributed according to $f_N(\frac{1}{\sigma^2}(\mu_1, \mu_2), \mathbbm{I}_2)$. With this insight, we can use the discussion in Section \ref{ssection:pn2} to compute the circular mean and variance. Then, we have the following lemmas.

\begin{lemma}
Let $Z_1, Z_2 \sim \mathcal{N}((\mu_0\cos(\alpha_\mu), \mu_0\sin(\alpha_\mu)), \sigma^2\mathbbm{I})_2$. Let $(r, m)$ be random variables, $r \in (0, \infty)$, $m \in [0, 2\pi)$, such that $Z_1 = r\cos(m)$ and $Z_2 = r\sin(m)$. In addition, set $a = \mu_0\cos(m - \alpha_\mu)$. Then,
\begin{align*}
PN_2(m \mid (\mu_0\cos(\alpha_\mu), \mu_0\sin(\alpha_\mu)), \sigma^2\mathbbm{I}_2) = 
\phi\left(\frac{\mu_0\sin(m - \alpha_\mu)}{\sigma}\right)\left(\phi\left(\frac{a}{\sigma}\right) + \frac{a}{\sigma}\Phi\left(\frac{a}{\sigma}\right)\right).
\end{align*}
\label{lemma:PN2_distr}
\end{lemma}
\begin{proof}
Following the derivation in Appendix A of \citep{WangGelfandDirectionalDataAnalysis2013}, we have that
\begin{align*}
& PN_2(m \mid (\mu_0\cos(\alpha_\mu), \mu_0\sin(\alpha_\mu)), \sigma^2\mathbbm{I}_2)\\ 
=& \int_0^{\infty} \frac{1}{2\pi\sigma^2}\exp(-\frac{1}{2\sigma^2}((z_1 - \mu_0\cos(\alpha_\mu))^2 + (z_2 - \mu_0\cos(\alpha_\mu))^2)) dr\\
=& \frac{1}{\sqrt{2\pi\sigma^2}}\exp(-\frac{\mu_0^2}{2\sigma^2} \sin^2(m - \alpha_\mu))\int_0^{\infty}\frac{1}{\sqrt{2\pi\sigma^2}} r\exp(-\frac{1}{2\sigma^2}(r - \mu_0\cos(m - \alpha_\mu))^2)dr.
\end{align*}

Applying the change of variables formula and still following the derivation in Appendix A, we have that
\begin{align*}
& \int_0^{\infty}\frac{1}{\sqrt{2\pi\sigma^2}} r\exp(-\frac{1}{2\sigma^2}(r - a)^2)dr\\
=& \int_0^{\infty}\frac{1}{\sqrt{2\pi\sigma^2}} (r - a)\exp(-\frac{1}{2\sigma^2}(r - a)^2)dr + \int_0^{\infty}\frac{1}{\sqrt{2\pi\sigma^2}} a\exp(-\frac{1}{2\sigma^2}(r - a)^2)dr\\
=& \sigma\phi\left(\frac{a}{\sigma}\right) + a\Phi\left(\frac{a}{\sigma}\right).
\end{align*}
Set $a = \mu_0\cos(m - \alpha_\mu)$, then
\begin{align*}
PN_2(m \mid (\mu_0\cos(\alpha_\mu), \mu_0\sin(\alpha_\mu)), \sigma^2\mathbbm{I}_2) = 
\phi\left(\frac{\mu_0\sin(m - \alpha_\mu)}{\sigma}\right)\left(\phi\left(\frac{a}{\sigma}\right) + \frac{a}{\sigma}\Phi\left(\frac{a}{\sigma}\right)\right).
\end{align*}
\end{proof}

\begin{lemma}
If $Y_{\ell}$ and $Y_{\ell'}$ are generated according to \textit{SvM} outlined in \eqref{model:SvM} with the random variables associated with them labeled accordingly and $(Z_1, Z_2) \sim \mathcal{N}((\mu_1, \mu_2), \sigma^2\mathbbm{I}_2)$ with $\mu_1 = \mu_0\cos(\alpha_\mu)$ and $\mu_2 = \mu_0\cos(\alpha_\mu)$, $\mu_0 \in \mathbbm{R}^+$, $\alpha_\mu \in [0, 2\pi)$, then
\begin{align}
    \E{Y_{\ell}} &= \alpha_\mu,
    \label{lemma:SvM_e_a}\\
    \Var{Y_{\ell}} &= 1 - \frac{I_1(\rho)}{I_0(\rho)}\left(\frac{\pi\beta}{2}\right)^{\frac{1}{2}}\exp(-\beta)(I_0(\beta) + I_1(\beta)) \qquad \beta = \frac{\mu_0^2}{4\sigma^2},
    \label{lemma:SvM_var_a}\\
    \textrm{Corr}(Y_\ell, Y_{\ell'}) &= \frac{\left(\frac{I_1(\rho)}{I_0(\rho)}\right)^2(\E{\cos(m_\ell - m_{\ell'})} - \E{\cos(m_\ell + m_{\ell'} - 2\alpha_\mu)})}{\sqrt{\left(1 - \frac{I_2(\rho)}{I_0(\rho)}\E{\cos(2(m_\ell - \alpha_\mu))}\right)\left(1 - \frac{I_2(\rho)}{I_0(\rho)}\E{\cos(2(m_{\ell'} - \alpha_\mu))}\right)}}.
    \label{lemma:SvM_corr_a}
\end{align}
\label{lemma:SvM_model_prop_a}
\end{lemma}
\begin{proof}
When deriving the circular mean and variance, we'll drop the $\ell$ subscript out of convenience. For the circular mean, i.e. the quantity in \ref{lemma:SvM_var_a}, by the law of iterated expectations, 
\begin{align*}
\E{e^{iY}} &= \E{\E{e^{iY} \mid m}}
= \frac{I_1(\rho)}{I_0(\rho)}\E{e^{im}}.
\end{align*}
Because 
\begin{multline*}
    PN_2(\alpha_u + \epsilon \mid (\mu_0\cos(\alpha_\mu), \mu_0\sin(\alpha_\mu)), \sigma^2\mathbbm{I}_2) \\
    = PN_2(\alpha_u - \epsilon \mid (\mu_0\cos(\alpha_\mu), \mu_0\sin(\alpha_\mu)), \sigma^2\mathbbm{I}_2)
\end{multline*} 
by Lemma \ref{lemma:PN2_distr}, $\E{\sin(m - \alpha)} = 0$. Thus, following the argument in Appendix A in \cite{WangGelfandDirectionalDataAnalysis2013},
\[
\arctan^*\left(\frac{\E{\sin(m)}}{\E{\cos(m)}}\right) = \arctan^*\left(\frac{\mu_0\sin(\alpha_\mu)}{\mu_0\cos(\alpha_\mu)}\right).
\]
As a result, the mean angle must be $\alpha_\mu$ by definition.

For the circular variance, i.e. the quantity in \ref{lemma:SvM_var_a}, we point out that if if $m$ is the mean for $Y$, $\E{\sin(Y - m)} = 0$ \cite{MardiaJuppDirectionalStatistics2010}. Then, by the law of iterated expectations, 
\begin{align*}
\E{\cos(Y - \alpha_\mu)} &= \E{\E{\cos((Y - m) + (m - \alpha_\mu))}}\\
&= \E{\E{\cos(Y - m)\cos(m - \alpha_\mu) - \sin(Y - m)\sin(m - \alpha_\mu)}}\\
&= \E{\cos(m - \alpha_\mu)\E{\cos(Y - m)} - \sin(m - \alpha_\mu)\E{\sin(Y - m)}}\\
&= \frac{I_1(\rho)}{I_0(\rho)}\E{\cos(m - \alpha_\mu)}.\\
\end{align*}
For $\alpha_\mu = 0$ and $\sigma = 1$, we have that $\E{\cos(m)} = (\frac{\pi\beta}{2})^{\frac{1}{2}}\exp(-\beta)(I_0(\beta) + I_1(\beta))$ for $\beta = \frac{\mu_0^2}{4}$ \cite{}. In other words, we have that for $a = \mu_0\cos(m)$, then
\[
\int_0^{2\pi} \cos(m) \phi(\mu_0\sin(m))(\phi(a) + a\Phi(a)) = \left(\frac{\pi\beta}{2}\right)^{\frac{1}{2}}\exp(-\beta)(I_0(\beta) + I_1(\beta)).
\]
By change of variable, $\E{\cos(m - \alpha_\mu)} = (\frac{\pi\beta}{2})^{\frac{1}{2}}\exp(-\beta)(I_0(\beta) + I_1(\beta))$ when $\alpha_\mu \neq 0$ and $\sigma = 1$. Finally, when $\sigma \neq 1$, we can define $\widetilde{\mu}_0 = \frac{\mu_0}{\sigma}$ based on Lemma \ref{lemma:PN2_distr}. Hence, $\E{\cos(m - \alpha_\mu)} = (\frac{\pi\widetilde{\beta}}{2})^{\frac{1}{2}}\exp(-\widetilde{\beta})(I_0(\widetilde{\beta}) + I_1(\widetilde{\beta}))$ with $\widetilde{\beta} = \frac{\mu_0^2}{4\sigma^2}$.

By definition, the circular correlation we are interested in calculating is 
\[
\frac{\E{\sin(Y_\ell - \alpha_\mu)\sin(Y_{\ell'} - \alpha_\mu)}}{\sqrt{\E{\sin^2(Y_\ell - \alpha_\mu)}\E{\sin^2(Y_{\ell'} - \alpha_\mu)}}}.
\]
Thus, we need to compute $\E{\sin(Y_\ell - \alpha_\mu)\sin(Y_{\ell'} - \alpha_\mu)}$ and without loss of generality, $\E{\sin^2(Y_\ell - \alpha_\mu)}$. To compute these quantities, we will use trigonometric identities and the following fact. Because the von Mises distribution is unimodal and symmetric, it is straightforward to show that $\E{\sin(2(y_1 - m_\ell))} = 0$ as well.

Using the half angle identity,
\[
\E{\sin^2(y_1 - \alpha_\mu)} = \frac{1}{2}(1 - \E{\cos(2(y_1 - \alpha_\mu))}).
\]
Based on what we discussed earlier and the definition of the Bessel function,
\begin{equation*}
    \begin{aligned}
        \E{\cos(2(Y_\ell - \alpha_\mu))} &= \E{\E{\cos(2(Y_\ell - m_\ell) + 2(m_\ell - \alpha_\mu)) \mid m_\ell}}\\
        &= \E{\E{\cos(2(Y_\ell - m_\ell))\cos(2(m_\ell - \alpha_\mu)) \mid m_\ell}}\\
        &= \frac{I_2(\rho)}{I_0(\rho)}\E{cos(2(m_\ell - \alpha_\mu))}.
    \end{aligned}
\end{equation*}
We can repeat this argument to compute $\E{\sin^2(Y_{\ell'} - \alpha_\mu)}$.

Meanwhile, we can calculate $\E{\sin(Y_\ell - \alpha_\mu)\sin(Y_{\ell'} - \alpha_\mu)}$ using the fact that
\[
\E{\sin(Y_\ell - \alpha_\mu)\sin(Y_{\ell'} - \alpha_\mu)} = 
\frac{1}{2}\E{\cos(Y_\ell - Y_{\ell'}) - \cos(Y_\ell + Y_{\ell'} - 2\alpha_\mu)}.
\]
For the former, we can use conditional independence to get that
\begin{equation*}
    \begin{aligned}
        \E{\cos(Y_\ell - Y_{\ell'})} &= \E{\E{\cos(Y_\ell - m_\ell + m_{\ell'} - Y_{\ell'} + m_\ell - m_{\ell'}) \mid m_\ell, m_{\ell'}}}\\
        &= \E{\E{\cos(Y_\ell - m_\ell)\cos(m_{\ell'} - Y_{\ell'})\cos(m_\ell - m_{\ell'}) \mid m_\ell, m_{\ell'}}}\\
        &= \left(\frac{I_1(\rho)}{I_0(\rho)}\right)^2\E{\cos(m_\ell - m_{\ell'})}.\\
    \end{aligned}
\end{equation*}
A similar argument holds for $\E{\cos(Y_\ell + Y_{\ell'} - 2\alpha_\mu)}$ except that
\[
\E{\cos(Y_\ell + Y_{\ell'} - 2\alpha_\mu)} = \left(\frac{I_1(\rho)}{I_0(\rho)}\right)^2\E{\cos(m_\ell + m_{\ell'} - 2\alpha_\mu)}.
\]
\end{proof}

Note that it becomes challenging to explicitly compute the circular correlation. Given the projected normal distribution for $m$, it is difficult to extend our previous calculations for $\cos(2(m - \alpha))$ because the density involves $\cos(m - \alpha)$ and $\sin(m - \alpha)$, leading to a mismatch. Further, it is tricky to derive the projected normal distribution for $m_\ell$ and $m_{\ell'}$ because of the correlation between $Z_{\ell, 1}$ and $Z_{\ell', 1}$ and the correlation between $Z_{\ell, 2}$ and $Z_{\ell', 2}$. This makes it hard to integrate out $r_\ell$ and $r_{\ell'}$ if $Z_{\ell, 1} = r_\ell\cos(m_\ell)$, $Z_{\ell, 2} = r_\ell\sin(m_\ell)$, $Z_{\ell', 1} = r_{\ell'}\cos(m_{\ell'})$, and $Z_{\ell', 2} = r_{\ell'}\sin(m_{\ell'})$ for $r_\ell, r_{\ell'} \in \mathbbm{R}^+$. Fortunately, it is straightforward to compute these values by simulation. Wang and Gelfand briefly explored this in their paper on projected Gaussian processes \citep{WangGelfandModelingSpaceSpaceTime2014}.

\subsection{SvM-p model}
\label{ssection:svm-p_lemma_computations}
\begin{lemma}
If $Y_{\ell}$ and $Y_{\ell'}$ are generated according to \textit{SvM-p} specified in \eqref{model:SvM-p} with the random variables associated with them labeled accordingly, $\arctan^*$ is defined in \eqref{fct:arctan_star}, $\alpha = \E{Y_{\ell}}$, which is given in \eqref{lemma:SvM-p_e}, and $s(\zeta_\ell) = 1 - \sum\limits_k P(\zeta_\ell = k) \frac{I_2(\rho_k)}{I_0(\rho_k)}\cos(2(m_k - \alpha))$, 
\begin{align}
    \E{Y_{\ell}} &= \arctan^*\left(\frac{\sum\limits_{k} P(\zeta = k)\frac{I_1(\rho_k)}{I_0(\rho_k)}\sin(m_k)}{\sum\limits_{k} P(\zeta = k)\frac{I_1(\rho_k)}{I_0(\rho_k)}\cos(m_k)}\right),
    \label{lemma:SvM-p_e_a}\\
    \Var{Y_\ell} &= 1 - \sum\limits_k P(\zeta = k)\frac{I_1(\rho_k)}{I_0(\rho_k)}\cos(m_k - \alpha),
    \label{lemma:SvM-p_var_a}\\
    \textrm{Corr}(Y_\ell, Y_{\ell'}) &= \frac{1}{\sqrt{s(\zeta_\ell) s(\zeta_{\ell'})}} \sum\limits_{k'}\sum\limits_{k} P(\zeta_\ell = k, \zeta_{\ell'} = k')\frac{I_1(\rho_k)}{I_0(\rho_k)}\frac{I_1(\rho_{k'})}{I_0(\rho_{k'})}\nonumber\\
    & \qquad\qquad\qquad\qquad\qquad\quad(\cos(m_k - m_{k'}) - \cos(m_k + m_{k'} - 2 \alpha)).
    \label{lemma:SvM-p_corr_a}
\end{align}
\label{lemma:SvM-p_model_prop_a}
\end{lemma}
\begin{proof}
Using the law of iterated expectations, we have that
\begin{equation*}
    \begin{aligned}
    \E{e^{iY_\ell}} &= \E{\E{e^{iY_\ell} \mid \zeta_\ell}}\\
    &= \sum\limits_{k} \frac{I_1(\rho_k)}{I_0(\rho_k)}e^{im_k} P(\zeta_\ell = k)\\
    \end{aligned}
\end{equation*}
By converting $e^{im_1}$ and $e^{im_1}$ to their respective sin and cos formulation and use a modified arctan function to get the angle between 0 and 2$\pi$, we get the result in \ref{lemma:SvM-p_e_a}.

From the definition of circular variance, we have that 
\begin{equation*}
    \begin{aligned}
    \textrm{Var}(Y_\ell) &= 1 - \E{\cos(Y_\ell - \alpha)}\\
    &= 1 - \E{\E{\cos(Y_\ell - \alpha) \mid \zeta_\ell}}\\
    &= 1 - \sum\limits_k (\E{\cos((Y_\ell - m_k) + (m_k - \alpha)) \mid \zeta_\ell = k} P(\zeta_\ell = k)\\
    &= 1 - \sum\limits_k P(\zeta_\ell = k)\frac{I_1(\rho_k)}{I_0(\rho_k)}\cos(m_k - \alpha).\\
    \end{aligned}
\end{equation*}

Finally, for the circular correlation, we first derive $\E{\sin^2(y_\ell - \alpha)}$ without loss of generality. As before, we need to compute $\E{cos(2(y_\ell - \alpha))}$. Breaking it into cases as before, we have that 
\begin{equation*}
    \begin{aligned}
        \E{\cos(2(y_\ell - \alpha))} &= \sum\limits_k P(\zeta_\ell = k)\E{\cos(2((y_\ell - m_k) + (m_k - \alpha))) \mid \zeta_\ell = k}\\
        &= \sum\limits_k P(\zeta_\ell = k)\E{cos(2((y_\ell - m_k)))\cos(2(m_k - \alpha)) \mid \zeta_\ell = k}\\
        &= \sum\limits_k P(\zeta_\ell = k)\frac{I_2(\rho_k)}{I_0(\rho_k)}\cos(2(m_k - \alpha)).
    \end{aligned}
\end{equation*}
We can repeat this calculation for $\E{\sin^2(y_{\ell'} - \alpha)}$.

For the numerator, we have that
\begin{equation*}
    \begin{aligned}
        \E{\sin(y_\ell - \alpha)\sin(y_{\ell'} - \alpha)} &= \frac{1}{2}(\E{\cos(y_\ell - \alpha - (y_{\ell'} - \alpha))) - \cos(y_\ell - \alpha + y_{\ell'} - \alpha)})\\
        &= \frac{1}{2}(\E{\cos(y_\ell - y_{\ell'})) - \cos(y_\ell + y_{\ell'} - 2\alpha)}).
    \end{aligned}
\end{equation*}
Let us derive each part separately. For $\E{\cos(y_\ell - y_{\ell'})}$, we can again use the same tricks.
\begin{equation*}
    \begin{aligned}
        & \E{\cos(y_\ell - y_{\ell'})}\\ 
        =& \E{\E{\cos(y_\ell - y_{\ell'}) \mid \zeta_\ell, \zeta_{\ell'}}}\\
        =& \sum\limits_k \sum\limits_{k'} P(\zeta_\ell = k, \zeta_{\ell'} = k')\E{\cos((y_\ell - m_k) + (m_k - y_{\ell'})) \mid \zeta_\ell = k, \zeta_{\ell'} = k'}\\
        =& \sum\limits_k \sum\limits_{k'} P(\zeta_\ell = k, \zeta_{\ell'} = k')\E{\cos(y_\ell - m_k)\cos(y_{\ell'} - m_k) \mid \zeta_\ell = k, \zeta_{\ell'} = k'}\\
        =& \sum\limits_k \sum\limits_{k'} P(\zeta_\ell = k, \zeta_{\ell'} = k')\frac{I_1(\rho_k)}{I_0(\rho_k)}\E{\cos(y_{\ell'} - m_{k'})\cos(m_{k'} - m_k) \mid \zeta_{\ell'} = k'}\\
        =& \sum\limits_k \sum\limits_{k'} P(\zeta_\ell = k, \zeta_{\ell'} = k')\frac{I_1(\rho_k)}{I_0(\rho_k)}\frac{I_1(\rho_{k'})}{I_0(\rho_{k'})}\cos(m_{k'} - m_k)
    \end{aligned}
\end{equation*}

We apply the same tricks for $\E{\cos(y_\ell + y_{\ell'} - 2 \alpha)}$.
\begin{equation*}
    \begin{aligned}
        & \E{\cos(y_\ell + y_{\ell'} - 2 \alpha)} \\
        =& \E{\E{\cos(y_\ell + y_{\ell'} - 2 \alpha) \mid \zeta_\ell, \zeta_{\ell'}}}\\
        =& \sum\limits_k \sum\limits_{k'} P(\zeta_\ell = k, \zeta_{\ell'} = k') \E{\cos((y_\ell - m_k) + (y_{\ell'} + m_k - 2\alpha)) \mid \zeta_\ell = k, \zeta_{\ell'} = k'}\\
        =& \sum\limits_k \sum\limits_{k'} P(\zeta_\ell = k, \zeta_{\ell'} = k') \E{\cos(y_\ell - m_k)\cos(y_{\ell'} + m_k - 2\alpha) \mid \zeta_\ell = k, \zeta_{\ell'} = k'}\\
        =& \sum\limits_k \sum\limits_{k'} P(\zeta_\ell = k, \zeta_{\ell'} = k')\frac{I_1(\rho_k)}{I_0(\rho_k)}\E{\cos(y_{\ell'} - m_{k'})\cos(m_{k'} + m_k - 2\alpha) \mid \zeta_{\ell'} = 2}\\
        =& \sum\limits_k \sum\limits_{k'} P(\zeta_\ell = k, \zeta_{\ell'} = k')\frac{I_1(\rho_k)}{I_0(\rho_k)}\frac{I_1(\rho_{k'})}{I_0(\rho_{k'})}\cos(m_{k'} + m_k - 2\alpha)
    \end{aligned}
\end{equation*}
Combining these quantities gives us the circular correlation.
\end{proof}

\subsection{SvM-p-2 model}
We first prove the useful lemmas mentioned in the part on \textit{SvM-p-2} from Section \ref{ssection:model_prop}.

\begin{lemma}
If $Z \sim \textrm{N}(0, 1)$, $\mathbbm{E}\left(\invlogit{Z}\right) = \frac{1}{2}$.
\label{lemma:logit_e_a}
\end{lemma}
\begin{proof}
We can also show that $\E{\invlogit{Z}} - \E{f(Z)} = 0$. For $z > \zEps$, we have that
\begin{equation*}
    \begin{aligned}
    \int_{\zEps}^\infty \invlogit{z} - f(z)\phi(z) dz &= \int_{\zEps}^\infty (\invlogit{z} - 1)\phi(z) dz\\
    &= -\int_{\zEps}^\infty (1 - \invlogit{z})\phi(z) dz\\
    &= -\int_{\zEps}^\infty \invlogit{-z}\phi(z) dz\\
    &= -\int_\infty^{-\zEps} \invlogit{z}\phi(z) dz.
    \end{aligned}
\end{equation*}
Hence, this cancels out $z < \zEps$. A similar argument holds for $z \in [-\zEps, 0]$ and $z \in [0, \zEps]$. As a result, we can compute $\E{\invlogit{Z}}$ using $\E{f(Z)}$. Then, we have that
\begin{equation*}
    \begin{aligned}
    \E{f(Z)} &= \int_{-\zEps}^{\zEps} \frac{1}{2\zEps}(z + \zEps)\frac{1}{\sqrt{2\pi}}\exp(-\frac{1}{2}z^2) +
    \int_{\zEps}^{\infty} \frac{1}{\sqrt{2\pi}}\exp(-\frac{1}{2}z^2)\\
    &= \frac{1}{2}\left(\Phi(\zEps) - \Phi(-\zEps)\right) + 1 - \Phi(\zEps)\\
    &= \frac{1}{2}.
    \end{aligned}
\end{equation*}
\end{proof}

While the calculations for the circular mean and variance are clear, the circular correlation requires some work to derive. 
\begin{lemma}
The circular correlation for the $\textit{SvM-p-2}$ model described in \eqref{model:SvM-p-2} is
\[
\frac{2\E{\sin(y_\ell - \alpha)\sin(y_{\ell'} - \alpha)}}{\sqrt{2\E{\sin^2(y_\ell - \alpha)}2\E{\sin^2(y_{\ell'} - \alpha)}}}
\]
where
\begin{align*}
\E{\sin^2(y_\ell - \alpha)} &= \frac{1}{2}\left(1 - P(\zeta_\ell = 1)\frac{I_2(\rho_1)}{I_0(\rho_1)}\cos(2(m_1 - \alpha)) - \right.\\
& \qquad \left. P(\zeta_\ell = 2)\frac{I_2(\rho_2)}{I_0(\rho_2)}\cos(2(m_2 - \alpha))\right),\\
\E{\sin^2(y_{\ell'} - \alpha)} &= \frac{1}{2}\left(1 - P(\zeta_{\ell'} = 1)\frac{I_2(\rho_1)}{I_0(\rho_1)}\cos(2(m_1 - \alpha)) - \right.\\
& \qquad \left. P(\zeta_{\ell'} = 2)\frac{I_2(\rho_2)}{I_0(\rho_2)}\cos(2(m_2 - \alpha))\right),\\
\E{\sin(y_\ell - \alpha)\sin(y_{\ell'} - \alpha)} &= P(\zeta_\ell, \zeta_{\ell'} = 1)\left(\frac{I_1(\rho_1)}{I_0(\rho_1)}\right)^2(1 - \cos(2(m_1 - \alpha))) +\\
& \qquad P(\zeta_\ell, \zeta_{\ell'} = 2)\left(\frac{I_1(\rho_2)}{I_0(\rho_2)}\right)^2(1 - \cos(2(m_2 - \alpha))) +\\
& \qquad (P(\zeta_\ell = 1, \zeta_{\ell'} = 2) + P(\zeta_\ell = 2, \zeta_{\ell'} = 1)  )\\
& \qquad\quad\frac{I_1(\rho_1)}{I_0(\rho_1)}\frac{I_1(\rho_2)}{I_0(\rho_2)}(\cos(m_1 - m_2) - \cos(m_1 + m_2 - 2\alpha)).
\end{align*}
\end{lemma}
\begin{proof}
We first derive $\E{\sin^2(y_\ell - \alpha)}$ without loss of generality. Like before, we need to compute $\E{cos(2(y_\ell - \alpha))}$. Breaking it into two cases, we have that 
\begin{equation*}
    \begin{aligned}
        \E{cos(2(y_\ell - \alpha))} &= P(\zeta_\ell = 1)\E{cos(2((y_\ell - m_1) + (m_1 - \alpha))) \mid \zeta_\ell = 1} +\\
        & \qquad P(\zeta_\ell = 2)\E{cos(2((y_\ell - m_2) + (m_2 - \alpha))) \mid \zeta_\ell = 2}\\
        &= P(\zeta_\ell = 1)\E{cos(2((y_\ell - m_1)))\cos(2(m_1 - \alpha)) \mid \zeta_\ell = 1} +\\
        & \qquad P(\zeta_\ell = 2)\E{cos(2((y_\ell - m_2)))\cos(2(m_2 - \alpha)) \mid \zeta_\ell = 2}\\
        &= P(\zeta_\ell = 1)\frac{I_2(\rho_1)}{I_0(\rho_1)}\cos(2(m_1 - \alpha)) + P(\zeta_\ell = 2)\frac{I_2(\rho_2)}{I_0(\rho_2)}\cos(2(m_2 - \alpha)).
    \end{aligned}
\end{equation*}
We can repeat this calculation for $\E{\sin^2(y_{\ell'} - \alpha)}$.

Note that $\E{\sin(y_\ell - \alpha)\sin(y_{\ell'} - \alpha)} = \frac{1}{2}(\E{\cos(y_\ell - \alpha - (y_{\ell'} - \alpha))) - \cos(y_\ell - \alpha + y_{\ell'} - \alpha)})$.
Meanwhile, $\E{\cos(y_\ell - y_{\ell'})}$ breaks down into four cases:
\begin{enumerate}
    \item $\zeta_\ell$, $\zeta_{\ell'}$ = 1
    \item $\zeta_\ell$, $\zeta_{\ell'}$ = 2
    \item $\zeta_\ell$ = 1, $\zeta_{\ell'}$ = 2
    \item $\zeta_\ell$ = 2, $\zeta_{\ell'}$ = 1.
\end{enumerate}
For the first case, 
\begin{equation*}
\begin{aligned}
    \E{\cos(y_1 - y_2) \mid \zeta_\ell, \zeta_{\ell'} = 1} &= \E{\cos((y_1 - m_1) + (m_1 - y_2)) \mid \zeta_\ell, \zeta_{\ell'} = 1}\\
    &= \E{\cos(y_1 - m_1)\cos(y_2 - m_1) \mid \zeta_\ell, \zeta_{\ell'} = 1}\\
    &= \left(\frac{I_1(\rho_1)}{I_0(\rho_1)}\right)^2.\\
\end{aligned}
\end{equation*}
This is weighted by $P(\zeta_\ell = 1, \zeta_{\ell'} = 1)$. By a similar calculation, $\E{\cos(y_1 - y_2) \mid \zeta_\ell, \zeta_{\ell'} = 2} = \left(\frac{I_1(\rho_2)}{I_0(\rho_2)}\right)^2$. 

Finally, the computation for the last two cases are identical. Without loss of generality, we have that
\begin{equation*}
\begin{aligned}
    \E{\cos(y_1 - y_2) \mid \zeta_\ell = 1, \zeta_{\ell'} = 2} &= \E{\cos((y_1 - m_1) + (m_1 - y_2)) \mid \zeta_\ell = 1, \zeta_{\ell'} = 2}\\
    &= \E{\cos(y_1 - m_1)\cos(y_2 - m_1) \mid \zeta_\ell = 1, \zeta_{\ell'} = 2}\\
    &= \frac{I_1(\rho_1)}{I_0(\rho_1)}\E{\cos(y_2 - m_2)\cos(m_2 - m_1) \mid \zeta_{\ell'} = 2}\\
    &= \frac{I_1(\rho_1)}{I_0(\rho_1)}\frac{I_1(\rho_2)}{I_0(\rho_2)}\cos(m_2 - m_1).\\
\end{aligned}
\end{equation*}
This is weighted by $P(\zeta_\ell = 1, \zeta_{\ell'} = 2)$ and $P(\zeta_\ell = 2, \zeta_{\ell'} = 1)$. If we put together these four cases using the law of iterated expectations, we get $\E{\cos(y_\ell - y_{\ell'})}$.

We can also modify these calculations for $\E{\cos(y_\ell + y_{\ell'} - 2\alpha)}$. In cases one and two, $m_1 - m_1$ and $m_2 - m_2$ have to be added in twice. As a result, without loss of generality, we have for the first case that 
\[
\E{\cos(y_\ell + y_{\ell'} - 2\alpha) \mid \zeta_\ell, \zeta_{\ell'} = 1} = \left(\frac{I_1(\rho_1)}{I_0(\rho_1)}\right)^2\cos(2(m_1 - \alpha)).
\]
For the last two cases, we have to insert $m_1 - m_1$ and $m_2 - m_2$. Hence, \[
\E{\cos(y_\ell + y_{\ell'} - 2\alpha) \mid \zeta_\ell = 1, \zeta_{\ell'} = 2} = \frac{I_1(\rho_1)}{I_0(\rho_1)}\frac{I_1(\rho_2)}{I_0(\rho_2)}\cos(m_1 + m_2 - 2\alpha).
\]
If we then weight these quantities with their region's probability and thus use the law of iterated expectations, we get $\E{\cos(y_\ell + y_{\ell'} - 2\alpha)}$.
\end{proof}

Next, for the circular correlation in \eqref{lemma:SvM-p_version_corr}, we have that
\begin{align*}
& 2\E{\sin(y_\ell - \alpha)\sin(y_{\ell'} - \alpha)} =\\
& \frac{I_1(\rho_1)}{I_0(\rho_1)}\frac{I_1(\rho_2)}{I_0(\rho_2)}(\cos(m_1 - m_2) - \cos(m_1 + m_2 - 2\alpha)) +\\
& \qquad \E{\invlogit{z_\ell}\invlogit{z_\ell'}} \left(\left(\frac{I_1(\rho_1)}{I_0(\rho_1)}\right)^2(1 - \cos(2(m_1 - \alpha))) + \right.\\
& \qquad \left. \left(\frac{I_1(\rho_2)}{I_0(\rho_2)}\right)^2(1 - \cos(2(m_2 - \alpha))) - \right.\\
& \qquad \left. 2\frac{I_1(\rho_1)}{I_0(\rho_1)}\frac{I_1(\rho_2)}{I_0(\rho_2)}(\cos(m_1 - m_2) - \cos(m_1 + m_2 - 2\alpha))\right).
\end{align*}
Note that this makes what we mentioned for \eqref{lemma:SvM-p_corr} more clear. If we look at the part attached to $\E{\invlogit{z_1}\invlogit{z_2}}$, it is the difference between the concentration of two observations belonging to the same component and the concentration of two observations belonging to different components. The part that is not attached is the concentration of two observations belonging to two different components. Intuitively, $\E{\invlogit{z_1}\invlogit{z_2}}$ will be close to 1 for "nearby" observations so the concentration of two observations belonging to the same component will dominate. 

In order to be more precise, one needs to compute $\E{\invlogit{Z_\ell}\invlogit{Z_{\ell'}}}$. This is intractable. 
Instead, we may use $\E{f(Z_\ell)f(Z_{\ell'})}$ to bound $\E{\invlogit{Z_\ell}\invlogit{Z_{\ell'}}}$. However, in order to do so, we need the following technical lemmas. 

\begin{lemma}
If $Z_\ell, Z_{\ell'} \sim \mathcal{N}(0, \Sigma)$ and $s = \textrm{Cov}\left(Z_\ell, Z_{\ell'}\right)$ and $\zEps \in (0, 2]$,
\begin{align*}
    \int_{-\zEps}^{\zEps}\int_{\zEps}^\infty & z_\ell \frac{1}{\sqrt{2\pi(1 - s^2)}}\exp(-\frac{(z_{\ell'} - s z_\ell)^2}{2(1 - s^2)}) \frac{1}{\sqrt{2\pi}}\exp(-\frac{z_\ell^2}{2})dz_{\ell'} dz_\ell = \qquad\\
    & -\phi(\zEps)\left(\Phi\left(-\zEps\sqrt{\frac{1 - s}{1 + s}}\right) - \Phi\left(-\zEps\sqrt{\frac{1 + s}{1 - s}}\right)\right) +\\
    &\qquad\phi(\zEps) \left(\Phi\left(\frac{\zEps - \vartheta}{\sqrt{1 - s^2}}\right) - \Phi\left(\frac{-\zEps - \vartheta}{\sqrt{1 - s^2}}\right)\right),
\end{align*}
where $\vartheta = \frac{(1 - s^2)\zEps}{s}$.
\label{lemma:one_z_bivariate_a}
\end{lemma}
\begin{proof}
Using integration of parts, we have that
\begin{align*}
    & \int_{-\zEps}^{\zEps}\int_{\zEps}^\infty z_\ell \frac{1}{\sqrt{2\pi(1 - s^2)}}\exp(-\frac{(z_{\ell'} - s z_\ell)^2}{2(1 - s^2)}) \frac{1}{\sqrt{2\pi}}\exp(-\frac{z_\ell^2}{2})dz_{\ell'} dz_\ell\\ 
    =& \int_{-\zEps}^{\zEps} z_\ell\left(1 - \Phi\left(\frac{\zEps - sz_\ell}{\sqrt{1 - s^2}}\right)\right) \frac{1}{\sqrt{2\pi}}\exp(-\frac{z_\ell^2}{2})\\
    =& \int_{-\zEps}^{\zEps} z_\ell\left(\Phi\left(\frac{sz_\ell - \zEps}{\sqrt{1 - s^2}}\right)\right) \frac{1}{\sqrt{2\pi}}\exp(-\frac{z_\ell^2}{2})\\
    =& -\frac{1}{\sqrt{2\pi}}\exp(-\frac{z_\ell^2}{2})\Phi\left(\frac{sz_\ell - \zEps}{\sqrt{1 - s^2}}\right)\bigg|_{-\zEps}^{\zEps} +\\ 
    & \qquad\int_{-\zEps}^{\zEps} \frac{s}{\sqrt{2\pi(1 - s^2)}}\exp(-\frac{(sz_\ell - \zEps)^2}{2(1 - s^2)}) \frac{1}{\sqrt{2\pi}}\exp(-\frac{z_\ell^2}{2})\\
    =& -\frac{1}{\sqrt{2\pi}}\exp(-\frac{z_\ell^2}{2})\Phi\left(\frac{sz_\ell - \zEps}{\sqrt{1 - s^2}}\right)\bigg|_{-\zEps}^{\zEps} +\\ 
    & \qquad\int_{-\zEps}^{\zEps} \frac{s}{\sqrt{2\pi(1 - s^2)}}\exp(-\frac{(z_\ell - \frac{\zEps}{s})^2}{2\frac{1 - s^2}{s^2}}) \frac{1}{\sqrt{2\pi}}\exp(-\frac{z_\ell^2}{2}).\\
\end{align*}

For the first half, we have that 
\begin{equation*}
    \begin{aligned}
    &-\frac{1}{\sqrt{2\pi}}\exp(-\frac{z_\ell^2}{2})\Phi\left(\frac{sz_\ell - \zEps}{\sqrt{1 - s^2}}\right)\bigg|_{-\zEps}^{\zEps}\\
    =&
     -\phi(\zEps)\Phi\left(\frac{s\zEps - \zEps}{\sqrt{1 - s^2}}\right) + \phi(-\zEps)\Phi\left(\frac{-s\zEps - \zEps}{\sqrt{1 - s^2}}\right)\\
    =& -\phi(\zEps)\left(\Phi\left(\frac{(s - 1)\zEps}{\sqrt{1 - s^2}}\right) - \Phi\left(\frac{-(s + 1)\zEps}{\sqrt{1 - s^2}}\right)\right)\\
    =& -\phi(\zEps)\left(\Phi\left(\frac{(s - 1)\zEps}{\sqrt{1 - s^2}}\right) - \Phi\left(\frac{-(s + 1)\zEps}{\sqrt{1 - s^2}}\right)\right)\\
    =& -\phi(\zEps)\left(\Phi\left(-\zEps\sqrt{\frac{(1 - s)^2}{1 - s^2}}\right) - \Phi\left(-\zEps\sqrt{\frac{(1 + s)^2}{1 - s^2}}\right)\right)\\
    =& -\phi(\zEps)\left(\Phi\left(-\zEps\sqrt{\frac{1 - s}{1 + s}}\right) - \Phi\left(-\zEps\sqrt{\frac{1 + s}{1 - s}}\right)\right).
    \end{aligned}
\end{equation*}

For the second half, we have the integral of the product of two Gaussian distributions. Hence, the mean is
\[
\frac{\frac{1 - s^2}{s^2}\frac{\zEps}{s} + (1)(0)}{1 + \frac{1 - s^2}{s^2}} =
\frac{(1 - s^2)\zEps}{s(s^2 + (1 - s^2))} = \frac{(1 - s^2)}{s}\zEps
\]
and the standard deviation is 
\[
\sqrt{\frac{\frac{1 - s^2}{s^2}}{1 + \frac{1 - s^2}{s^2}}} =
\sqrt{\frac{1 - s^2}{s^2 + (1 - s^2)}} =
\sqrt{1 - s^2}.
\]
Meanwhile, for the scaling term, we have that the mean is essentially zero and the standard deviation is 
\[
\sqrt{1 + \frac{1 - s^2}{s^2}} = \sqrt{\frac{s^2 + 1 - s^2}{s^2}} = \sqrt{\frac{1}{s^2}} = \frac{1}{s}.
\]
We then standardize using the mean and standard deviation to get the quantity in the lemma.
\end{proof}

\begin{lemma}
If $Z_\ell, Z_{\ell'} \sim \mathcal{N}(0, \Sigma)$ and $s = \textrm{Cov}\left(Z_\ell, Z_{\ell'}\right)$ and $\zEps \in (0, 2]$,
\begin{align*}
    \int_{-\zEps}^{\zEps} \int_{-\zEps}^{\zEps} & z_\ell z_{\ell'} \frac{1}{\sqrt{2\pi(1 - s^2)}}\exp(-\frac{(z_{\ell'} - s z_\ell)^2}{2(1 - s^2)}) \frac{1}{\sqrt{2\pi}}\exp(-\frac{z_\ell^2}{2})dz_{\ell'} dz_\ell = \qquad\\
    & \frac{1 - s^2}{s}\phi(\zEps)\vartheta\left(\Phi\left(\frac{\zEps + \vartheta}{\sqrt{1 - s^2}}\right) - \Phi\left(\frac{-\zEps + \vartheta}{\sqrt{1 - s^2}}\right) - \right.\\
    & \quad\qquad\qquad\qquad \left. \left(\Phi\left(\frac{\zEps - \vartheta}{\sqrt{1 - s^2}}\right) - \Phi\left(\frac{-\zEps - \vartheta}{\sqrt{1 - s^2}}\right)\right)\right) -\\
    & \quad s\phi(\zEps)\left(-\zEps\phi\left(-\zEps\sqrt{\frac{1 - s}{1 + s}}\right) + \Phi\left(\zEps\sqrt{\frac{1 - s}{1 + s}}\right)\right) +\\ 
    & \quad s\phi(\zEps)\left(\zEps\phi\left(-\zEps\sqrt{\frac{1 + s}{1 - s}}\right) + \Phi\left(\zEps\sqrt{\frac{1 + s}{1 - s}}\right)\right) -\\
    & \quad s\phi(\zEps)\left(\left(\Phi\left(\frac{\zEps - \vartheta}{\sqrt{1 - s^2}}\right) - \Phi\left(\frac{-\zEps - \vartheta}{\sqrt{1 - s^2}}\right)\right)\vartheta - \right.\\
    & \quad\qquad\qquad\left.\sqrt{1 - s^2}\left(\phi\left(\frac{\zEps - \vartheta}{\sqrt{1 - s^2}}\right) - \phi\left(\frac{-\zEps - \vartheta}{\sqrt{1 - s^2}}\right)\right)\right) +\\ 
    & \quad \textrm{P}(z_\ell \in [-\zEps, \zEps], z_{\ell'} \leq \zEps \mid \v{0}, \Sigma),
\end{align*}
where $\vartheta = \frac{(1 - s^2)\zEps}{s}$.
\label{lemma:two_z_bivariate_a}
\end{lemma}
\begin{proof}
To set up change of variable, we have that 
\begin{equation}
    \begin{aligned}
        & \int_{-\zEps}^{\zEps} \int_{-\zEps}^{\zEps} z_\ell z_{\ell'} \frac{1}{\sqrt{2\pi(1 - s^2)}}\exp(-\frac{(z_{\ell'} - s z_\ell)^2}{2(1 - s^2)}) \frac{1}{\sqrt{2\pi}}\exp(-\frac{z_\ell^2}{2})dz_{\ell'} dz_\ell\\
        =& \int_{-\zEps}^{\zEps} \int_{-\zEps}^{\zEps} z_\ell \left(\frac{z_{\ell'} - s z_\ell}{1 - s^2}\right) \sqrt{\frac{1 - s^2}{2\pi}}\exp(-\frac{(z_{\ell'} - s z_\ell)^2}{2(1 - s^2)}) \frac{1}{\sqrt{2\pi}}\exp(-\frac{z_\ell^2}{2})dz_{\ell'} dz_\ell +\\
        & \quad \int_{-\zEps}^{\zEps} \int_{-\zEps}^{\zEps} z_\ell (s z_\ell) \frac{1}{\sqrt{2\pi(1 - s^2)}}\exp(-\frac{(z_{\ell'} - s z_\ell)^2}{2(1 - s^2)}) \frac{1}{\sqrt{2\pi}}\exp(-\frac{z_\ell^2}{2})dz_{\ell'} dz_\ell.
    \end{aligned}
\label{two_z_prod_eq}
\end{equation}
Again, let us work on the two pieces of Equation \ref{two_z_prod_eq}. Let $\widetilde{\zEps} = \frac{(\zEps - s z_\ell)^2}{2(1 - s^2)}$ and $\widetilde{z}_{-\epsilon} = \frac{(-\zEps - s z_\ell)^2}{2(1 - s^2)}$. Then, by change of variable,
\begin{equation}
    \begin{aligned}
         & \int_{-\zEps}^{\zEps} \int_{\widetilde{z}_{-\epsilon}}^{\widetilde{\zEps}} z_\ell \sqrt{\frac{1 - s^2}{2\pi}}\exp(-u) \frac{1}{\sqrt{2\pi}}\exp(-\frac{z_\ell^2}{2})du dz_\ell\\
         &= \int_{-\zEps}^{\zEps} -z_\ell\sqrt{\frac{1 - s^2}{2\pi}}(\exp(\widetilde{\zEps}) - \exp(\widetilde{z}_{-\epsilon}))\frac{1}{\sqrt{2\pi}}\exp(-\frac{z_\ell^2}{2})\\
         &= \int_{-\zEps}^{\zEps} -z_\ell\sqrt{\frac{1 - s^2}{2\pi}}\exp(\widetilde{\zEps})\frac{1}{\sqrt{2\pi}}\exp(-\frac{z_\ell^2}{2}) +\\   
         & \qquad\int_{-\zEps}^{\zEps} -z_\ell\sqrt{\frac{1 - s^2}{2\pi}}\exp(\widetilde{z}_{-\epsilon}))\frac{1}{\sqrt{2\pi}}\exp(-\frac{z_\ell^2}{2}).
    \end{aligned}
\label{two_z_prod_eq_left}
\end{equation}
We can rearrange $\widetilde{\zEps}$ to be $\frac{(z_\ell - \frac{\zEps}{s})^2}{2\frac{1 - s^2}{s^2}}$ and $\widetilde{z}_{-\epsilon}$ to be $\frac{(s z_\ell - (-\frac{\zEps}{s}))^2}{2\frac{1 - s^2}{s^2}}$. By doing so, we see that we again have the product of two Gaussians for both parts of Equation \ref{two_z_prod_eq_left}. 

If we look at the left part of Equation \ref{two_z_prod_eq_left}, our scaling term is again $\phi(\zEps)$ and that the mean and standard deviation of the Gaussian are the same as that in Lemma \ref{lemma:one_z_bivariate_a}. As a result, we are taking the expectation of an unnormalized truncated Gaussian. Hence, 
\begin{equation*}
\begin{aligned}
    &\int_{-\zEps}^{\zEps} -z_\ell\sqrt{\frac{1 - s^2}{2\pi}}\exp(\widetilde{\zEps})\frac{1}{\sqrt{2\pi}}\exp(-\frac{z_\ell^2}{2}) \\
    =&\frac{1 - s^2}{s}\int_{-\zEps}^{\zEps} -z_\ell\frac{s}{\sqrt{2\pi(1 - s^2)}}\exp(\widetilde{\zEps})\frac{1}{\sqrt{2\pi}}\exp(-\frac{z_\ell^2}{2})\\
    =&-\frac{1 - s^2}{s}\phi(\zEps)\left(\left(\Phi\left(\frac{\zEps - \vartheta}{\sqrt{1 - s^2}}\right) - \Phi\left(\frac{-\zEps - \vartheta}{\sqrt{1 - s^2}}\right)\right)\vartheta - \right.\\
    & \qquad\qquad\qquad\left. \sqrt{1 - s^2}\left(\phi\left(\frac{\zEps - \vartheta}{\sqrt{1 - s^2}}\right) - \phi\left(\frac{-\zEps - \vartheta}{\sqrt{1 - s^2}}\right)\right)\right).
\end{aligned}    
\end{equation*}
Meanwhile for the right part of Equation \ref{two_z_prod_eq_left}, our scaling term is $\phi(-\zEps) = \phi(\zEps)$. While the standard deviation is the same as before, the mean is $-\vartheta$. As a result,
\begin{align*}
    \int_{-\zEps}^{\zEps} & -z_\ell\sqrt{\frac{1 - s^2}{2\pi}}\exp(\widetilde{\zEps})\frac{1}{\sqrt{2\pi}}\exp(-\frac{z_\ell^2}{2}) =\\
    &\frac{1 - s^2}{s}\phi(\zEps)\left(\left(\Phi\left(\frac{\zEps + \vartheta}{\sqrt{1 - s^2}}\right) - \Phi\left(\frac{-\zEps + \vartheta}{\sqrt{1 - s^2}}\right)\right)\vartheta + \right.\\
    & \qquad\qquad\qquad\left.\sqrt{1 - s^2}\left(\phi\left(\frac{\zEps + \vartheta}{\sqrt{1 - s^2}}\right) - \phi\left(\frac{-\zEps + \vartheta}{\sqrt{1 - s^2}}\right)\right)\right).
\end{align*}
Hence, combining the two, we get that for the left part of Equation \ref{two_z_prod_eq},
\begin{align}
     \frac{1 - s^2}{s}\phi(\zEps)\vartheta\left(\Phi\left(\frac{\zEps + \vartheta}{\sqrt{1 - s^2}}\right) - \Phi\left(\frac{-\zEps + \vartheta}{\sqrt{1 - s^2}}\right) - \left(\Phi\left(\frac{\zEps - \vartheta}{\sqrt{1 - s^2}}\right) - \Phi\left(\frac{-\zEps - \vartheta}{\sqrt{1 - s^2}}\right)\right)\right).
    \label{two_z_prod_eq_right_result}
\end{align}
For the right part, we have that
\begin{equation}
    \begin{aligned}
    & \int_{-\zEps}^{\zEps} \int_{-\zEps}^{\zEps} z_\ell (s z_\ell) \frac{1}{\sqrt{2\pi(1 - s^2)}}\exp(-\frac{(z_{\ell'} - s z_\ell)^2}{2(1 - s^2)}) \frac{1}{\sqrt{2\pi}}\exp(-\frac{z_\ell^2}{2})dz_{\ell'} dz_\ell\\
    =& \int_{-\zEps}^{\zEps} s z_\ell^2 \left(\Phi\left(\frac{\zEps - s z_\ell}{\sqrt{1 - s^2}}\right) - \Phi\left(\frac{-\zEps - s z_\ell}{\sqrt{1 - s^2}}\right)\right) \frac{1}{\sqrt{2\pi}}\exp(-\frac{z_\ell^2}{2}) dz_\ell\\
    =& \int_{-\zEps}^{\zEps} s z_\ell^2 \Phi\left(\frac{\zEps - s z_\ell}{\sqrt{1 - s^2}}\right)\frac{1}{\sqrt{2\pi}}\exp(-\frac{z_\ell^2}{2}) - \int_{-\zEps}^{\zEps} \Phi\left(\frac{-\zEps - s z_\ell}{\sqrt{1 - s^2}}\right) \frac{1}{\sqrt{2\pi}}\exp(-\frac{z_\ell^2}{2}) dz_\ell. 
    \end{aligned}
\label{two_z_prod_eq_right}
\end{equation}
Let us focus on the right part of Equation \ref{two_z_prod_eq_right}. By integration of parts, we have that
\begin{equation}
    \begin{aligned}
        & \int_{-\zEps}^{\zEps} s z_\ell^2 \Phi\left(\frac{\zEps - s z_\ell}{\sqrt{1 - s^2}}\right)\frac{1}{\sqrt{2\pi}}\exp(-\frac{z_\ell^2}{2})\\
        =& -s\left(z_\ell \frac{-s}{\sqrt{2\pi(1 - s^2)}}\exp(-\frac{(s z_\ell - \zEps)^2}{2(1 - s^2)}) + \Phi\left(\frac{\zEps - s z_\ell}{\sqrt{1 - s^2}}\right)\right)\frac{1}{\sqrt{2\pi}}\exp(-\frac{z_\ell^2}{2})\bigg|_{-\zEps}^{\zEps} +\\
        & \quad\int_{-\zEps}^{\zEps} s\left(z_\ell \frac{-s}{\sqrt{2\pi(1 - s^2)}}\exp(-\frac{(s z_\ell - \zEps)^2}{2(1 - s^2)}) + \Phi\left(\frac{\zEps - s z_\ell}{\sqrt{1 - s^2}}\right)\right)\frac{1}{\sqrt{2\pi}}\exp(-\frac{z_\ell^2}{2})\\
        =& -s\phi(\zEps)\left(-\zEps\phi\left(-\zEps\sqrt{\frac{1 - s}{1 + s}}\right) + \Phi\left(\zEps\sqrt{\frac{1 - s}{1 + s}}\right)\right) +\\
        &\quad s\phi(\zEps)\left(\zEps\phi\left(-\zEps\sqrt{\frac{1 + s}{1 - s}}\right) + \Phi\left(\zEps\sqrt{\frac{1 + s}{1 - s}}\right)\right) +\\
        & \quad\int_{-\zEps}^{\zEps} sz_\ell \frac{-s}{\sqrt{2\pi(1 - s^2)}}\exp(-\frac{(s z_\ell - \zEps)^2}{2(1 - s^2)})\frac{1}{\sqrt{2\pi}}\exp(-\frac{z_\ell^2}{2}) +\\
        &\quad\int_{-\zEps}^{\zEps}s\Phi\left(\frac{\zEps - s z_\ell}{\sqrt{1 - s^2}}\right)\frac{1}{\sqrt{2\pi}}\exp(-\frac{z_\ell^2}{2}).
    \end{aligned}
\label{two_z_prod_eq_right_left}
\end{equation}
The third term of the Equation \ref{two_z_prod_eq_right_left} is again a truncated normal with a scaling term of $\phi(\zEps)$ and the same $\vartheta$ and $\widetilde{\sigma}$ as before. Hence, 
\begin{align*}
    \int_{-\zEps}^{\zEps} & sz_\ell \frac{-s}{\sqrt{2\pi(1 - s^2)}}\exp(-\frac{(s z_\ell - \zEps)^2}{2(1 - s^2)})\frac{1}{\sqrt{2\pi}}\exp(-\frac{z_\ell^2}{2}) = \\
    &\quad-s\phi(\zEps)\left(\left(\Phi\left(\frac{\zEps - \vartheta}{\sqrt{1 - s^2}}\right) - \Phi\left(\frac{-\zEps - \vartheta}{\sqrt{1 - s^2}}\right)\right)\vartheta - \right.\\
    &\qquad\qquad\qquad\left. \sqrt{1 - s^2}\left(\phi\left(\frac{\zEps - \vartheta}{\sqrt{1 - s^2}}\right) - \phi\left(\frac{-\zEps - \vartheta}{\sqrt{1 - s^2}}\right)\right)\right).
\end{align*}
The last term of the Equation \ref{two_z_prod_eq_right_left} is the product of the conditional normal distribution based on $z_\ell$ and a normal distribution of $z_\ell$. As a result, it is the probability of $z_\ell \in [-\zEps, \zEps]$ and $z_{\ell'} \leq \zEps$ where $z_\ell$ and $z_{\ell'}$ are distributed according to
\[
\mathcal{N}\bigg(
\begin{pmatrix}
z_\ell\\
z_{\ell'}
\end{pmatrix}
\bigg|
\begin{pmatrix}
0\\
0
\end{pmatrix},
\begin{pmatrix}
1 & s\\
s & 1
\end{pmatrix}
\bigg).
\]
\end{proof}

\begin{lemma}
If $\widetilde{E}_{z_\ell}$ denotes the quantity in Lemma \ref{lemma:one_z_bivariate_a} and $\widetilde{E}_{z_\ell, z_{\ell'}}$ denotes the quantity in Lemma \ref{lemma:two_z_bivariate_a}, then for $z_\ell \neq z_{\ell'}$,
\[
\E{f{z_\ell}f{z_{\ell'}}} = \frac{1}{4}(\textrm{P}(z_\ell > \zEps, z_{\ell'} > \zEps) + \textrm{P}(z_\ell > -\zEps, z_{\ell'} > - \zEps)) + \frac{1}{2\zEps}\widetilde{E}_{z_\ell} + \frac{1}{4\zEps^2}\widetilde{E}_{z_\ell, z_{\ell'}}.
\]
\end{lemma}
\begin{proof}
We have that
\begin{equation*}
    \begin{aligned}
        \E{f{z_\ell}f{z_{\ell'}}} &= \mathbbm{E}\left(\left(g(z_\ell) + \frac{1}{2}\right)\left(g(z_{\ell'}) + \frac{1}{2}\right)\indFct{z_\ell > -\zEps, z_{\ell'} > - \zEps}\right)\\
         &= \mathbbm{E}\left(g(z_\ell)g(z_{\ell'}))\indFct{z_\ell > -\zEps, z_{\ell'} > - \zEps}\right) + \frac{1}{4}\textrm{P}(z_\ell > -\zEps, z_{\ell'} > - \zEps).
    \end{aligned}
\end{equation*}
Applying the law of iterated expectations to regions, we have that
\begin{equation*}
    \begin{aligned}
        \mathbbm{E}\left(g(z_\ell)g(z_{\ell'}))\right) =&
        \mathbbm{E}\left(g(z_\ell)g(z_{\ell'})) \mid z_\ell > \zEps, z_{\ell'} > \zEps\right)P(z_\ell > \zEps, z_{\ell'} > \zEps) +\\
        &\quad\mathbbm{E}\left(g(z_\ell)g(z_{\ell'})) \mid z_\ell \in [-\zEps, \zEps], z_{\ell'} > \zEps\right)P(z_\ell \in [-\zEps, \zEps], z_{\ell'} > \zEps) +\\
        &\quad \mathbbm{E}\left(g(z_\ell)g(z_{\ell'})) \mid z_\ell > \zEps, z_{\ell'} \in [-\zEps, \zEps]\right)P(z_\ell > \zEps, z_{\ell'} \in [-\zEps, \zEps]) +\\
        &\quad\mathbbm{E}\left(g(z_\ell)g(z_{\ell'})) \mid z_\ell \in [-\zEps, \zEps], z_{\ell'} \in [-\zEps, \zEps]\right)\\
        &\qquad P(z_\ell \in [-\zEps, \zEps], z_{\ell'} \in [-\zEps, \zEps]).
    \end{aligned}
\end{equation*}
Because of Lemma \ref{lemma:one_z_bivariate_a}, we get that 
\[
\frac{1}{4\zEps}\widetilde{E}_{z_\ell} = \mathbbm{E}\left(g(z_\ell)g(z_{\ell'})) \mid z_\ell \in [-\zEps, \zEps], z_{\ell'} > \zEps\right)P(z_\ell \in [-\zEps, \zEps], z_{\ell'} > \zEps)
\]
and 
\[
\frac{1}{4\zEps}\widetilde{E}_{z_\ell} = \mathbbm{E}\left(g(z_\ell)g(z_{\ell'})) \mid z_\ell > \zEps, z_{\ell'} \in [-\zEps, \zEps]\right)P(z_\ell > \zEps, z_{\ell'} \in [-\zEps, \zEps]).
\]
Similarly, Lemma \ref{lemma:two_z_bivariate_a} gives us that
\[
\frac{1}{4\zEps^2}\widetilde{E}_{z_\ell, z_{\ell'}} = \mathbbm{E}\left(g(z_\ell)g(z_{\ell'})) \mid z_\ell \in [-\zEps, \zEps], z_{\ell'} \in [-\zEps, \zEps]\right)P(z_\ell \in [-\zEps, \zEps], z_{\ell'} \in [-\zEps, \zEps]).
\] 
If we plug those values in, we get the value stated in the lemma.
\end{proof}

\begin{lemma}
If $Z_\ell, Z_{\ell'} \sim \mathcal{N}(0, \Sigma)$ and $s = \textrm{Cov}\left(Z_\ell, Z_{\ell'}\right)$ and $\zEps \in (0, 2]$, then
\begin{align*}
-2\invlogit{-\zEps}&\left(\Phi(-\zEps) + \left(\frac{1}{2} - \Phi(-\zEps)\right)\invlogit{\zEps}\right) \leq\\ &\E{\invlogit{Z_\ell}\invlogit{Z_{\ell'}}} -  \E{f(Z_\ell)f(Z_{\ell'})} \leq 0.
\end{align*}
\label{lemma:logistic_e_prod_approx_bound}
\end{lemma}
\begin{proof}
To take advantage of the conditional Normal distribution, we can write
\[
\invlogit{Z_\ell}\invlogit{Z_{\ell'}}-  f(Z_\ell)f(Z_{\ell'})
\]
as
\[
\frac{1}{2}\left((\invlogit{Z_\ell} - f(Z_\ell))(\invlogit{Z_{\ell'}} + f(Z_{\ell'})) + (\invlogit{Z_\ell} + f(Z_\ell))(\invlogit{Z_{\ell'}} - f(Z_{\ell'}))\right).
\]
Without loss of generality, we focus on $\E{\invlogit{Z_\ell} - f(Z_\ell))(\invlogit{Z_{\ell'}} + f(Z_{\ell'}))}$.

Given $z_{\ell'}$, we can show that 
\begin{equation*}
    \begin{aligned}
        &\int_{0}^{\infty} (\invlogit{z_\ell} - f(z_\ell))\phi\left(\frac{z_\ell - sz_{\ell'}}{\sqrt{1 - s^2}}\right)d z_\ell\\
        =& -\int_{0}^{\infty} (1 - \invlogit{z_\ell} - (1 - f(z_\ell)))\phi\left(\frac{z_\ell - sz_{\ell'}}{\sqrt{1 - s^2}}\right) d z_\ell\\
        =& -\int_{0}^{\infty} (\invlogit{-z_\ell} -  f(-z_\ell))\phi\left(\frac{z_\ell - sz_{\ell'}}{\sqrt{1 - s^2}}\right) d z_\ell\\
        =& -\int_{-\infty}^{0} (\invlogit{z'_\ell} -  f(z'_\ell))\phi\left(\frac{-z'_\ell - sz_{\ell'}}{\sqrt{1 - s^2}}\right) d z'_\ell.\\
    \end{aligned}
\end{equation*}
Using this, we have that
\[
\E{\invlogit{z_\ell} - f(z_\ell) \mid z_{\ell'}} = \int_{-\infty}^{0} (\invlogit{z_\ell} - f(z_\ell))\left(\phi\left(\frac{z_\ell - sz_{\ell'}}{\sqrt{1 - s^2}}\right) - \phi\left(\frac{-z_\ell - sz_{\ell'}}{\sqrt{1 - s^2}}\right)\right)d z_\ell.
\]
We can also demonstrate that $\E{\invlogit{z_\ell} - f(z_\ell) \mid z_{\ell'}} = -\E{\invlogit{z_\ell} - f(z_\ell) \mid -z_{\ell'}}$ because
\[
\phi\left(\frac{z_\ell - sz_{\ell'}}{\sqrt{1 - s^2}}\right) - \phi\left(\frac{-z_\ell - sz_{\ell'}}{\sqrt{1 - s^2}}\right) = -\left(\phi\left(\frac{z_\ell - s(-z_{\ell'})}{\sqrt{1 - s^2}}\right) - \phi\left(\frac{-z_\ell - s(-z_{\ell'})}{\sqrt{1 - s^2}}\right)\right).
\]
As a result, because $\invlogit{z_\ell} - f(z_\ell) > 0$ for $z_\ell < 0$ and $\invlogit{z_{\ell'}} + f(z_{\ell'}) < \invlogit{-z_{\ell'}} + f(-z_{\ell'})$ for $z_{\ell'} < 0$, $\E{\invlogit{Z_\ell}\invlogit{Z_{\ell'}}} -  \E{f(Z_\ell)f(Z_{\ell'})} \leq 0$.

Now that we have established an upper bound, we need to establish a lower bound. From our discussion above, $\E{\invlogit{Z_\ell}\invlogit{Z_{\ell'}}} -  \E{f(Z_\ell)f(Z_{\ell'})}$ can be written as 
\begin{align*}
\int_{-\infty}^0 \int_{-\infty}^{\infty} & (\invlogit{z_\ell} - f(z_\ell))\\
& \quad(\invlogit{z_{\ell'}} + f(z_{\ell'}) - (\invlogit{-z_{\ell'}} + f(-z_{\ell'})))\\
&\quad \phi\left(\frac{z_\ell - s z_{\ell'}}{\sqrt{1 - s^2}}\right)\phi(z_{\ell'}) d z_\ell d z_{\ell'}.
\end{align*}
If we consider the parts, we have that 
\[
-\invlogit{-\zEps} \leq \invlogit{z_\ell} - f(z_\ell) \leq \invlogit{-\zEps}
\]
and
\begin{align*}
    -2 \leq \invlogit{z_{\ell'}} + f(z_{\ell'}) - (\invlogit{-z_{\ell'}} + f(-z_{\ell'})) \leq -2\invlogit{\zEps} & \qquad z_{\ell'} \in (-\infty, -\zEps)\\
    -2\invlogit{\zEps} \leq \invlogit{z_{\ell'}} + f(z_{\ell'}) - (\invlogit{-z_{\ell'}} + f(-z_{\ell'})) \leq 0 & \qquad z_{\ell'} \in [-\zEps, 0].\\
\end{align*}
As a result, a lower bound is 
\[
-2\invlogit{-\zEps}\left(\Phi(-\zEps) + \left(\frac{1}{2} - \Phi(-\zEps)\right)\invlogit{\zEps}\right).
\]
Because a similar argument holds for $(\invlogit{Z_\ell} + f(Z_\ell))(\invlogit{Z_{\ell'}} - f(Z_{\ell'}))$, we get the bound described in the lemma.
\end{proof}

Using these technical lemmas, we have the following lemma. 
\begin{lemma}
The circular correlation for the $\textit{SvM-p-2}$ model described in \eqref{model:SvM-p-2} and under conditions of Lemma \ref{lemma:logistic_e_prod_approx_bound} is
\[
\frac{2\E{\sin(y_\ell - \alpha)\sin(y_{\ell'} - \alpha)}}{
1 - \frac{1}{2}\left(\frac{I_2(\rho_1)}{I_0(\rho_1)}\cos(2(m_1 - \alpha)) + \frac{I_2(\rho_2)}{I_0(\rho_2)}\cos(2(m_2 - \alpha))\right)}
\]
where $2\E{\sin(y_\ell - \alpha)\sin(y_{\ell'} - \alpha)}$ is equal to
\begin{align*}
& \frac{I_1(\rho_1)}{I_0(\rho_1)}\frac{I_1(\rho_2)}{I_0(\rho_2)}(\cos(m_1 - m_2) - \cos(m_1 + m_2 - 2\alpha)) +\\
& \qquad \E{\invlogit{z_{\ell}}\invlogit{z_{\ell'}}} \left(\left(\frac{I_1(\rho_1)}{I_0(\rho_1)}\right)^2(1 - \cos(2(m_1 - \alpha))) + \right.\\
& \qquad \left. \left(\frac{I_1(\rho_2)}{I_0(\rho_2)}\right)^2(1 - \cos(2(m_2 - \alpha))) - \right.\\
& \qquad \left. 2\frac{I_1(\rho_1)}{I_0(\rho_1)}\frac{I_1(\rho_2)}{I_0(\rho_2)}(\cos(m_1 - m_2) - \cos(m_1 + m_2 - 2\alpha))\right) \\
\end{align*}
and $\E{\invlogit{z_{\ell}}\invlogit{z_{\ell'}}}$ is approximated in Lemma \ref{lemma:logistic_e_prod_approx_bound}.
\end{lemma}
\begin{proof}
For the denominator, according to \ref{lemma:logit_e_a}, $P(\zeta_\ell = 1) = P(\zeta_{\ell'} = 1) = \frac{1}{2}$. Further, $P(\zeta_\ell = 2) = 1 - P(\zeta_\ell = 1)$ and $P(\zeta_{\ell'} = 2) = 1 - P(\zeta_{\ell'} = 1)$ so $P(\zeta_\ell = 2) = P(\zeta_\ell = 1) = \frac{1}{2}$ as well. Hence, we get the value in the denominator if we plug in these probabilities. 

Meanwhile, for the numerator, we need to compute the probability for the four cases:
\begin{enumerate}
    \item $\zeta_1$, $\zeta_2$ = 1
    \item $\zeta_1$, $\zeta_2$ = 2
    \item $\zeta_1$ = 1, $\zeta_2$ = 2
    \item $\zeta_1$ = 2, $\zeta_2$ = 1.
\end{enumerate}
For the first case, $P(\zeta_1 = 1, \zeta_2 = 1)$, which is approximated in Lemma \ref{lemma:logistic_e_prod_approx_bound}. 

For the second, through our approximation, 
\begin{equation*}
\begin{aligned}
    P(\zeta_\ell = 2, \zeta_{\ell'} = 2) &= \E{\indFct{\zeta_\ell, \zeta_{\ell'} = 2}}\\
    &= \E{\E{\indFct{\zeta_\ell, \zeta_{\ell'} = 2} \mid \lambda_1, \psi_2}}\\
    &= \E{(1 - \lambda_1)(1 - \psi_2)}\\
    &= \E{(1 - \invlogit{z_{\ell}})(1 - \invlogit{z_{\ell'}})}\\
    &= 1 - \E{\invlogit{z_{\ell}}} -  \E{\invlogit{z_{\ell'}}} + \E{\invlogit{z_{\ell}}\invlogit{z_{\ell'}}}\\
    &\approx \E{\invlogit{z_{\ell}}\invlogit{z_{\ell'}}}.
\end{aligned}
\end{equation*}
As a result, it too is weighted by the approximation in Lemma \ref{lemma:logistic_e_prod_approx_bound}.

Finally, for the last two cases, the calculations are identical. Following the same calculation as above and using Lemma \ref{lemma:logistic_e_prod_approx_bound}, we get that
\begin{equation*}
    \begin{aligned}
        P(\zeta_\ell = 1, \zeta_{\ell'} = 2) &= \E{\invlogit{z_{\ell}}(1 - \invlogit{z_{\ell'}})}\\
        &= \E{\invlogit{z_{\ell}}} - \E{\invlogit{z_{\ell}}\invlogit{z_{\ell'}}}\\
        &\approx \frac{1}{2} - \E{\invlogit{z_{\ell}}\invlogit{z_{\ell'}}}.\\
    \end{aligned}
\end{equation*}
Similarly, $P(\zeta_\ell = 2, \zeta_{\ell'} = 1)$ is also approximately equal to $\frac{1}{2} - \E{\invlogit{z_{\ell}}\invlogit{z_{\ell'}}}$. We can again use Lemma \ref{lemma:logit_e_a} to calculate this quantity.

If we plug these probabilities into the numerator, we get the quantity in the lemma.
\end{proof}

We explicate at a high level the upper and lower bound. As shown above, $\E{f(Z_\ell)f(Z_{\ell'})}$ provides an upper bound on $\E{\invlogit{Z_\ell}\invlogit{Z_{\ell'}}}$. This upper bound is comprised of computing the expectation over four areas: (1) $f(Z_\ell)f(Z_{\ell'}) = 0$, (2) $f(Z_\ell)f(Z_{\ell'}) = 1$, (3) either $f(Z_{\ell}) = 1$ or $f(Z_{\ell'}) = 1$, and (4) $Z_\ell, Z_{\ell'} \in [-\zEps, \zEps]\cross[-\zEps, \zEps]$. The technical lemmas discussed are for the calculation of $\frac{1}{2\zEps}\widetilde{E}_{z_\ell}$ and $\frac{1}{4\zEps^2}\widetilde{E}_{z_\ell, z_{\ell'}}$, which are the expectations from areas (3) and (4) respectively. In contrast, the lower bound without $\E{f(Z_\ell)f(Z_{\ell'})}$ is essentially the product of two terms. The first part, $-\invlogit{-\zEps}$, represents the greatest amount our approximation can be larger than the inverse logit function. Doing so allows us to marginalize out one of the variables from the bivariate normal distribution. After some rearrangement, the other part again represents the greatest differences in particular regions multiplied by the probability of those regions. It can be more "refined" because we are back to considering one variable. However, this part also reflects the algebraic choices that allowed us to transform the difference of products to a sum of products of terms that only depend on $Z_\ell$ or $Z_{\ell}$. Because this bound in Lemma \ref{lemma:logistic_e_prod_approx_bound} gets smaller as $\zEps$ increases, the bound is minimized at $\zEps = 2$.

\section{MCMC Sampling}
To sample from our models, we turn to Hamiltonian Monte Carlo (HMC). In HMC, each parameter that we are interested in sampling is given a momentum random variable \citep{NealMCMCUsingHamiltonian2011}. If we use the notation from Neal's paper to introduce HMC, a parameter is represented by $q_i$ and its momentum is represented by $p_i$. Instead of sampling from the posterior, we are now interested in sampling the joint distribution based on the Hamiltonian, i.e. $\frac{1}{\widetilde{T}}\exp(\frac{-U(\v{q}) + -K(\v{p})}{T})$. For this distribution based on the Hamiltonian, the potential energy, $U(\v{q})$, is defined to be the negative log posterior. On the other hand, the momentum, $K(\v{p})$, is defined to be $\v{p}^T M \v{p}$ for some mass matrix $M$. If we use the leapfrog integrator, the sampling updates are the following:
\begin{align*}
    p_i(t + \frac{\epsilon}{2}) &= p_i(t) - \epsilon\frac{\partial U}{\partial q_i}(q(t))\\
    q_i(t + \epsilon) &= q_i(t) + \epsilon (M^{-1}p(t + \frac{\epsilon}{2}))_i\\
    p_i(t + \epsilon) &= p_i(t + \frac{\epsilon}{2}) - \frac{\epsilon}{2}\frac{\partial U}{\partial q_i}(q(t + \epsilon))
\end{align*}

To understand why we need these different approaches,
we will derive the updates for the momentum vector to better understand how HMC will fit our models because this update is the derivative of the log posterior with respect to the current parameters \citep{NealMCMCUsingHamiltonian2011}. This is one of the key interactions between the sampler and the model. In particular, we will focus on how the momentum terms related to the Gaussian Process term will be updated. While this doesn't give us a complete picture of the HMC sampler, we do so because the updates for the other parameters are relatively straightforward. It will also illustrate the challenges of fitting these models.

\subsection{SvM and SvM-c Models}
\label{ssection:model_fit_svm_svm_c}
\textbf{SvM model} As discussed earlier, $\v{m}$ is distributed according to a Projected Gaussian Process. In their paper, Wang and Gelfand suggest sampling an auxiliary variable $\v{r} \in \mathbbm{R}^{N+}$ such that if $\v{Z_1}$ and $\v{Z_2}$ are distributed according to the description in \eqref{model:SvM}, $z_{1, \ell} = r_\ell\cos(m_\ell)$ and $z_{2, \ell} = r_\ell\sin(m_\ell)$ for $\ell \in 1, 2, \dots, N$ \citep{WangGelfandModelingSpaceSpaceTime2014}. In doing so, we have that
\begin{align}
p(\v{r}, \v{m}) = \left(\prod_{\ell = 1}^N r_\ell\right) \mathcal{N}\bigg(\big(\v{r}\cos(\v{m}), \v{r}\sin(\v{m})\big) \bigg| (\mu_1, \mu_2), \mathbbm{I}_2 \otimes \Sigma\bigg).
\label{eq:aug_projected_gp_centered}
\end{align}

If we refer to the parametrization in \eqref{eq:aug_projected_gp_centered} as the $\textit{circular centered parametrization}$, the updates for the momentum vector corresponding to $m_\ell$ and $r_\ell$ are respectively
\begin{align}
    -\frac{\partial U}{\partial q_{m_\ell}}(q(t)) &= \bigg(\rho_\ell\cos(y_\ell) + r_\ell\left(\Sigma^{-1}(\v{r}\cos(\v{m}) - \mu_1) \right)_\ell\bigg) \sin(m_\ell) - \nonumber\\
    & \qquad \bigg(\rho_\ell\sin(y_\ell) + r_\ell\left(\Sigma^{-1}(\v{r}\sin(\v{m}) - \mu_2) \right)_\ell\bigg)\cos(m_\ell)
    \label{lemma:SvM_centered_m_update}\\
    -\frac{\partial U}{\partial q_{r_\ell}}(q(t)) &= \frac{1}{r_\ell} - \left(\Sigma^{-1}(\v{r}\cos(\v{m}) - \mu_1) \right)_\ell \cos(m_\ell) - \nonumber\\ 
    & \qquad(\Sigma^{-1}(\v{r}\sin(\v{m}) - \mu_1))_\ell \sin(m_\ell)
    \label{lemma:SvM_centered_r_update}
\end{align}

We do not necessarily have to transform $\v{r}$ and $\v{m}$ directly to draws from the bivariate Gaussian Process. Instead, we can first transform them to draws from $2N$ independent $\textrm{N}(0, 1)$ distribution. In other words, if $(\v{\widetilde{Z}_1}, \v{\widetilde{Z}_2}) \sim \mathcal{N}(0, \mathbbm{I}_2 \otimes \mathbbm{I}_N)$, let $\widetilde{z}_{1, \ell} = \widetilde{r}_{\ell}\cos(\widetilde{m}_{\ell})$ and $\widetilde{z}_{2, \ell} = \widetilde{r}_{\ell}\sin(\widetilde{m}_{\ell})$ for $\v{\widetilde{r}} \in \mathbbm{R}^{N+}$, $\v{\widetilde{m}} \in [0, 2\pi)^N$, and $\ell = 1, 2, \dots, N$. Then, in a simplification of \ref{eq:aug_projected_gp_centered}, we have that for $\ell = 1, 2, \dots, N$,
\begin{align}
    p(\widetilde{r}_\ell) = r_\ell\exp{-\frac{1}{2}r_\ell^2} \qquad\qquad
    p(\widetilde{m}_\ell) = \textrm{Unif}(0, 2\pi).
    \label{eq:aug_projected_gp_noncentered}
\end{align}
Let $L$ denote the Cholesky decomposition of the covariance matrix of the Gaussian Process, $\Sigma$. We can then set $\v{Z_1} = L\v{\widetilde{Z}_1} + \mu_1$  and $\v{Z_2} = L\v{\widetilde{Z}_1} + \mu_2$ to change these draws from independent normal distributions to draws from our bivariate Gaussian Process. Finally, we can use the $\arctan^*$ defined in \ref{fct:arctan_star} to transform $\v{Z_1}$ and $\v{Z_2}$ to $\v{m}$.

If we refer to the parametrization in \eqref{eq:aug_projected_gp_noncentered} as the $\textit{circular non-centered parametrization}$ and use notation from before, the updates for the momentum vector corresponding to $\widetilde{m}_\ell$ and $\widetilde{r}_\ell$ are respectively
\begin{align}
    -\frac{\partial U}{\partial q_{\widetilde{m}_\ell}}(q(t)) &= \sum_{\ell' = 1}^N  -\rho_{\ell'}\sin(y_{\ell'} - m_{\ell'})\left(\frac{z_{2, \ell'}}{z^2_{1, \ell'} + z^2_{2, \ell'}} L_{\ell', \ell} \widetilde{r}_\ell \sin(\widetilde{m}_{\ell}) + \right.\nonumber\\
    & \qquad\qquad\qquad\qquad\qquad\qquad\quad\left. \frac{z_{1, \ell'}}{z^2_{1, \ell'} + z^2_{2, \ell'}} L_{\ell', \ell} \widetilde{r}_\ell \cos(\widetilde{m}_{\ell}) \right)
    \label{lemma:SvM_non_centered_m_update}\\
    -\frac{\partial U}{\partial q_{r_\ell}}(q(t)) &= \sum_{\ell' = 1}^N  -\rho_{\ell'}\sin(y_{\ell'} - m_{\ell'})\left(\frac{-z_{2, \ell'}}{z^2_{1, \ell'} + z^2_{2, \ell'}} L_{\ell', \ell} \cos(\widetilde{m}_{\ell}) + \right.\nonumber\\
    & \qquad\qquad\qquad\qquad\qquad\qquad\quad \left. \frac{z_{1, \ell'}}{z^2_{1, \ell'} + z^2_{2, \ell'}} L_{\ell', \ell} \sin(\widetilde{m}_{\ell}) \right) + \frac{1}{\widetilde{r}_\ell} - \widetilde{r}_\ell.
    \label{lemma:SvM_non_centered_r_update}
\end{align}

We will now compare the updates in \eqref{lemma:SvM_centered_m_update} and \eqref{lemma:SvM_centered_r_update} to the updates in \eqref{lemma:SvM_non_centered_m_update} and \eqref{lemma:SvM_non_centered_r_update}. We immediately see that the centered updates are much more expensive because they involve inverting the covariance matrix. If we look closer and set the momentum updates for $m_\ell$ and $\widetilde{m}_\ell$ to be 0, we see that 
\begin{align*}
    \frac{\sin(m_\ell)}{\cos(m_\ell)} &= \frac{\bigg(\rho_\ell\sin(y_\ell) + r_\ell\left(\Sigma^{-1}(\v{r}\sin(\v{m}) - \mu_2) \right)_\ell\bigg)}{ \bigg(\rho_\ell\cos(y_\ell) + r_\ell\left(\Sigma^{-1}(\v{r}\cos(\v{m}) - \mu_1) \right)_\ell\bigg)}\\
    \frac{\sin(\widetilde{m}_\ell)}{\cos(\widetilde{m}_\ell)} &= -\frac{\sum_{\ell' = 1}^N  \rho_{\ell'}\sin(y_{\ell'} - m_{\ell'}) \frac{z_{1, \ell'}}{z^2_{1, \ell'} + z^2_{2, \ell'}} L_{\ell', \ell} \widetilde{r}_\ell}{\sum_{\ell' = 1}^N  \rho_{\ell'}\sin(y_{\ell'} - m_{\ell'})\frac{z_{2, \ell'}}{z^2_{1, \ell'} + z^2_{2, \ell'}} L_{\ell', \ell} \widetilde{r}_\ell}.
\end{align*}
From this, we see that the update for $m_\ell$ tries to balance information from the data weighted by the concentration parameter $\rho_\ell$ against the prior information weighted by $r_\ell$. However, how the prior information is used is not fully clear because the difference between the current positions and the prior means are multiplied by the inverse of the covariance matrix. On the other hand, the update for $\widetilde{m}_\ell$ separates the terms dependent on $\sin(\widetilde{m}_\ell)$ and $\cos(\widetilde{m}_\ell)$. Of note, these terms are weighted by the concentration parameter at the location, the Cholesky decomposition, and $\widetilde{r}_\ell$. We see that the ratio is reversed and negated compared to the update for $m_\ell$. The derivative of $\arctan^*(z_{1, \ell}, z_{2, \ell})$ might cause this because as seen in Lemma \ref{lemma:arctan_star_derivative} in the Appendix, it assigns information from the other coordinates. 

On the other hand, the update for $r_\ell$ does not explicitly involve the data whereas the update for $\widetilde{r}_\ell$ does. This update makes it clear that $r_\ell$ is an auxiliary variable. This might be an issue because as discussed earlier, $r_\ell$ affects how strongly to weigh prior information and it might be good for the data to affect such a term. The update still does provide some control due to the $\frac{1}{r_\ell}$ term, which reduces the gradient for large values. Unfortunately, this may be counteracted by prior information and/or the current position of $m_\ell$. Meanwhile, the update for $\widetilde{r}_\ell$ involves both $-\widetilde{r}_\ell$ and $\frac{1}{r_\ell}$. This will help discourage both large and small values. Further, all observations weigh in unless as discussed previously, the ratio of $\frac{\sin(\widetilde{m}_\ell)}{\cos(\widetilde{m}_\ell)}$ matches the ratio of the information from the data.  

There is another additional fundamental issue with the noncentered circular parametrization. We sample $\v{r}$ and $\v{m}$ only because this gives us control over $\v{m}$ while providing an easier space to sample from. When we sample $\v{\widetilde{r}}$ and $\v{\widetilde{m}}$ as an analogue to this approach and transform it to $\v{m}$, this transformation is taking a convex combination of $\v{\widetilde{r}}\cos(\v{\widetilde{m}})$ and $\v{\widetilde{r}}\sin(\v{\widetilde{m}})$ with the corresponding Cholesky decomposition row as weights. It is as if we are sampling $\v{m}$ through $Z_1$ and $Z_2$. This becomes problematic because linear moves in $Z_1$ and $Z_2$ do not translate to linear moves in $\v{m}$. For instance, a linear move along a ray from the origin does not change $\v{m}$.

As a result, we turn to elliptical slice sampling \citep{MurrayEtAlEllipticalSliceSampling2010}.

\textbf{SvM-c} For this model, we have to marginalize out the cluster membership in order to sample for $z_\ell$. In other words, the probability for an observation $y_\ell$ is the following.
\[
p(y_\ell \mid \v{m_{k}}, \v{\rho_k}, \v{\lambda}) = \sum_k \lambda_k \vM{y_\ell}{m_{k, \ell'}}{\rho_{k, \ell}}.
\]

The updates are slightly altered. The centered update for the momentum vector for the SvM-c described in \eqref{model:SvM-c} corresponding to $z_{k, \ell}$ becomes 
\begin{align}
-&\frac{\partial U}{\partial q_{k, 1, \ell}}q_{k, 1, \ell}(t) = \nonumber\\
& \quad\frac{\exp(\rho_{k, \ell}cos(y_\ell - m_{k, \ell}))}{\sum_k \lambda_k \vM{y_\ell}{m_{k, \ell}}{\rho_{k, \ell}}}\left(\frac{-z_{k, 2, \ell}}{z^2_{k, 1, \ell} + z^2_{k, 2, \ell}} \rho_{k, \ell} \sin(y_\ell - m_{k, \ell})\right) - (\Sigma_k^{-1} \v{z_k})_\ell
\label{lemma:SvM-c_centered_m_update_1}\\
-&\frac{\partial U}{\partial q_{k, \ell, 2}}q_{k, \ell, 2}(t) = \nonumber\\
& \quad\frac{\exp(\rho_{k, \ell}cos(y_\ell - m_{k, \ell}))}{\sum_k \lambda_k \vM{y_\ell}{m_{k, \ell}}{\rho_{k, \ell}}}\left(\frac{z_{k, 1, \ell}}{z^2_{k, 1, \ell} + z^2_{k, 2, \ell}} \rho_{k, \ell} \sin(y_\ell - m_{k, \ell})\right) - (\Sigma_k^{-1} \v{z_k})_\ell
\label{lemma:SvM-c_centered_m_update_2} 
\end{align}
The non-centered update changes to 
\begin{align}
-&\frac{\partial U}{\partial q_{k, 1, \ell}}  q_{k, 1, \ell}(t) = \nonumber\\
& \quad \sum_{\ell' = 1}^N  \frac{\exp(\rho_{k, \ell'}cos(y_{\ell'} - m_{k, \ell'}))}{\sum_k \lambda_k \vM{y_{\ell'}}{m_{k, \ell'}}{\rho_{k, \ell'}})}\left(\frac{-z_{k, 2, \ell'}}{z^2_{k, 1, \ell'} + z^2_{k, 2, \ell'}} \rho_{k, \ell'} \sin(y_{\ell'} - m_{k, \ell'})\right) L_{\ell', \ell}  - \widetilde{z}_{k, \ell}
\label{lemma:SvM-c_non_centered_m_update_1}\\
-&\frac{\partial U}{\partial q_{k, 2, \ell}}q_{k, 2, \ell}(t) = \nonumber\\
& \quad\sum_{\ell' = 1}^N  \frac{\exp(\rho_{k, \ell'}cos(y_{\ell'} - m_{k, \ell'}))}{\sum_k \lambda_k \vM{y_{\ell'}}{m_{k, \ell'}}{\rho_{k, \ell'}})}\left(\frac{z_{k, 1, \ell'}}{z^2_{k, 1, \ell'} + z^2_{k, 2, \ell'}} \rho_{k, \ell'} \sin(y_{\ell'} - m_{k, \ell'})\right) L_{\ell', \ell}  - \widetilde{z}_{k, \ell}
\label{lemma:SvM-c_non_centered_m_update_2}.
\end{align}

These updates are similar to the \textit{SvM} updates from \ref{lemma:SvM_centered_m_update}, \ref{lemma:SvM_centered_r_update}, \ref{lemma:SvM_non_centered_m_update}, and \ref{lemma:SvM_non_centered_r_update}. The difference is that these updates are also weighted by how much an observation belongs to that particular's component. Because these distinctions aren't substantial, we again will use the elliptical slice sampling approach to sample for $(Z_{k, 1} - \mu_{k, 1}, Z_{k, 2} - \mu_{k, 2}) \mid \v{\varphi}_k, \nu_k, \v{\zeta}$ for $k \in 1, 2, \ldots K$. The only modification from the proposal in the previous section is that the likelihood is now defined to be the following:
\[
L(\v{y} \mid \v{m}_k, \v{\varphi}_k, \nu_k, \v{\zeta}) = \prod_{\ell} (\vM{y_\ell}{m_\ell}{\rho_\ell})^{\indFct{\zeta_\ell = k}}.
\]


\subsection{SvM-p Models} 
Here, the calculations are much more straightforward. Because the Gaussian Process assigns probability through a multivariate Normal distribution, there are two ways to sample $z_\ell$ without loss of generality. We can sample $z_\ell$ directly from GP(0, $\Sigma$) or the so-called centered parametrization. Alternatively, if we let $\Sigma = LL^T$ and $z = L\widetilde{z}$, we can sample $z_\ell$ according to the non-centered parametrization. We will present the updates from both parametrizations.

We again marginalize out the cluster assignment to sample for $z_\ell$. The probability for a single observation $y_\ell$ is the following:
\[
p(y_\ell \mid \v{z_{., \ell}}, \v{m}, \v{\rho}) = \sum\limits_k \lambda_{k, \ell}\vM{y_\ell}{m_k}{\rho_k}.
\]

Then, the centered parametrization update for the model described in \eqref{model:SvM-p} is
\begin{align}
    -\frac{\partial U}{\partial q_{k, \ell}}(q(t)) = \lambda_{k, \ell}\frac{\vM{y_\ell}{m_k}{\rho_k} - \sum_{k' = 1}^K\lambda_{k', \ell}\vM{y_\ell}{m_{k'}}{\rho_{k'}}}{p(y_\ell \mid \v{z_{., \ell}}, \v{m}, \v{\rho})}  - (\Sigma^{-1} z)_\ell
\label{lemma:mixture_prob_centered_m_update_a}
\end{align}
In the two components case, this simplifies to
\begin{align}
    -\frac{\partial U}{\partial q_{k, \ell}}(q(t)) &= \lambda_{1, \ell} (1 - \lambda_{1, \ell}) \frac{\vM{y_\ell}{m_1}{\rho_1} - \vM{y_\ell}{m_2}{\rho_2}}{p(y_\ell \mid z_\ell, m_1, m_2, \rho_1, \rho_2)}  - (\Sigma^{-1} z)_\ell.
\label{lemma:mixture_prob_centered_m_update-2_a}
\end{align}
The non-centered parametrization update is
\begin{align}
    -\frac{\partial U}{\partial q_\ell}(q(t)) &= \sum_{\ell'} L_{\ell', \ell} \lambda_{k, \ell'}\frac{\vM{y_{\ell'}}{m_k}{\rho_k} - \sum_{k' = 1}^K\lambda_{k', \ell'}\vM{y_{\ell'}}{m_{k'}}{\rho_{k'}}}{p(y_\ell' \mid \v{z_{., \ell'}}, \v{m}, \v{\rho})}  -  \widetilde{z}_\ell..
    \label{lemma:mixture_prob_noncentered_m_update_a}
\end{align}
The non-centered parametrization update in the two components case reduces to 
\begin{align}
     -\frac{\partial U}{\partial q_\ell}(q(t)) &= \sum_{\ell'} \lambda_{1, \ell'} (1 -  \lambda_{1, \ell'}) \frac{\vM{y_{\ell'}}{m_1}{\rho_1} - \vM{y_{\ell'}}{m_2}{\rho_2}}{p(y_{\ell'} \mid z_{\ell'}, m_1, m_2, \rho_1, \rho_2)} L_{\ell', \ell}  -  \widetilde{z}_\ell.
    \label{lemma:mixture_prob_noncentered_m_update-2_a}
\end{align}

From doing so, we see that the centered and noncentered parametrization differ in how they use the likelihood for updates and how they use covariance information. The centered update based on the likelihood compares the probability of an observation belonging to component $k$ against the probability of the observation belonging to the current mixture of von Mises distributions. This update is also weighted by the currently probability of the observation belongs to component $k$. In the two component case, this becomes a comparison the probability of the observation belonging to either component. On the other hand, the non-centered update uses the covariance information to pool these weighted comparisons. Because these updates seem sensible, we use HMC to sample from $\textit{SvM-p}$.

\subsection{Regularized Expectation Maximization}
To help with the sampling, we initialized it with results from applying Expectation Maximization (EM) to the unnormalized posterior or a regularized EM. Because $\Sigma$ might be close to singular, we ran regularized EM on the non-centered parametrization for our updates.

\textbf{SvM-c} Within the regularized EM framework, we let $\zeta_\ell$ for $\ell = 1, 2, \ldots, N$ be the latent variable and $\lambda_k$, $\v{\varphi_k}$, $\nu_k$, and $\v{\widetilde{z}_k}$ for $k = 1, 2, \ldots, K$ be the parameters. If we summarize the parameters as $\Theta$, the posterior can be written as following.
\begin{align*}
    \widetilde{p}(\v{\zeta}, \Theta \mid \v{y}, \v{x}) &:= \prod_\ell \prod_k \left(\lambda_k \vM{y_\ell}{m_{k, \ell}}{\rho_{k, \ell}}\right)^{\indFct{\zeta_\ell = k}} \prod_k \textrm{N}(\widetilde{z} \mid 0, \mathbbm{I}_n)\\ &\qquad\prod_\ell \prod_k \textrm{N}(\varphi_{k, \ell} \mid \nu_k, \varsigma^2) \prod_k \textrm{N}(\nu_k \mid 0, \tau^2).
\end{align*}

For the expected conditional log unnormalized posterior, we need a term to represent $P(\zeta_{\ell} = k \mid y_{\ell}, \lambda_k, \v{\widetilde{z}_k}, \v{\varphi_{., \ell}})$. We let $r_{k, \ell}$ be this term. To be explicit, we are attaching the draws from the Gaussian Process, $\v{\widetilde{z}_k}$, to $\v{m_{., \ell}}$. Then, it is defined in the following manner.
\begin{align}
    r_{k, \ell} &:= \frac{\lambda_k \vM{y_\ell}{m_{k, \ell}}{\rho_{k, \ell}}}{\sum_{k'} \lambda_{k'} \vM{y_\ell}{m_{k', \ell}}{\rho_{k', \ell}}}.
\label{eq:SvM-c_r_k_m}
\end{align}
With this term, the expected conditional log unnormalized posterior is the following.
\begin{equation}
    \begin{aligned}
    \E{\log(\widetilde{p}(\v{\zeta}, \Theta \mid \v{y}, \v{x}))} &= \sum_\ell \sum_k r_{k, \ell}\left(\log(\lambda_k) + \log(\vM{y_\ell}{m_{k, \ell}}{\rho_{k, \ell}})\right) +\\
    & \qquad\sum_\ell \sum_k\log(\textrm{N}(\widetilde{z}_{k, \ell} \mid 0, 1)) +\\ 
    & \qquad \sum_\ell\sum_k \log(\textrm{N}(\varphi_{k, \ell} \mid \nu_k, \varsigma^2)) + \sum_k \log(\textrm{N}(\nu_k \mid 0, \tau^2)).
\end{aligned}    
\label{eq:SvM-c_exp_cond_log_post}
\end{equation}

With this, we can now state the regularized EM algorithm.
\begin{itemize}
    \item \textbf{E step:} Calculate the expected conditional log unnormalized posterior from \ref{eq:SvM-c_exp_cond_log_post} and $r_{k, \ell}$ from \ref{eq:SvM-c_r_k_m}.
    \item \textbf{M step:}
    \begin{itemize}
        \item For $k = 1, 2, \dots, K$, set $\lambda_k = \frac{1}{N} \sum\limits_\ell r_{k, \ell}$.
        \item For $k = 1, 2, \dots, K$, alternate updating $\v{\widetilde{z}_{k, 1, .}}$ and $\v{\widetilde{z}_{k, 2, .}}$ using coordinate gradient ascent and the following gradients.
        \begin{equation*}
            \begin{aligned}
            \frac{\partial}{\partial \widetilde{z}_{k, 1, \ell}} &  \E{\log(\widetilde{p}(\v{\zeta}, \Theta \mid \v{y}, \v{x}))} =\\
            & \quad\sum_{\ell'} L_{\ell', \ell} r_{k, \ell'} \frac{-z_{k, 2, \ell'}}{z^2_{k, 1, \ell'} + z^2_{k, 2, \ell'}}\rho_{k, \ell'}\sin(y_{\ell'} - m_{k, \ell'}) - \widetilde{z}_{k, \ell}.\\
            \frac{\partial}{\partial \widetilde{z}_{k, 2, \ell}} & \E{\log(\widetilde{p}(\v{\zeta}, \Theta \mid \v{y}, \v{x}))} =\\
            & \quad \sum_{\ell'} L_{\ell', \ell} r_{k, \ell'} \frac{z_{k, 1, \ell'}}{z^2_{k, 1, \ell'} + z^2_{k, 2, \ell'}}\rho_{k, \ell'}\sin(y_{\ell'} - m_{k, \ell'}) - \widetilde{z}_{k, \ell}.
            \end{aligned}
        \end{equation*}
        \item For $k = 1, 2, \dots, K$ and $\ell = 1, 2, \dots N$, update $\varphi_{k, \ell}$ using gradient ascent and the gradient below.
        \begin{equation*}
            \begin{aligned}
                \frac{\partial}{\partial \varphi_{k, \ell}} & \E{\log(\widetilde{p}(\v{\zeta}, \Theta \mid \v{y}, \v{x}))} =\\
                & r_{k, \ell}\rho_{k, \ell} \left(\cos(y_\ell - m_{k, \ell}) - \frac{I_{-1}(\rho_{k, \ell})}{I_{0}(\rho_{k, \ell})}\right) - \frac{1}{\varsigma^2}(\varphi_{k, \ell} - \nu_k).
            \end{aligned}
        \end{equation*}
        \item For $k = 1, 2, \dots K$, set $\nu_k = \frac{\sum_\ell \frac{\varphi_{k, \ell}}{\varsigma^2}}{\frac{N}{\varsigma^2} + \frac{1}{\tau^2}}$.
    \end{itemize}
\end{itemize}
We wish to discuss the updates in further detail. The updates for $\lambda_k$ and $\nu_k$ are similar to the maximum likelihood estimator for a multinomial distribution's probabilities and a posterior conjugate Normal distribution's mean. However, because we don't know the assignments, we weight each observation with $r_{k, \ell}$ instead. Next, to understand the update for $\varphi_{k, \ell}$, we note that according to the definition of the modified Bessel function in \ref{eq:bessel}, $I_{-1}(\rho_k) = I_{1}(\rho_k)$ or the expected concentration for a von Mises distribution. The update is intuitively trying to find the concentration parameter that best matches the expected concentration of that observation to the observation's sample concentration weighted by how much the observations belongs to that cluster. The hierarchical prior helps the update because the update is also trying to match the concentration parameter against the hierarchical mean, $\nu_k$, and can borrow strength from other observations. 

Even though we did not derive the Hamiltonian Monte Carlo $\textit{SvM}$ non-centered parametrization momentum update for $\widetilde{z}_{k, \ell}$, that update is similar to the regularized EM update for $\widetilde{z}_{k, \ell}$. At a high level, this makes sense because $\zeta_\ell$ assigns an observation to one component. Since we can only compute the conditional probability of $\zeta_\ell$, we have to weight the update by $r_{k, \ell}$. It is still interesting that we can use this standard parametrization. It might be possible to do so because we are only interested in optimizing a point estimate and not interested in exploring the space. One consequence of this similarity is that like Hamiltonian Monte Carlo, we choose to use a coordinate gradient ascent approach to update all $\widetilde{z}_k$ because an update for a coordinate can be written as the matrix multiplication of the transpose of the Cholesky and the other gradient information.

\textbf{SvM-p} Within the regularized EM framework, we let $\zeta_\ell$ for $\ell = 1, 2, \dots, N$ be the latent variable and $\widetilde{z_k}$, $m_k$, and $\rho_k$ for $k = 1, 2, \ldots, K$ be the parameters. To be clear, we attach the draw from the Gaussian Process, $\widetilde{z_k}$, to $\v{\lambda_k}$. If we again summarize the parameters as $\Theta$, the unnormalized posterior can be written as following.
\begin{align*}
    \widetilde{p}(\v{\zeta}, \Theta \mid \v{y}, \v{x}) &:= \prod_\ell \prod_k \left(\lambda_{k, \ell} \vM{y_\ell}{m_{k}}{\rho_{k}}\right)^{\indFct{\zeta_\ell = k}}\prod_k \textrm{N}(\v{\widetilde{z}_k} \mid 0, \mathbbm{I}_N)\\
    &\qquad \prod_k \Gamma(\rho_k \mid 1, 1) \prod_k \textrm{Unif}(m_k \mid 0, 2\pi).
\end{align*}

For the expected conditional log likelihood, we need to tweak our definition of $r_{k, \ell}$ to remove the observation index from the von Mises distribution parameters and put them on the mixing probability. In other words, we let $r_{k, \ell}$ be the following. 
\begin{align}
    r_{k, \ell} := \frac{\lambda_{k, \ell} \vM{y_\ell}{m_{k}}{\rho_{k}}}{\sum\limits_{k'} \lambda_{k', \ell} \vM{y_\ell}{m_{k'}}{\rho_{k'}}}.
    \label{eq:SvM-p_r_k_m}
\end{align}
Then, the expected conditional log unnormalized posterior is the following:
\begin{equation}
    \begin{aligned}
        \E{\log( \widetilde{p}(\v{\zeta}, \Theta \mid \v{y}, \v{x}))} &= \sum_\ell \sum_k r_{k, \ell}\left(\log(\lambda_{k, \ell}) + \log(\vM{y_\ell}{m_{k}}{\rho_{k}})\right) +\\ 
        & \qquad\sum_k \log(\textrm{N}(\v{\widetilde{z}_k} \mid 0, \mathbbm{I}_N)) + \\ 
        & \qquad\sum_k \log(\textrm{Unif}(m_k \mid 0, 2\pi)) + \sum_k \log(\Gamma(\rho_k \mid 1, 1)).
    \end{aligned}
\label{eq:SvM-p_exp_cond_log_post}
\end{equation}

Our regularized EM algorithm is the following.
\begin{itemize}
    \item \textbf{E step:} Calculate the expected conditional log unnormalized posterior from \ref{eq:SvM-p_exp_cond_log_post} and $r_{k, \ell}$ from \ref{eq:SvM-p_r_k_m}.
    \item \textbf{M step:}
    \begin{itemize}
        \item For $k = 1, 2, \dots, K - 1$, update $\v{\widetilde{z}_k}$ using coordinate gradient ascent and the following gradient.
        \begin{align*}
            \frac{\partial}{\partial \widetilde{z}_{k, \ell}}  \E{\log(\widetilde{p}(\v{\zeta}, \Theta \mid \v{y}, \v{x}))} &= \sum_{\ell'}  (r_{k, \ell'} - \lambda_{k, \ell'}) L_{\ell', \ell} - \widetilde{z}_{k, \ell}.
        \end{align*}
        \item For $k = 1, 2, \dots, K$, set $m_k = \textrm{arctan}^*\left(\frac{\sum_\ell r_{k, \ell}\sin(y_\ell)}{\sum_\ell r_{k, \ell}\cos(y_\ell)}\right)$.
        \item For $k = 1, 2, \dots, K$, update $\rho_k$ using gradient ascent and the following gradient.
        \begin{align*}
            \frac{\partial}{\partial \rho_k}  \E{\log(\widetilde{p}(\v{\zeta}, \Theta \mid \v{y}, \v{x}))} &= \sum_{\ell = 1}^{n} r_{k, \ell} \left(\cos(y_\ell - m_k) - \frac{I_{-1}(\rho_k)}{I_{0}(\rho_k)}\right) - 1
        \end{align*}
    \end{itemize}
\end{itemize}
Note that the order of means can be kept by sorting the results at the end. Then, we make a few remarks about these updates. First, the updates for $m_k$ and $\rho_k$ in the $\textit{SvM-p}$ model are simpler than the updates in the $\textit{SvM-c}$ model due to our model choices. For instance, because the prior on $m_k$ does not depend on $m_k$, $m_k$ is maximized when it is set equal to the weighted circular mean. The update becomes slightly more complicated if we use a von Mises distribution as a prior instead. Next, the update for $\rho_k$ reflects the fact that the best $\rho_k$ for all the weighted observation makes the expected concentration match the weighted sample concentration. This update does not consider the effect of the prior. It seems that in this case, the prior is providing a constant pull back to 0. This constant pull back is controlled by the $b$ parameter in the prior Gamma distribution. If $a$ is not 1, then the gradient will also include a $\frac{1}{\rho_k}$ term.


\textbf{iVM} While this is an inversion of how we introduced the models, we do so because we can collect the parts independent of location information from the EM algorithm for the previous two models for the $\textit{iVM}$ model's regularized EM algorithm. As before, we let $\zeta_\ell$ for $\ell = 1, 2, \dots, N$ be the latent variable and $\lambda_k$, $m_k$, and $\rho_k$ for $k = 1, 2, \ldots, K$ be the parameters. If we again summarize the parameters as $\Theta$, the unnormalized posterior can be written as following.
\begin{align*}
    \widetilde{p}(\v{\zeta}, \Theta \mid \v{y}, \v{x}) &:= \prod_\ell \prod_k \left(\lambda_{k} \vM{y_\ell}{m_{k}}{\rho_{k}}\right)^{\indFct{\zeta_\ell = k}} \prod_k \Gamma(\rho_k \mid a_k, b_k) \prod_k \vM{m_k}{u_k}{c_k}.
\end{align*}

We again need to tweak the definition of $r_{k, \ell}$. In other words, we let $r_{k, \ell}$ be the following. 
\begin{align}
    r_{k, \ell} := \frac{\lambda_{k} \vM{y_\ell}{m_{k}}{\rho_{k}}}{\sum\limits_{k'} \lambda_{k', \ell} \vM{y_\ell}{m_{k'}}{\rho_{k'}}}.
    \label{eq:iVM_r_k_m}
\end{align}
Then the expected conditional log unnormalized posterior is the following:
\begin{equation}
    \begin{aligned}
        \E{\log( \widetilde{p}(\v{\zeta}, \Theta \mid \v{y}, \v{x}))} &= \sum_\ell \sum_k r_{k, \ell}\left(\log(\lambda_{k}) + \log(\vM{y_\ell}{m_{k}}{\rho_{k}})\right) + \\ 
        & \qquad\sum_k \log(\vM{m_k}{u_k}{c_k}) + \sum_k \log(\Gamma(\rho_k \mid a_k, b_k)).
    \end{aligned}
\label{eq:iVM_exp_cond_log_post}
\end{equation}

For simplicity, we assume that $c_k = 0$ and $a_k = 1$ and $b_k = 0$. At a high level, we have an uniform prior on the mean direction and have a flat prior on the concentration parameter. While the latter is an improper prior, the posterior is still proper. Our EM algorithm is the following.
\begin{itemize}
    \item \textbf{E step:} Calculate the expected conditional log unnormalized posterior from \ref{eq:iVM_exp_cond_log_post} and $r_{k, \ell}$ from \ref{eq:iVM_r_k_m}.
    \item \textbf{M step:}
    \begin{itemize}
        \item For $k = 1, 2, \dots, K$, set $\lambda_k = \frac{1}{N} \sum\limits_\ell r_{k, \ell}$.
        \item For $k = 1, 2, \dots, K$, set $m_k = \textrm{arctan}^*\left(\frac{\sum_\ell r_{k, \ell}\sin(y_\ell)}{\sum_\ell r_{k, \ell}\cos(y_\ell)}\right)$.
        \item For $k = 1, 2, \dots, K$, update $\rho_k$ using gradient ascent and the following gradient.
        \begin{align*}
            \frac{\partial}{\partial \rho_k}  \E{\log(\widetilde{p}(\v{\zeta}, \Theta \mid \v{y}, \v{x}))} &= \sum_{\ell = 1}^{n} r_{k, \ell} \left(\cos(y_\ell - m_k) - \frac{I_{-1}(\rho_k)}{I_{0}(\rho_k)}\right)
        \end{align*}
    \end{itemize}
\end{itemize}

\subsection{Calculating posterior predictive probability}
Let $x^*$ represent the withheld locations, $y^*$ the withheld data, $\theta$ a posterior draw for the parameters based on $x$ and $y$, and $\theta^*$ a draw for the parameters for $x^*$ and $y^*$. Set $N$ to be the number of locations and $N^*$ to be the number of withheld locations. The posterior predictive probability is given as follows.
\begin{align}
p(\v{y^*} \mid x, x^*, \v{y}) := \int \int \textrm{p}(y^* \mid \theta^*) \textrm{p}(\theta^* \mid \theta, x, x^*)  \textrm{p}(\theta \mid x, y) d\theta^* d\theta.
\label{eq:post_pred_prob_a}    
\end{align}
We now discuss how to calculate it for our models.

\noindent\textbf{SvM-c:} Here, $\theta = \{\{\v{z}_{k, 1}, \v{z}_{k, 2}\}_{k = 1}^K, \{\v{\varphi_{k}}\}_{k = 1}^K, \{\nu_k\}_{k = 1}^K, {\lambda}_{k = 1}^K\}$. Further, for the two Gaussian processes of component k, let $\Sigma^*_{k, 1}$ and $\Sigma^*_{k, 2}$ be the covariance matrices based on $x^*$ and $\widetilde{\Sigma}_{k, 1}$ and $\widetilde{\Sigma}_{k, 2}$ be the covariance matrices based on $x$ and $x^*$. In other words, $\widetilde{\Sigma}_{k, 1, i, j} = K(x_i, x^*_j)$ for some kernel function $K(\cdot, \cdot)$. Then, the posterior predictive probability for posterior draw $i$ can be found in the following procedure:
\begin{enumerate}
    \item Draw $\v{z^*_{k, 1}} \sim \textrm{N}(\v{\mu_{k, 1}} + \widetilde{\Sigma}_{k, 1}^T \Sigma_{k, 1}^{-1} (\v{z_{k, 1}} - \v{\mu_{k, 1}}), \Sigma^*_{k, 1} - \widetilde{\Sigma}_{k, 1}^T \Sigma_{k, 1}^{-1} \widetilde{\Sigma}_{k, 1})$ and \\ $\v{z^*_{k, 2}} \sim \textrm{N}(\v{\mu_{k, 2}} + \widetilde{\Sigma}_{k, 2}^T \Sigma_{k, 2}^{-2} (\v{z_{k, 2}} - \v{\mu_{k, 2}}), \Sigma^*_{k, 2} - \widetilde{\Sigma}_{k, 2}^T \Sigma_{k, 2}^{-2} \widetilde{\Sigma}_{k, 2})$ for $k = 1, 2, \ldots, K$. 
    \item Set $\v{m^*_k} = \arctan^*(\v{z^*_{k, 1}}, \v{z^*_{k, 2}})$.
    \item Draw $\v{\varphi^*_k} \sim \textrm{N}(\nu_k, \varsigma^2)$. Set $\v{\rho^*} = \exp(\v{\varphi^*_k})$.
    \item Let $p^*_{i, 1} = \prod_{n^* = 1}^{N^*} \left(\sum_k \lambda_k\vM{y^*_{n^*}}{m^*_{k, n^*}}{\rho^*_{k, n^*}}\right)$.
    \item Repeat the previous steps $M - 1$ more times for some $M \in \mathbbm{N}$ and denote the $m^{th}$ result as $p^*_{i, m}$.
    \item If there are $I$ iterations, return $\frac{1}{M} \sum_m \frac{1}{I} \sum_i p^*_{i, m}$ for the posterior predictive probability.
\end{enumerate}
\noindent\textbf{SvM-p:} Here, $\theta = \{\{\v{z}_{k, 1}, \v{z}_{k, 2}\}_{k = 1}^K, \{\v{m_{k}}\}_{k = 1}^K, \{\rho_k\}_{k = 1}^K\}$. Further, for the Gaussian processes of component k, let $\Sigma^*_{k}$ be the covariance matrix based on $x^*$ and $\widetilde{\Sigma}_{k}$ be the covariance matrices based on $x$ and $x^*$. In other words, $\widetilde{\Sigma}_{k, i, j} = K(x_i, x^*_j)$ for some kernel function $K(\cdot, \cdot)$. Then, the posterior predictive probability for posterior draw $i$ can be found in the following procedure:
\begin{enumerate}
    \item Draw $\v{z^*_{k}} \sim \textrm{N}(\v{\mu_{k}} + \widetilde{\Sigma}_{k}^T \Sigma_{k}^{-1} (\v{z_{k}} - \v{\mu_{k}}), \Sigma^*_{k} - \widetilde{\Sigma}_{k}^T \Sigma_{k}^{-1} \widetilde{\Sigma}_{k})$ for $k = 1, 2, \ldots, K - 1$. 
    \item For $n^* \in 1, 2, \ldots, N^*$, set $\v{\lambda^*_{\cdot, n^*}} = \Psi^{-1}(z^*_{1, n^*}, z^*_{2, n^*}, \ldots, z^*_{K - 1, n^*})$.
    \item Let $p^*_{i, 1} = \prod_{n^*} \left(\sum_k \lambda_{k, n^*}\vM{y^*_{n^*}}{m_{k}}{\rho_{k}}\right)$.
    \item Repeat the previous steps $M - 1$ more times for some $M \in \mathbbm{N}$ and denote the $m^{th}$ result as $p^*_{i, m}$.
    \item If there are $I$ iterations, return $\frac{1}{M} \sum_m \frac{1}{I} \sum_i p^*_{i, m}$ for the posterior predictive probability.
\end{enumerate}

.

\section{MCMC Sampling Proofs}
\subsection{HMC}
\subsubsection{SvM-c}
\begin{lemma}
The derivatives for $\arctan^*(z_1, z_2)$ are 
\begin{align*}
    \frac{d}{d z_1} \arctan^*(z_1, z_2) = \frac{-z_2}{z_1^2 + z_2^2}\\
    \frac{d}{d z_2} \arctan^*(z_1, z_2) = \frac{z_1}{z_1^2 + z_2^2}.
\end{align*}
\label{lemma:arctan_star_derivative}
\end{lemma}
\begin{proof}
By the chain rule, the derivative for $\arctan^*(z_1, z_2)$ first involves taking the derivative of $\arctan(\frac{z_2}{z_1})$. The rest follows by taking the appropriate derivative of $\frac{z_2}{z_1}$ and suitable rearrangement of the gradient.
\end{proof}

\begin{lemma}
For the model specified in \eqref{model:SvM} and using the parametrization given in \eqref{eq:aug_projected_gp_centered}, the update for the momentum vector are respectively
\begin{align}
    -\frac{\partial U}{\partial q_{m_\ell}}(q(t)) &= \bigg(\rho_\ell\cos(y_\ell) + r_\ell\left(\Sigma^{-1}(\v{r}\cos(\v{m}) - \mu_1) \right)_\ell\bigg) \sin(m_\ell) - \nonumber\\
    & \qquad \bigg(\rho_\ell\sin(y_\ell) + r_\ell\left(\Sigma^{-1}(\v{r}\sin(\v{m}) - \mu_2) \right)_\ell\bigg)\cos(m_\ell)
    \label{lemma:SvM_centered_m_update_a}\\
    -\frac{\partial U}{\partial q_{r_\ell}}(q(t)) &= \frac{1}{r_\ell} - \left(\Sigma^{-1}(\v{r}\cos(\v{m}) - \mu_1) \right)_\ell \cos(m_\ell) - \nonumber\\ 
    & \qquad(\Sigma^{-1}(\v{r}\sin(\v{m}) - \mu_1))_\ell \sin(m_\ell)
    \label{lemma:SvM_centered_r_update_a}
\end{align}
\end{lemma}
\begin{proof}
Suppressing all other parts that do not depend on $m_\ell$ and $r_\ell$, we have that 
\begin{align*}
\log(p(m_\ell, r_\ell \mid \v{y})) &= \rho_\ell\cos(y_\ell - m_\ell) + \sum_{\ell = 1}^N \log(r_\ell) - \frac{1}{2}(\v{r}\cos(\v{m}) - \mu_1)^{T}\Sigma^{-1}(\v{r}\cos(\v{m}) - \mu_1) -\\
& \qquad \frac{1}{2}(\v{r}\sin(\v{m}) - \mu_2)^{T}\Sigma^{-1}(\v{r}\sin(\v{m}) - \mu_2) + C.
\end{align*}

Without loss of generality, let $\v{v} = \v{r}\cos(\v{m}) - \mu_1$. Note that $\frac{d v_{\ell'}}{d m_\ell} = -r_\ell\sin(m_\ell)$ and $\frac{d v_{\ell'}}{d r_\ell} = \cos(m_\ell)$ for $\ell' = \ell$ and 0 otherwise. Hence,
\[
\frac{d}{d m_\ell} \frac{1}{2}(\v{r}\cos(\v{m}) - \mu_1)^{T}\Sigma^{-1}(\v{r}\cos(\v{m}) - \mu_1) = \left(\Sigma^{-1}(\v{r}\cos(\v{m}) - \mu_1)\right)_\ell -r_\ell\sin(m_\ell)
\] 
and
\[
\frac{d}{d r_\ell} \frac{1}{2}(\v{r}\cos(\v{m}) - \mu_1)^{T}\Sigma^{-1}(\v{r}\cos(\v{m}) - \mu_1) = \left(\Sigma^{-1}(\v{r}\cos(\v{m}) - \mu_1)\right)_\ell \cos(m_\ell).
\]
A similar result holds for $\frac{1}{2}(\v{r}\sin(\v{m}) - \mu_2)^{T}\Sigma^{-1}(\v{r}\sin(\v{m}) - \mu_2)$.

If we then take the derivative of $\log(r_\ell)$, we get the quantity in \ref{lemma:SvM_centered_r_update_a}. On the other hand, for the quantity in \ref{lemma:SvM_centered_m_update_a},  
\begin{equation*}
    \begin{aligned}
        \frac{d}{d m_\ell} \rho_\ell \cos(y_\ell - m_\ell) &= -\rho_\ell \sin(y_\ell - m_\ell)\\
        &= -\rho_\ell (\sin(y_\ell)\cos(m_\ell) - \sin(m_\ell)\cos(y_\ell)).
    \end{aligned}
\end{equation*}
Some rearrangement then gives us the quantity in \ref{lemma:SvM_centered_m_update_a}. 
\end{proof}

\begin{lemma}
For the model specified in \eqref{model:SvM} and using the parametrization given in \eqref{eq:aug_projected_gp_noncentered}, the update for the momentum vectors are respectively
\begin{align}
    -&\frac{\partial U}{\partial q_{\widetilde{m}_\ell}}(q(t)) = \nonumber\\
    & \sum_{\ell' = 1}^N  -\rho_{\ell'}\sin(y_{\ell'} - m_{\ell'})\left(\frac{z_{2, \ell'}}{z^2_{1, \ell'} + z^2_{2, \ell'}} L_{\ell', \ell} \widetilde{r}_\ell \sin(\widetilde{m}_{\ell}) + \frac{z_{1, \ell'}}{z^2_{1, \ell'} + z^2_{2, \ell'}} L_{\ell', \ell} \widetilde{r}_\ell \cos(\widetilde{m}_{\ell}) \right)
    \label{lemma:SvM_non_centered_m_update_a}\\
    -& \frac{\partial U}{\partial q_{r_\ell}}(q(t)) = \nonumber\\
    & \sum_{\ell' = 1}^N  -\rho_{\ell'}\sin(y_{\ell'} - m_{\ell'})\left(\frac{-z_{2, \ell'}}{z^2_{1, \ell'} + z^2_{2, \ell'}} L_{\ell', \ell} \cos(\widetilde{m}_{\ell}) + \frac{z_{1, \ell'}}{z^2_{1, \ell'} + z^2_{2, \ell'}} L_{\ell', \ell} \sin(\widetilde{m}_{\ell}) \right) + \frac{1}{\widetilde{r}_\ell} - \widetilde{r}_\ell.
    \label{lemma:SvM_non_centered_r_update_a}
\end{align}
\end{lemma}
\begin{proof}
Using the chain rule, we have that
\[
\frac{d}{d\widetilde{m}_\ell} \rho_{\ell'}\cos(y_{\ell'} - m_{\ell'}) = (\frac{d}{d m_{\ell'}} \rho_{\ell'}\cos(y_{\ell'} - m_{\ell'}))\left(\frac{d m_{\ell'}}{d z_{1, \ell'}} \frac{d z_{1, \ell'}}{d \widetilde{z}_{1, \ell}} \frac{d \widetilde{z}_{1, \ell}}{d\widetilde{m}_\ell} + \frac{d m_{\ell'}}{d z_{2, \ell'}} \frac{d z_{2, \ell'}}{d \widetilde{z}_{2, \ell}} \frac{d \widetilde{z}_{2, \ell}}{d\widetilde{m}_\ell}\right)
\]
and
\[
\frac{d}{d\widetilde{r}_\ell} \rho_{\ell'}\cos(y_{\ell'} - m_{\ell'}) = (\frac{d}{d m_{\ell'}} \rho_{\ell'}\cos(y_{\ell'} - m_{\ell'}))\left(\frac{d m_{\ell'}}{d z_{1, \ell'}} \frac{d z_{1, \ell'}}{d \widetilde{z}_{1, \ell}} \frac{d \widetilde{z}_{1, \ell}}{d\widetilde{r}_\ell} + \frac{d m_{\ell'}}{d z_{2, \ell'}} \frac{d z_{2, \ell'}}{d \widetilde{z}_{2, \ell}} \frac{d \widetilde{z}_{2, \ell}}{d\widetilde{r}_\ell}\right).
\]
Again, $\frac{d \widetilde{z}_{1, \ell'}}{d\widetilde{m}_\ell} = -r_\ell\sin(m_\ell)$ and $\frac{d \widetilde{z}_{2, \ell'}}{d\widetilde{m}_\ell} = r_\ell\cos(m_\ell)$ if $\ell' = \ell$ and $0$ otherwise. A similar result holds for $\frac{d \widetilde{z}_{1, \ell'}}{d\widetilde{r}_\ell}$ and $\frac{d \widetilde{z}_{2, \ell'}}{d\widetilde{r}_\ell}$. Then, by definition, $\frac{d z_{1, \ell'}}{d \widetilde{z}_{1, \ell}} = \frac{d z_{2, \ell'}}{d \widetilde{z}_{2, \ell}} = L_{j, i}$. Finally, Lemma \ref{lemma:arctan_star_derivative} gives us $\frac{d m_{\ell'}}{d z_{1, \ell'}}$ and $\frac{d m_{\ell'}}{d z_{2, \ell'}}$. 

Because the distribution on $\widetilde{m}_\ell$ is uniform, We get the quantity in \ref{lemma:SvM_non_centered_m_update_a} by summing over all observations. For $\widetilde{r}_\ell$, we also have to add in the derivative of $\log(r_\ell) - \frac{1}{2}r_\ell^2$.
\end{proof}

\subsubsection{SvM-p}
\begin{lemma}
For the model specified in \eqref{model:SvM-p}, the update for the momentum vector corresponding to $z_\ell$ in the centered parametrization is 
\[
\lambda_{k, \ell}\frac{\vM{y_\ell}{m_k}{\rho_k} - \sum_{k' = 1}^K\lambda_{k', \ell}\vM{y_\ell}{m_{k'}}{\rho_{k'}}}{p(y_\ell \mid \v{z_{., \ell}}, \v{m}, \v{\rho})}  - (\Sigma^{-1} z)_\ell
\]
\end{lemma}
\begin{proof}
We have that
\begin{equation*}
    \begin{aligned}
        & \frac{d}{d z_{k, \ell}} \log(p(z, \v{m}, \v{\rho} \mid \bm{y}))\\
        &\propto \frac{d}{d z_{k, \ell}} (\log(p(y_\ell \mid z, \v{m}, \v{\rho})) + \log(\textrm{GP}(z \mid 0, \Sigma)))\\
        &= \sum_{k' = 1}^K \frac{d\lambda_{k', \ell}}{d z_{k, \ell}} \frac{d}{d\lambda_{k', \ell}} \log(p(y_\ell \mid z, \v{m}, \v{\rho})) + \frac{d}{d z_\ell}\log(\textrm{GP}(z \mid 0, \Sigma))\\
        &= \frac{\sum_{k' = 1}^{K} \vM{y_\ell}{m_{k'}}{\rho_{k'}}\frac{d\lambda_{k', \ell}}{d z_{k', \ell}}}{p(y_\ell \mid z, \v{m}, \v{\rho})} - (\Sigma^{-1} z)_\ell.
    \end{aligned}
\end{equation*}
If $k' = k$, then
\begin{equation*}
    \begin{aligned}
        \frac{d\lambda_{k, \ell}}{d z_{k, \ell}}
        &= \frac{d}{d z_{k, \ell}} \frac{\exp(z_{k, \ell})}{1 + \sum_{k' = 1}^{K - 1}\exp(z_{k', \ell})}\\ 
        &= \frac{\exp(z_{k, \ell})}{1 + \sum_{k' = 1}^{K - 1}\exp(z_{k', \ell})} - \left(\frac{\exp(z_{k, \ell})}{1 + \sum_{k' = 1}^{K - 1}\exp(z_{k', \ell})}\right)^2\\
        &= \lambda_{k, \ell}(1 - \lambda_{k, \ell})\\
    \end{aligned}
\end{equation*}
If $k' \neq k$, $k' \neq K$ then
\begin{equation*}
    \begin{aligned}
        \frac{d\lambda_{k', \ell}}{d z_{k, \ell}}
        &= \frac{d}{d z_{k', \ell}} \frac{\exp(z_{k', \ell})}{1 + \sum_{k' = 1}^{K - 1}\exp(z_{k', \ell})}\\ 
        &= -\frac{\exp(z_{k, \ell})\exp(z_{k', \ell})}{\left(1 + \sum_{k' = 1}^{K - 1}\exp(z_{k', \ell})\right)^2}\\
        &= -\lambda_{k, \ell}\lambda_{k', \ell}.\\
    \end{aligned}
\end{equation*}
If $k' = K$, then
\begin{equation*}
    \begin{aligned}
        \frac{d\lambda_{k', \ell}}{d z_{k, \ell}}
        &= \frac{d}{d z_{k', \ell}} \frac{1}{1 + \sum_{k' = 1}^{K - 1}\exp(z_{k', \ell})}\\ 
        &= -\frac{\exp(z_{k, \ell})}{\left(1 + \sum_{k' = 1}^{K - 1}\exp(z_{k', \ell})\right)^2}\\
        &= -\lambda_{k, \ell}\lambda_{K, \ell}.\\
    \end{aligned}
\end{equation*}
\end{proof}

\begin{lemma}
For the model specified in \eqref{model:SvM-p-2}, the update for the momentum vector corresponding to $z_\ell$ in the centered parametrization is
\[
\lambda_{1, \ell} (1 - \lambda_{1, \ell}) \frac{\vM{y_\ell}{m_1}{\rho_1}  - \vM{y_\ell}{m_2}{\rho_2}}{p(y_\ell \mid z_\ell, m_1, m_2, \rho_1, \rho_2)}  - (\Sigma^{-1} z)_\ell.
\]
\end{lemma}
\begin{proof}
We have that
\begin{equation*}
    \begin{aligned}
        & \frac{d}{d z_\ell} \log(p(z, m_1, m_2, \rho_1, \rho_2 \mid \bm{y}))\\
        &\propto \frac{d}{d z_\ell} \log(p(y_\ell \mid z, m_1, m_2, \rho_1, \rho_2)) + \log(\textrm{GP}(z \mid 0, \Sigma))\\
        &= \frac{d\lambda_\ell}{d z_\ell} \frac{d}{d\lambda_\ell} \log(p(y_\ell \mid z, m_1, m_2, \rho_1, \rho_2)) + \frac{d}{d z_\ell}\log(\textrm{GP}(z \mid 0, \Sigma))\\
        &= \frac{d\lambda_\ell}{d z_\ell} \frac{\vM{y_\ell}{m_1}{\rho_1} - \vM{y_\ell}{m_2}{\rho_2}}{p(y_\ell \mid z_\ell, m_1, m_2, \rho_1, \rho_2)} - (\Sigma^{-1} z)_\ell.
    \end{aligned}
\end{equation*}
As $\lambda_\ell = \invlogit{z_\ell}$,
\begin{equation*}
    \begin{aligned}
        \frac{d\lambda_\ell}{d z_\ell}
        &= \frac{d}{d z_{\ell}} \frac{1}{1 + \exp(-z_{\ell})}\\ 
        &= \frac{-1}{(1 + \exp(-z_{\ell}))^2}-\exp(-z_{\ell})\\
        &= \frac{\exp(-z_{\ell})}{1 + \exp(-z_{\ell})} \frac{1}{1 + \exp(-z_{\ell})}\\
        &= \left(1 - \frac{1}{1 + \exp(-z_{\ell})}\right)\frac{1}{1 + \exp(-z_{\ell})}\\
        &= (1 - \invlogit{z_\ell})\invlogit{z_\ell}.
    \end{aligned}
\end{equation*}

Putting these terms together gives us the quantity in the lemma.
\end{proof}

\begin{lemma}
For the model specified in \ref{model:SvM-p}, the update for the momentum vector corresponding to $\widetilde{z_\ell}$ in the noncentered parametrization is
\[
\sum_{\ell'} L_{\ell', \ell} \lambda_{k, \ell'}\frac{\vM{y_{\ell'}}{m_k}{\rho_k} - \sum_{k' = 1}^K\lambda_{k', \ell'}\vM{y_{\ell'}}{m_{k'}}{\rho_{k'}}}{p(y_\ell' \mid \v{z_{., \ell'}}, \v{m}, \v{\rho})}  -  \widetilde{z}_\ell.
\]
\end{lemma}
\begin{proof}
We have that
\begin{equation*}
    \begin{aligned}
        & \frac{d}{d \widetilde{z}_{k, \ell}} \log(P(z, \v{m}, \v{\rho} \mid \v{y}))\\
        &\propto \frac{d}{d \widetilde{z}_{k,\ell}} \log(P(y_\ell \mid z, m_1, m_2, \rho_1, \rho_2)) + \log(\mathcal{N}(\widetilde{z}_\ell \mid 0, 1))\\
        &= \sum_{\ell' = 1}^N\sum_{k' = 1}^K \frac{d z_{k, \ell'}}{d \widetilde{z}_{k, \ell}} \frac{d\lambda_{k', \ell'}}{d z_{k, \ell'}} \frac{d}{d\lambda_{k', \ell'}} \log(P(y_\ell \mid z, m_1, m_2, \rho_1, \rho_2)) -  \widetilde{z}_{k, \ell}\\
        &= \sum_{\ell' = 1}^N \sum_{k' = 1}^K  \frac{\frac{d z_{k, \ell'}}{d \widetilde{z}_{k, \ell}} \frac{d\lambda_{k', \ell}}{d z_{k, \ell'}}\vM{y_{\ell'}}{m_k'}{\rho_k'}}{p(y_\ell \mid z_\ell, m_1, m_2, \rho_1, \rho_2)} -  \widetilde{z}_{k, \ell}.
    \end{aligned}
\end{equation*}
As before, we need $\frac{d \psi_{\lambda'}}{d z_{\ell'}}$ and $\frac{d z_{\ell'}}{d \widetilde{z}_\ell}$ for the chain rule. We get $\frac{d \lambda_{\ell'}}{d z_{\ell'}}$ from before. Then, as $z_{\ell'} = L_{j, .} \widetilde{z}$, $\frac{d z_{\ell'}}{d \widetilde{z}_\ell} = L_{\ell', \ell}$. Putting these terms together gives us the quantity in the lemma.
\end{proof}

\begin{lemma}
For the model specified in \ref{model:SvM-p-2}, the update for the momentum vector corresponding to $\widetilde{z_\ell}$ in the noncentered parametrization is
\[
\sum_{\ell'} \lambda_{1, \ell} (1 - \lambda_{1, \ell}) \frac{\vM{y_{\ell'}}{m_k}{\rho_k} - \vM{y_{\ell'} }{m_K}{\rho_K}}{p(y_\ell' \mid \v{z_{., \ell'}}, \v{m}, \v{\rho})} L_{\ell', \ell}  -  \widetilde{z}_\ell.
\]
\end{lemma}
\begin{proof}
We have that
\begin{equation*}
    \begin{aligned}
        & \frac{d}{d \widetilde{z}_\ell} \log(P(z, m_1, m_2, \rho_1, \rho_2 \mid \bm{y}))\\
        &\propto \frac{d}{d \widetilde{z}_\ell} \log(P(y_\ell \mid z, m_1, m_2, \rho_1, \rho_2)) + \log(\mathcal{N}(\widetilde{z}_\ell \mid 0, 1))\\
        &= \sum_{\ell' = 1}^N \frac{d z_{\ell'}}{d \widetilde{z}_\ell} \frac{d\lambda_{1, \ell'}}{d z_{\ell'}} \frac{d}{d\lambda_{1, \ell'}} \log(P(y_\ell \mid z, m_1, m_2, \rho_1, \rho_2)) -  \widetilde{z}_\ell\\
        &= \sum_{\ell' = 1}^N \frac{d z_{\ell'}}{d \widetilde{z}_\ell} \frac{d\lambda_{1, \ell'}}{d z_{\ell'}} \frac{\vM{y_{\ell'}}{m_k}{\rho_k} - \vM{y_{\ell'} }{m_K}{\rho_K}}{p(y_\ell \mid z_\ell, m_1, m_2, \rho_1, \rho_2)} -  \widetilde{z}_\ell.
    \end{aligned}
\end{equation*}
As before, we need $\frac{d \psi_{\lambda'}}{d z_{\ell'}}$ and $\frac{d z_{\ell'}}{d \widetilde{z}_\ell}$ for the chain rule. We get $\frac{d \lambda_{\ell'}}{d z_{\ell'}}$ from before. Then, as $z_{\ell'} = L_{j, .} \widetilde{z}$, $\frac{d z_{\ell'}}{d \widetilde{z}_\ell} = L_{\ell', \ell}$. Putting these terms together gives us the quantity in the lemma.
\end{proof}

\subsection{Expectation Maximization}
\subsubsection{SvM-c}
\begin{lemma}
The $\nu_k$ that maximizes the conditional log posterior in \eqref{eq:SvM-c_exp_cond_log_post} is
\begin{equation}
    \begin{aligned}
        \frac{\sum_\ell \frac{r_{k, \ell} \varphi_{k, \ell}}{\varsigma^2}}{\sum_\ell \frac{r_{k, \ell}}{\varsigma^2} + \frac{1}{\tau^2}}.
    \end{aligned}
\end{equation}
\end{lemma}
\begin{proof}
Gathering the terms that depend on $\nu_k$ and taking the derivative, we have that 
\begin{equation*}
    \begin{aligned}
        \frac{d}{d\nu_k} \left(\sum_\ell -\frac{r_{k, \ell}}{2 * \varsigma^2}(\varphi_k - \nu_k)^2 - \frac{1}{2 * \tau^2}\nu_k^2\right) &= 0\\
        \sum_\ell \frac{r_{k, \ell}}{\varsigma^2}(\varphi_k - \nu_k) - \frac{1}{\tau^2}\nu_k &= 0\\
        \nu_k(\sum_\ell \frac{r_{k, \ell}}{\varsigma^2} - \frac{1}{\tau^2}) &= \sum_\ell \frac{r_{k, \ell}}{\varsigma^2}\varphi_k\\
        \nu_k &= \frac{\sum_\ell \frac{r_{k, \ell} \varphi_{k, \ell}}{\varsigma^2}}{\sum_\ell \frac{r_{k, \ell}}{\varsigma^2} + \frac{1}{\tau^2}}.
    \end{aligned}
\end{equation*}
\end{proof}

\begin{lemma}
The gradient of $\varphi_{k, \ell}$ with respect to the conditional log posterior in \eqref{eq:SvM-c_exp_cond_log_post} is
\begin{equation*}
    \begin{aligned}
        r_{k, \ell}\rho_{k, \ell} \left(\cos(y_\ell - m_{k, \ell}) - \frac{I_{-1}(\rho_{k, \ell})}{I_{0}(\rho_{k, \ell})}\right) - \frac{1}{\varsigma^2}(\varphi_{k, \ell} - \nu_k).
    \end{aligned}
\end{equation*}
\end{lemma}
\begin{proof}
Gathering the terms that depend on $\varphi_{k, \ell}$, we have that 
\begin{equation*}
    \begin{aligned}
        r_{k, \ell}(\rho_{k, \ell}\cos(y_\ell - m_{k, \ell}) - \log(I_{0}(\rho_{k, \ell}))) - \frac{1}{2 * \varsigma^2}(\varphi_{k, \ell} - \nu_k)^2).
    \end{aligned}
\end{equation*}
Then, because
\[
\frac{d}{d\varphi_i} \rho_{k, \ell} = \frac{d}{d\varphi_i} \exp(\varphi_{k, \ell}) = \exp(\varphi_{k, \ell}) = \rho_{k, \ell},
\]
and
\[
\frac{d}{d\rho_{k, \ell}} \log(I_{0}(\rho_{k, \ell})) = \frac{I_{-1}(\rho_{k, \ell})}{I_{0}(\rho_{k, \ell})}
\]
according to Wolfram Research \citep{WolframResearchModifiedBesselFunction}, we get the result in the lemma from the chain rule.
\end{proof}

\begin{lemma}
The gradient of $\widetilde{z}_{k, 1, \ell}$ and $\widetilde{z}_{k, 2, \ell}$ from the conditional log posterior given in \eqref{eq:SvM-c_exp_cond_log_post} is
\begin{equation*}
    \begin{aligned}
        \frac{\partial}{\partial \widetilde{z}_{k, 1, \ell}}  \E{\log(\widetilde{p}(\v{\zeta}, \Theta \mid \v{y}, \v{x}))} &= \sum_{\ell'} L_{\ell', \ell} r_{k, \ell'} \frac{-z_{k, 2, \ell'}}{z^2_{k, 1, \ell'} + z^2_{k, 2, \ell'}}\rho_{k, \ell'}\sin(y_{\ell'} - m_{k, \ell'}) - \widetilde{z}_{k, \ell}.\\
        \frac{\partial}{\partial \widetilde{z}_{k, 2, \ell}}  \E{\log(\widetilde{p}(\v{\zeta}, \Theta \mid \v{y}, \v{x}))} &= \sum_{\ell'} L_{\ell', \ell} r_{k, \ell'} \frac{z_{k, 1, \ell'}}{z^2_{k, 1, \ell'} + z^2_{k, 2, \ell'}}\rho_{k, \ell'}\sin(y_{\ell'} - m_{k, \ell'}) - \widetilde{z}_{k, \ell}.
    \end{aligned}
\end{equation*}
\end{lemma}
\begin{proof}
Without loss of generality, gathering the terms that depend on $\widetilde{z}_{k, 1, \ell}$, we have that the derivative is
\begin{equation*}
    \begin{aligned}
        \sum_{\ell'} \left(r_{k, \ell'}\rho_{k, \ell'}\sin(y_{\ell'} - m_{k, \ell'})\right) \frac{d m_{k, \ell'}}{d z_{k, 1, \ell'}} \frac{d z_{k, 1, \ell'}}{d \widetilde{z}_{k, 1 \ell'}} -  \widetilde{z}_{k, 1, \ell}.
    \end{aligned}
\end{equation*}
We need $\frac{d m_{k, \ell'}}{d z_{k, 1, \ell'}}$ and $\frac{d z_{k, 1, \ell'}}{d \widetilde{z}_{k, 1 \ell'}}$ for the chain rule. As before, because $z_{k, 1, \ell'} = L_{\ell', .} \widetilde{z}_{k, 1}$, $\frac{d z_{k, 1, \ell'}}{d \widetilde{z}_{k, 1, \ell'}} = L_{\ell', \ell}$. Then, because $m_{k, \ell'} = \arctan^*(z_{k, 1, \ell'}, z_{k, 2, \ell'})$ and if $C$ is a constant in $\mathbbm{R}$,
\begin{equation*}
    \begin{aligned}
        \frac{d m_{k, \ell'}}{d z_{k, 1, \ell'}}
        &= \frac{d}{d z_{k, 1, \ell'}} \arctan(z_{k, 1, \ell'}; z_{k, 2, \ell'}) + C\\ 
        &= \frac{1}{1 + \left(\frac{z_{k, 2, \ell'}}{z_{k, 1, \ell'}}\right)^2}\left(-\frac{z_{k, 2, \ell'}}{z_{k, 1, \ell'}^2}\right)\\
        &= \frac{-z_{k, 2, \ell'}}{z_{k, 1, \ell'}^2 + z_{k, 2, \ell'}^2}.
    \end{aligned}
\end{equation*}
Meanwhile, following a similar calculation, we get that
\[
\frac{d m_{k, \ell'}}{d z_{k, 2, \ell'}} = \frac{z_{k, 1, \ell'}}{z_{k, 1, \ell'}^2 + z_{k, 2, \ell'}^2}.
\]
Putting these terms together gives us the quantity in the lemma.
\end{proof}

\subsubsection{SvM-p}
\begin{lemma}
The gradient of $\widetilde{z}_{k, \ell}$ from the conditional log posterior in \eqref{eq:SvM-p_exp_cond_log_post} is
\begin{align*}
\sum_{\ell'}  (r_{k, \ell'} - \lambda_{k, \ell'}) L_{\ell', \ell} - \widetilde{z}_{k, \ell}.
\end{align*}
\end{lemma}
\begin{proof}
Gathering terms that depend on $\widetilde{z}_{k, \ell}$ and using work from earlier, we have that 
\begin{equation*}
    \begin{aligned}
        \frac{\partial}{\partial \widetilde{z}_{k, \ell}}  \E{\log(\widetilde{p}(\v{\zeta}, \Theta \mid \v{y}, \v{x}))} &= \frac{\partial}{\partial \widetilde{z}_{k, \ell}} \left(\sum_{\ell'} \sum_{k'} r_{k', \ell'}\log(\lambda_{k', \ell'}) - \frac{1}{2} \widetilde{z}^T_{k}\widetilde{z}_{k} \right)\\
        &= \sum_{\ell'} \sum_{k'} \frac{\partial z_{k, \ell'}}{\partial \widetilde{z}_{k, \ell}} \frac{\partial \lambda{z}_{k', \ell'}}{\partial z_{k, \ell'}} \frac{\partial}{\partial \lambda{z}_{k', \ell'}}  r_{k', \ell'} \log(\lambda_{k', \ell'}) - \widetilde{z}_{k, \ell}\\
        &= \sum_{\ell'} \sum_{k'} \frac{\partial z_{k, \ell'}}{\partial \widetilde{z}_{k, \ell}} \frac{\partial \lambda{z}_{k', \ell'}}{\partial z_{k, \ell'}} \frac{r_{k', \ell'}}{\lambda_{k', \ell'}} - \widetilde{z}_{k, \ell}\\
        &= \sum_{\ell'} L_{\ell', \ell} \lambda_{k, \ell'}(1 - \lambda_{k, \ell'})\frac{r_{k, \ell'}}{\lambda_{k, \ell'}} -\\
        & \qquad\sum_{k' \neq k} L_{\ell', \ell}\lambda_{k, \ell'}\lambda_{k', \ell'}\frac{r_{k', \ell'}}{\lambda_{k', \ell'}}  - \widetilde{z}_{k, \ell}\\
        &= \sum_{\ell'} L_{\ell', \ell} (1 - \lambda_{k, \ell'})r_{k, \ell'} - \sum_{k' \neq k} L_{\ell', \ell}\lambda_{k, \ell'}r_{k', \ell'}  - \widetilde{z}_{k, \ell}\\
        &= \sum_{\ell'} L_{\ell', \ell} (r_{k, \ell'} - \lambda_{k, \ell'})  - \widetilde{z}_{k, \ell}\\
    \end{aligned}
\end{equation*}
\end{proof}

\begin{lemma}
The gradient of $\widetilde{z}_{k, \ell}$ from the conditional log posterior in \eqref{eq:SvM-p_exp_cond_log_post} is
\begin{equation*}
    \begin{aligned}
        \sum_{\ell'} & L_{\ell', \ell} (r_{i, 1} - \lambda_{1, \ell}) - \widetilde{z}_\ell.
    \end{aligned}
\end{equation*}
\end{lemma}
\begin{proof}
Gathering the terms that depend on $\widetilde{z}_{k, \ell}$, we have that 
\begin{equation*}
    \begin{aligned}
        \sum_{\ell'} r_{1, k}\log(\lambda_{1, \ell}) + (1 - r_{1, k})\log(1 - \lambda_{1, \ell}) - \widetilde{z}_\ell.
    \end{aligned}
\end{equation*}
As before, we need $\frac{d m_{k, \ell'}}{d z_k[j]}$ and $\frac{d z_k[j]}{d \widetilde{z}_{k, \ell'}}$ for the chain rule. The latter is still the same, i.e. $L_{\ell', \ell}$. Then, because $\lambda_{1, \ell'} = \invlogit{z_{\ell'}}$, we have from earlier that
\[
\frac{d \lambda_{1, \ell'}}{d z_{\ell'}} = (1 - \lambda_{1, \ell'})(\lambda_{1, \ell'}).
\]
As a result, we get that 
\[
\frac{d}{d z_{\ell'}} r_{1, k}\log(\lambda_{1, \ell}) + (1 - r_{1, k})\log(1 - \lambda_{1, \ell}) =
r_{1, k}(1 - \lambda_{1, \ell}) - (1 - r_{1, k})\lambda_{1, \ell},
\]
which simplifies to the quantity in the lemma.
\end{proof}

\section{Additional Simulation and Results}
\subsection{Real Data Duplicates and Posterior Plots}
\begin{figure}[!t]
\centering
\begin{subfigure}{.3\textwidth}
\centering
\includegraphics[width = .8\textwidth]{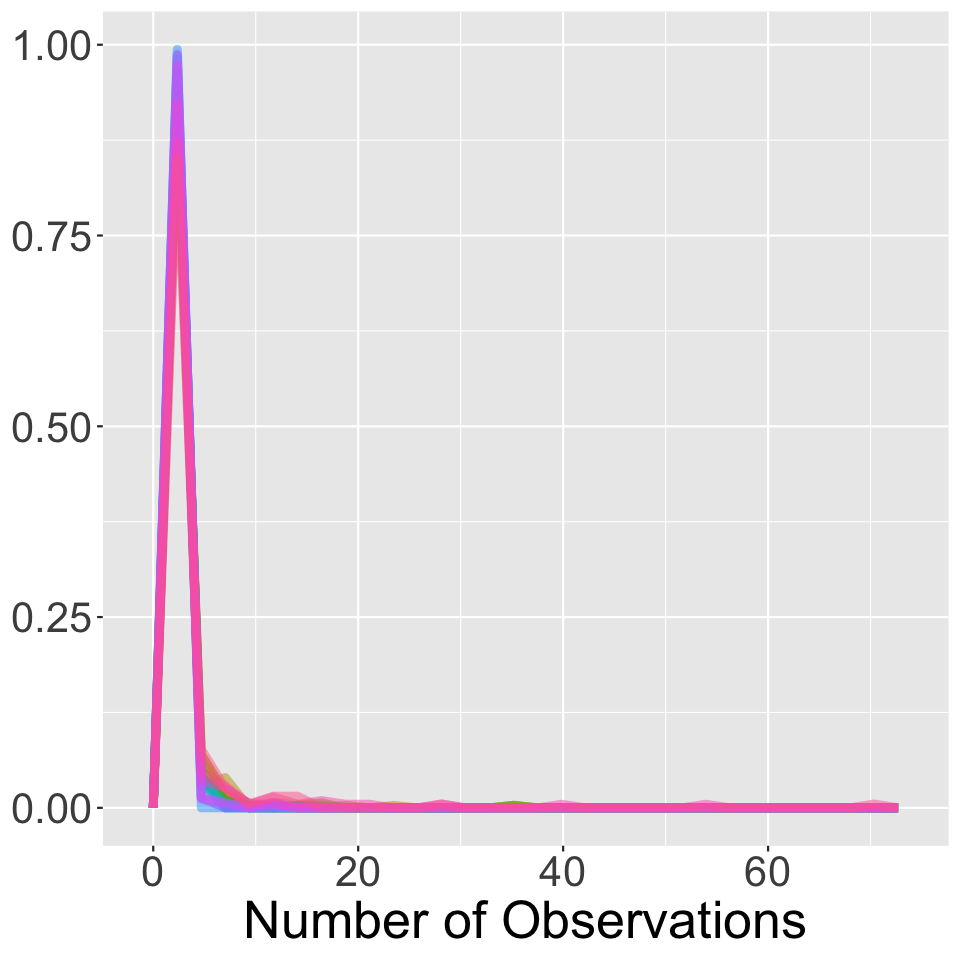}
\caption{Frequency plot of duplicate combinations}
\end{subfigure}
\begin{subfigure}{.3\textwidth}
\centering
\includegraphics[width = .8\textwidth]{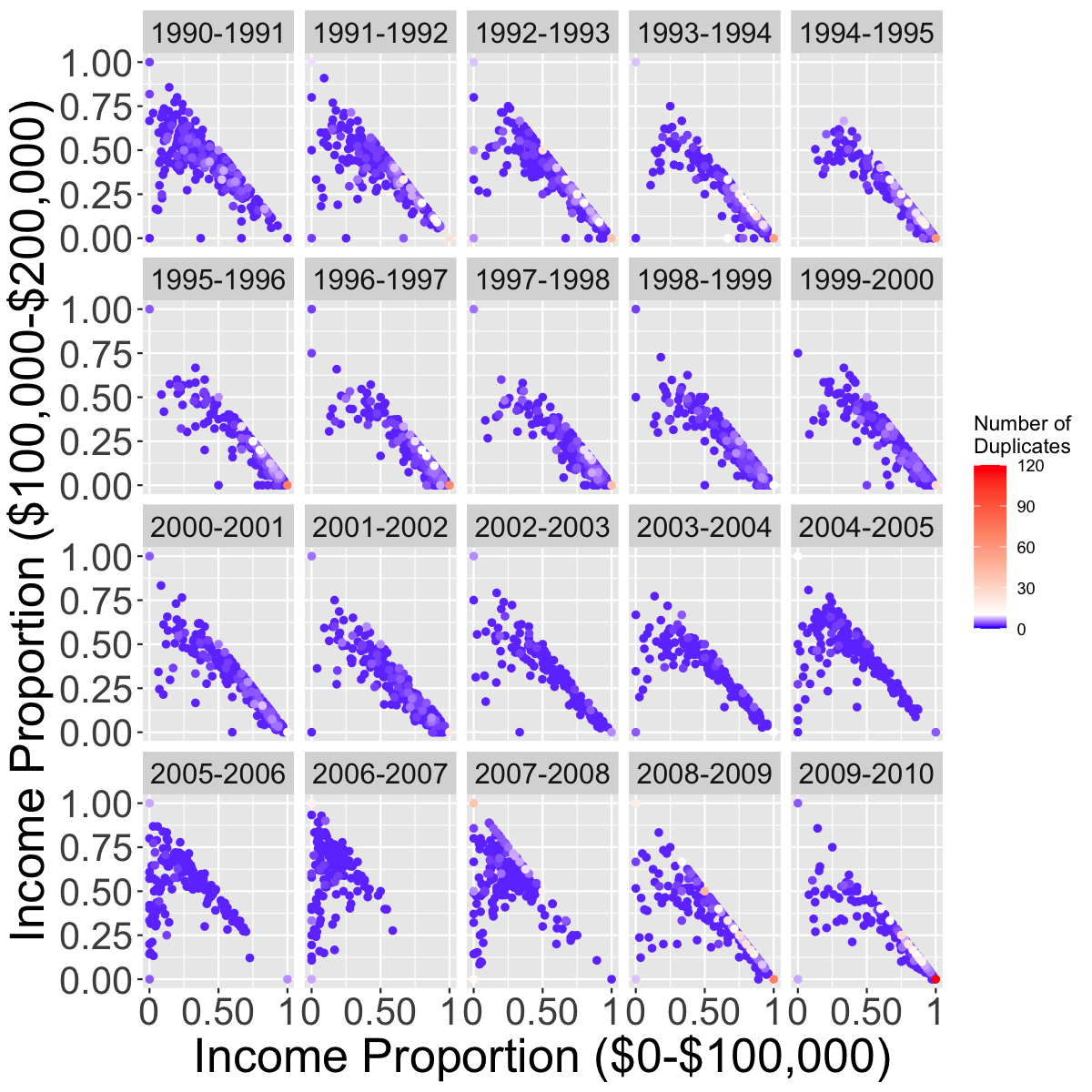}
\caption{Location plot colored by the number of duplicates}
\label{fig:real_data_loc_plot_by_year}
\end{subfigure}
\begin{subfigure}{.3\textwidth}
\centering
\includegraphics[width = .8\textwidth]{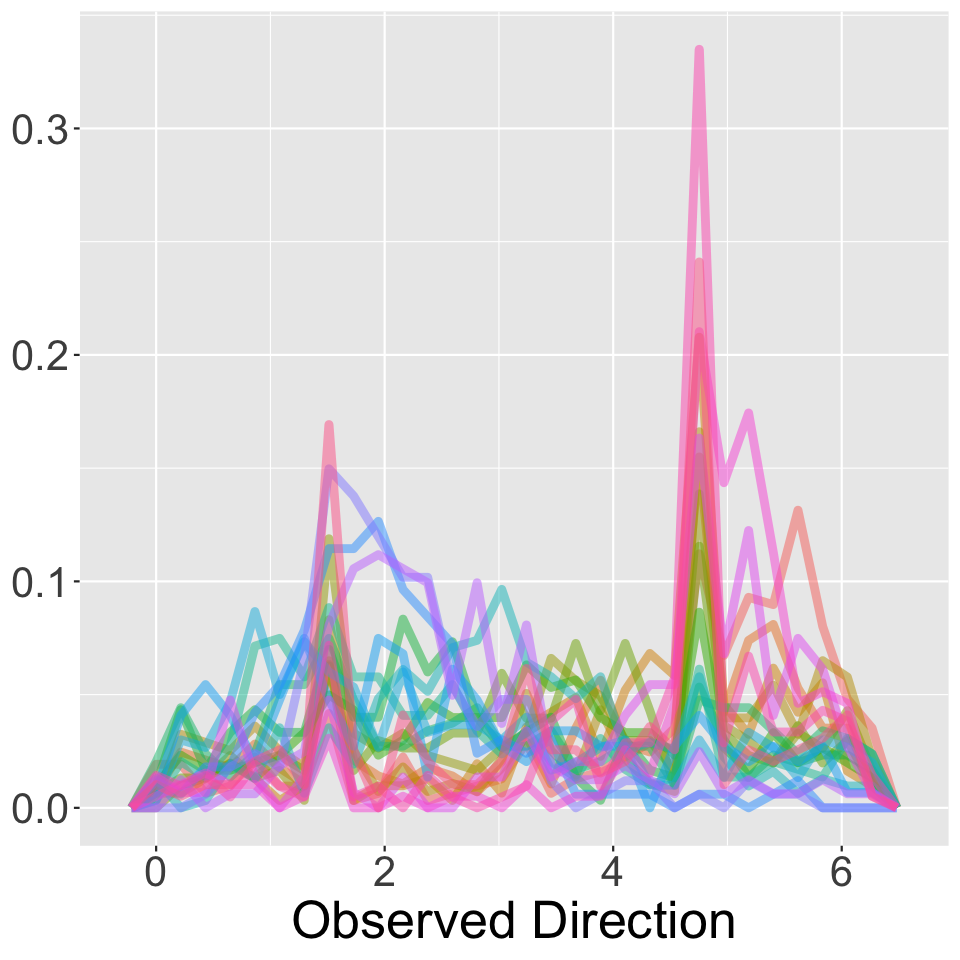}
\caption{Frequency plot of duplicate observation values}
\end{subfigure}\\
\centering
\begin{subfigure}{.3\textwidth}
\centering
\includegraphics[width = .8\textwidth]{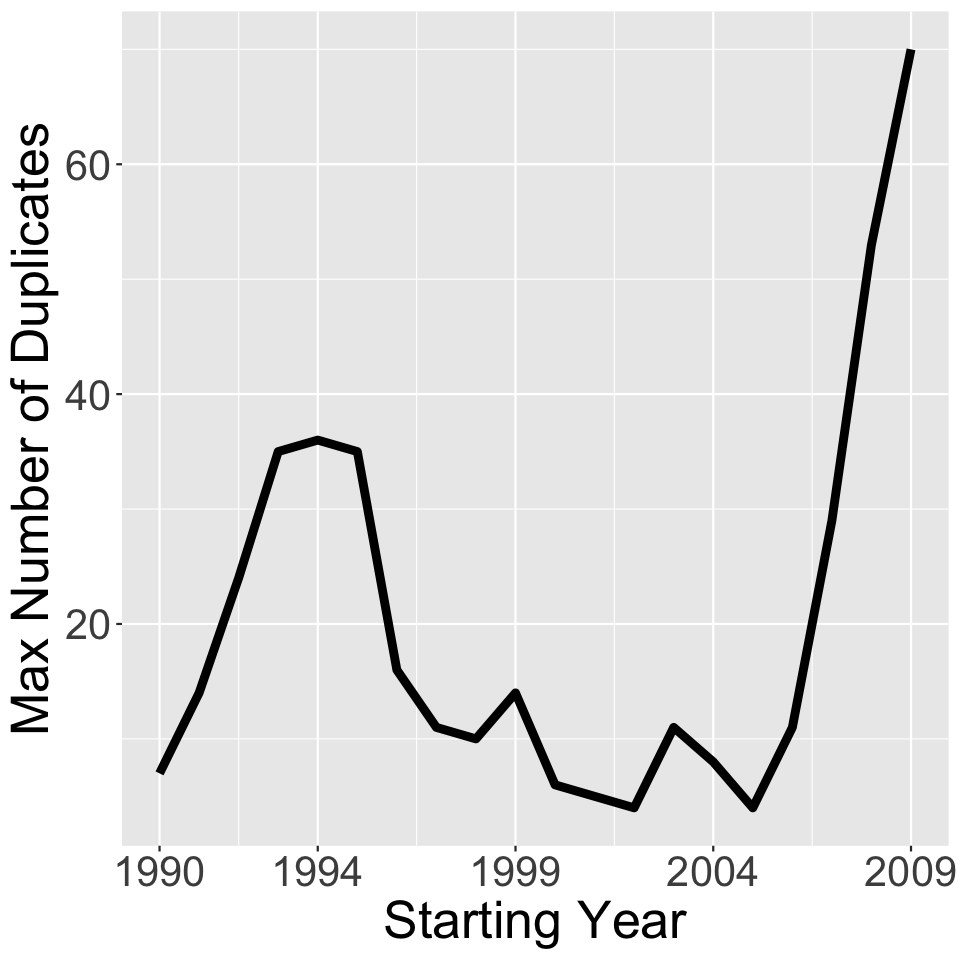}
\caption{Max number of duplicates by starting year}
\end{subfigure}
\begin{subfigure}{.3\textwidth}
\centering
\includegraphics[width = .8\textwidth]{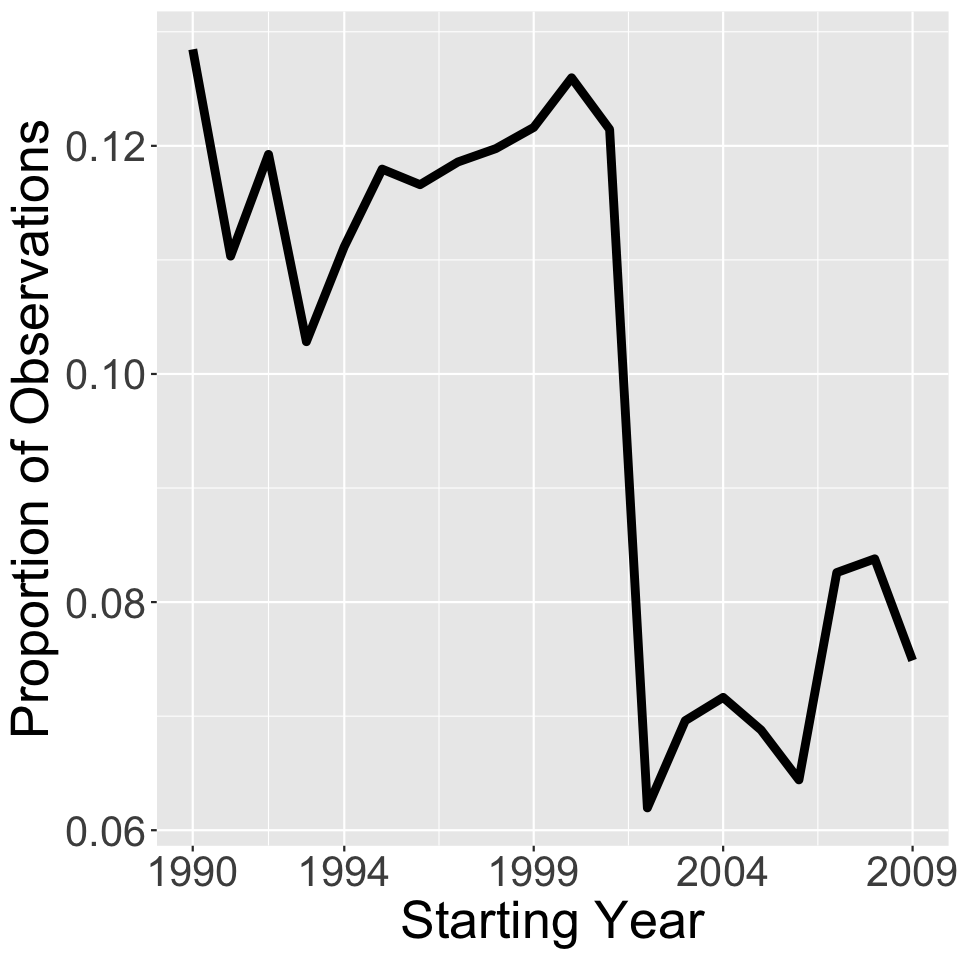}
\caption{Proportion of duplicate observations by starting year}
\end{subfigure}
\begin{subfigure}{.3\textwidth}
\centering
\includegraphics[width = .8\textwidth]{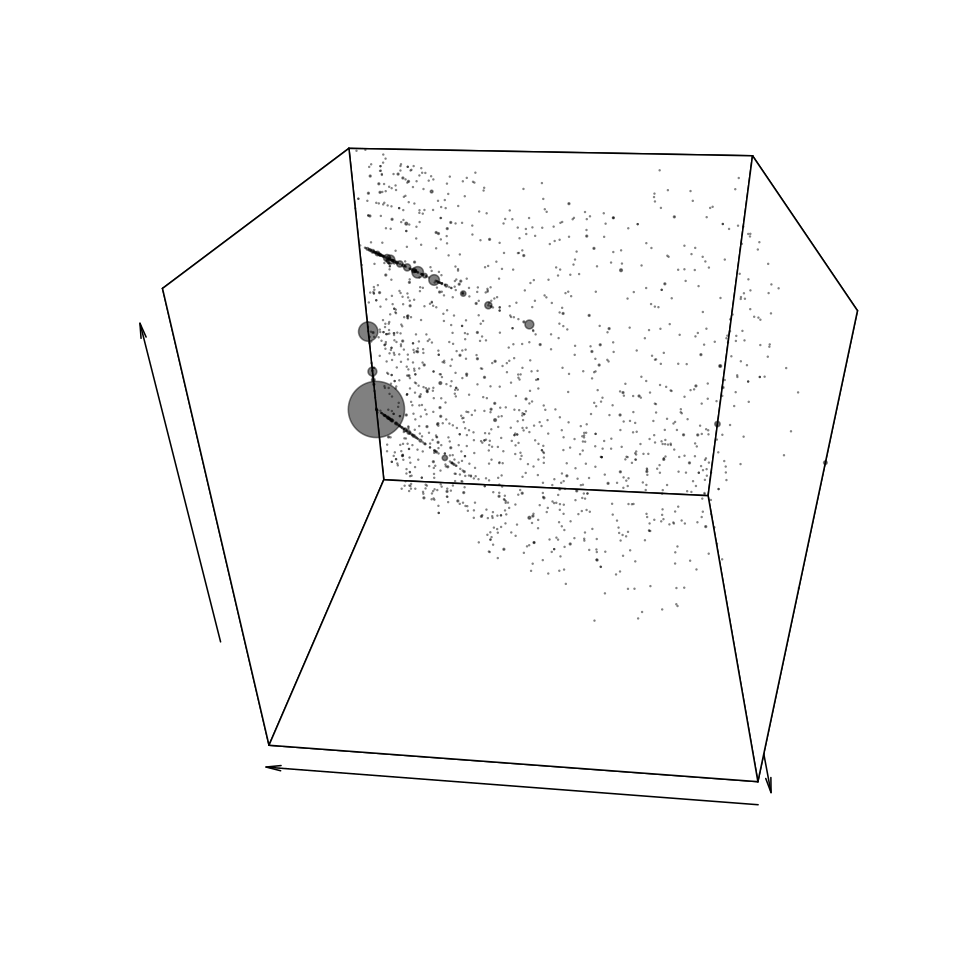}
\caption{Plot of the observed direction for 2009-2010 weighted by the number of appearances}
\label{fig:real_data_example_repeats}
\end{subfigure}
\caption{Plots in the top row show the duplicate observations broken down by different categories across different years with further context provided by plots in the bottom row.}
\label{fig:real_data_repeats_info}

\centering
\begin{subfigure}{.22\textwidth}
\centering
\includegraphics[width = 1\textwidth]{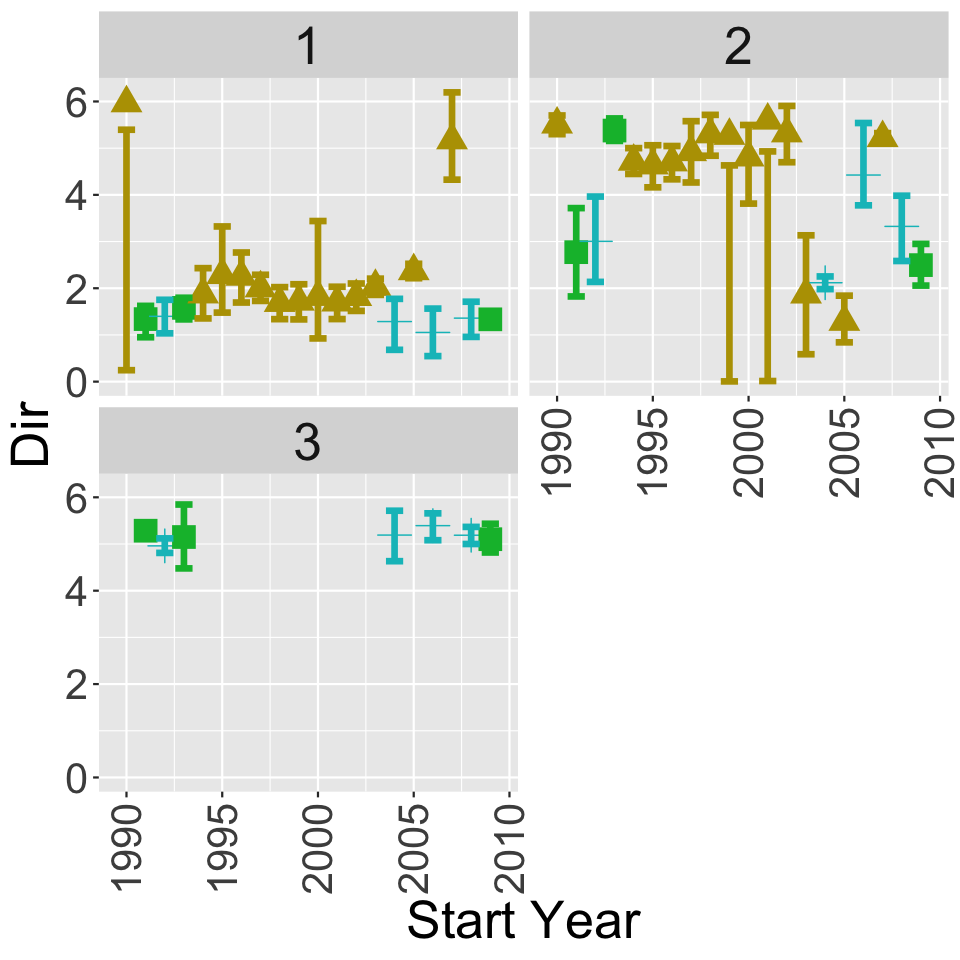}
\end{subfigure}
\begin{subfigure}{.22\textwidth}
\centering
\includegraphics[width = 1\textwidth]{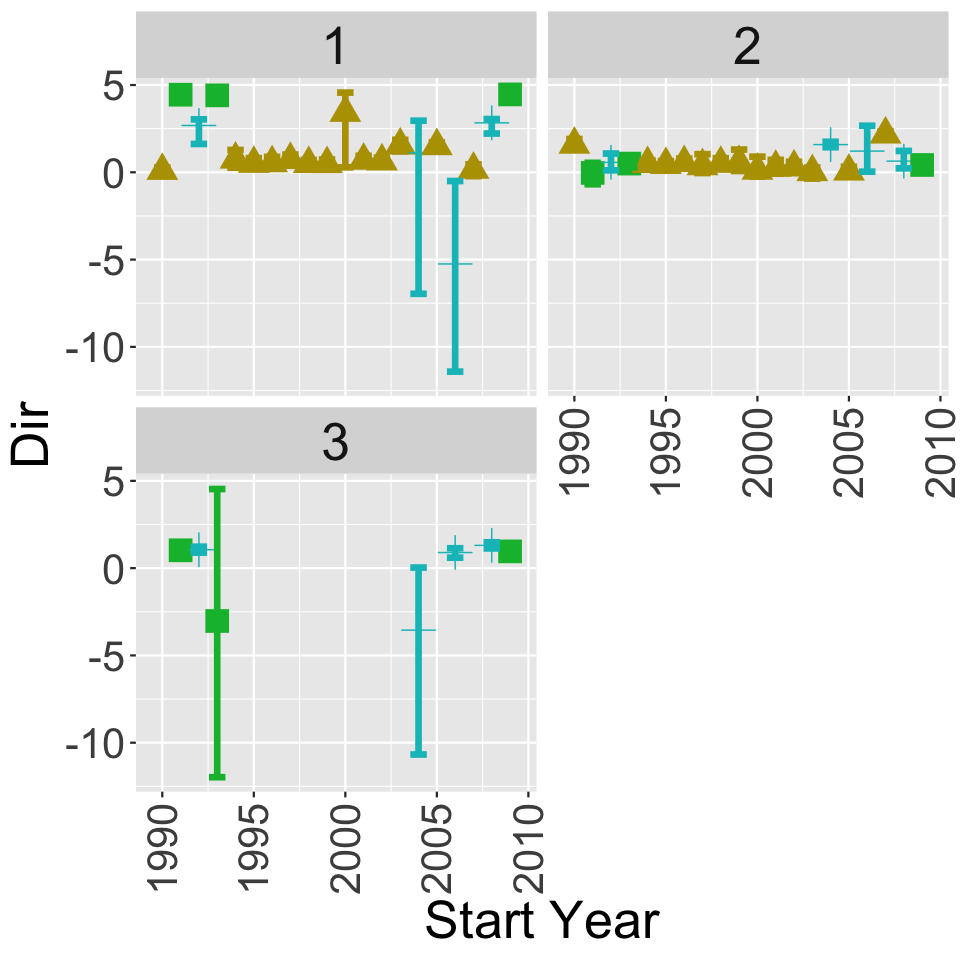}
\end{subfigure}
\begin{subfigure}{.22\textwidth}
\centering
\includegraphics[width = 1\textwidth]{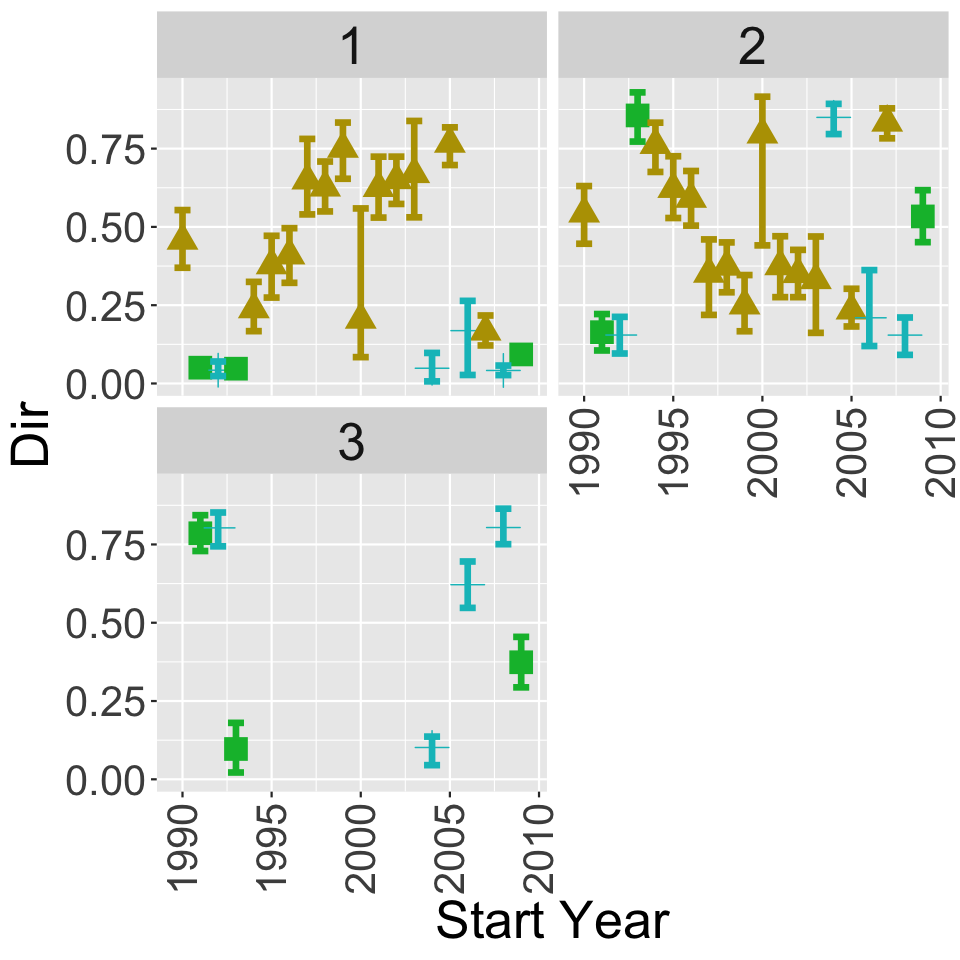}
\end{subfigure}
\centering
\begin{subfigure}{.22\textwidth}
\centering
\includegraphics[width = 1\textwidth]{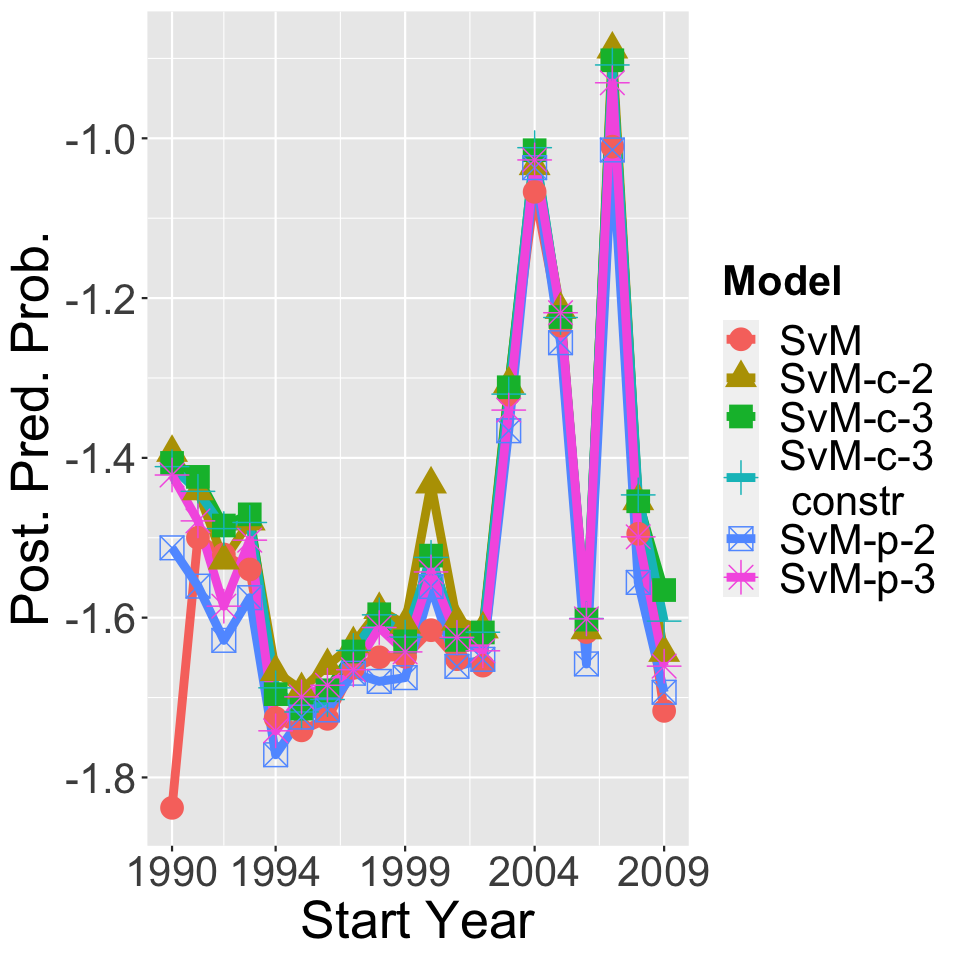}
\end{subfigure}
\caption{Plots showing the posterior predictive log probability defined in \eqref{eq:post_pred_prob} and the von Mises distributions' parameters and the mixing probability of the first component selected according to that probability and the models' mixing for each year. For $\textit{SvM}$ and $\textit{SvM-c}$, the mean and concentration parameters' mean and credible interval are averaged over all locations. On the other hand, the mean and credible interval for the probability of the first component is averaged across all location for $\textit{SvM-p}$. In addition, the second component's values are shown with an opacity of 0.5.}
\label{fig:real_data_fitted_results_summary}
\end{figure}

Figure \ref{fig:real_data_example_repeats} shows us that the data violates a key assumption. There are locations with observed angles that are seen multiple times in the data set. Due to the use of a continuous distribution for the observations, i.e. the von Mises distribution, the probability of this happening should be zero. It is more problematic for the $\textit{SvM}-c$ model if we use the hierarchical prior for the concentration parameter because repeated observations might induce large values, which will pull up all values for the concentration parameter. 

Fortunately, it is straightforward to identify these duplicates. Figure \ref{fig:real_data_repeats_info} shows how observations are duplicated if we break it down by number, location, and/or the value of the observed angle. Despite there being some variation in the exact value, most duplicate observations are only observed twice. Further, the proportion of duplicated observations decreases over time because it appears to oscillate around 0.12 before 2002-2003 and then fluctuates around 0.07 afterwards. Interestingly enough, this decrease in the proportion also roughly corresponds to a spike in the maximum number of duplicates. This suggests a concentration of duplicates at fewer locations, which is observed in Figure \ref{fig:real_data_loc_plot_by_year} from 2004 onward. However, the maximum number of duplicate combinations by year and the duplicate by locations appear to be correlated and exhibit trends that are independent of the number of observations. The maximum number of duplicates peaks at 36 in 1994-1995 before coming back down. Simultaneously, the locations with duplicates goes from a diffuse spread across the simplex to a concentrated spread on the lower half of the boundary in which the proportion of the third income category is close to zero. It then returns to a more diffuse set of locations by 2003-2004. Finally, while there is a lot of noise in the duplicated observed direction value across the years, there does appear to be a peak around $\frac{\pi}{2}$ and $\frac{3\pi}{2}$. This is seen in the plot for observed directions and their locations weighted by the number of appearances for 2008-2009.

\subsection{Real data results}
\begin{figure}[!tb]
\centering
\begin{subfigure}{.18\textwidth}
\centering
\includegraphics[width = 1\textwidth]{logistic_gp_dir_mixture_comp_hier_p/l_02_real_results/year_1_gp_dir_location_comp_mixture_results.png}
\caption{1990-1991\\ (\textit{SvM-c-2})}
\end{subfigure}
\begin{subfigure}{.18\textwidth}
\centering
\includegraphics[width = 1\textwidth]{logistic_gp_dir_mixture_comp_hier_p/l_02_mixture_3_real_results/year_2_gp_mixture_comp_K_results.png}
\caption{1991-1992\\ (\textit{SvM-c-3})}
\end{subfigure}
\begin{subfigure}{.18\textwidth}
\centering
\includegraphics[width = 1\textwidth]{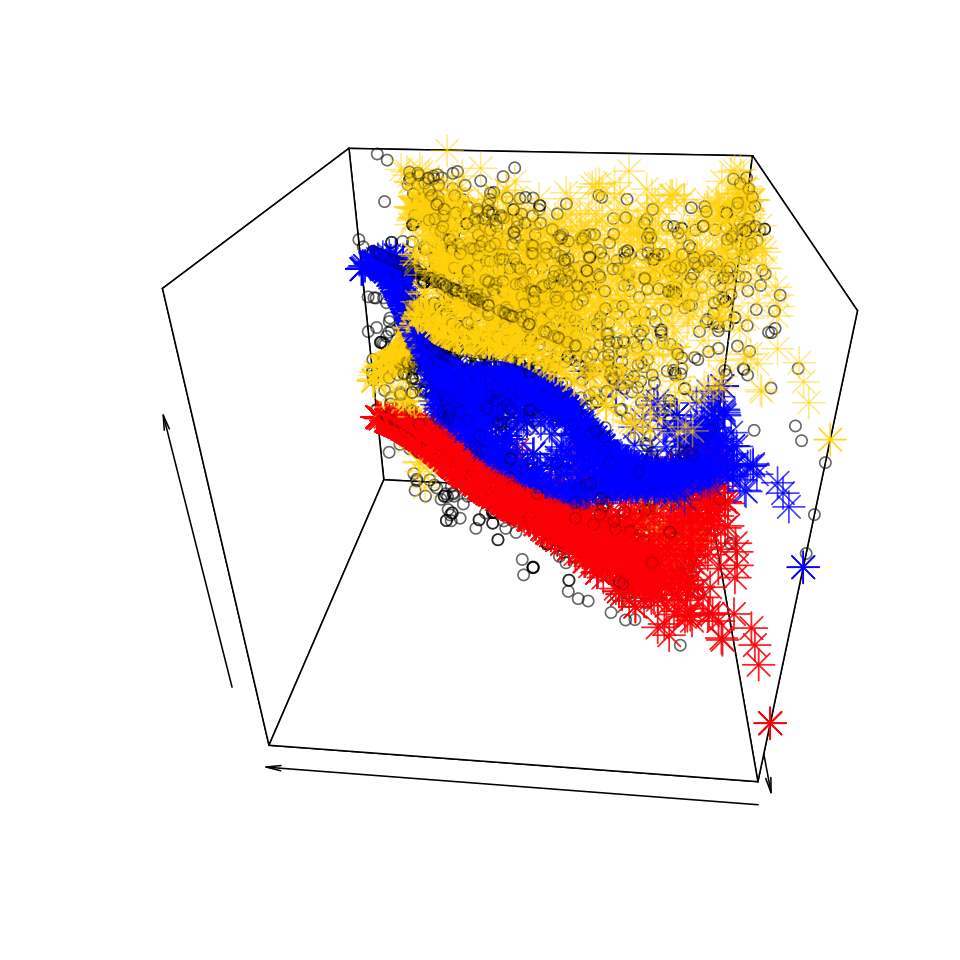}
\caption{1992-1993\\ (\textit{SvM-c-3})}
\end{subfigure}
\begin{subfigure}{.18\textwidth}
\centering
\includegraphics[width = 1\textwidth]{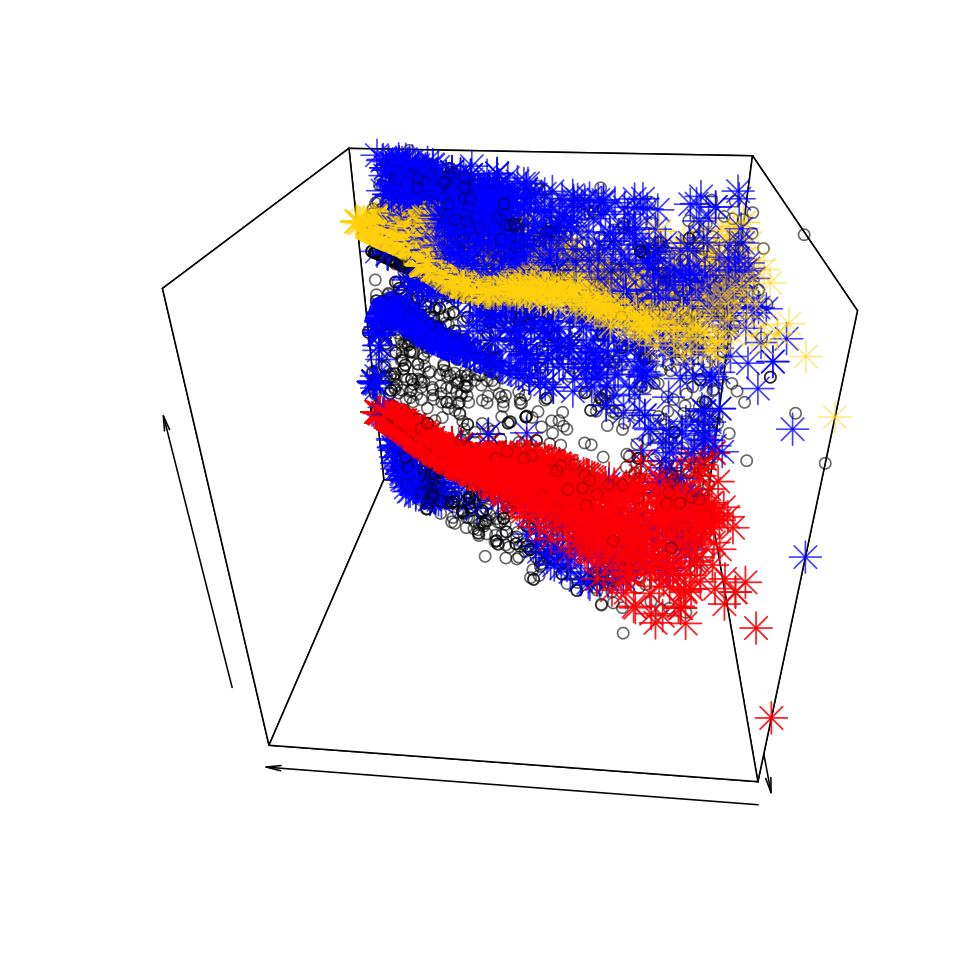}
\caption{1993-1994\\ (\textit{SvM-c-3})}
\end{subfigure}
\begin{subfigure}{.18\textwidth}
\centering
\includegraphics[width = 1\textwidth]{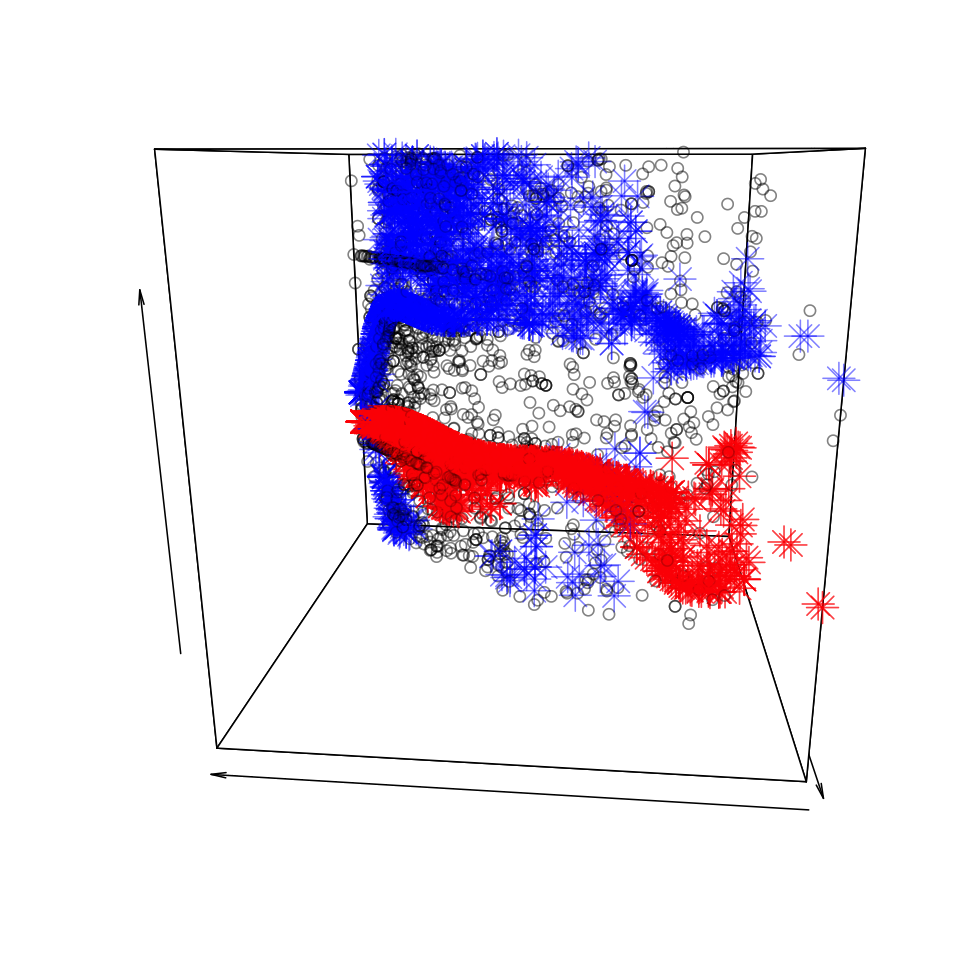}
\caption{1994-1995\\ (\textit{SvM-c-2})}
\end{subfigure}\\
\begin{subfigure}{.18\textwidth}
\centering
\includegraphics[width = 1\textwidth]{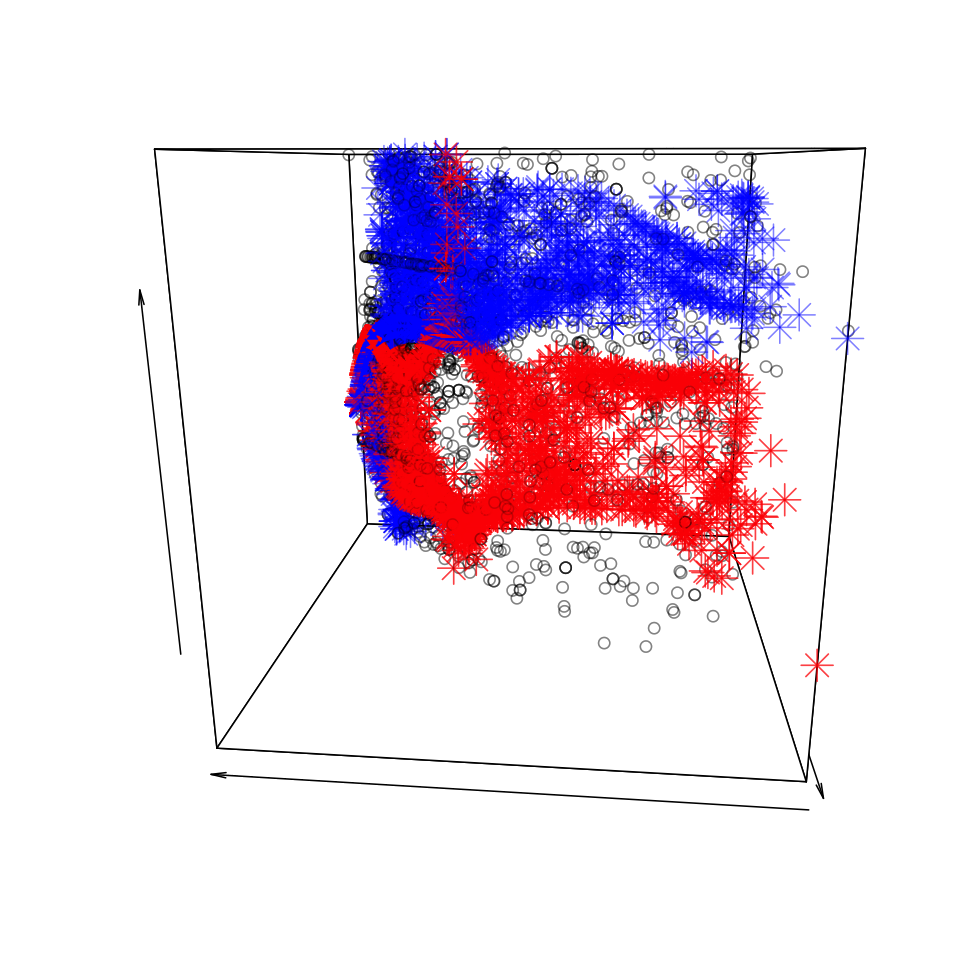}
\caption{1995-1996\\ (\textit{SvM-c-2})}
\end{subfigure}
\begin{subfigure}{.18\textwidth}
\centering
\includegraphics[width = 1\textwidth]{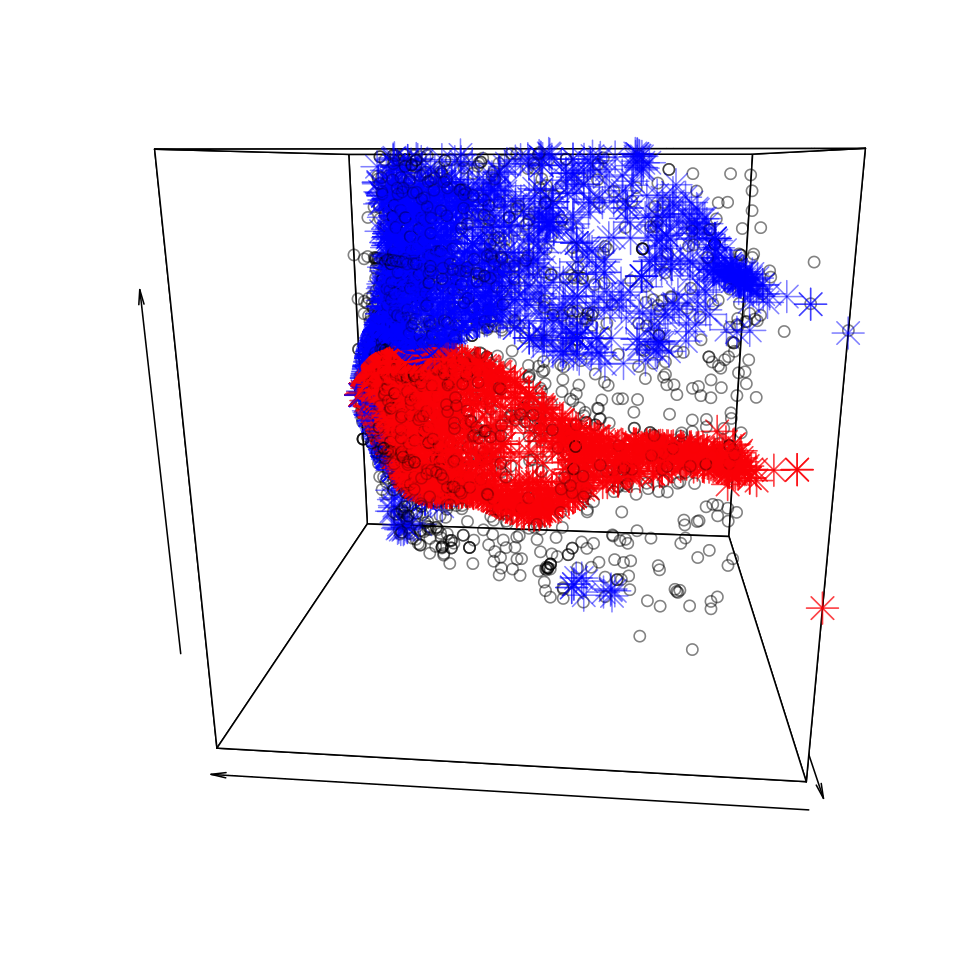}
\caption{1996-1997\\ (\textit{SvM-c-2})}
\end{subfigure}
\begin{subfigure}{.18\textwidth}
\centering
\includegraphics[width = 1\textwidth]{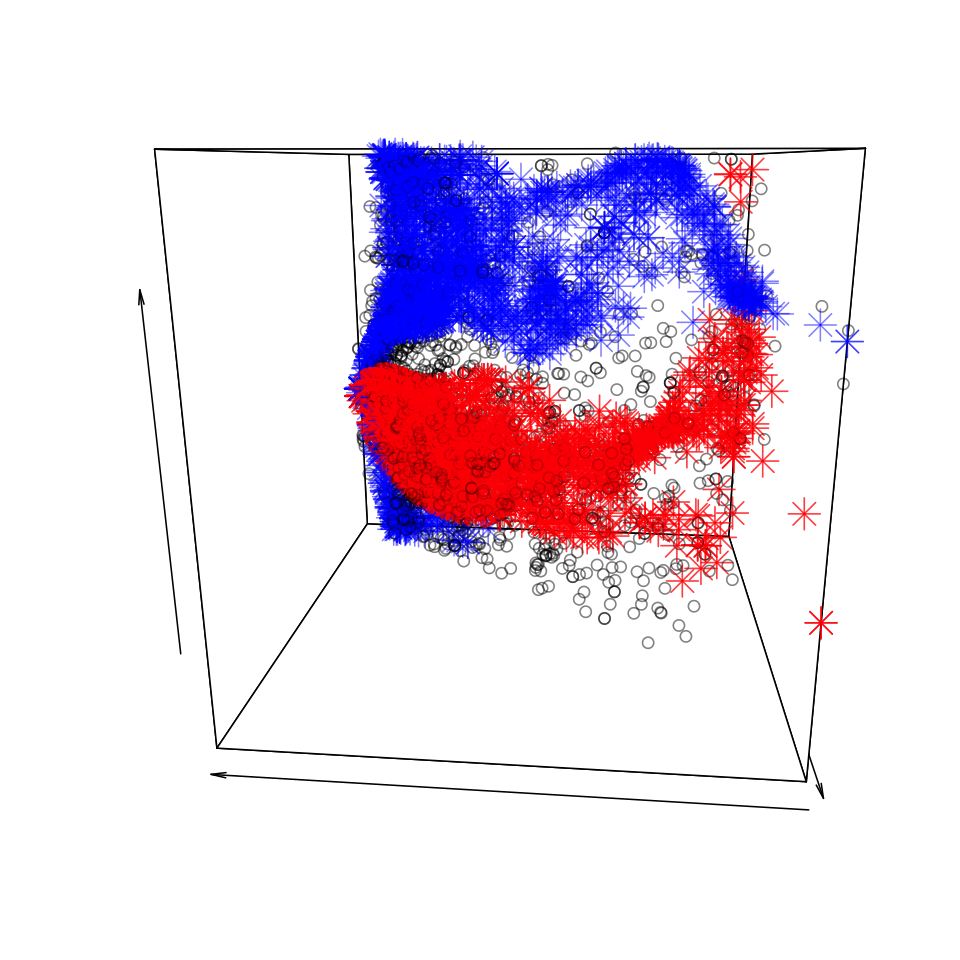}
\caption{1997-1998\\ (\textit{SvM-c-2})}
\end{subfigure}
\begin{subfigure}{.18\textwidth}
\centering
\includegraphics[width = 1\textwidth]{logistic_gp_dir_mixture_comp_hier_p/l_02_real_results/year_9_gp_dir_location_comp_mixture_results.png}
\caption{1998-1999\\ (\textit{SvM-c-2})}
\end{subfigure}
\begin{subfigure}{.18\textwidth}
\centering
\includegraphics[width = 1\textwidth]{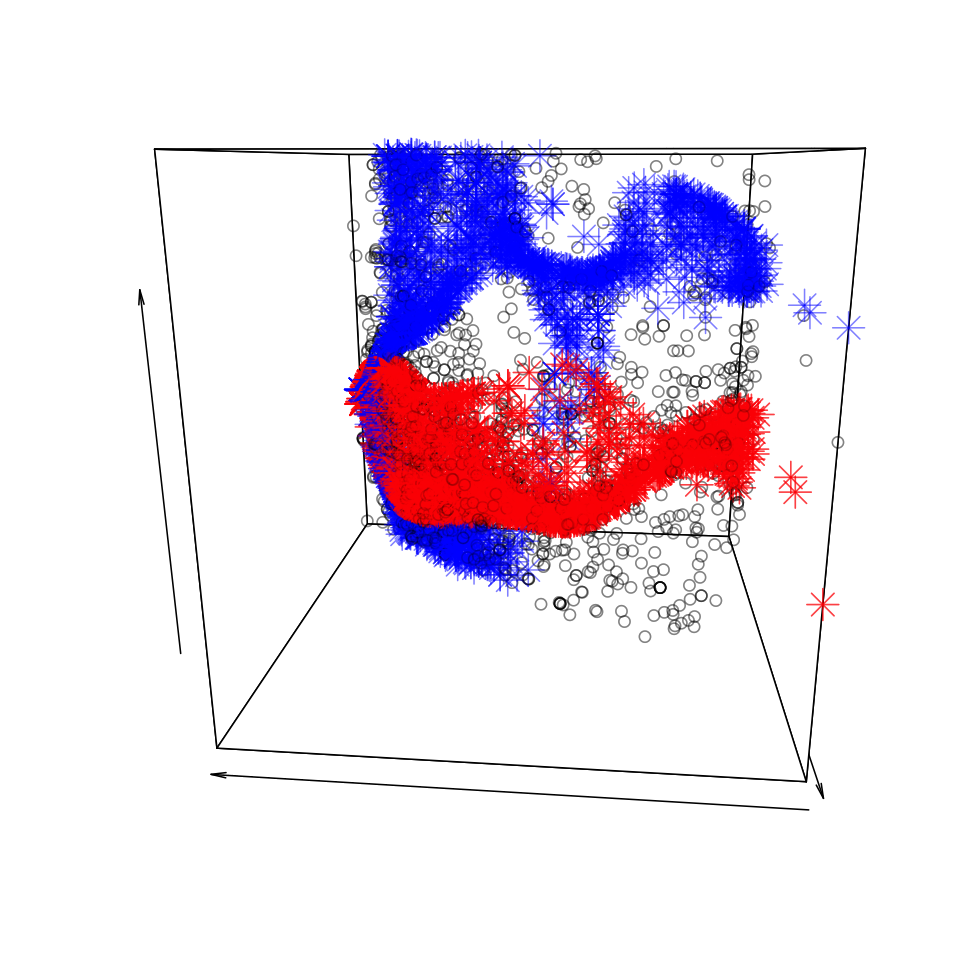}
\caption{1999-2000\\ (\textit{SvM-c-2})}
\end{subfigure}\\
\begin{subfigure}{.18\textwidth}
\centering
\includegraphics[width = 1\textwidth]{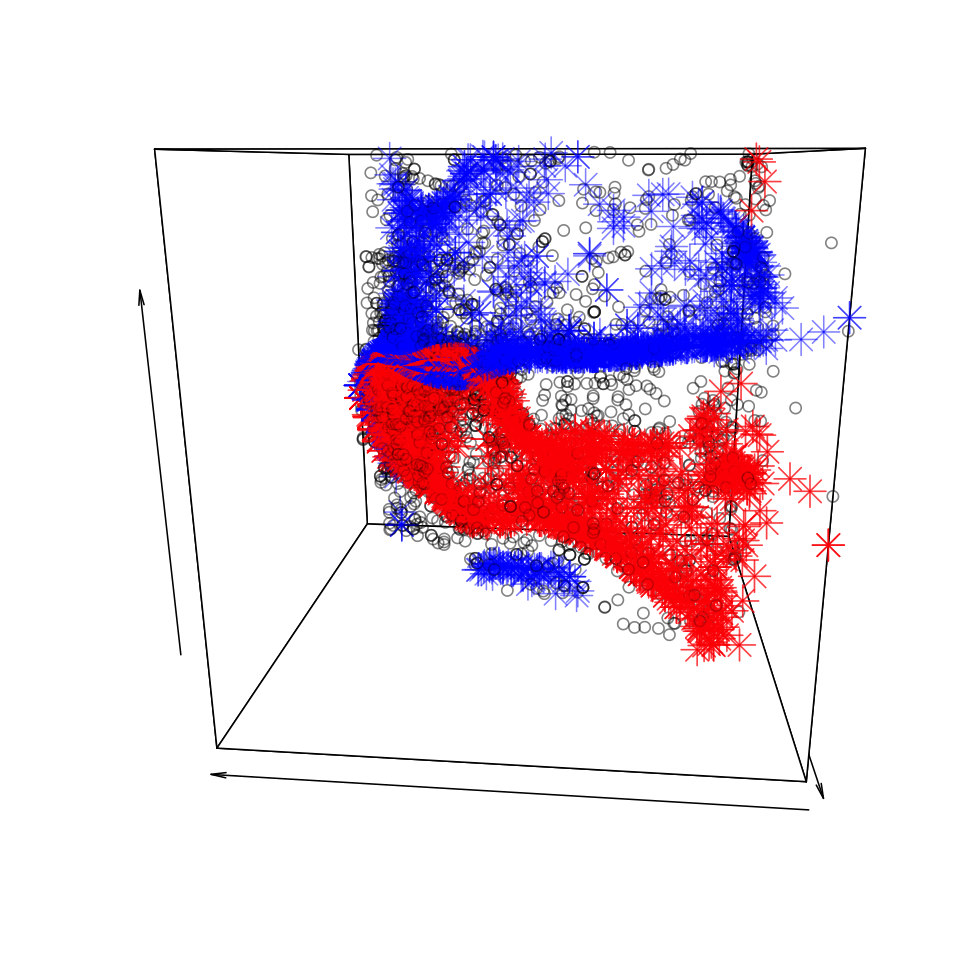}
\caption{2000-2001\\ (\textit{SvM-c-2})}
\end{subfigure}
\begin{subfigure}{.18\textwidth}
\centering
\includegraphics[width = 1\textwidth]{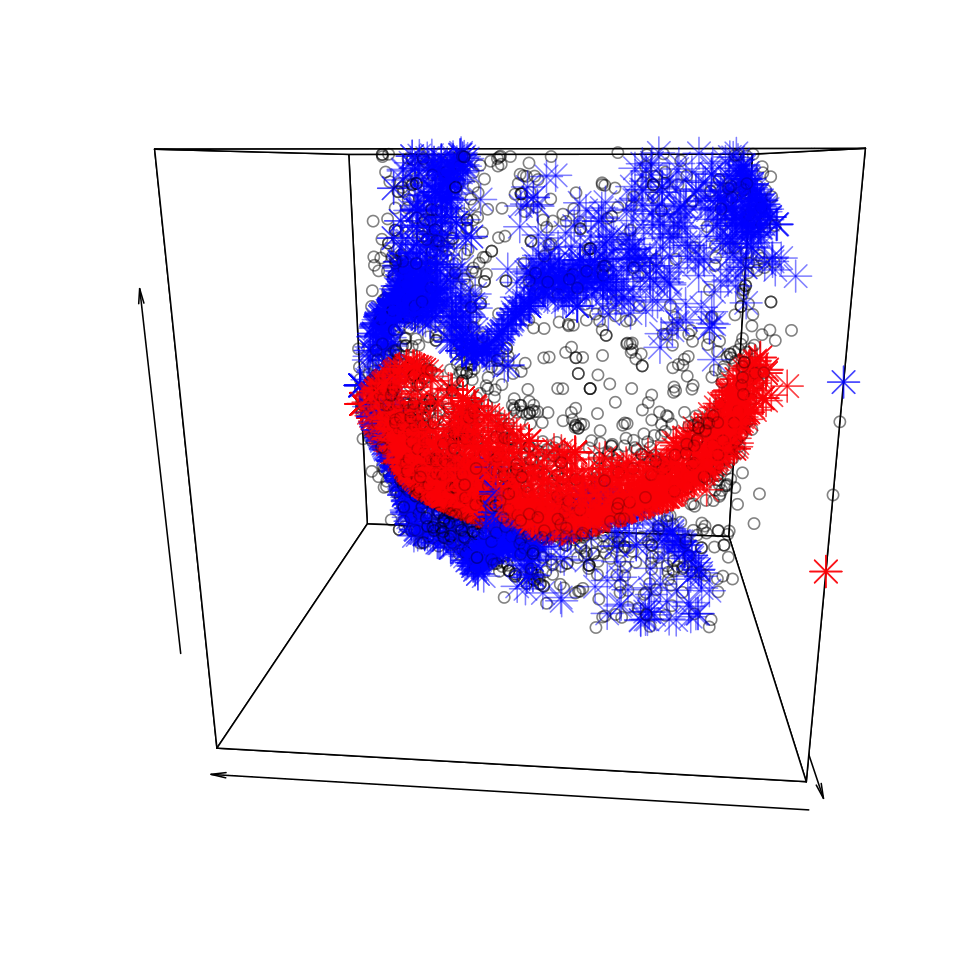}
\caption{2001-2002\\ (\textit{SvM-c-2})}
\end{subfigure}
\begin{subfigure}{.18\textwidth}
\centering
\includegraphics[width = 1\textwidth]{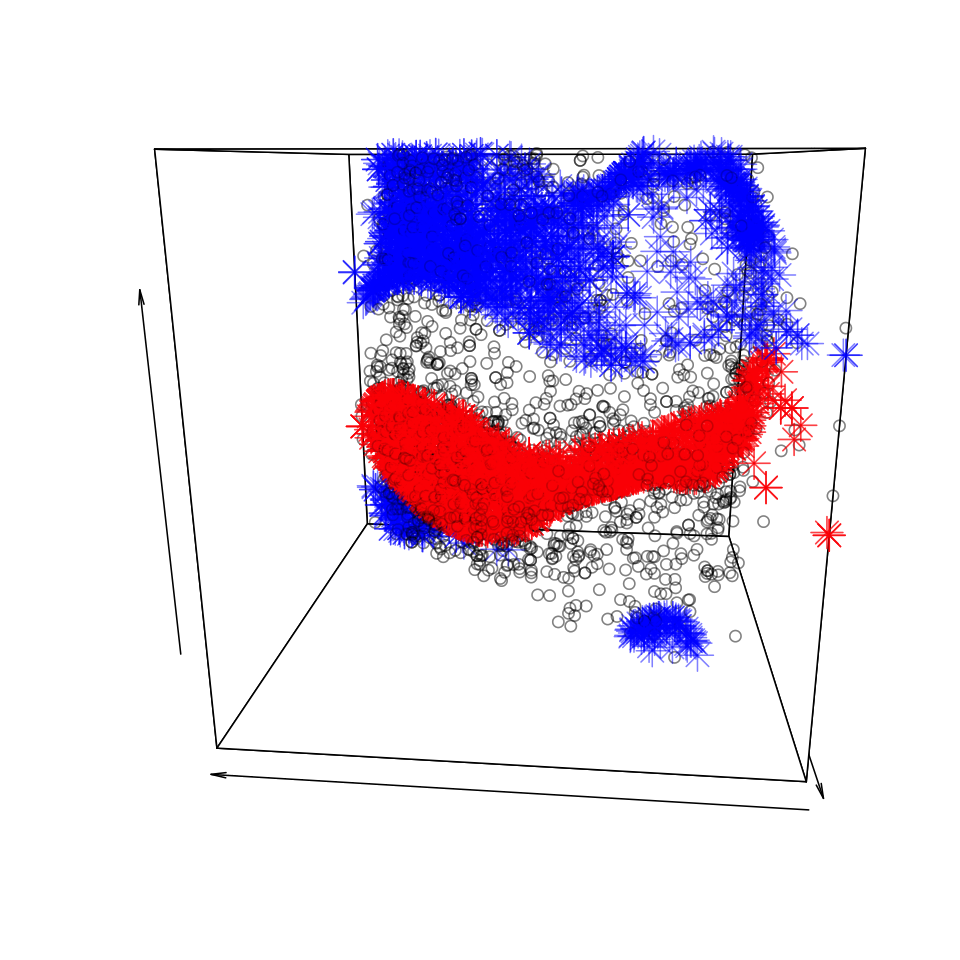}
\caption{2002-2003\\ (\textit{SvM-c-2})}
\end{subfigure}
\begin{subfigure}{.18\textwidth}
\centering
\includegraphics[width = 1\textwidth]{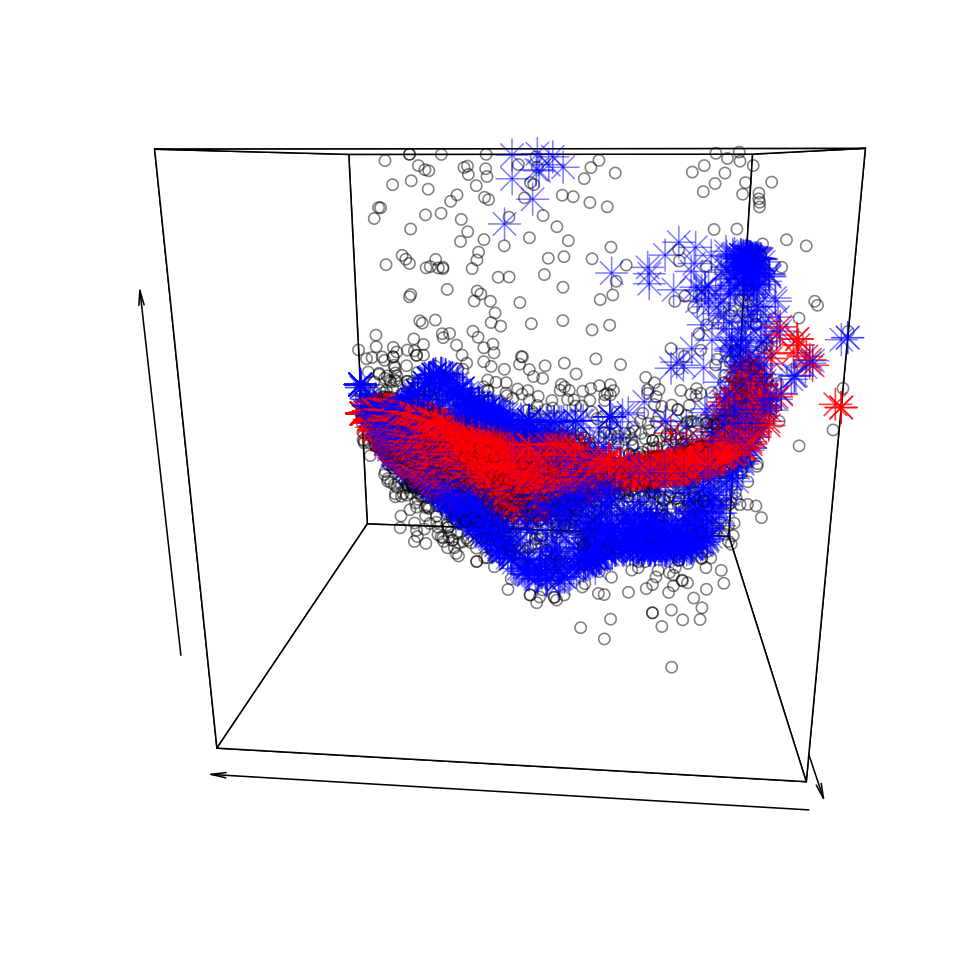}
\caption{2003-2004\\ (\textit{SvM-c-2})}
\end{subfigure}
\begin{subfigure}{.18\textwidth}
\centering
\includegraphics[width = 1\textwidth]{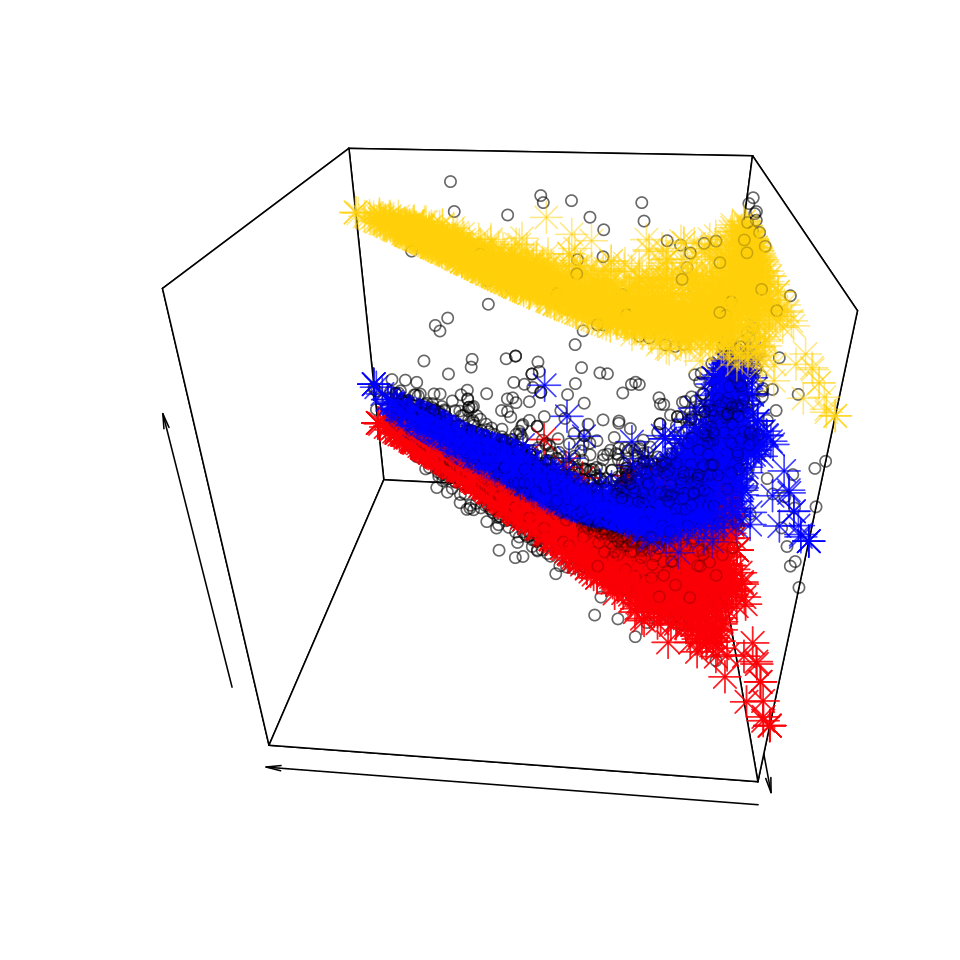}
\caption{2004-2005\\ (\textit{SvM-c-3})}
\end{subfigure}\\
\begin{subfigure}{.18\textwidth}
\centering
\includegraphics[width = 1\textwidth]{logistic_gp_dir_mixture_comp_hier_p/l_02_real_results/year_16_gp_dir_location_comp_mixture_results.png}
\caption{2005-2006\\ (\textit{SvM-c-2})}
\end{subfigure}
\begin{subfigure}{.18\textwidth}
\centering
\includegraphics[width = 1\textwidth]{logistic_gp_dir_mixture_comp_hier_p/l_02_mixture_3_real_results/year_17_gp_mixture_comp_K_results.png}
\caption{2006-2007\\ (\textit{SvM-p-3})}
\end{subfigure}
\begin{subfigure}{.2\textwidth}
\centering
\includegraphics[width = 1\textwidth]{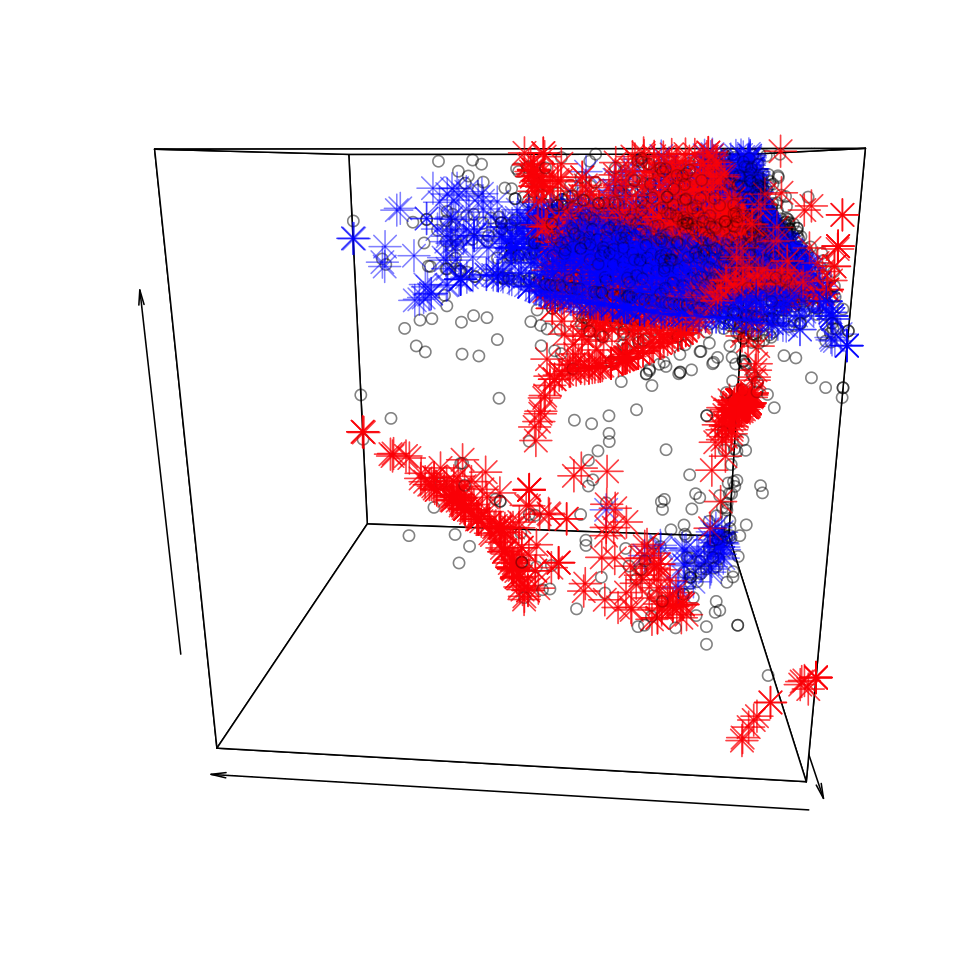}
\caption{2007-2008\\\textit{SvM-c}}
\end{subfigure}
\begin{subfigure}{.18\textwidth}
\centering
\includegraphics[width = 1\textwidth]{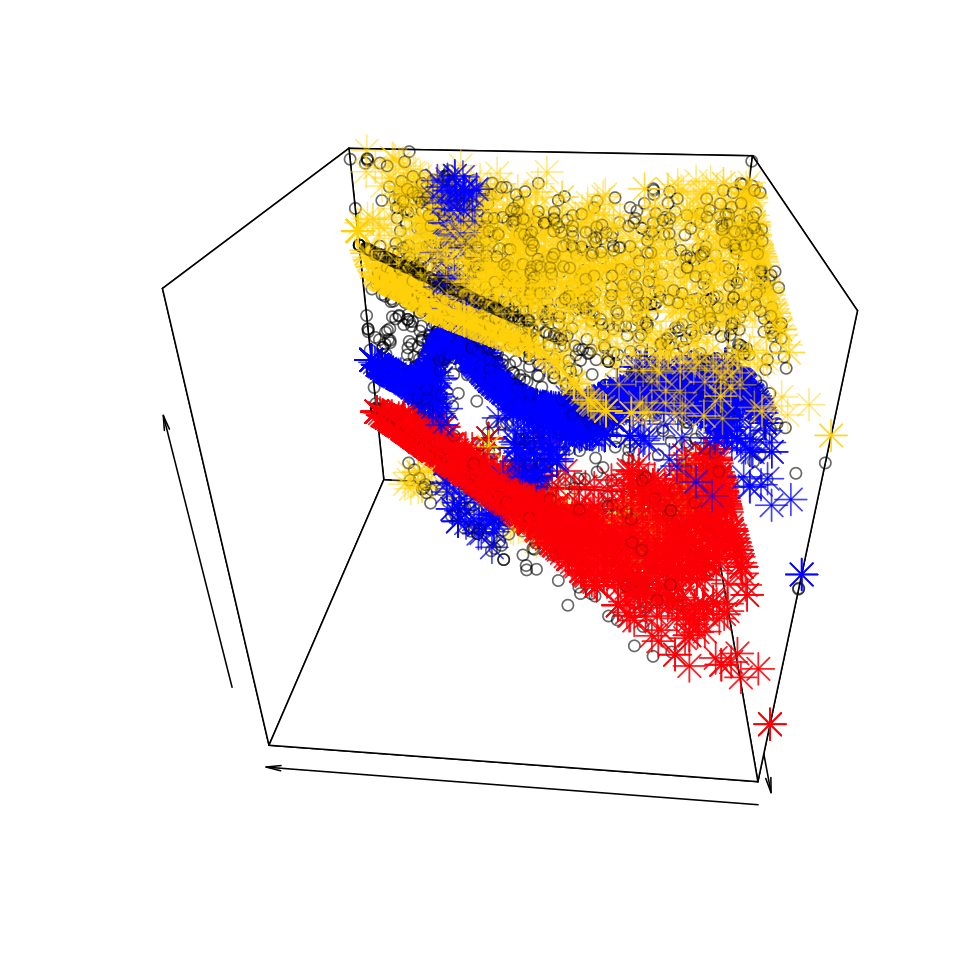}
\caption{2008-2009\\ (\textit{SvM-c 3})}
\end{subfigure}
\begin{subfigure}{.18\textwidth}
\centering
\includegraphics[width = 1\textwidth]{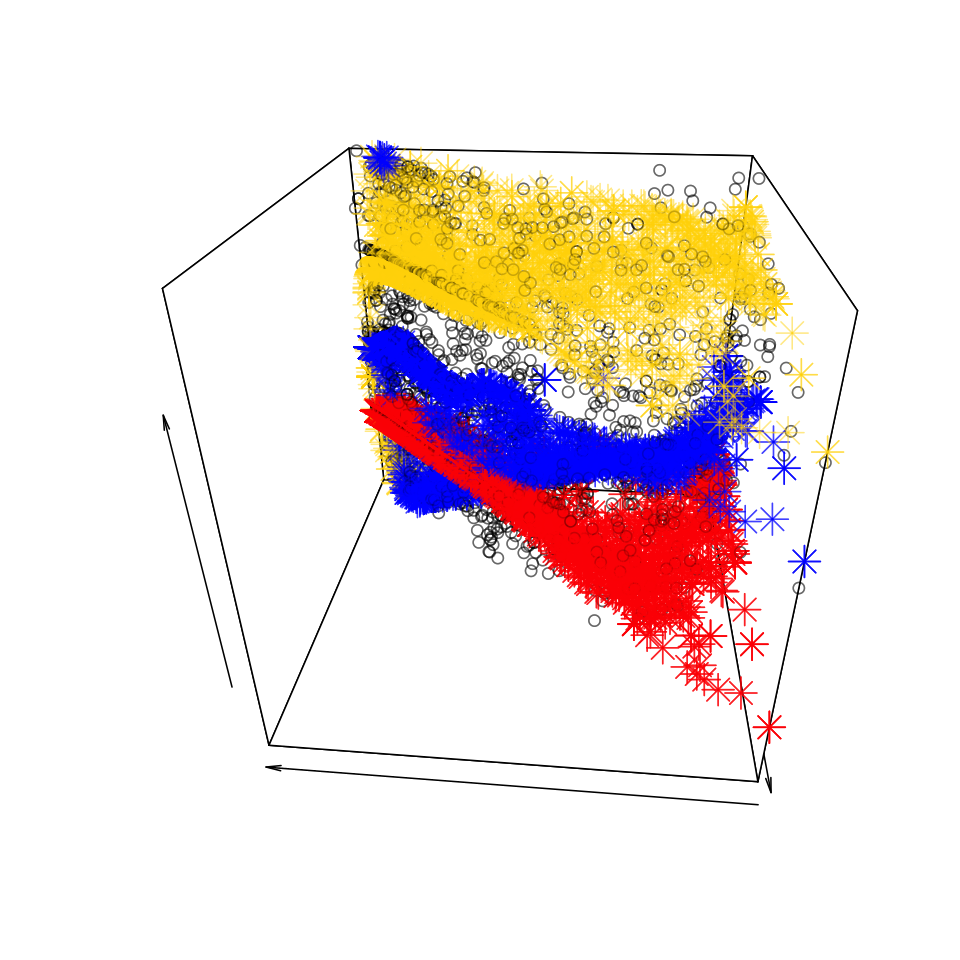}
\caption{2009-2010\\ (\textit{SvM-c-3})}
\end{subfigure}
\caption{Plots showing the observed random direction not withheld and the fitted mean surface of the model selected by the posterior predictive log probability in \eqref{eq:post_pred_prob}. The front axis represents the proportion in the first income category and the side axis represents the proportion in the second income category. The up-down axis represents the direction. The start of the arrow indicates a value of zero whereas the end indicates a value of $2\pi$ for angles and 1 for proportions.}
\label{fig:real_data_fitted_results_all}
\end{figure}

\begin{table}[!htbp]
\resizebox{0.75\textwidth}{!}{
\begin{tabular}{@{\extracolsep{5pt}} cccc} 
\\[-1.8ex]\hline 
\hline \\[-1.8ex] 
\multirow{2}{*}{\shortstack[c]{Year\\(Model)}} & \multirow{2}{*}{\shortstack[c]{ $\bm{m}$}} & \multirow{2}{*}{\shortstack[c]{ $\bm{\rho}$}} & \multirow{2}{*}{\shortstack[c]{ $\bm{\lambda}$}} \\ 
& & &\\
\hline \\[-1.8ex] 
& & &\\
\multirow{2}{*}{\shortstack[c]{1990-1991\\ \textit{SvM-c-2}}} & $\overline{m}_1$ = 5.97 (5.39, 0.25) & $\overline{\rho}_1$ = 1.10 (0.90, 1.31) & $\lambda_1$ = 0.46 (0.37, 0.55)\\
& $\overline{m}_2$ = 5.51 (5.30, 5.70) & $\overline{\rho}_2$ = 4.92 (3.82, 6.73) & $\lambda_2$ = 0.54 (0.45, 0.63)\\
& & &\\
\multirow{3}{*}{\shortstack[c]{1991-1992\\ \textit{SvM-c-3}}} & $\overline{m}_1$ = 1.34 (0.95, 1.62) & $\overline{\rho}_1$ = 85.20 (70.34, 97.40) & $\lambda_1$ = 0.05 (0.04, 0.07)\\
& $\overline{m}_2$ = 2.77 (1.82, 3.72) & $\overline{\rho}_2$ = 0.97 (0.50, 1.72) & $\lambda_2$ = 0.16 (0.11, 0.22)\\
& $\overline{m}_3$ = 5.28 (5.11, 5.45) & $\overline{\rho}_3$ = 2.78 (2.30, 3.35) & $\lambda_3$ = 0.79 (0.73, 0.84)\\
& & &\\
\multirow{3}{*}{\shortstack[c]{1992-1993\\ \textit{SvM-c-3}}} & $\overline{m}_1$ = 1.40 (1.03, 1.75) & $\overline{\rho}_1$ = 14.61  (5.07, 20.76) & $\lambda_1$ = 0.04 (0.02, 0.07)\\
& $\overline{m}_2$ = 3.00 (2.14, 3.96) & $\overline{\rho}_2$ = 1.77 (1.12, 2.93) & $\lambda_2$ = 0.15 (0.10, 0.21)\\
& $\overline{m}_3$ = 4.96 (4.81, 5.12) & $\overline{\rho}_3$ = 2.87 (2.43, 3.42) & $\lambda_3$ = 0.80 (0.74, 0.85)\\
& & &\\
\multirow{3}{*}{\shortstack[c]{1993-1994\\ \textit{SvM-c-3}}} & $\overline{m}_1$ = 1.56 (1.13, 1.83) & $\overline{\rho}_1$ = 82.14 (61.00, 96.64) & $\lambda_1$ = 0.05 (0.03, 0.06)\\
& $\overline{m}_2$ = 5.38 (5.16, 5.64) & $\overline{\rho}_2$ = 1.65 (1.33, 2.09) & $\lambda_2$ = 0.86 (0.77, 0.93)\\
& $\overline{m}_3$ = 5.15 (4.48, 5.85) & $\overline{\rho}_3$ = 0.05 (6.29e-6, 92.93) & $\lambda_3$ = 0.10 (0.02, 0.18)\\
& & &\\
\multirow{2}{*}{\shortstack[c]{1994-1995\\ \textit{SvM-c-2}}} & $\overline{m}_1$ = 1.87 (1.35, 2.43) & $\overline{\rho}_1$ = 2.10 (1.30, 3.60) & $\lambda_1$ = 0.24 (0.17, 0.32)\\
& $\overline{m}_2$ = 4.71 (4.45, 5.00) & $\overline{\rho}_2$ = 1.62 (1.35, 1.95) & $\lambda_2$ = 0.76 (0.83, 0.68)\\
& & &\\
\multirow{2}{*}{\shortstack[c]{1995-1996\\ \textit{SvM-c-2}}} & $\overline{m}_1$ = 2.29 (1.48, 3.32) & $\overline{\rho}_1$ = 1.65 (1.22, 2.25) & $\lambda_1$ = 0.38 (0.27, 0.47)\\
& $\overline{m}_2$ = 4.64 (4.16, 5.06) & $\overline{\rho}_2$ = 1.54 (1.25, 1.90) & $\lambda_2$ = 0.62 (0.53, 0.73)\\
& & &\\
\multirow{2}{*}{\shortstack[c]{1996-1997\\ \textit{SvM-c-2}}} & $\overline{m}_1$ = 2.25 (1.69, 2.77) & $\overline{\rho}_1$ = 1.70 (1.24, 2.42) & $\lambda_1$ = 0.41 (0.32, 0.50)\\
& $\overline{m}_2$ = 4.70 (4.33, 5.05) & $\overline{\rho}_2$ = 1.71 (1.35, 2.20) & $\lambda_2$ = 0.59 (0.50, 0.68)\\
& & &\\
\multirow{2}{*}{\shortstack[c]{1997-1998\\ \textit{SvM-c-2}}} & $\overline{m}_1$ = 2.00 (1.73, 2.29) & $\overline{\rho}_1$ = 2.05 (1.49, 2.76) & $\lambda_1$ = 0.65 (0.54, 0.78)\\
& $\overline{m}_2$ = 4.92 (4.27, 5.58) & $\overline{\rho}_2$ = 1.43 (0.96, 2.82) & $\lambda_2$ = 0.35 (0.22, 0.46)\\
& & &\\
\multirow{2}{*}{\shortstack[c]{1998-1999\\ \textit{SvM-c-2}}} & $\overline{m}_1$ = 1.69 (1.34, 2.02) & $\overline{\rho}_1$ = 1.63 (1.31, 2.02) & $\lambda_1$ = 0.63 (0.55, 0.71)\\
& $\overline{m}_2$ = 5.30 (4.84, 5.71) & $\overline{\rho}_2$ = 1.64 (1.24, 2.24) & $\lambda_2$ = 0.37 (0.29, 0.45)\\
& & &\\
\multirow{2}{*}{\shortstack[c]{1999-2000\\ \textit{SvM-c-2}}} & $\overline{m}_1$ = 1.72 (1.33, 2.09) & $\overline{\rho}_1$ = 1.62 (1.21, 2.08) & $\lambda_1$ = 0.75 (0.65, 0.83)\\
& $\overline{m}_2$ = 5.27 (4.63, 0.01) & $\overline{\rho}_2$ = 1.72 (1.08, 3.68) & $\lambda_2$ = 0.25 (0.17, 0.35)\\
& & &\\
\multirow{2}{*}{\shortstack[c]{2000-2001\\ \textit{SvM-c-2}}} & $\overline{m}_1$ = 1.82 (0.93, 3.44) & $\overline{\rho}_1$ = 30.49 (1.33, 96.14) & $\lambda_1$ = 0.20 (0.08, 0.56)\\
& $\overline{m}_2$ = 4.81 (3.81, 5.50) & $\overline{\rho}_2$ = 1.11 (0.78, 2.45) & $\lambda_2$ = 0.80 (0.44, 0.92)\\
& & &\\
\multirow{2}{*}{\shortstack[c]{2001-2002\\ \textit{SvM-c-2}}} & $\overline{m}_1$ = 1.70 (1.34, 2.03) & $\overline{\rho}_1$ = 1.97 (1.48, 2.56) & $\lambda_1$ = 0.64 (0.52, 0.73)\\
& $\overline{m}_2$ = 5.59 (4.93, 0.01) & $\overline{\rho}_2$ = 1.31 (0.93, 2.04) & $\lambda_2$ = 0.62 (0.53, 0.72)\\
& & &\\
\multirow{2}{*}{\shortstack[c]{2002-2003\\ \textit{SvM-c-2}}} & $\overline{m}_1$ = 1.82 (1.51, 2.10) & $\overline{\rho}_1$ = 1.88 (1.53, 2.32) & $\lambda_1$ = 0.65 (0.57, 0.72)\\
& $\overline{m}_2$ = 5.32 (4.70, 5.90)" & $\overline{\rho}_2$ = 1.33 (0.99, 1.83) & $\lambda_2$ = 0.35 (0.28, 0.43)\\
& & &\\
\multirow{2}{*}{\shortstack[c]{2003-2004\\ \textit{SvM-c-2}}} & $\overline{m}_1$ = 2.03 (1.86, 2.21) & $\overline{\rho}_1$ = 4.68 (3.33, 6.42) & $\lambda_1$ = 0.67 (0.53, 0.84)\\
& $\overline{m}_2$ = 1.87 (0.59, 3.13) & $\overline{\rho}_2$ = 1.02 (0.70, 1.31) & $\lambda_2$ = 0.33 (0.16, 0.47)\\
& & &\\
\multirow{3}{*}{\shortstack[c]{2004-2005\\ \textit{SvM-c-3}}} & $m_1$ = 1.29 (0.68, 1.77) & $\rho_1$ = 3.07 (9.50e-4, 19.22) & $\overline{\lambda}_1$ = 0.05 (0.01, 0.10)\\
& $m_2$ = 2.12 (1.98, 2.25) & $\rho_2$ = 4.89 (4.21, 5.75) & $\overline{\lambda}_2$ = 0.85 (0.80, 0.89)\\
& $m_3$ = 5.19 (4.63, 5.71) & $\rho_3$ = 0.03 (2.31e-5, 1.03) & $\overline{\lambda}_3$ = 0.10 (0.05, 0.14)\\
& & &\\
\multirow{2}{*}{\shortstack[c]{2005-2006\\ \textit{SvM-c-2}}} & $\overline{m}_1$ = 2.37 (2.21, 2.52) & $\overline{\rho}_1$ = 4.52 (3.77, 5.54) & $\lambda_1$ = 0.77 (0.70, 0.82)\\
& $\overline{m}_2$ = 1.29 (0.84, 1.84) & $\overline{\rho}_2$ = 1.05 (0.78, 1.37) & $\lambda_2$ = 0.23 (0.18, 0.30)\\
& & &\\
\multirow{3}{*}{\shortstack[c]{2006-2007\\ \textit{SvM-c-3}}} & $m_1$ = 1.05 (0.55, 1.56) & $\rho_1$ = 0.01 (1.10e-5, 0.61) & $\overline{\lambda}_1$ = 0.17 (0.03, 0.26)\\
& $m_2$ = 4.42 (3.78, 5.54) & $\rho_2$ = 3.37 (1.04, 14.43) & $\overline{\lambda}_2$ = 0.21 (0.12, 0.36)\\
& $m_3$ = 5.39 (5.08, 5.66) & $\rho_3$ = 2.45 (1.82, 3.12) & $\overline{\lambda}_3$ = 0.62 (0.55, 0.70)\\
& & &\\
\multirow{2}{*}{\shortstack[c]{2007-2008\\ \textit{SvM-c-2}}} & $\overline{m}_1$ = 5.17 (4.33, 6.19) & $\overline{\rho}_1$ =  1.17 (0.80, 1.58) & $\lambda_1$ = 0.17 (0.12, 0.22)\\
& $\overline{m}_2$ = 5.22 (5.12, 5.32) & $\overline{\rho}_2$ = 8.81 (7.50, 10.38) & $\lambda_2$ = 0.83 (0.78, 0.88)\\
& & &\\
\multirow{3}{*}{\shortstack[c]{2008-2009\\ \textit{SvM-c-3}}} & $\overline{m}_1$ = 1.36 (0.96, 1.71) & $\overline{\rho}_1$ = 16.99 (9.22, 21.09) & $\lambda_1$ = 0.04 (0.03, 0.06)\\
& $\overline{m}_2$ = 3.32 (2.58, 3.98) & $\overline{\rho}_2$ = 1.90 (1.26, 3.40) & $\lambda_2$ = 0.15 (0.09, 0.21)\\
& $\overline{m}_3$ = 5.19 (5.00, 5.37) & $\overline{\rho}_2$ = 3.68 (3.09, 4.39) & $\lambda_2$ = 0.80 (0.75, 0.86)\\
& & &\\
\multirow{2}{*}{\shortstack[c]{2008-2009\\ \textit{SvM-c-3}}} & $\overline{m}_1$ = 1.33 (1.17, 1.49) & $\overline{\rho}_1$ = 86.74 (75.96, 97.37) & $\lambda_1$ = 0.09 (0.07, 0.11)\\
& $\overline{m}_2$ = 2.50 (2.06, 2.95) & $\overline{\rho}_2$ = 1.52 (1.21, 1.94) & $\lambda_2$ = 0.53 (0.45, 0.62)\\
& $\overline{m}_3$ = 5.10 (4.82, 5.43) & $\overline{\rho}_3$ = 2.61 (1.88, 3.89) & $\lambda_3$ = 0.37 (0.29, 0.45)\\
& & &\\
\hline \\[-1.8ex] 
\end{tabular} 
}
\caption{Circular mean posterior values and 95\% credible intervals for the parameters in the Von Mises distribution and mixing probability are shown for the models fitted to the income proportions in Los Angeles County from 1990 to 2010. Parameters with bars over them were averaged across all locations.}
\label{table:real_data_results_param_summary}
\end{table}

\subsection{Additional Simulation Plots}
\begin{table}
\resizebox{0.95\textwidth}{!}{
\begin{tabular}{@{\extracolsep{5pt}} ccccccc} 
\\[-1.8ex]\hline 
\hline \\[-1.8ex] 
 & \textit{iV} & \textit{iVM} & \textit{SvM} mean $\pi$ & \textit{SvM-c} & \textit{SvM-p} &  \textit{SvM} mean $0$\\ 
\hline \\[-1.8ex] 
\textit{SvM} cf.~\eqref{model:SvM} & 4528 & 3039 & 5209 & 2663 & 3376
& 5488\\
& & & & &\\
\textit{SvM-c} cf.~\eqref{model:SvM-c} & 23738
& 16442 & 29486 & 22624 & 16546 & 30571\\
& & & & &\\
\textit{SvM-p} cf.~\eqref{model:SvM-p} & 813 & 1188 & 816 & 556 & 866 & 535\\
& & &\\
\hline \\[-1.8ex] 
\end{tabular} 
}
\caption{Time it takes for our models to run four chains for different models in seconds according to the sampling scheme outlined in Section \ref{ssection:sim_results} on four cores of a cluster using 2x 3.0 GHz Intel Xeon Gold 6154 as its processor.}
\end{table}

\begin{figure}[!h]
\captionsetup[subfigure]{justification=centering}
\centering
\begin{subfigure}{.23\textwidth}
\centering
\includegraphics[width = 1\textwidth]{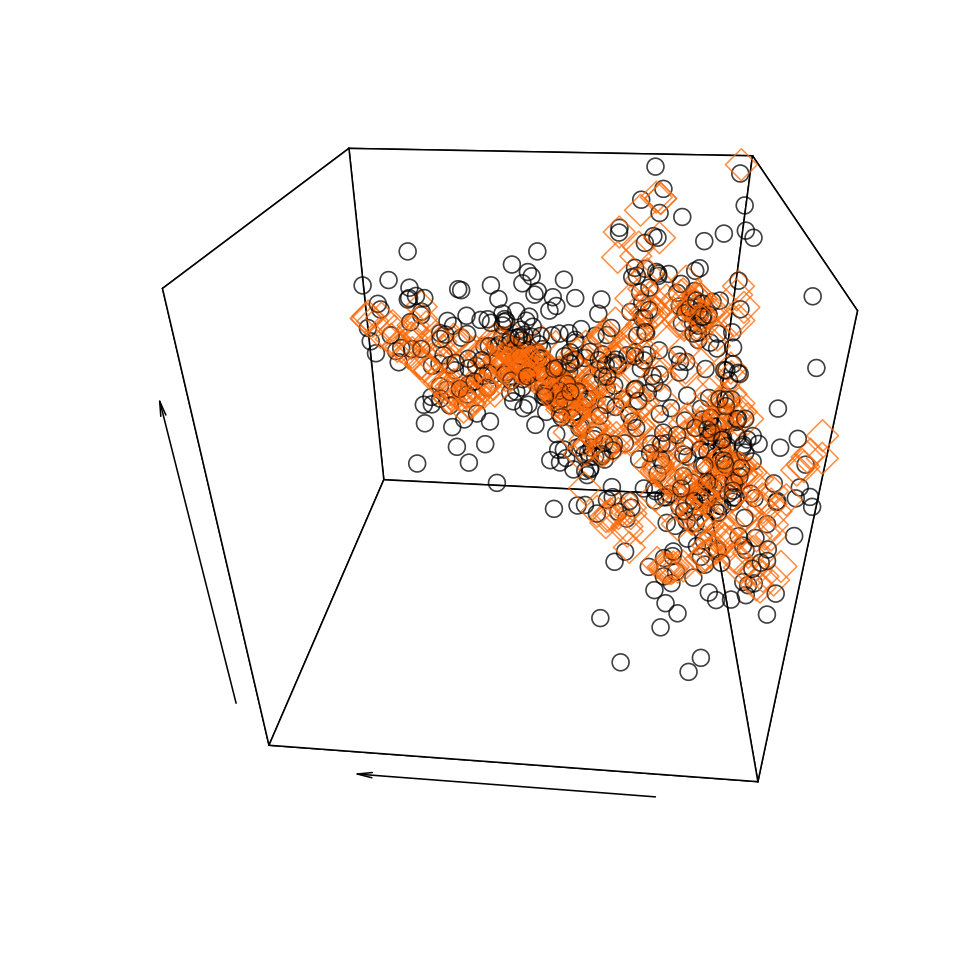}
\caption{Simulated data with\\means in orange}
\end{subfigure}
\begin{subfigure}{.23\textwidth}
\centering
\includegraphics[width = 1\textwidth]{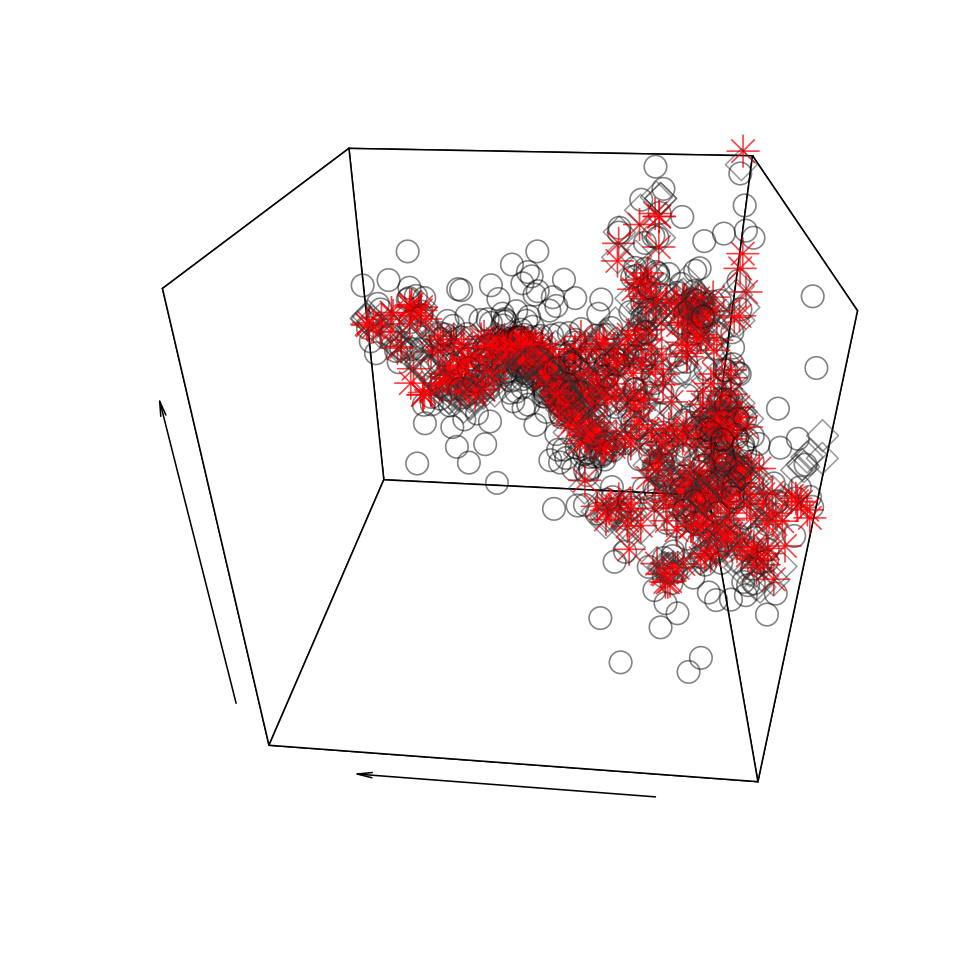}
\caption{SvM model (cf.~\eqref{model:SvM})\\fitted means in red}
\end{subfigure}
\begin{subfigure}{.23\textwidth}
\centering
\includegraphics[width = 1\textwidth]{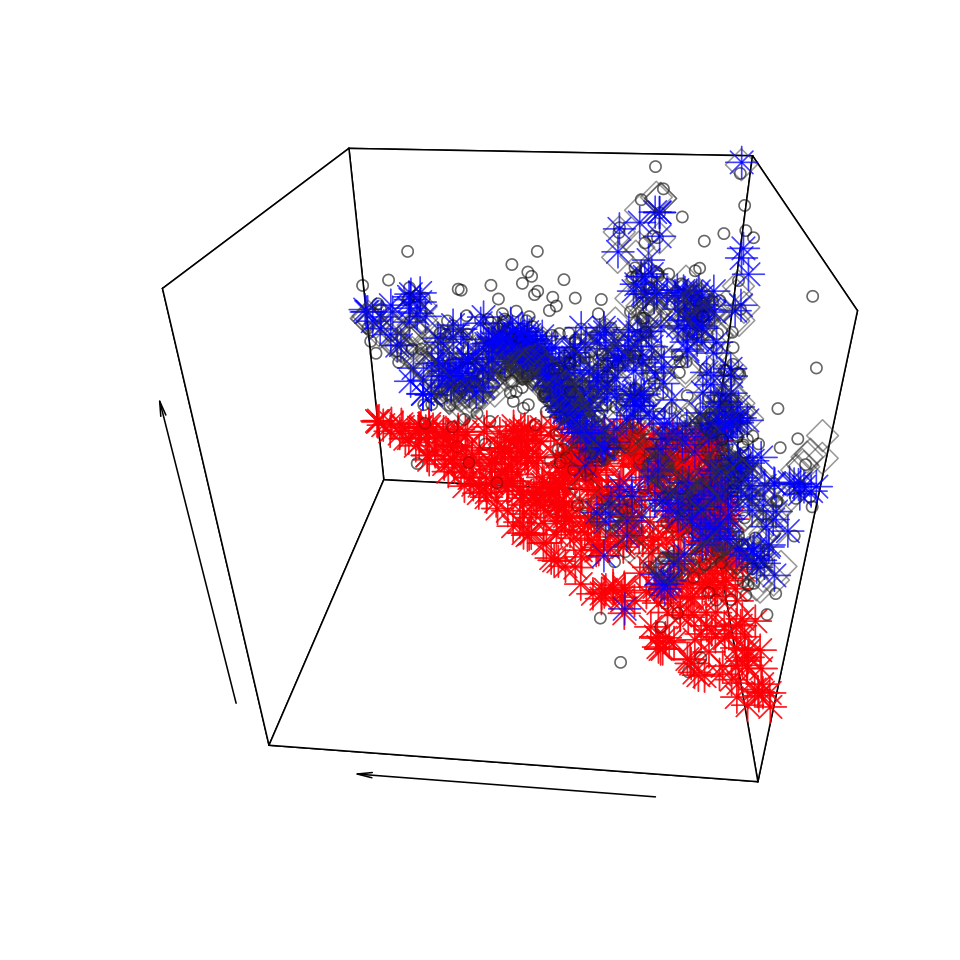}
\caption{SvM-c model (cf.~\eqref{model:SvM-c})\\fitted means in red and blue}
\end{subfigure}
\begin{subfigure}{.23\textwidth}
\centering
\includegraphics[width = 1\textwidth]{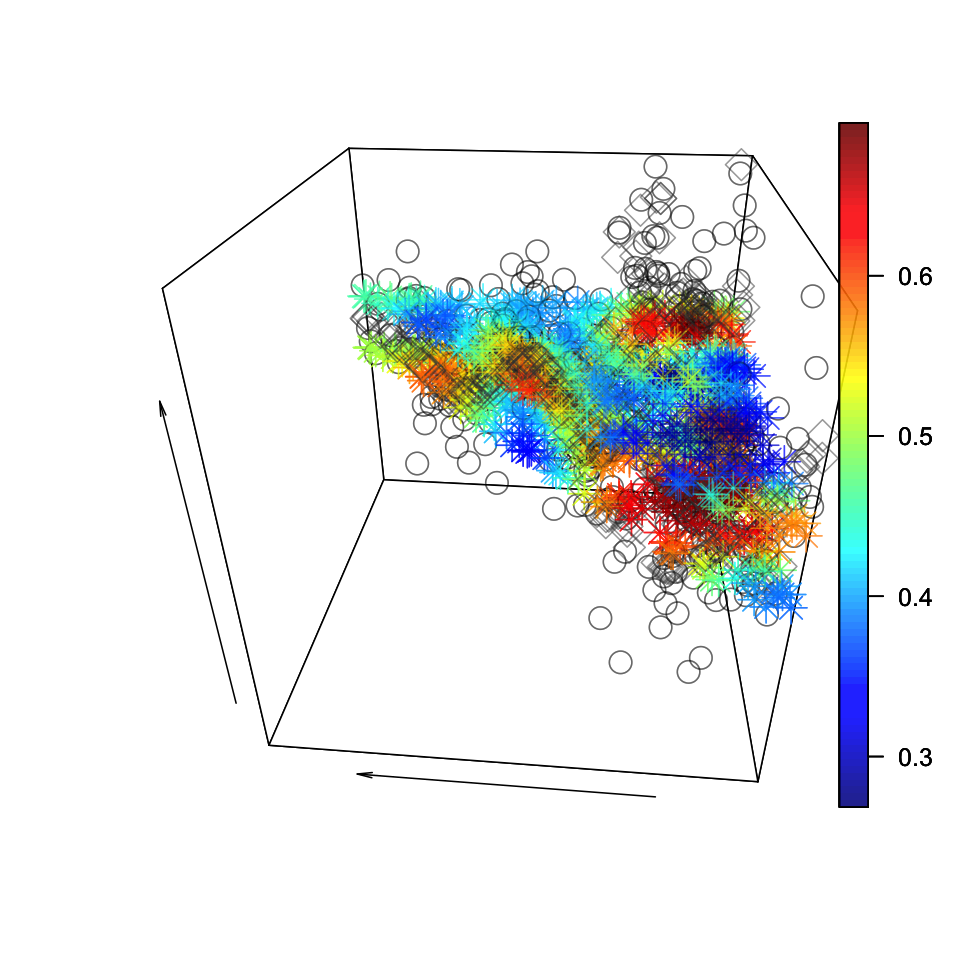}
\caption{SvM-p model (cf.~\eqref{model:SvM-p})\\fitted mean colored by $\lambda_1$}
\end{subfigure}\\
\caption{Plots of the mean component fitted by different models for observations simulated according to the $\textit{SvM}$ model in \eqref{model:SvM}. The fitted mean components are shown in colored asterisks with the simulated observations as circles and the simulated means as grey rhombuses for different scenarios. The x-y axis are the observations' location in the first two dimensions of a three dimensional simplex; the z-axis represent the observations.}
\label{fig:model_fitted_SvM_plots}

\resizebox{0.95\textwidth}{!}{
\begin{tabular}{@{\extracolsep{5pt}} cccc} 
\\[-1.8ex]\hline 
\hline \\[-1.8ex] 
 & $\bm{m}$ & $\bm{\rho}$ & $\bm{\lambda}$ \\ 
\hline \\[-1.8ex] 
\textbf{Simulation} & $\overline{m}$ = 3.23 (1.78, 4.86)  & $\overline{\rho} = 3.00 (2.54, 3.50)$ & ---\\
& & &\\
\textit{SvM} cf.~\eqref{model:SvM} & $\overline{m}$ = 3.28 (2.89, 3.68) & $\overline{\rho}$ = 3.14 (2.70, 3.63) & ---\\
& & &\\
\multirow{2}{*}{\shortstack[c]{\textit{SvM-c} cf.~\eqref{model:SvM-c}}} & $\overline{m}_1$ = 1.58 (0.39, 2.78) & $\overline{\rho}_1$ = 0.37 (0.00, 89.98) & $\lambda_1$ = 0.02 (9.07e-4, 0.04)\\
& $\overline{m}_2$ = 3.32 (2.90, 3.76) & $\overline{\rho}_2$ = 3.45 (2.93, 4.09) & $\lambda_2$ = 0.98 (0.96, 0.9991)\\
& & &\\
\multirow{2}{*}{\shortstack[c]{\textit{SvM-p} cf.~\eqref{model:SvM-p}}} & $m_1$ = 2.94 (2.78, 3.10) & $\rho_1$ = 2.90 (2.15, 3.78) & $\overline{\lambda_1}$ = 0.51 (0.31, 0.72)\\
& $m_2$ = 3.81 (3.60, 3.99) & $\rho_2$ = 1.65 (1.15, 2.30) & $\overline{\lambda_2}$ = 0.49 (0.28, 0.69)\\
& & &\\
\hline \\[-1.8ex] 
\end{tabular} 
}
\captionof{table}{Circular mean posterior values and 95\% credible intervals for the parameters in the Von Mises distribution and mixing probability are shown for the models fitted to the simulation data in Figure \ref{fig:model_fitted_SvM_plots}. Parameters with bars over them were averaged across all locations.}
\end{figure}

\begin{figure}[!h]
\captionsetup[subfigure]{justification=centering}
\centering
\begin{subfigure}{.23\textwidth}
\centering
\includegraphics[width = 1\textwidth]{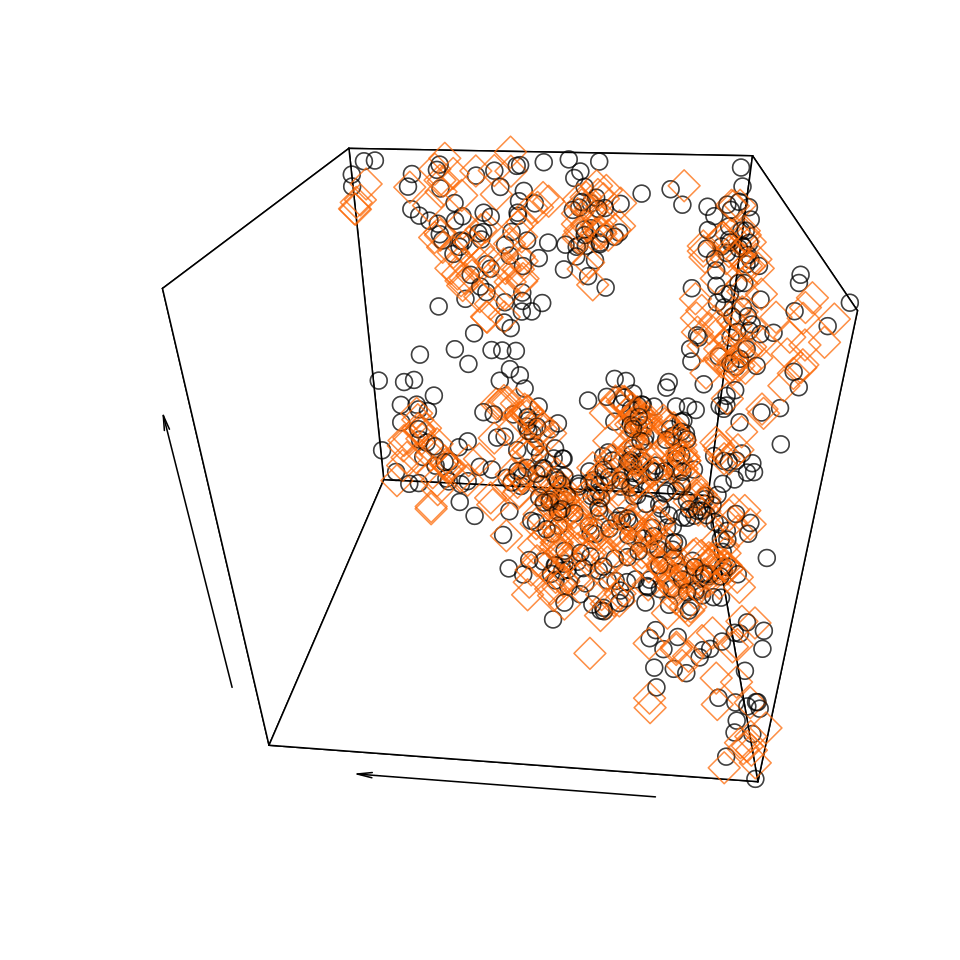}
\caption{Simulated data with\\means in orange}
\end{subfigure}
\begin{subfigure}{.23\textwidth}
\centering
\includegraphics[width = 1\textwidth]{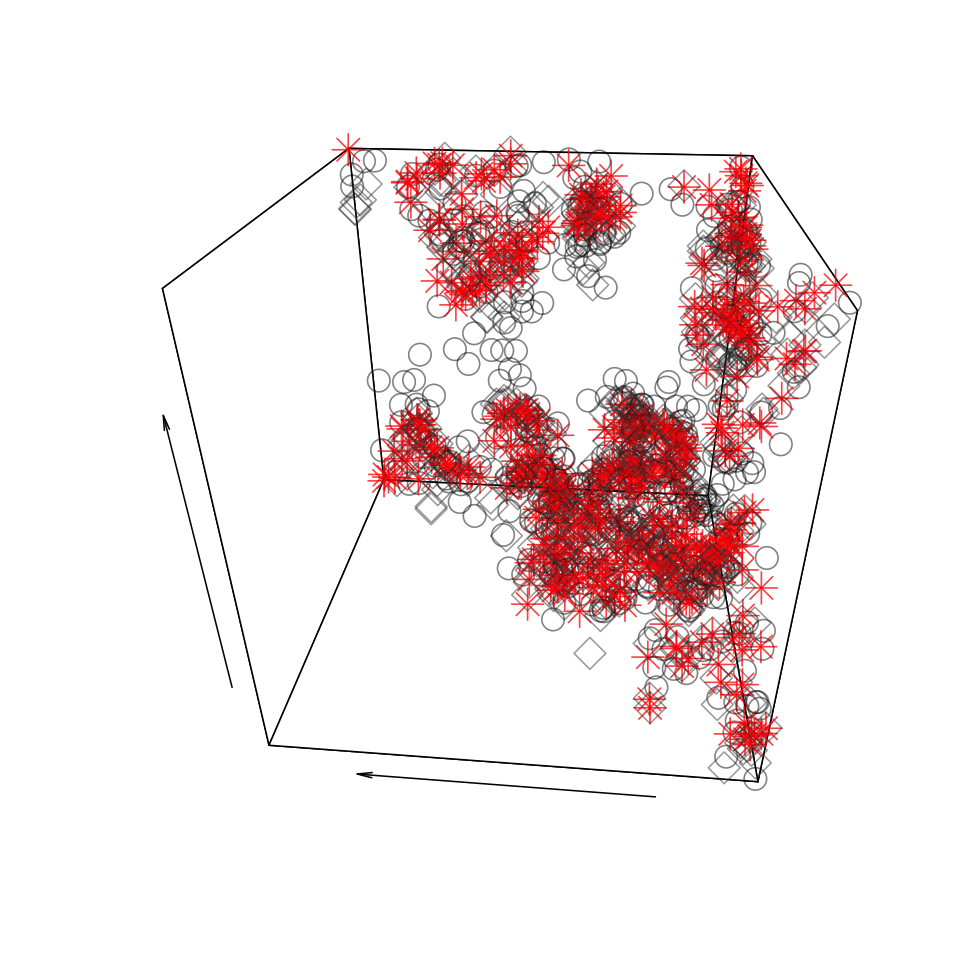}
\caption{SvM model (cf.~\eqref{model:SvM})\\fitted means in red}
\end{subfigure}
\begin{subfigure}{.23\textwidth}
\centering
\includegraphics[width = 1\textwidth]{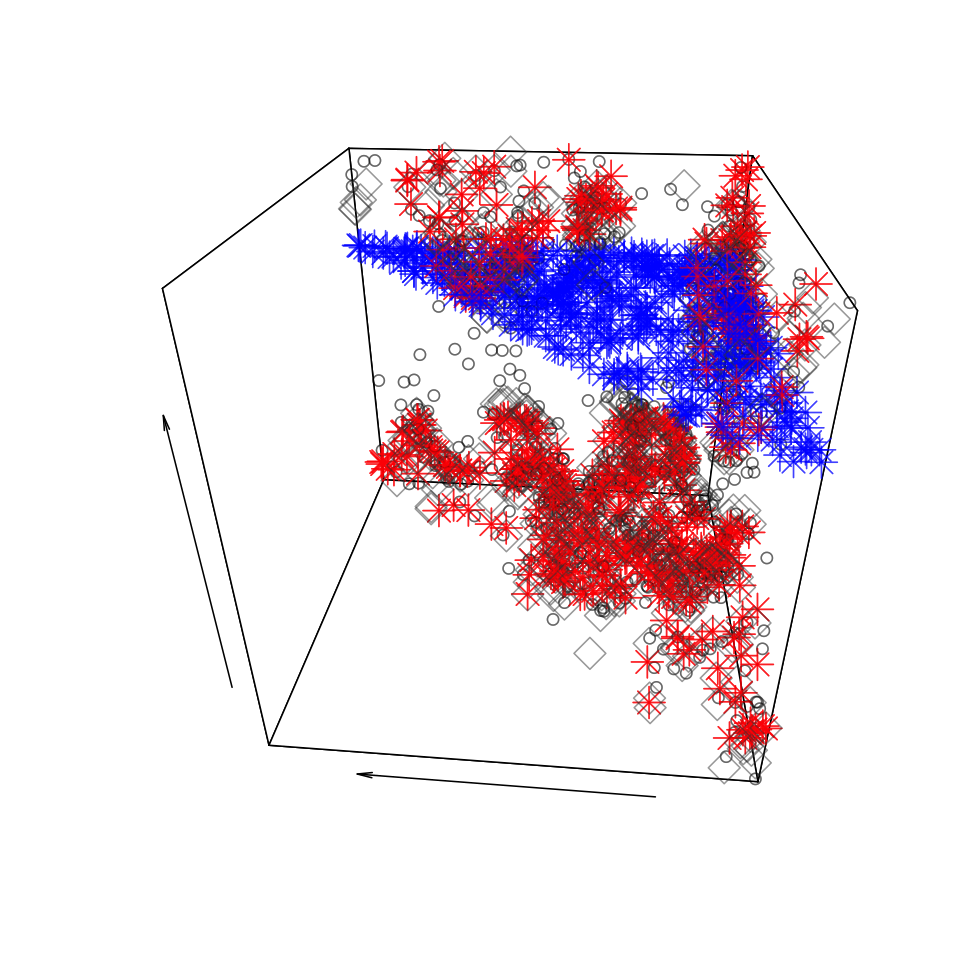}
\caption{SvM-c model (cf.~\eqref{model:SvM-c})\\fitted means in red and blue}
\end{subfigure}
\begin{subfigure}{.23\textwidth}
\centering
\includegraphics[width = 1\textwidth]{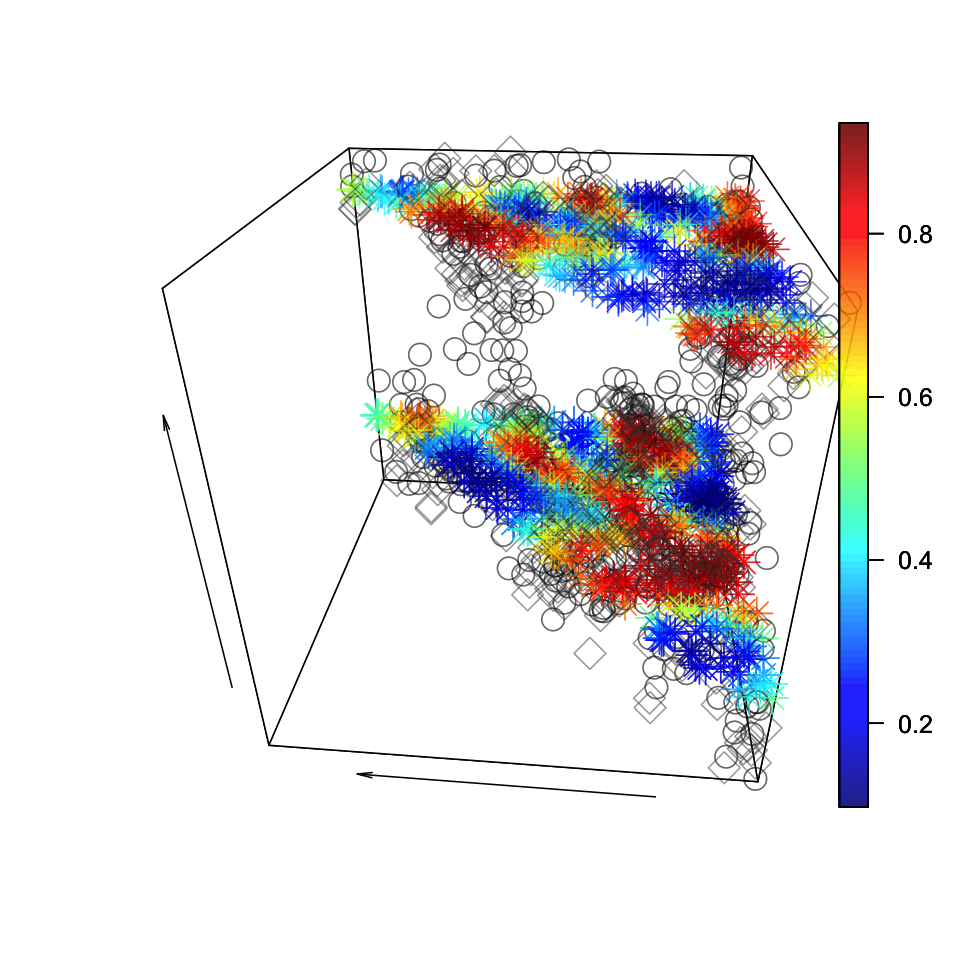}
\caption{SvM-p model (cf.~\eqref{model:SvM-p})\\fitted mean colored by $\lambda_1$}
\end{subfigure}\\
\caption{Plots of the mean component fitted by different models for observations simulated according to the $\textit{SvM}$ model in \eqref{model:SvM}. The fitted mean components are shown in colored asterisks with the simulated observations as circles and the simulated means as grey rhombuses for different scenarios. The x-y axis are the observations' location in the first two dimensions of a three dimensional simplex; the z-axis represent the observations.}
\label{fig:model_fitted_alt_SvM_plots}

\resizebox{0.95\textwidth}{!}{
\begin{tabular}{@{\extracolsep{5pt}} cccc} 
\\[-1.8ex]\hline 
\hline \\[-1.8ex] 
 & $\bm{m}$ & $\bm{\rho}$ & $\bm{\lambda}$ \\ 
\hline \\[-1.8ex] 
\textbf{Simulation} & $\overline{m}$ = 0.68 (4.17, 2.38)  & $\overline{\rho} = 2.99 (2.51, 3.45)$ & ---\\
& & &\\
\textit{SvM} cf.~\eqref{model:SvM} & $\overline{m}$ = 0.64 (6.27, 1.27) & $\overline{\rho}$ = 2.65 (2.24, 3.12) & ---\\
& & &\\
\multirow{2}{*}{\shortstack[c]{\textit{SvM-c} cf.~\eqref{model:SvM-c}}} & $\overline{m}_1$ = 0.67 (0.10, 1.24) & $\overline{\rho}_1$ = 2.66 (2.24, 3.19)  & $\lambda_1$ = 0.99 (0.95, 0.9996)\\
& $\overline{m}_2$ = 4.72 (3.52, 5.91) & $\overline{\rho}_2$ = 0.12 (3.40e-5, 225.56) & $\lambda_2$ = 0.01 (4.51e-4, 0.05\\
& & &\\
\multirow{2}{*}{\shortstack[c]{\textit{SvM-p} cf.~\eqref{model:SvM-p}}} & $m_1$ = 1.48 (1.34, 1.61) & $\rho_1$ = 2.18 (1.77, 2.64) & $\overline{\lambda_1}$ = 0.53 (0.31, 0.75)\\
& $m_2$ = 5.64 (5.46, 5.83) & $\rho_2$ = 1.60 (1.26, 1.98) & $\overline{\lambda_2}$ = 0.47 (0.25, 0.69)\\
& & &\\
\hline \\[-1.8ex] 
\end{tabular} 
}
\captionof{table}{Circular mean posterior values and 95\% credible intervals for the parameters in the Von Mises distribution and mixing probability are shown for the models fitted to the simulation data in Figure \ref{fig:model_fitted_SvM_plots}. Parameters with bars over them were averaged across all locations.}
\end{figure}

\begin{figure}[!h]
\begin{subfigure}{.23\textwidth}
\centering
\includegraphics[width = 1\textwidth]{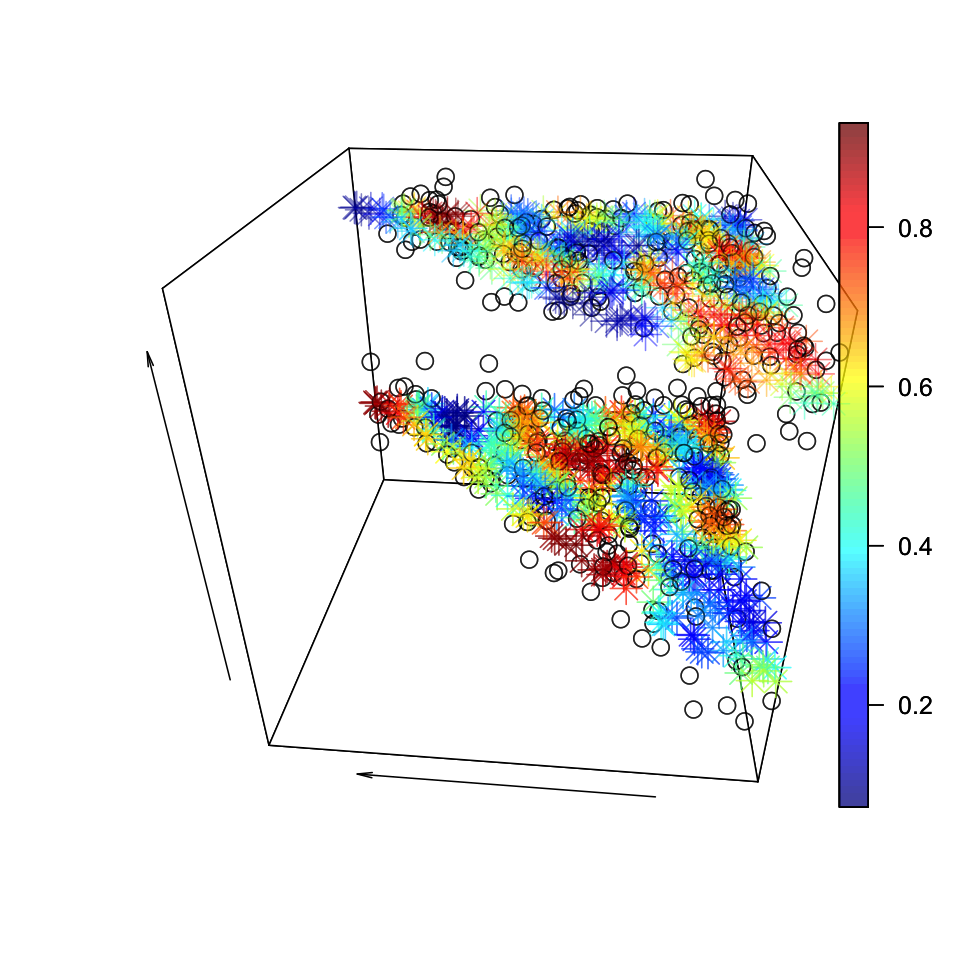}
\caption{Simulated data with\\means colored by $\lambda_1$}
\end{subfigure}
\begin{subfigure}{.23\textwidth}
\centering
\includegraphics[width = 1\textwidth]{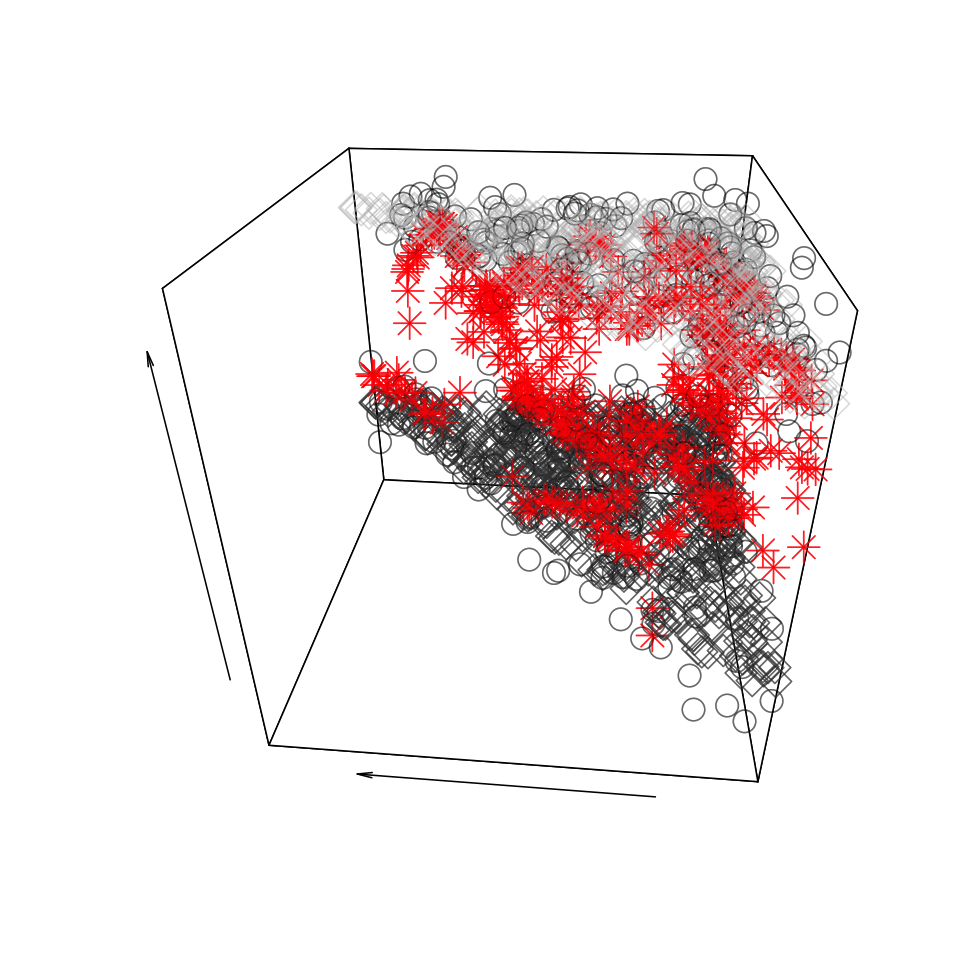}
\caption{SvM model,\\cf.~\eqref{model:SvM}}
\end{subfigure}
\begin{subfigure}{.23\textwidth}
\centering
\includegraphics[width = 1\textwidth]{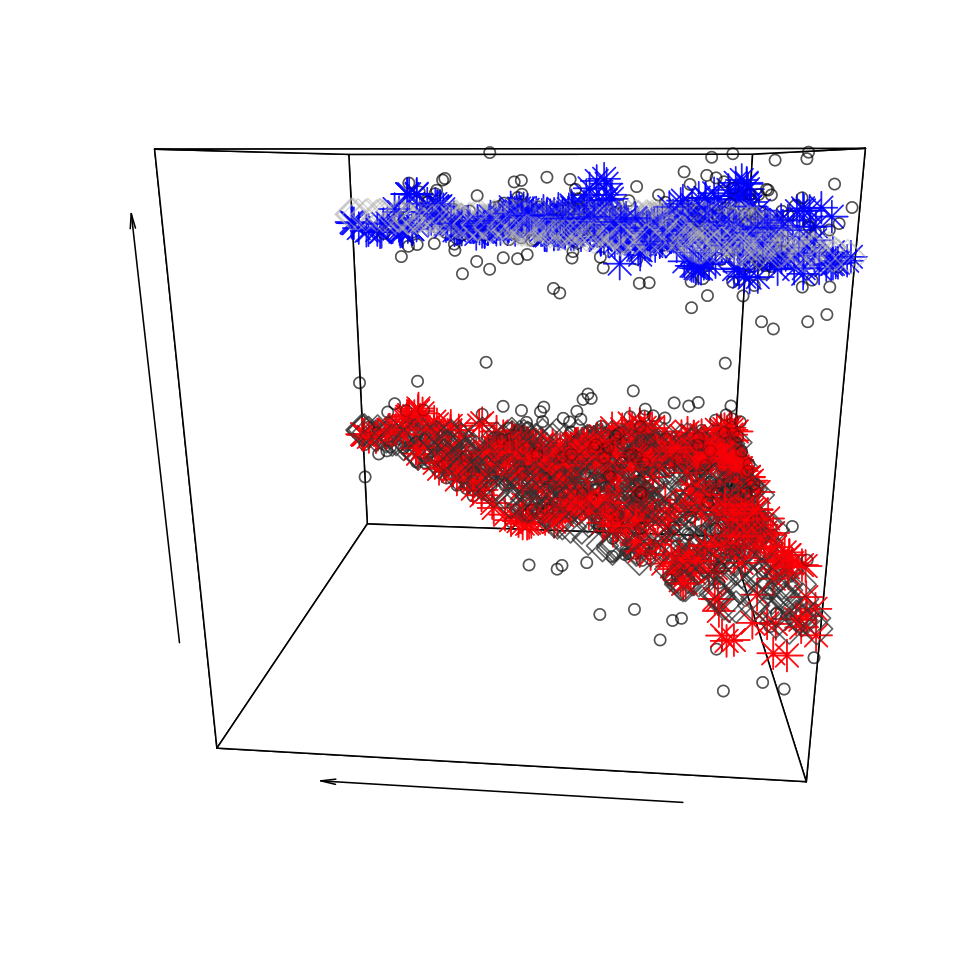}
\caption{SvM-c model,\\cf.~\eqref{model:SvM-c}}
\end{subfigure}
\begin{subfigure}{.23\textwidth}
\centering
\includegraphics[width = 1\textwidth]{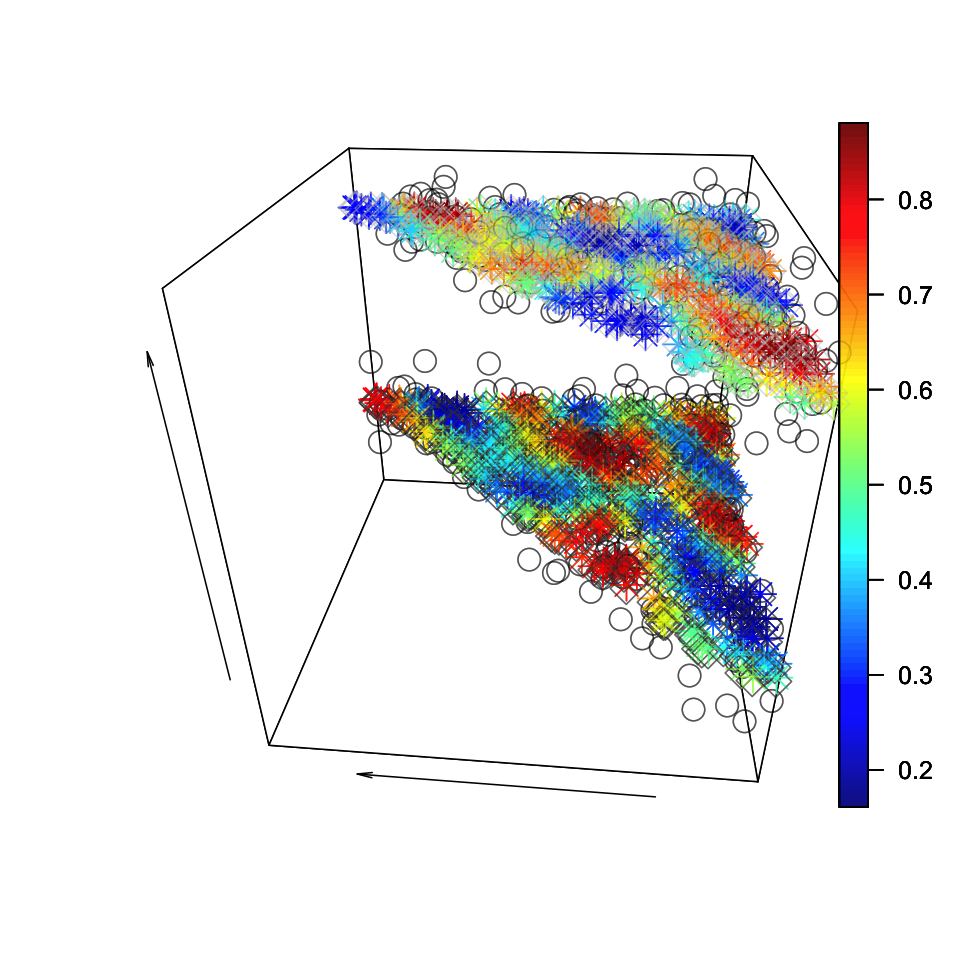}
\caption{SvM-p model,\\cf.~\eqref{model:SvM-p}}
\end{subfigure}\\
\caption{Plots of the mean component fitted by different models for observations simulated according to the $\textit{SvM-p}$ model \eqref{model:SvM} in the bottom row. The fitted mean components are shown in red asterisks with the simulated observations as circles and the simulated means as grey rhombuses for different scenarios. The x-y axis are the observations' location in the first two dimensions of a three dimensional simplex; the z-axis represent the observations.}
\label{fig:model_fitted_SvM-p_plots}

\centering
\resizebox{0.95\textwidth}{!}{
\begin{tabular}{@{\extracolsep{5pt}} cccc} 
\\[-1.8ex]\hline 
\hline \\[-1.8ex] 
 & $\bm{m}$ & $\bm{\rho}$ & $\bm{\lambda}$ \\ 
\hline \\[-1.8ex] 
\multirow{2}{*}{\textbf{Simulation}} & $m_1$ = 1.57 & $\rho_1 = 5$ & $\overline{\lambda}_1 = 0.51 (0.16, 0.91)$\\
& $m_2$ = 4.71 & $\rho_2 = 10$ & $\overline{\lambda}_2 = 0.49 (0.09, 0.84)$\\
& & &\\
\textit{SvM} cf.~\eqref{model:SvM} & $\overline{m}$ = 3.22 (1.87, 4.56) & $\overline{\rho}$ = 0.91 (0.71, 1.12) & ---\\
& & &\\
\multirow{2}{*}{\shortstack[c]{\textit{SvM-c} cf.~\eqref{model:SvM-c}}} & $\overline{m}_1$ = 1.63 (1.23, 2.02) & $\overline{\rho}_1$ = 5.98 (4.86, 7.30) & $\lambda_1$ = 0.53 (0.48, 0.57)\\
& $\overline{m}_2$ =  4.72 (4.34, 5.08) & $\overline{\rho}_2$ = 11.39 (9.02, 14.31) & $\lambda_2$ = 0.47 (0.43, 0.52)\\
& & &\\
\multirow{2}{*}{\shortstack[c]{\textit{SvM-p} cf.~\eqref{model:SvM-p}}} & $m_1$ = 1.62 (1.57, 1.68) & $\rho_1$ = 5.36 (4.55, 6.25) & $\overline{\lambda_1}$ = 0.52 (0.30, 0.73)\\
& $m_2$ = 4.71 (4.67, 4.75) & $\rho_2$ = 9.85 (8.19, 11.67) & $\overline{\lambda_2}$ = 0.48 (0.27, 0.70)\\
& & &\\
\hline \\[-1.8ex] 
\end{tabular} 
}
\captionof{table}{Circular mean posterior values and 95\% credible intervals for the parameters in the Von Mises distribution and mixing probability are shown for the models fitted to the simulation data in Figure \ref{fig:model_fitted_SvM-p_plots}. Parameters with bars over them were averaged across all locations.}
\end{figure}

\end{document}